\documentclass[a4paper,11pt]{article}
\pdfoutput=1 

\usepackage{jheppub} 

\usepackage[T1]{fontenc} 
\usepackage{subcaption}
\usepackage{verbatim}
\usepackage{mathrsfs}
\usepackage{physics}
\usepackage{tikz}
\usetikzlibrary{calc}
\usetikzlibrary{shapes.geometric,positioning}
\usepackage{amsthm}
\usepackage{mathtools}
\usepackage{enumerate}

\definecolor{ForestGreen}{RGB}{34,139,34}

\newcommand{\del}{\partial}

\renewcommand{\tilde}{\widetilde}
\renewcommand{\hat}{\widehat}
\renewcommand{\bar}{\overline}

\newcommand{\reals}{\mathbb{R}}
\newcommand{\comps}{\mathbb{C}}
\renewcommand{\Im}{\text{Im}}

\newcommand{\A}{\mathcal{A}}
\newcommand{\M}{\mathcal{M}}
\renewcommand{\H}{\mathcal{H}}
\newcommand{\K}{\mathcal{K}}

\newcommand{\im}{\operatorname{im}}

\newcommand{\B}{\mathcal{B}}

\newcommand{\Z}{\mathcal{Z}}

\newtheorem{theorem}{Theorem}

\newtheorem{prop}[theorem]{Proposition}

\theoremstyle{definition}
\newtheorem{definition}[theorem]{Definition}
\newtheorem{remark}[theorem]{Remark}

\numberwithin{theorem}{section}

\title{Excitability in quantum field theory}
\author[a]{Jacqueline Caminiti,}
\author[b]{Federico Capeccia,}
\author[c]{and Jonathan Sorce}
\affiliation[a]{Perimeter Institute for Theoretical Physics}
\affiliation[b]{University of California, Davis}
\affiliation[c]{Princeton Gravity Initiative}
\abstract{
In quantum field theory, it is not always possible to excite one state out of another using only local operators.
This paper establishes abstract algebraic criteria for (local) excitability in general quantum theories, and computes these criteria explicitly for zero-mean Gaussian states in (generalized) free field theories.
We find that in this context, due to the special nature of Gaussian states, one-way excitability always implies two-way excitability, and our results generalize the ``quasiequivalence theorems'' of Powers, St{\o}rmer, van Daele, Araki, and Yamagami.
A key role in our proof is played by the information-theoretic tool of canonical purification.
In appendices, we provide a pedagogical introduction to the algebraic formulation of (generalized) free field theory.
}

\begin{document}
\maketitle

\section{Introduction}

When can one state of a quantum field be excited out of another?
In the algebraic description of quantum field theory, this question can be made precise.
One starts with an abstract algebra $\A_0$ describing the field content of the theory, then constructs different Hilbert space ``sectors'' by selecting algebraic states and performing the GNS construction.
Each GNS sector ${\cal H}_{\omega}$ can be thought of as the set of local excitations around a single chosen state $\omega$, so the question of whether $\omega_2$ can be excited from $\omega_1$ is a question of whether $\omega_2$ ``lives inside'' the space $\H_{\omega_1}.$

Concretely, every GNS sector carries a representation of the field algebra $\A_0,$ and we say that $\omega_2$ can be excited out of $\omega_1$ if there is a density matrix $\rho_{\omega_2}$ on $\H_{\omega_1}$ that reproduces the $\omega_2$ correlation functions:
\begin{equation}
	\tr_{\H_{\omega_1}}(\rho_{\omega_2} a)
		= \omega_2(a), \qquad a \in \A_0.
\end{equation}
Beyond asking when excitability is possible, one can also study when the correlation functions of $\omega_2$ can be reproduced by a pure state,
\begin{equation}
	\rho_{\omega_2}
		= |\omega_2\rangle \langle \omega_2|,
\end{equation}
and one can further ask, ``if $\omega_2$ can be excited out of $\omega_1,$ does that necessarily mean that $\omega_1$ can be excited out of $\omega_2$?''

These are physically important questions at the foundation of quantum field theory, and the answers are generally nontrivial.
One expects, for example, that infrared effects will prevent the excitation of a thermal state out of the global vacuum in Minkowski spacetime --- see figure \ref{fig:thermalstates} --- while no such obstruction should occur within a finite causal diamond.
More generally, one expects that any two states that are ``well behaved in the ultraviolet'' should be excitable into one other within compact subregions.
While these statements are physically intuitive, it turns out to be quite difficult to establish them with any degree of rigor, since the abstract question of excitability is hard to translate into standard tools of quantum field theory like scattering amplitudes and correlation functions.

\begin{figure}
    \centering
    \hspace{-5cm}
    \begin{tikzpicture}[scale=2, transform shape]
    \node[diamond, draw=black, line width=1pt, fill=red, minimum size=1.6cm] (D) {};
        \node[anchor=south west] at ($(D.45)+(-.25,0)$) {\tiny$\mathscr{I}^+$};
        \node[anchor=north west] at ($(D.-45)+(-.25,.15)$) {\tiny$\mathscr{I}^-$};
        \node[anchor=east] at ($(D.-180)$) {\textcolor{red}{\tiny$\beta$}};
    \hspace{2cm}
    \raisebox{-0.2cm}{\huge$\not\prec$}
    \hspace{2cm}
    \node[diamond, draw=black, line width=1pt, fill=blue, minimum size=1.6cm] (D2) {};
        \node[anchor=south west] at ($(D2.45)+(-.25,0)$) {\tiny  $\mathscr{I}^+$};
        \node[anchor=north west] at ($(D2.-45)+(-.25,.15)$) {\tiny  $\mathscr{I}^-$};
        \node[anchor=west] at ($(D2.0)$) {\textcolor{blue}{\tiny$\Omega$}};
\end{tikzpicture}
\\
\vspace{.5cm}
\hspace{-5cm}
\begin{tikzpicture}[scale=2, transform shape]
    \node[diamond, draw=black, line width=1pt, fill=blue, minimum size=1.6cm] (D) {};
        \node[anchor=south west] at ($(D.45)+(-.25,0)$) {\tiny$\mathscr{I}^+$};
        \node[anchor=north west] at ($(D.-45)+(-.25,.15)$) {\tiny$\mathscr{I}^-$};
        \node[anchor=east] at ($(D.-180)$) {$\;\;$};
    \node[diamond, draw=black, line width=1pt, fill=red, minimum size=.8cm] {};
    \hspace{2cm}
    \raisebox{-0.2cm}{\huge$\prec$}
    \hspace{2cm}
    \node[diamond, draw=black, line width=1pt, fill=blue, minimum size=1.6cm] (D2) {};
        \node[anchor=south west] at ($(D.45)+(-.25,0)$) {\tiny$\mathscr{I}^+$};
        \node[anchor=north west] at ($(D.-45)+(-.25,.15)$) {\tiny$\mathscr{I}^-$};
        \node[anchor=west] at ($(D2.0)$) {\textcolor{blue}{$\;\;$}};
\end{tikzpicture}
    \caption{While a thermal state cannot be excited (indicated by $\not \prec$) from the Minkowski vacuum globally, 
    one expects that this should be possible in a local subregion, due to the absence of infrared effects.}
    \label{fig:thermalstates}
\end{figure}
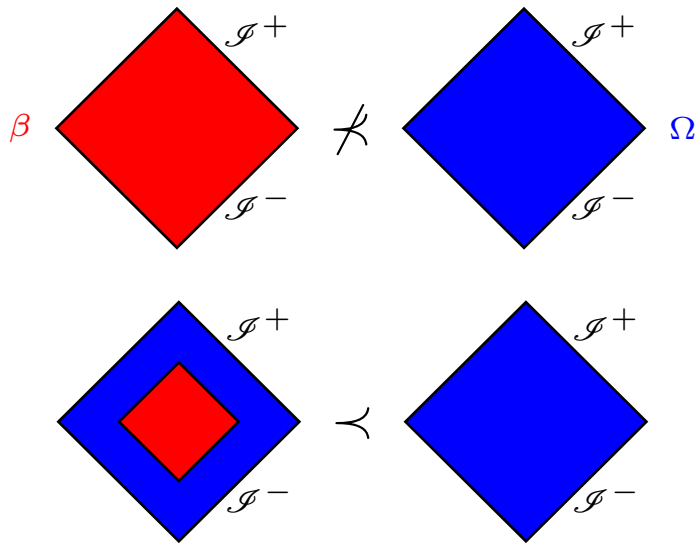

In the context of (generalized) free field theory, the work of \cite{Powers:quasi, Araki:quasi, VanDaele:quasi, Araki-Yamagami} proved a theorem determining when two zero-mean Gaussian states $\omega_1$ and $\omega_2$ can be mutually excited out of one another.
In honor of the paper \cite{Araki-Yamagami} that established the result in its most general form, this is typically called the ``Araki-Yamagami theorem.''
In Klein-Gordon theory, Verch \cite{Verch:hadamard} used this theorem to establish that any two ``Hadamard'' states --- i.e., Klein-Gordon states that are well behaved in the ultraviolet --- can be mutually excited out of one another within compact regions.

Interestingly, if one reads the papers \cite{Powers:quasi, Araki:quasi, VanDaele:quasi, Araki-Yamagami, Verch:hadamard}, one will not see a transparent connection between the questions posed there and the ``excitability'' questions we have introduced above.
Instead, the Araki-Yamagami theorem concerns the problem of \textit{quasiequivalence of sectors}.
In free field theory, once one chooses a GNS representation $\H_{\omega},$ the abstract algebra $\A_0$ of free fields can be completed to a von Neumann algebra $\A_{\omega}$.
Within this algebra, the original abstract fields appear as the unitary ``Weyl operators'' $e^{i \phi[f]},$ where $f$ is a real smearing function.
Two abstract states $\omega_1$ and $\omega_2$ are said to be quasiequivalent if the von Neumann algebras $\A_{\omega_1}$ and $\A_{\omega_2}$ are related by an algebraic isomorphism that does not modify any of the Weyl operators.
As we will explain in section \ref{sec:excitability}, quasiequivalence implies mutual excitability, but the argument for this implication is not direct, and it is not apparent from the analysis of \cite{Powers:quasi, Araki:quasi, VanDaele:quasi, Araki-Yamagami} whether mutual excitability implies quasiequivalence.

The point of this paper is to revisit the quasiequivalence work in  \cite{Powers:quasi, Araki:quasi, VanDaele:quasi, Araki-Yamagami} using a different perspective and an updated set of tools.
Our motivation for doing so comes from problems like the geometric modular flow conjecture \cite{Jensen:JSS, Sorce:analyticity, Caminiti:2025hjq}, and the holographic construction of large-$N$ sectors with inequivalent bulk duals \cite{Leutheusser:HSMT, Witten:largeN, Leutheusser:subalgebra}.
These are interesting contemporary problems that would profit from a better understanding of the conditions under which one state can be excited into another.
The main new contribution of our paper is a series of intermediate results that establish criteria for excitability even when full quasiequivalence is not satisfied, and a demonstration that excitability always implies quasiequivalence for Gaussian states in free theories.
Additionally, because our approach is motivated by a different set of questions and makes use of information-theoretic tools, we also believe that our work provides a substantial conceptual simplification of the proofs developed in \cite{Powers:quasi, Araki:quasi, VanDaele:quasi, Araki-Yamagami}.

To understand most of this paper, it is necessary that the reader have a background in algebraic free field theory.
To make the paper approachable to readers without this background, we have provided a complete pedagogical introduction to this theory in appendix \ref{app:review}.
The reader may approach section \ref{sec:excitability} without consulting this appendix, but the rest of the paper relies on its ideas.

The plan of the paper is given below.
\begin{itemize}
\item
In section \ref{sec:history}, we provide a discussion of prior work and some justification as to why we think our approach is simpler than existing ones.
The main technical results of the paper are then summarized in section \ref{sec:summary}.
\item
In section \ref{sec:excitability}, we explain the web of connections between quasiequivalence and mutual excitability, using a purely algebraic description of quantum field theory.
We also show that excitability for a generic pair of states is equivalent to excitability of their canonical purifications.
\item
In section \ref{sec:general-necessary}, we establish a few necessary criteria for excitability of Gaussian states, which will be used throughout the rest of the paper.
\item
In section \ref{sec:pure}, we review the conditions under which one pure Gaussian state can be excited out of another, using a modification of arguments from \cite{Shale:unitary, Wald:particle-creation, Wald:s-matrix}.
We also present new proofs of theorems relating the standard ``Bogoliubov'' criteria to other criteria involving the two-point functions.
\item
In section \ref{sec:passing-to-canonical}, we establish concrete formulas relating the two-point functions of Gaussian states to the two-point functions of their canonical purifications.
Combined with the results of the prior section, this gives a general theorem controlling excitability out of states for which the canonical purification is ``pure'' in a certain technical sense.
\item
In section \ref{sec:centrally-pure}, we fill a gap in the generality of our results by treating excitability for ``centrally pure'' states, which are not pure in the technical sense, and which appear when one canonically purifies a ``non-factorial'' state.
\item 
In section \ref{sec:other-way}, we use our prior results to show that for zero-mean Gaussian states in free field theory, one-way excitability always implies mutual excitability; in other words, while it is possible in general to create excitability relations like $\omega_2 \prec \omega_1$ with $\omega_1 \not\prec \omega_2,$ this cannot happen when both states are Gaussian.
\item
Finally, in section \ref{sec:discussion}, we comment on our results, our approach, and future directions.
\end{itemize}
\noindent
In addition to the pedagogical appendix \ref{app:review}, several other appendices fill in technical material that has been deferred to keep the main text readable.
Appendix \ref{app:one-particle-modular} also provides a pedagogical introduction to one-particle modular theory.

\subsection{Discussion of prior work}
\label{sec:history}

The main papers that inspired the present work are \cite{Shale:unitary, Wald:particle-creation, Wald:s-matrix, Powers:quasi, Woronowicz:purification, VanDaele:quasi, Araki:quasi, Araki-Yamagami, Longo:quasi, Conti:quasi}.
In \cite{Shale:unitary, Wald:particle-creation, Wald:s-matrix}, criteria were established for unitary equivalence of GNS representations coming from pure Gaussian states.
In \cite{Powers:quasi}, the problem of quasiequivalence for ``factorial states'' of free fermions was solved using a ``doubling trick'' that was later connected, by \cite{Woronowicz:purification}, to what is now called canonical purification.
In \cite{VanDaele:quasi, Araki:quasi}, a similar technique was applied to free bosons.
The most general theorem encompassing prior work was established by Araki and Yamagami in \cite{Araki-Yamagami}, which relaxed the factorial assumption.\footnote{It was proved in \cite{Verch:algebra} that non-factorial states do not appear in physical settings in Klein-Gordon theory, but (i) no such result is known for generalized free field theories, and (ii) non-factorial features appear in a state-independent way in theories with higher-form symmetries. It therefore seems worthwhile to deal with the non-factorial case here.}
Recently, the interesting papers \cite{Longo:quasi, Conti:quasi} have revisited the Araki-Yamagami theorem in the context of modular theory; we use some aspects of their approach in our arguments below.

This paper provides a complete proof of the general version of the Araki-Yamagami theorem, using techniques that we believe constitute a substantial simplification over prior approaches.\footnote{It is very difficult to quantify ``simplicity,'' so a statement like this one is necessarily ideological. After all, the present paper is over a hundred pages long including appendices! On the other hand, the present paper is mostly self-contained, and is written in a pedagogical ``physics style'' as opposed to a more concise ``mathematical style.'' The Araki-Yamagami paper \cite{Araki-Yamagami} is 56 math-pages long and depends on both conceptual and technical results from the prior papers \cite{Shale:unitary, Segal:distributions, Araki:quasi, Araki-Shiraishi}.}
One main advantage of our approach is that it allows us to discuss conditions for excitability as a stepping stone towards the full question of quasiequivalence; this is useful because it allows us to treat the possibility that $\omega_2$ can be excited out of $\omega_1$, but not vice versa.

Below we summarize the tools that distinguish our approach from previous ones.

\subsubsection{Simplifications as compared to pre-Araki-Yamagami work}

Prior to the work by Araki and Yamagami, quasiequivalence was established in the factorial setting in the papers \cite{Shale:unitary, Powers:quasi, Araki:quasi, VanDaele:quasi, Woronowicz:purification}.
Our approach in this setting is very closely related to that introduced by Woronowicz in \cite{Woronowicz:purification}, but with a few improvements.

\begin{enumerate}
	\item We use a characterization of unitary equivalence that is closer to that of Wald \cite{Wald:particle-creation, Wald:s-matrix} than to that of Shale \cite{Shale:unitary}.
	Shale's theorem tells you how to determine whether a particular symplectic transformation is unitarily implementable, so to apply this theorem to the question of whether two states are unitarily equivalent, you must first construct a symplectic transformation between the states.
	By contrast, Wald's approach is direct; the symplectic transformation is constructed ``along the way'' as part of a different proof strategy.
	\item We use the construction from \cite{Sorce:paper1} of canonical purifications on $*$-algebras to work directly with field operators $\phi[f]$ when convenient, instead of having to translate everything into the bounded ``Weyl operators'' $e^{i \phi[f]}.$
    This simplifies several key arguments.
	\item Even though we eventually show that one-way excitability $\omega_2 \prec \omega_1$ always implies quasiequivalence for Gaussian states, we find that it is conceptually cleaner to study the conditions for one-way excitability directly, rather than seeking the conditions for quasiequivalence from the start.
\end{enumerate}

\subsubsection{Simplifications as compared to Araki-Yamagami}

The real power of our approach comes when simplifying and generalizing the theorem of Araki and Yamagami \cite{Araki-Yamagami}.
Our progress results from the following improvements.

\begin{enumerate}
    \item Following Woronowicz \cite{Woronowicz:purification}, we obtain the ``doubling trick'' used in prior work as a special case of a general canonical purification formula \cite{Sorce:paper1}.
    We extend the validity of Woronowicz's approach by applying canonical purifications to non-factorial states as in \cite{Sorce:paper1}.
    \item Since the canonical purifications of general states are ``centrally pure,'' we tackle the question of excitability for general states by passing to the canonical purification, then determining the conditions under which centrally pure states can be excited out of one another.
    There is a sense in which any centrally pure state breaks up into a ``pure piece'' and an ``abelian piece,'' and we decompose the excitability problem into these two components.
    \item We produce a version of Wald's characterization of unitary equivalence that applies to the ``abelian'' excitability problem just mentioned.
    This allows us to avoid invoking Segal's lemma on absolute continuity of Gaussian random processes \cite{Segal:distributions}, so that no part of the present work is dependent on results from constructive quantum field theory.
\end{enumerate}

\subsection{Summary of results}
\label{sec:summary}

The details of the paper are rather technical, so we summarize here all essential results that may be useful in future applications.
We write this section, and only this section, in definition-theorem-proof format, where after each theorem we point to the place in the paper where the proof can be found.
Results are presented in the order in which they are proved, with dependencies between theorems pointed out, so that the logical structure of the arguments is clear.

\subsubsection{Abstract results about excitability}

\begin{definition}
\label{def:exc}
	Let $\A_0$ be an abstract $*$-algebra, and let $\omega : \A_0 \to \comps$ be a state.
	Within the GNS representation $\H_{\omega}$, one can define a von Neumann algebra $\A_{\omega}$ via the procedure from \cite[appendix A]{Sorce:paper1}.
	
	Given two states $\omega_1, \omega_2 : \A_0 \to \comps,$ we say 
    \begin{equation*}
        \textit{$\omega_2$ can be excited out of $\omega_1$}, 
    \end{equation*}
    written $\omega_2 \prec \omega_1,$ if there is a density matrix $\rho_{\omega_2}$ on $\H_{\omega_1}$ that reproduces all $\omega_1$ correlation functions.
	Concretely, for any $a \in \A_0$ with polar decomposition $u_{a}^{(1)} |a|^{(1)}$ on $\H_{\omega_1}$ and $u_{a}^{(2)} |a|^{(2)}$ on $\H_{\omega_2},$ we require
	\begin{equation}
		\tr_{\H_{\omega_1}}(\rho_{\omega_2} u_{a}^{(1)}) = \langle \omega_2 | u_a^{(2)} |\omega_2\rangle,
	\end{equation}
	as well as
	\begin{equation}
		\tr_{\H_{\omega_1}}(\rho_{\omega_2} f(|a|^{(1)})) = \langle \omega_2 | f(|a|^{(2)}) |\omega_2\rangle, \quad f \text{ bounded},
	\end{equation}
	and similarly for all finite products of operators like $u_a$ and $f(|a|),$ as well as $u_a^{\dagger}.$
\end{definition}

\begin{prop} \label{prop:alpha-definition}
	The above definition, $\omega_2 \prec \omega_1,$ is equivalent to the statement that the map $\alpha : \A_{\omega_1} \to \A_{\omega_2}$ defined by
	\begin{align}
		\alpha(u_a^{(1)})
			& = u_a^{(2)}, \\
		\alpha(f(|a|^{(1)}))
			& = f(|a|^{(2)})\quad f \text{ bounded}
	\end{align}
	extends, via linearity, multiplication, and adjoints, to a normal homomorphism from $\A_{\omega_1}$ to $\A_{\omega_2}.$
	(A homomorphism is said to be normal if it is continuous with respect to the ultraweak topology.)
\end{prop}
\begin{proof}
	See section \ref{sec:one-way-excitability}.
\end{proof}

\begin{remark} \label{rem:alpha-bounded}
	If $\A_0$ is a C$^*$ algebra and $\omega : \A_0 \to \comps$ is norm-continuous, then for every $a \in \A_0,$ the representative operators on $\H_{\omega_1}$ and $\H_{\omega_2}$ are bounded.
	In this case one does not need any of the fancy polar decomposition language used above, and one says $\omega_2 \prec \omega_1$ if there is a density matrix $\rho_{\omega_2}$ satisfying
	\begin{equation}
		\tr_{\H_{\omega_1}}(\rho_{\omega_2} a) = \langle \omega_2 | a | \omega_2\rangle \qquad \text{for any $a \in \A_0$.}
	\end{equation}
	In this setting, the excitability relation $\omega_2 \prec \omega_1$ is equivalent to ultraweak continuity of the ``identity map'' that takes the representative of $a$ on $\H_{\omega_1}$ to the representative of $a$ on $\H_{\omega_2}.$
\end{remark}

\begin{definition}
	A state $\omega : \A_0 \to \comps$ is said to be factorial if $\A_\omega$ is a factor.
\end{definition}

\begin{prop} \label{prop:two-way-excitability}
	If $\omega_1$ is factorial and we have $\omega_2 \prec \omega_1,$ then we automatically have $\omega_1 \prec \omega_2.$
\end{prop}
\begin{proof}
	See section \ref{sec:one-way-excitability}.
\end{proof}

\begin{prop}
	If $\omega_1$ is not factorial, then there always exists a state $\omega_2$ with $\omega_2 \prec \omega_1$ but $\omega_1 \not\prec \omega_2.$
\end{prop}
\begin{proof}
	See section \ref{sec:one-way-excitability}.
\end{proof}

\begin{prop}
	If $\omega_2 \prec \omega_1$ and $\omega_1 \prec \omega_2,$ then the ``field-preserving map'' $\alpha$ from proposition \ref{prop:alpha-definition} is a normal isomorphism from $\A_{\omega_1}$ onto $\A_{\omega_2}.$
	In this case we write $\omega_1 \sim \omega_2$ and say the two states are quasiequivalent.
    \label{prop:mutual-exc-is-quasi}
\end{prop}
\begin{proof}
	See section \ref{sec:mutual-excitability}.
\end{proof}

\begin{definition}
	The canonical purification $\hat{\omega}$ is defined on the algebraic tensor product $\A_0 \otimes \A_0^{\text{op}}$ by
	\begin{equation} \label{eq:omega-hat-definition}
		\hat{\omega}(a \otimes b)
			\equiv \langle a^{\dagger} \omega | J_{\omega} |b^{\dagger} \omega\rangle,
	\end{equation}
	where $J_{\omega}$ is the modular conjugation of the GNS vector $|\omega\rangle \in \H_{\omega}$ with respect to the von Neumann algebra $\A_{\omega}$.
    \label{def:omega_hat}
\end{definition}

\begin{remark}
Some intuition for this definition comes from the fact that a mixed state $\omega$ induces a nontrivial commutant ${\cal A}_{\omega}'$ in the GNS representation, which can be thought of as containing the ``purifying degrees of freedom'' introduced by the GNS construction.
When $|\omega\rangle$ is separating for $\A_{\omega}$, one has ${\cal A}_{\omega}'=J_{\omega}{\cal A}_{\omega} J_{\omega}$, so the purifying degrees of freedom can be identified with a second copy of ${\cal A}_{\omega}$.
In the same way that $\A_{\omega}$ is generated by elements $a \in \A_0,$ the commutant $\A_{\omega}'$ is then generated by elements of the opposite algebra $\A_0^{\text{op}}$ via the identification
\begin{equation}
    b^{\text{op}} \leftrightarrow J_{\omega} b^{\dag} J_{\omega}\,.
\end{equation}
This motivates the definition in equation \eqref{eq:omega-hat-definition} if we rewrite it heuristically\footnote{We say ``heuristically'' because when the operator $a$ is unbounded, one must be careful about domain issues in writing down the vector $a (J_{\omega} b^{\dagger} J_{\omega}) |\omega\rangle.$} as
\begin{equation}
    \hat{\omega}(a \otimes b^{\text{op}})
			=
        \langle \omega | a (J_{\omega} b^{\dag} J_{\omega}) |\omega\rangle\,.
        \label{eq:conj_by_j}
\end{equation}
For more details, see \cite{Sorce:paper1}.
\label{rem:omega_hat}
\end{remark}

\begin{remark}
	It was shown in \cite{Sorce:paper1} that $\hat{\omega}$ genuinely defines a state, i.e., it is a positive linear functional.
	The state $\hat{\omega}$ is not actually always pure in the technical sense.
	Rather, it is pure if and only if $\omega$ is factorial.
	More generally, when $\omega$ is a non-factorial state, the state $\hat{\omega}$ is ``centrally pure,'' meaning one has the inclusion of von Neumann algebras
	\begin{equation}
		(\A_0 \otimes \A_0^{\text{op}})_{\hat{\omega}}' \subseteq (\A_0 \otimes \A_0^{\text{op}})_{\hat{\omega}}.
	\end{equation}
	Intuitively, this means that while $\hat{\omega}$ can introduce a commutant for $\A_0 \otimes \A_0^{\text{op}}$ in the GNS representation, that commutant contains only certain limits of operators that were already in $\A_0 \otimes \A_0^{\text{op}}.$
	In this setting, the GNS construction does not introduce any nontrivial purifying degrees of freedom.
\label{rem:centrally-pure-def}
\end{remark}

\begin{theorem} \label{thm:canonical-checking}
	One has $\omega_2 \prec \omega_1$ if and only if one has $\hat{\omega}_2 \prec \hat{\omega}_1.$
\end{theorem}
\begin{proof}
	See section \ref{sec:canonical-checking}.
	This result generalizes a result due to Woronowicz \cite{Woronowicz:purification} away from the setting of factorial states on C$^*$ algebras.
\end{proof}

\subsubsection{Results about Gaussian states}

\begin{definition}
	On a spacetime $\M$ with an antisymmetric bidistribution $\Omega,$ one says that the \textit{free field algebra} $\A_0$ is the $*$-algebra generated by symbols $\phi[f]$ with $f$ a compactly supported smooth function, subject to the canonical commutation relation
	\begin{equation}
		[\phi[f], \phi[g]] = - i \Omega[f, g].
	\end{equation}
\end{definition}

\begin{remark}
	In Klein-Gordon theory, one also imposes an equation of motion, but this is not at all important for the present work.
\end{remark}

\begin{definition}
	Given a free field algebra $\A_0,$ the state $\omega : \A_0 \to \comps$ is said to be \textit{Gaussian} if all the connected correlators, defined formally via
	\begin{equation}
		\omega_n^c(\phi[f_1] \dots \phi[f_n])
		= \left.\frac{\del^n}{\del t_1 \dots \del t_n} \log\omega(e^{t_1 \phi[f_1]} \dots e^{t_n \phi[f_n]})\right
		|_{t_1=\dots=t_n=0},
	\end{equation}
	vanish for $n \geq 3.$
	
	A Gaussian state is said to be \textit{zero-mean} if the one-point functions vanish, $\omega(\phi[f]) = 0.$
\end{definition}

\begin{remark}
    \textit{For the rest of this paper we will work only with zero-mean Gaussian states.}
	We will often say simply ``Gaussian state.''
	We intend to deal with more general Gaussian states in future work.
\end{remark}

\begin{definition}
	Given a Gaussian state $\omega,$ the corresponding ``phase-space inner product'' $\mu$ on the space $C^{\infty}_0$ of smearing functions is defined by
	\begin{equation}
		\langle f | g \rangle_{\mu}
			\equiv \frac{1}{2} \omega(\phi[f]^* \phi[g] + \phi[g] \phi[f]^*).
	\end{equation}
	The space obtained by quotienting out by null states and completing with respect to $\mu$ is denoted $\K_{\mu}$.
\end{definition}

\begin{remark}
	Many properties of $\K_{\mu}$, and its relation with the GNS space $\H_{\omega},$ are derived in appendix \ref{app:review}.
	In particular, $\H_{\omega}$ splits into $n$-particle sectors $\H_{\omega} = \oplus_{n} \H_{\omega}^{(n)},$ and the one-particle sector maps isometrically into $\K_{\mu}$ via the formula
	\begin{equation}
		\phi[f] |\omega\rangle \mapsto \sqrt{1 - i R} |f\rangle_{\mu},
	\end{equation}
	with $R$ defined by 
	\begin{equation}
		\langle f | R | g\rangle_{\mu} = \frac{1}{2} \Omega[f^*, g].
	\end{equation}
\end{remark}

\begin{theorem} \label{thm:both-pure}
	If $\omega_1$ and $\omega_2$ are pure Gaussian states of a free field algebra, then one has $\omega_2 \prec \omega_1$ if and only if the following two conditions are satisfied:
	\begin{enumerate}[(i)]
		\item One has $\mu_2 \prec \mu_1,$ i.e., there is a constant $c$ such that we have
		\begin{equation}
			\langle f | f \rangle_{\mu_2} \leq c \langle f | f \rangle_{\mu_1}\,, \qquad 
            f \in C_0^{\infty}\,.
		\end{equation}
		\item The operator $Q$ that implements the $\mu_2$ inner product, defined by
		\begin{equation}
			\langle f | Q | g\rangle_{\mu_1} = \langle f | g \rangle_{\mu_2}, \qquad f, g \in C^{\infty}_0,
		\end{equation}
		has the property that $Q-1$ is Hilbert-Schmidt.
		
		(An operator $O$ is said to be Hilbert-Schmidt under the condition $\tr(O^{\dagger} O) < \infty.$)
	\end{enumerate}
\end{theorem}
\begin{proof}
	This is the main theorem of section \ref{sec:pure}.
\end{proof}

\begin{remark}
	The intuition behind the above theorem is that when $\omega_2$ and $\omega_1$ are both pure, one can write down an ansatz for the representative vector $|\omega_2\rangle \in \H_{\omega_1}$ and solve the ansatz explicitly.
	Condition (i) is the requirement that one can even write down the ansatz in the first place; condition (ii) is the requirement that the solution to the ansatz gives a normalizable vector in $\H_{\omega_1}.$
\end{remark}

\begin{remark}
	If we define the operators $X_j$ as the operators on $\K_{\mu_1}$ that implement the two-point functions, i.e., by
	\begin{equation}
		\langle f | X_j | g \rangle_{\mu_1} = \omega_j(\phi[f]^* \phi[g]),
	\end{equation}
	then one has (as computed in section \ref{sec:pure-two-point}) the formulas
	\begin{equation}
		X_1 = 1 - i R_1
	\end{equation}
	and
	\begin{equation}
		X_2 = Q - i R_1.
	\end{equation}
	So condition (ii) in the above theorem may be stated equivalently as the condition that $X_2 - X_1$ is Hilbert-Schmidt.
\end{remark}

\begin{prop}\label{prop:pure-with/out-roots}
	For pure states $\omega_1$ and $\omega_2,$ the condition that $X_2 - X_1$ is Hilbert-Schmidt implies $\sqrt{X_2} - \sqrt{X_1}$ Hilbert-Schmidt.
\end{prop}
\begin{proof}
	See section \ref{sec:pure-two-point}.
	This proposition is important for logical consistency, because for non-pure states, the general condition will be that $\sqrt{X_2} - \sqrt{X_1}$ is Hilbert-Schmidt.
    For non-pure states, the Hilbert-Schmidt property for $X_2 - X_1$ will not be sufficient.
\end{proof}

\begin{prop}
	There exist counterexamples where $\omega_1$ is pure, $\omega_2$ is mixed, and one has $\mu_2 \prec \mu_1$ with $X_2 - X_1$ Hilbert-Schmidt, but $\sqrt{X_2} - \sqrt{X_1}$ is not Hilbert-Schmidt.
    
	(In these cases, thanks to theorem \ref{thm:mixed-from-factor}, one has $\omega_2 \not\prec \omega_1.$)
\end{prop}
\begin{proof}
	See appendix \ref{app:counterexample-1}.
\end{proof}

\begin{prop}
	The canonical purification of a Gaussian state is Gaussian.
\end{prop}
\begin{proof}
	See section \ref{sec:gaussian-canonical-formulas}.
\end{proof}

\begin{theorem} \label{thm:passing-to-canonical}
	The canonical purification conditions
	\begin{enumerate}[(i)]
		\item $\hat{\mu}_2 \prec \hat{\mu}_1$,
		\item $\hat{X}_2 - \hat{X}_1$ Hilbert-Schmidt, and
		\item $\hat{Q}$ has trivial kernel
	\end{enumerate}
	are equivalent to the following conditions on $\omega_1$ and $\omega_2$:
	\begin{enumerate}[(i)]
		\item $\mu_2 \prec \mu_1$,
		\item $\sqrt{X_2} - \sqrt{X_1}$ Hilbert-Schmidt, and
		\item $Q$ has trivial kernel.
	\end{enumerate}
	
	In the second set of conditions, item (ii) may be replaced by the pair of conditions that $X_2 - X_1$ be Hilbert-Schmidt and that $\sqrt{1 + (v^{\dagger} R_2 v)^2} - \sqrt{1 + R_1^2}$ be Hilbert-Schmidt, where $v$ is the partial isometry that appears in the polar decomposition of the map $L$ that canonically maps $\K_{\mu_1}$ to $\K_{\mu_2}.$
\end{theorem}
\begin{proof}
	See section \ref{sec:passing-to-canonical}.
\end{proof}
\begin{remark}
    The two sets of conditions (i-ii) above are actually equivalent to the one another, independently of the pair of conditions (iii).
	We do not provide a proof of this, but we indicate how the proof works in section \ref{sec:passing-to-canonical}.
\end{remark}

\begin{theorem} \label{thm:mixed-from-factor}
	If $\omega_1$ is a factorial state, then we have $\omega_2 \prec \omega_1$ if and only if the second set of conditions in the above theorem is satisfied.
    In fact, one only needs to check conditions (i) and (ii), since in the setting of factorial $\omega_1$, these conditions automatically imply condition (iii).
\end{theorem}
\begin{proof}
	The canonical purification of a factorial state is pure, so we may combine theorems \ref{thm:both-pure} and \ref{thm:passing-to-canonical} to reach this conclusion, with a few extra subtleties.
	See section \ref{sec:factorial-consequences}.
\end{proof}

\begin{theorem} \label{thm:mixed-from-pure}
	If $\omega_1$ is pure, then in the set of conditions for $\omega_2 \prec \omega_1$, the assumption that $\sqrt{X_2} - \sqrt{X_1}$ is Hilbert-Schmidt is equivalent to the condition that the diagonal elements of $X_2 - X_1$, within the block decomposition for $R_1,$ be trace class.
	
	Equivalently, one may check that $Q - 1$ is Hilbert-Schmidt and that $Q + R_1 Q^{-1} R_1$ is trace-class.
	
	(Note that $R_1$ has eigenvalues $\pm i$ in this setting; these are the eigenvalues with respect to which the block decomposition is taken.)
\end{theorem}
\begin{proof}
	See section \ref{sec:exciting-from-pure}.
\end{proof}

\begin{prop}
	There exist examples where $\omega_1$ is pure, $\omega_2$ is mixed, the block-diagonal elements of $X_2 - X_1$ are trace class, but $X_2 - X_1$ is not itself trace-class.
\end{prop}
\begin{proof}
	See appendix \ref{app:counterexample-2}.
\end{proof}

\begin{theorem} \label{thm:abelian}
	Given a free field algebra with $\Omega = 0,$ i.e., an abelian free field algebra, one has $\omega_2 \prec \omega_1$ if and only if
	\begin{enumerate}[(i)]
	\item $\mu_2 \prec \mu_1$,
	\item $Q-1$ is Hilbert-Schmidt, and
    \item $Q$ has trivial kernel.
	\end{enumerate}
\end{theorem}
\begin{proof}
	See section \ref{sec:abelian}.
\end{proof}

\begin{theorem} \label{thm:c-pure-from-c-pure}
	When $\omega_1$ and $\omega_2$ are centrally pure, the excitability relation $\omega_2 \prec \omega_1$ is equivalent to
	\begin{enumerate}[(i)]
		\item $\mu_2 \prec \mu_1$,
		\item $Q - 1$ is Hilbert-Schmidt, and
		\item $Q$ has trivial kernel.
	\end{enumerate}
\end{theorem}
\begin{proof}
	This is proved in section \ref{sec:centrally-pure} by breaking a centrally pure state into a ``pure piece'' and an ``abelian piece''  (section \ref{sec:excitability-breakdown}), and then combining theorems \ref{thm:mixed-from-pure} and \ref{thm:abelian} (section \ref{sec:general-proof}).
\end{proof}

\begin{theorem} \label{thm:general-theorem}
	For zero-mean Gaussian states $\omega_1$ and $\omega_2,$ the excitability relation $\omega_2 \prec \omega_1$ is equivalent to
	\begin{enumerate}[(i)]
		\item $\mu_2 \prec \mu_1$,
		\item $\sqrt{X_2} - \sqrt{X_1}$ is Hilbert-Schmidt, and
		\item $Q$ has trivial kernel.
	\end{enumerate}
\end{theorem}
\begin{proof}
	This follows immediately by combining theorem \ref{thm:c-pure-from-c-pure} with theorem \ref{thm:passing-to-canonical}.
\end{proof}

\begin{theorem} \label{thm:always-quasi}
	For zero-mean Gaussian states, $\omega_2 \prec \omega_1$ always implies $\omega_1 \prec \omega_2.$
\end{theorem}
\begin{proof}
This theorem is proved in section \ref{sec:other-way} from the necessary conditions (i-iii) of theorem \ref{thm:general-theorem}.
\end{proof}

\begin{remark}
    By combining this last theorem with proposition \ref{prop:mutual-exc-is-quasi}, we learn that for quasiequivalence, it is necessary and sufficient to check the conditions of theorem \ref{thm:general-theorem}.
    This should be contrasted with the work of Araki and Yamagami \cite{Araki-Yamagami}, in which quasiequivalence is shown to be equivalent to conditions (i) and (ii) of theorem \ref{thm:general-theorem}, \textit{together} with the condition $\mu_1 \prec \mu_2.$

    To see consistency of our results with those of Araki and Yamagami, note that once one has $\mu_2 \prec \mu_1$ and $\mu_1 \prec \mu_2,$ then it is obvious that $Q$ must be invertible, so it cannot have a kernel --- this gives condition (iii) in our theorem.
    On the other hand, as explained in section \ref{sec:passing-to-canonical}, condition (ii) of our theorem implies that $Q$ is invertible away from its kernel, so when condition (iii) is satisfied, we learn that $Q$ is invertible, which gives the Araki-Yamagami condition $\mu_1 \prec \mu_2.$ 
\end{remark}

\section{Quasiequivalence and excitability}
\label{sec:excitability}

In the introduction, we alluded to the problem of ``quasiequivalence,'' which has its roots in \cite{Haag-Kastler} and was explored extensively in \cite{Powers:quasi, Woronowicz:purification, Araki:quasi, VanDaele:quasi, Araki-Yamagami, Verch:hadamard}.
This is the question of when two sectors of the same quantum field theory have physically isomorphic von Neumann algebras.
The goals of the present section are: (i) to explain the connection between quasiequivalence and ``mutual excitability;'' (ii) to establish weaker algebraic criteria that are equivalent to ``one-way excitability;'' (iii) to demonstrate that one-way excitability of generic states is equivalent to one-way excitability of the canonical purifications; and (iv) to accomplish points (i)-(iii) in a framework that is appropriate for general interacting field theories, rather than the C$^*$-algebraic framework of \cite{Woronowicz:purification} that is only appropriate for free theories.

While later sections of this paper are specifically about free field theory, this section is written in the fashion of \cite{Sorce:paper1}, where a $*$-algebraic framework is used to encourage future generalizations beyond the setting of free fields.

\subsection{Quasiequivalence and mutual excitability}
\label{sec:mutual-excitability}

In \cite{Haag-Kastler, Woronowicz:purification}, the problem of quasiequivalence is formulated as follows.
One has an abstract C$^*$-algebra $\A_0$ of quantum fields, on which one specifies a continuous quantum state $\omega$ and constructs the GNS representation $\H_{\omega}.$
Within this representation, the C$^*$-algebra $\A_0$ can be completed to a von Neumann algebra $\A_{\omega}.$
Two states $\omega_1$ and $\omega_2$ are said to be ``quasiequivalent'' if there is a $*$-isomorphism between $\A_{\omega_1}$ and $\A_{\omega_2}$ that acts as the identity on $\A_0.$
We will begin by explaining the equivalence of this criterion to ``mutual excitability,'' then generalize to the setting where we may only assume that $\A_0$ is a $*$-algebra.

As explained in the introduction, it is natural to say that the state $\omega_2$ can be ``excited'' out of $\omega_1$ if there is a density matrix on $\H_{\omega_1}$ that reproduces all of the $\omega_2$ correlation functions.
Concretely, we say ``$\omega_2$ can be excited out of $\omega_1$'' if there is a density matrix $\rho_{\omega_2}$ on $\H_{\omega_1}$ satisfying\footnote{We note that this density matrix representative will \textit{not} generally live in $\A_{\omega_1}$ itself --- in fact this is impossible if $\A_{\omega_1}$ is not type I (see \cite{Sorce:types}).}
\begin{equation}
	\tr_{\H_{\omega_1}}(\rho_{\omega_2} a)
	= \omega_2(a) \qquad \forall\, a \in \A_0.
\end{equation}
We write this excitability relation as $\omega_2 \prec \omega_1.$

Suppose now that the states $\omega_1$ and $\omega_2$ are quasiequivalent.
This means there is a $*$-isomorphism $\alpha : \A_{\omega_1} \to \A_{\omega_2}$ that acts as the identity on $\A_0.$
In this case, one can make use of a foundational theorem about von Neumann algebras --- see e.g.  \cite[corollary 5.13]{Stratila:book} --- which says that every $*$-isomorphism is automatically continuous with respect to the ultraweak topology.\footnote{A review of the ultraweak topology can be found in \cite[section 2.1]{Sorce:modular}, but the basic idea is that in the ultraweak topology on a von Neumann algebra $\A$, a sequence (or net) of operators $a_n \in \A$ converges to an operator $a \in \A$ if and only if we have
	\begin{equation*}
		\tr(\rho a_n) \to \tr(\rho a)
	\end{equation*}
	for every density matrix $\rho$ acting on $\H$.}
Since $\omega_2$ is automatically ultraweakly continuous as a functional on $\A_{\omega_2},$ the pullback $\omega_2 \circ \alpha$ is ultraweakly continuous as a functional on $\A_{\omega_1}.$
Moreover, because $\alpha$ preserves $\A_0,$ one has
\begin{equation}
	(\omega_2 \circ \alpha)(a)
	= \omega_2(a) \qquad \forall\, a \in \A_0.
\end{equation}
This means that there is an ultraweakly continuous functional on $\A_{\omega_1}$ that reproduces the correlation functions of $\omega_2.$
It is a basic fact about von Neumann algebras --- see e.g. \cite[theorem 46.4]{Conway:book} --- that every positive, ultraweakly continuous functional has a density matrix representative.
Putting these pieces together, we see that when $\omega_1$ and $\omega_2$ are quasiequivalent, we have the excitability relation $\omega_2 \prec \omega_1.$
In fact, because $\alpha$ is a $*$-isomorphism, it is invertible, and we can run the argument the other way --- this implies that when $\omega_1$ and $\omega_2$ are quasiequivalent, we have $\omega_2 \prec \omega_1$ \textit{and} $\omega_1 \prec \omega_2,$ i.e., we have mutual excitability.

In fact, it turns out that mutual excitability is \textit{the same statement} as quasiequivalence.
That is, if we assume $\omega_2 \prec \omega_1$ and $\omega_1 \prec \omega_2,$ then we may conclude that there is a $*$-isomorphism from $\A_{\omega_1}$ to $\A_{\omega_2}$ that preserves $\A_0.$
First, beginning with the assumption $\omega_2 \prec \omega_1,$ one can go to ``standard form'' for $\A_{\omega_1}$ by selecting a particular faithful functional $\phi_1 : \A_{\omega_1} \to \comps$ and going to its GNS representation $\H_{\phi_1}.$
A review of this procedure can be found in \cite[section 4.1]{Sorce:paper1}; the important point is that one has $\A_{\omega_1} \cong \A_{\phi_1},$ i.e., no new operators are introduced in the GNS representation $\H_{\phi_1}.$
Since we have assumed $\omega_2 \prec \omega_1,$ there is a natural extension of $\omega_2$ to an ultraweakly continuous functional on $\A_{\omega_1},$ and hence on $\A_{\phi_1}.$
Moreover, because $|\phi_1\rangle$ is cyclic and separating for $\A_{\phi_1}$ within $\H_{\phi_1}$,  every ultraweakly continuous functional on $\A_{\phi_1}$ has a vector representative --- see e.g. \cite[corollary 5.24]{Stratila:book}.
So there is a representative $|\omega_2\rangle_{\phi_1} \in \H_{\phi_1}.$
Defining the operator $U : \H_{\phi_1} \to \H_{\omega_2}$ by
\begin{equation}
    U a |\omega_2\rangle_{\phi_1} = a |\omega_2\rangle_{\omega_2}, \qquad a \in \A_0,
\end{equation}
and defining $U$ to vanish on the orthocomplement of the set $\A_{0} |\omega_2\rangle_{\phi_1},$ one easily sees that $U$ is a partial isometry.
One also easily sees that $U$ conjugates elements of $\A_0$ acting on $\H_{\phi_1}$ to the same elements of $\A_0$ acting on $\H_{\omega_2}.$
So conjugation by $U$ serves as the ``identity map'' on $\A_0$ from $\A_{\omega_1}$ to $\A_{\omega_2}$, and conjugation by a partial isometry is always ultraweakly continuous.
One therefore concludes, from the assumption $\omega_2 \prec \omega_1,$ that the identity map from $\A_0 \subseteq \A_{\omega_1}$ to $\A_0 \subseteq \A_{\omega_2}$ is ultraweakly continuous.
Repeating the same argument for the assumption $\omega_1 \prec \omega_2,$ one forms a chain of ultraweakly continuous maps
\begin{equation}
	\A_{\omega_1} \to \A_{\omega_2} \to \A_{\omega_1}
\end{equation}
that compose to the identity.
It follows that the ``identity map'' on $\A_0$ provides a $*$-isomorphism from $\A_{\omega_1}$ to $\A_{\omega_2}$.

Thus far we have established that ``quasiequivalence'' and ``mutual excitability'' are identical notions for the C$^*$-algebraic framework considered in \cite{Haag-Kastler, Woronowicz:purification}.
However, this framework is not believed to be appropriate for general quantum field theories \cite{Hollands:axioms, Hollands:review}.
In free field theories, because the commutator of two fields is proportional to the identity, the Baker-Campbell-Hausdorff formula allows one to create a C$^*$-algebra of ``exponentiated fields'' $e^{i \phi[f]}$ \cite{Slawny:1972iq, Bratteli:1996xq}.
For a general field theory, however, we can only assume that we have a $*$-algebra generated by smeared fields $\phi[f],$ and they cannot be exponentiated in an abstract way that is independent of representation.

To generalize the above discussion, we now assume that $\A_0$ is only an abstract $*$-algebra.
As reviewed in \cite[appendix A.3]{Sorce:paper1}, any algebraic state $\omega : \A_0 \to \comps$ gives rise to a GNS representation $\H_{\omega}$ on which the elements of $\A_0$ act as closed, unbounded operators.
We follow the choice made in \cite{Sorce:paper1}, where $\A_{\omega}$ is taken to be the smallest von Neumann algebra on $\H_{\omega}$ such that all elements of $\A_0$ are ``affiliated.''
Equivalently, as proved in \cite[appendix A.4]{Sorce:paper1}, one may take arbitrary polar decompositions of elements in $\A_0,$
\begin{equation}
	a = u_a |a|,
\end{equation}
and take $\A_{\omega}$ to be generated by all partial isometries $u_a,$ their adjoints $u_a^{\dagger},$ and bounded functions $f(|a|)$ of the self-adjoint operators $|a|.$

In this setting, what should we mean when we say that $\omega_1$ is quasiequivalent to $\omega_2$?
We want to say that this happens whenever there is a $*$-isomorphism from $\A_{\omega_1}$ to $\A_{\omega_2}$ that ``preserves the field algebra $\A_0.$''
But because $\A_0$ acts in an unbounded way on each space $\H_{\omega_j},$ it is not a subalgebra of either von Neumann algebra $\A_{\omega_j}.$
Instead, it makes sense to say that a $*$-isomorphism $\alpha : \A_{\omega_1} \to \A_{\omega_2}$ is a quasiequivalence if it preserves the polar decompositions of elements in $\A_0$.
So if we decompose $a \in \A_0$ as $u_a^{(1)} |a|^{(1)}$ on $\H_{\omega_1}$ and as $u_a^{(2)} |a|^{(2)}$ on $\H_{\omega_2},$ then for quasiequivalence we should require
\begin{equation}
	\alpha(u_a^{(1)})
	= u_a^{(2)}
\end{equation}
as well as 
\begin{equation}
	\alpha(f(|a|^{(1)}))
	= f(|a|^{(2)}), \qquad f \text{ bounded}.
\end{equation}
As for the excitability relation, it makes sense to say that $\omega_2$ can be excited out of $\omega_1$ if there is a density matrix $\rho_{\omega_2}$ on $\A_{\omega_1}$ satisfying
\begin{align}
	\tr_{\H_{\omega_1}}(\rho_{\omega_2} u_a^{(1)})
	= \langle \omega_2 | u_a^{(2)} | \omega_2\rangle,
\end{align}
as well as 
\begin{equation}
	\tr_{\H_{\omega_1}}(\rho_{\omega_2} f(|a|^{(1)}))
	= \langle \omega_2 | f(|a|^{(2)}) | \omega_2\rangle, \qquad f \text{ bounded},
\end{equation}
and similarly for all finite products of operators like $u_a$ and $f(|a|),$ as well as $u_a^{\dagger}.$

By a similar logic to the C$^*$-algebra case, it is obvious that quasiequivalence of $\omega_1$ and $\omega_2$ implies mutual excitability.
The converse is slightly more subtle.
As in the C$^*$-algebra case, one begins with the assumption $\omega_2 \prec \omega_1$ and passes to standard form for $\A_{\omega_1}$ by going to an equivalent von Neumann algebra $\A_{\phi_1}$ acting on $\H_{\phi_1}.$
In the C$^*$-algebraic case, this isomorphism would map $\A_0 \subseteq \A_{\omega_1}$ to an equivalent subalgebra of $\A_{\phi_1}.$
In the general, $*$-algebraic case, we must deal with the fact that $\A_0$ is not a subalgebra of $\A_{\omega_1}.$
Nevertheless, the unbounded operators in $\A_0$ actually do map, in a natural way, into operators acting on $\H_{\phi_1}$.
This map, constructed in \cite[section 4.1]{Sorce:paper1}, uses the polar decomposition explicitly.
As a consequence, the polar decompositions of elements of $\A_0$ acting on $\H_{\omega_1}$ are canonically isomorphic, within $\A_{\omega_1},$ to the polar decompositions acting on $\H_{\phi_1}$ within $\A_{\phi_1}.$
One can then proceed as in the C$^*$ algebraic case to construct a map $U$ that conjugates polar decompositions of $\A_0$ elements on $\A_{\omega_1}$ to the corresponding polar decompositions of $\A_0$ elements on $\A_{\omega_2}.$
This map is ultraweakly continuous, and repeating the argument for $\omega_1 \prec \omega_2$ provides an inverse, so the conjugation map is a quasiequivalence isomorphism, as desired.

\subsection{One-way excitability and normal homomorphisms}
\label{sec:one-way-excitability}

We have shown that algebraic quasiequivalence is the same as the mutual excitability relation $\omega_2 \prec \omega_1$ and $\omega_1 \prec \omega_2.$
Here, we show the following weakened statements:
\begin{itemize}
	\item We have $\omega_2 \prec \omega_1$ if and only if there exists an ultraweakly continuous $*$-\textit{homo}morphism from $\A_{\omega_1}$ to $\A_{\omega_2}$ that acts as the identity on field operators.\footnote{``Acting as the identity'' is meant in the ``polar decomposition sense'' of the previous subsection.}
	An ultraweakly continuous $*$-homomorphism is called a ``normal homomorphism,'' and we will use this terminology below.\footnote{Let us remark that any sensible definition of excitability should be transitive; that is, $\omega_3\prec \omega_2$ and $\omega_2 \prec \omega_1$ should imply $\omega_3 \prec \omega_1$.
By recasting excitability in terms of normal homomorphisms, this transitivity property becomes automatic; clearly the maps ${\cal A}_{\omega_1}\to {\cal A}_{\omega_2}$ and ${\cal A}_{\omega_2} \to {\cal A}_{\omega_3}$ compose to a normal homomorphism ${\cal A}_{\omega_1}\to {\cal A}_{\omega_3}$.}
	\item If we have $\omega_2 \prec \omega_1$, and $\A_{\omega_1}$ is a factor, then we automatically have mutual excitability $\omega_1 \prec \omega_2.$
	Conversely, if $\A_{\omega_1}$ is not a factor, then there always exists a state $\omega_2$ that can be excited out of $\omega_1$ but not vice versa.
\end{itemize}

The proof of the first statement is almost trivial given what we already know.
When we proved that having a quasiequivalence map from $\A_{\omega_1}$ to $\A_{\omega_2}$ induces the relation $\omega_2 \prec \omega_1,$ we actually did not use the assumption that the homomorphism was invertible; we only used its ultraweak continuity and the fact that it preserves the field algebra.
The converse relation, $\omega_1 \prec \omega_2,$ was the only part of the argument that used the invertibility assumption.
So by repeating the logic of the previous subsection, one finds that whenever we have a field-preserving normal homomorphism from $\A_{\omega_1}$ to $\A_{\omega_2},$ we have $\omega_2 \prec \omega_1.$

Conversely, given $\omega_2 \prec \omega_1,$ and using the trick of passing to standard form, we already showed in the preceding subsection that one obtains a normal homomorphism from $\A_{\omega_1}$ to $\A_{\omega_2}$ that preserves polar decompositions; the converse assumption $\omega_1 \prec \omega_2$ was only needed if we wanted to show that this normal homomorphism had an inverse.

As for the second bullet point above, we first assume that $\omega_1$ is a state on $\A_0$ such that the algebra $\A_{\omega_1}$ is not a factor.
This means that there is a nontrivial center $\A_{\omega_1} \cap \A_{\omega_1}'.$
Take $P$ to be a nontrivial projection operator in this algebra, and define the abstract algebraic state
\begin{equation}
	\omega_2(a)
	\equiv \frac{\langle P \omega_1 | a | P \omega_1\rangle_1}{\langle P \omega_1 | P \omega_1\rangle_1}, \qquad a \in \A_0.
\end{equation}
This expression makes sense even though $a$ is unbounded, because $P$ is in the commutant of $\A_{\omega_1}$ and $a$ is affiliated to $\A_{\omega_1}$; consequently, $|P \omega_1\rangle$ is in the domain of $a$ and in fact we have
\begin{equation}
	\omega_2(a)
	= \frac{\langle P \omega_1 | a | \omega_1\rangle_1}{\langle P \omega_1 | P \omega_1\rangle_1}.
\end{equation}
The GNS representation of $\omega_2$ embeds isometrically within $\H_{\omega_1}$ via the formula
\begin{equation}
	V a |\omega_2\rangle_2
	= a \frac{P |\omega_1\rangle_1}{\lVert P |\omega_1\rangle_1 \rVert}.
\end{equation}
Conjugation by $V^{\dagger}$ maps operators in $\A_0$ acting on $\H_{\omega_1}$ to the equivalent operators acting on $\H_{\omega_2},$ since we have
\begin{equation}
	V^{\dagger} a V b |\omega_2\rangle_2
	= V^{\dagger} a b \frac{P |\omega_1\rangle_1}{\lVert P |\omega_1\rangle_1 \rVert}
	= a b |\omega_2\rangle_2.
\end{equation}
This conjugation preserves polar decompositions, since we have\footnote{A sneaky trick here is that $VV^{\dagger}$ is the projection onto the image of $V$, which is the closure of the space $\A_{\omega_1} P |\omega_1\rangle.$
	This projection commutes with $\A_{\omega_1}$ --- as does any projection onto an invariant space for $\A_{\omega_1}$ --- so we can perform the manipulation
	\begin{equation*}
		V^{\dagger} u_a^{(1)} |a|^{(1)} V
		= V^{\dagger} u_a^{(1)} |a|^{(1)} (VV^{\dagger}) V
		= V^{\dagger} u_a^{(1)} (VV^{\dagger})|a|^{(1)} V.
\end{equation*}
\label{foot:sneakyfoot}}
\begin{equation}
	u_a^{(2)} |a|^{(2)}
	= V^{\dagger} u_a^{(1)} |a|^{(1)} V
	= (V^{\dagger} u_a^{(1)} V) (V^{\dagger} |a|^{(1)} V),
\end{equation}
and from uniqueness of the polar decomposition it is easy to show
\begin{equation}
	u_a^{(2)}
	= V^{\dagger} u_a^{(1)} V
\end{equation}
and similarly for $|a|.$
Consequently, the state $\omega_2$ can be excited out of $\omega_1$ as $P |\omega_1\rangle_1 / \lVert P |\omega_1\rangle_1 \rVert,$ and conjugation by $V^{\dagger}$ provides an explicit field-preserving normal homomorphism from $\A_{\omega_1}$ to $\A_{\omega_2}$.

On the other hand, we can show that for the state $\omega_2$ constructed above, one has $\omega_1 \nprec \omega_2.$
Otherwise, if $\omega_1$ \textit{could} be excited out of $\omega_2,$ then by the analysis of the preceding subsection, conjugation by $V^{\dagger}$ would have to be a bijection from $\A_{\omega_1}$ to $\A_{\omega_2}.$
But this is not the case, since we have $V^{\dagger} (1 - P) V = 0.$
To see this, one simply checks
\begin{equation}
	V^{\dagger} (1 - P) V a |\omega_2\rangle_2
	= V^{\dagger} (1 - P) a \frac{P |\omega_1\rangle_1}{\lVert P |\omega_1\rangle_1 \rVert}
	= V^{\dagger} (1 - P) P \frac{a |\omega_1\rangle_1}{\lVert P |\omega_1\rangle_1 \rVert}
	= 0,
\end{equation}
where we have used that $P$ is in the center of $\A_{\omega_1},$ so it commutes with each $a \in \A_0.$

We have now proved one part of the second bullet point from the beginning of this subsection --- we have proved that when $\A_{\omega_1}$ is not a factor, it is possible to construct a state $\omega_2$ with $\omega_2 \prec \omega_1$ but $\omega_1 \nprec \omega_2.$
We now show the converse --- when $\A_{\omega_1}$ is a factor and $\omega_2 \prec \omega_1,$ we always have $\omega_1 \prec \omega_2.$

By the same kind of ``standard form'' argument used several times in this section, it suffices to assume that $|\omega_1\rangle_1$ is separating for $\A_{\omega_1}.$
So we assume that $\omega_2$ is represented by a vector $|\omega_2\rangle_1$ in $\H_{\omega_1}$, and construct the isometry from $\H_{\omega_2}$ to $\H_{\omega_1}$ given by
\begin{equation}
	V a |\omega_2\rangle_2 = a |\omega_2\rangle_1.
\end{equation}
Conjugation by $V^{\dagger}$ maps $\A_{\omega_1}$ to $\A_{\omega_2}.$
We aim to show that when $\A_{\omega_1}$ is a factor, this map is both injective and surjective.
Surjectivity follows from the fact that $V$ is an isometry, so conjugation by $V$ is a right-inverse for conjugation by $V^{\dagger}.$
As for injectivity, we assume toward contradiction that there exists a nontrivial operator $L \in \A_{\omega_1}$ with $V^{\dagger} L V = 0.$
The image of $V$ is the closure of $\A_{\omega_1} |\omega_2\rangle_1,$ and the projection onto this space, $P'$, is an element of the commutant $\A_{\omega_1}'.$
The formula $V^{\dagger} L V = 0$ is equivalent to $P' L P' = 0,$ or simply $L P' = 0$ since $P'$ commutes with $L$.

The assumption that there is a nontrivial $L \in \A_{\omega_1}$ with $L P' = 0$ is the statement that the support of $P'$ is \textit{non-separating} for the von Neumann algebra $\A_{\omega_1}.$
It is an easy exercise\footnote{Proof: Let $a \in \A_{\omega_1}$ be a nonzero operator with $a b |\omega_2\rangle_1 = 0$ for all $b \in \A_{\omega_1}.$ The operator $a^{\dagger} a$ is positive, so we can pick some nonzero spectral projection $\Pi \in \A_{\omega_1}$ that is bounded above by a constant multiple of $a^{\dagger} a,$ giving $\langle b \omega_2| \Pi | b \omega_2\rangle_1 \leq 0$ and hence $\Pi b |\omega_2\rangle_1 = 0$ for all $b \in \A_{\omega_1}.$ Since $\Pi$ is in $\A_{\omega_1},$ this gives $\Pi a' b |\omega_2\rangle_1 = 0$ for all $b \in \A_{\omega_1}$ and $a' \in \A_{\omega_1}',$ so the support of $\Pi$ is orthogonal to the subspace on the left side of equation \eqref{eq:proper-inclusion-cyclic}.} to show that this means the support of $P'$ is \textit{non-cyclic} for $\A_{\omega_1}',$ i.e.,
\begin{equation} \label{eq:proper-inclusion-cyclic}
	(\text{smallest closed vector space containing $\A_{\omega_1}' \A_{\omega_1} |\omega_2\rangle_1$}) \subsetneq \H_{\omega_1}.
\end{equation}
But the left-hand side of the above inclusion is $(\A_{\omega_1} \vee \A_{\omega_1}') |\omega_2\rangle_1,$ where $\A_{\omega_1} \vee \A_{\omega_1}'$ is the von Neumann algebra generated by $\A_{\omega_1}$ and $\A_{\omega_1}'.$
It is easy to check that when $\A_{\omega_1}$ is a factor, this ``doubly generated'' algebra is the full algebra of bounded operators on Hilbert space, and the proper inclusion in \eqref{eq:proper-inclusion-cyclic} cannot hold.
We therefore conclude that when $\A_{\omega_1}$ is a factor, conjugation by $V^{\dagger}$ is injective.
Putting everything together, we conclude that when we have $\omega_2 \prec \omega_1$ and $\A_{\omega_1}$ is a factor, we also have $\omega_1 \prec \omega_2.$

\subsection{Checking excitability with canonical purifications}
\label{sec:canonical-checking}

In the rest of the paper, we will be concerned with establishing concrete conditions for $\omega_2 \prec \omega_1$ when $\omega_1$ and $\omega_2$ are Gaussian states of a free field theory.
The technique we will use is to begin by studying the case where $\omega_1$ and $\omega_2$ are pure, then studying more general states by using canonical purifications.
This technique goes back to \cite{Powers:quasi} for the specific case of ``factorial states'' in free fermion theories.
A general framework was worked out by Woronowicz in \cite{Woronowicz:purification} for factorial states on C$^*$ algebras.
Here we generalize Woronowicz's result away from factorial states, and away from C$^*$ algebras.

In particular, we consider a general $*$-algebra $\A_0$ and an algebraic state $\omega : \A_0 \to \comps.$
In \cite{Sorce:paper1}, the canonical purification of $\omega$ was defined as a particular algebraic state $\hat{\omega} : \A_0 \otimes \A_0^{\text{op}} \to \comps.$
We emphasize that in technical terms, the state $\hat{\omega}$ is not actually pure unless $\A_{\omega}$ is a factor.
In general, $\hat{\omega}$ will be a ``centrally pure'' state, as explained in detail in \cite{Sorce:paper1}.
This means that the von Neumann algebra $(\A_0 \otimes \A_{0}^{\text{op}})_{\hat{\omega}}$ may not be a factor, but its commutant will always be contained within the algebra itself; in other words, one has
\begin{equation}
	(\A_0 \otimes \A_{0}^{\text{op}})_{\hat{\omega}}'
	\subseteq (\A_0 \otimes \A_{0}^{\text{op}})_{\hat{\omega}}.
\end{equation}
Nevertheless, we prove in this subsection the following statement: $\omega_2 \prec \omega_1$ if and only if $\hat{\omega}_2 \prec \hat{\omega}_1$.

For the ``if'' direction, one assumes that $\hat{\omega}_2$ can be excited out of $\hat{\omega}_1.$
This means that there is a normal homomorphism
\begin{equation}
	\hat{\alpha} : (\A_0 \otimes \A_{0}^{\text{op}})_{\hat{\omega}_1}
	\to (\A_0 \otimes \A_{0}^{\text{op}})_{\hat{\omega}_2}
\end{equation}
that preserves the polar decompositions of elements of the $*$-algebra.
Obviously the restriction of this map to $(\A_0 \otimes 1)_{\hat{\omega}_1}$ preserves the polar decompositions of elements of the $*$-algebra $\A_0 \otimes 1,$ and maps to $(\A_0 \otimes 1)_{\hat{\omega}_2}$.
The structure theorem of \cite[section 4.4]{Sorce:paper1} establishes an isomorphism between $(\A_0 \otimes 1)_{\hat{\omega}_j}$ and $\A_{\omega_j}$, and by pulling back through these isomorphisms, we obtain a field-preserving normal homomorphism from $\A_{\omega_1}$ to $\A_{\omega_2},$ which gives $\omega_2 \prec \omega_1.$

For the ``only if'' direction, one assumes $\omega_2 \prec \omega_1.$
Using the structure theorem of \cite[section 4.4]{Sorce:paper1}, it suffices to assume that the states $|\omega_j\rangle_j$ are separating for the algebras $\A_{\omega_j},$ in which case one can map $\H_{\hat{\omega}_j}$ unitarily to $\H_{\omega_j}$.
In this setting, what we are trying to prove is the existence of a normal homomorphism
\begin{equation}
	\A_{\omega_1} \vee \A_{\omega_1}' \to \A_{\omega_2} \vee \A_{\omega_2}'
    \label{eq:purified_homo}
\end{equation}
that maps polar decompositions of elements $a J_{\omega_1} b J_{\omega_1}$ on $\H_{\omega_1}$ to polar decompositions of corresponding elements  $a J_{\omega_2} b J_{\omega_2}$ on $\H_{\omega_2}.$
The assumption we are allowed to use is that $\omega_2$ can be excited out of $\omega_1$.
Since $|\omega_1\rangle_1$ is separating for $\A_{\omega_1}$, we can assume that the representative of $\omega_2$ is a vector, and we may as well take it to be the (unique!) vector representative within the natural cone of $|\omega_1\rangle_1$ \cite{Araki:natural-1}.
What this means in practice is that there is a vector $|\omega_2^{\natural}\rangle_1 \in \H_{\omega_1}$ that (i) reproduces the correlators of $\omega_2$ on $\A_{\omega_1}$, and (ii) is fixed by the modular conjugation $J_{\omega_1}.$

We can construct an isometry from $\H_{\omega_2}$ to $\H_{\omega_1}$ using
\begin{equation} \label{eq:natural-def-V}
	V a |\omega_2\rangle_2
	= a |\omega_2^{\natural}\rangle_1, \qquad a \in \A_0.
\end{equation}
We would now like to construct the desired homomorphism \eqref{eq:purified_homo}
as the map which conjugates by $V^{\dag}$; that is, we would like to show
\begin{equation}
    V^{\dag} (aJ_{\omega_1} b J_{\omega_1})V c\ket{\omega_2}_2
    =
    (aJ_{\omega_2}b J_{\omega_2})c\ket{\omega_2}_2.
    \label{eq:conj_aJbJ}
\end{equation}
In fact, it is enough to just show
\begin{equation}
    V^{\dag} (J_{\omega_1} b J_{\omega_1})V c\ket{\omega_2}_2
    =
    (J_{\omega_2}b J_{\omega_2})c\ket{\omega_2}_2.
    \label{eq:conj_JbJ}
\end{equation}
To see this, note that one can compute $V^{\dagger} a V = a$ using equation \eqref{eq:natural-def-V}.
So as long as we are allowed to insert a factor of $VV^{\dag}$ after the first $a$ in equation \eqref{eq:conj_aJbJ}, we recover equation \eqref{eq:conj_aJbJ} using equation \eqref{eq:conj_JbJ}.
To perform this insertion,
note that since the image of $V$ is an invariant space for $\A_{\omega_1},$  the projection $P' = V V^{\dagger}$ is in $\A_{\omega_1}'$.
Hence, we can perform the manipulation $V^{\dag}a=V^{\dag}VV^{\dag}a=V^{\dag}aVV^{\dag}$.

To prove equation \eqref{eq:conj_JbJ}, we will use the identity $V^{\dagger} J_{\omega_1} V = J_{\omega_2}.$
The proof of that identity is straightforward but tedious, so we present it in appendix \ref{app:modular-conjugation-proof}.
Once we know this identity, we prove equation \eqref{eq:conj_JbJ} using the following steps:
\begin{align}
	\begin{split}
		V^{\dagger} (J_{\omega_1} a J_{\omega_1}) V b |\omega_2\rangle_2
		& = V^{\dagger} J_{\omega_1} a J_{\omega_1} b |\omega_2^{\natural}\rangle_1 \\
		& = V^{\dagger} J_{\omega_1} a (J_{\omega_1} b J_{\omega_1}) |\omega_2^{\natural}\rangle_1 \\
		& = V^{\dagger} J_{\omega_1} (J_{\omega_1} b J_{\omega_1}) a |\omega_2^{\natural}\rangle_1 \\
		& = V^{\dagger} b J_{\omega_1} a |\omega_2^{\natural}\rangle_1 \\
		& = V^{\dagger} (V V^{\dagger}) b J_{\omega_1} (V V^{\dagger}) a |\omega_2^{\natural}\rangle_1 \\
        & = V^{\dagger} b (V V^{\dagger}) J_{\omega_1} V (V^{\dagger} a |\omega_2^{\natural}\rangle_1) \\
		& = (V^{\dagger} b V) (V^{\dagger} J_{\omega_1} V) a |\omega_2\rangle_2 \\
		& = b J_{\omega_2} a |\omega_2\rangle_2 \\
		& = J_{\omega_2} a J_{\omega_2} b |\omega_2\rangle_2.
	\end{split}
\end{align}
So conjugation by $V^{\dagger}$ maps $J_{\omega_1} a J_{\omega_1}$ to $J_{\omega_2} a J_{\omega_2}.$
This provides a normal homomorphism from $\A_{\omega_1} \vee \A_{\omega_1}'$ to $\A_{\omega_2} \vee \A_{\omega_2}'$ that preserves polar decompositions, completing the proof that $\omega_2 \prec \omega_1$ implies $\hat{\omega}_2 \prec \hat{\omega}_1.$

\section{Some necessary conditions for Gaussian states}
\label{sec:general-necessary}

So far we have been looking at conditions for excitability for general states of general quantum theories.
Here, we establish some necessary conditions for excitability of zero-mean Gaussian states in free theories.
These conditions will form the basis of the detailed analysis of the rest of the paper.

\subsection{Boundedness of the symmetrized inner product}
\label{sec:mu-boundedness}

Suppose that $\omega_1$ and $\omega_2$ are zero-mean Gaussian states, and suppose that there is a density matrix $\rho_{\omega_2}$ in the GNS space $\H_{\omega_1}$ that reproduces the correlation functions of $\omega_2$ on the Weyl algebra.
The existence of this density matrix turns out to imply a useful relationship between the symmetrized inner products $\mu_1$ and $\mu_2$ that were defined in appendix \ref{app:phase-space}.
Recall, from that appendix, that one has
\begin{equation}
	\langle f | g \rangle_{\mu_1}
	= \frac{\omega_1(\phi[f^*] \phi[g] + \phi[g] \phi[f^*])}{2},
\end{equation}
that the Hilbert space of test functions endowed with this inner product is called $\K_{\mu_1},$ and that one defines an operator $R_1$ that implements the commutator,
\begin{equation}
	\langle f | R_1 | g\rangle_{\mu_1}
	= \frac{1}{2} \Omega[f^*, g].
\end{equation}
As explained in appendix \ref{app:vN-generating}, the map from test functions to Weyl operators, $f \mapsto e^{i \phi[f]},$ is continuous with respect to the $\mu_1$ topology on the domain and the strong-$\H_{\omega_1}$ topology on the image.\footnote{A quick introduction to operator topologies can be found in \cite[section 2.1]{Sorce:modular}.}
Since the image of this map is in the set of unitary operators, which all have norm one, this map is actually continuous in the ultrastrong topology --- it is a basic fact that the strong and ultrastrong topologies agree on bounded sets.
Finally, because the ultraweak topology is weaker than the ultrastrong topology, one may conclude that the map $f \mapsto e^{i\phi[f]}$ is continuous when one uses the ultraweak topology in the image space.
Consequently, if one has a sequence of real test functions $f_n$ that converges to zero with respect to $\mu_1$, then one has
\begin{equation}
	\tr_{\H_{\omega_1}}(\rho_{\omega_2} e^{i \phi[f_n]})
	\to 1.
\end{equation}
Expanding $e^{i \phi[f_n]}$ as a power series, one finds
\begin{equation}
	e^{- \langle f_n | f_n\rangle_{\mu_2}/2} \to 1,
\end{equation}
hence $\langle f_n | f_n\rangle_{\mu_2} \to 0.$

So given any real sequence $f_n$ with $\langle f_n | f_n\rangle_{\mu_1} \to 0,$ one has $\langle f_n | f_n\rangle_{\mu_2} \to 0.$
Since one has $\langle u + i v | u + i v\rangle_{\mu_1} = \langle u | u \rangle_{\mu_1} + \langle v | v \rangle_{\mu_1}$ for $u$ and $v$ real, one can reach the same conclusion for general complex $f_n.$
It follows that the inner product $\mu_2$ is bounded by a multiple of the inner product $\mu_1,$ i.e., there exists some number $c$ with
\begin{equation} \label{eq:mu2-bound}
	\langle f | f \rangle_{\mu_2} \leq c\, \langle f | f\rangle_{\mu_1}
\end{equation}
for any test function $f.$
One writes this relation as $\mu_2 \prec \mu_1,$ and says that $\mu_2$ is ``dominated by $\mu_1.$''
Note that in the special case where $\omega_1$ is factorial, meaning that $\A_{\omega_1}$ is a factor, we showed in section \ref{sec:one-way-excitability} that $\omega_2 \prec \omega_1$ automatically implies $\omega_1 \prec \omega_2.$
So under the assumption $\omega_2 \prec \omega_1,$ we automatically have that $\mu_1$ is bounded by a multiple of $\mu_2$:
\begin{equation} \label{eq:mu1-bound}
	\langle f | f \rangle_{\mu_1} \leq c'\, \langle f | f\rangle_{\mu_2}.
\end{equation}
In particular, this will be the case when $\omega_1$ is pure, since then we have $\A_{\omega_1} = \B(\H_{\omega_1}),$ which is a factor.

Whenever we have $\mu_2 \prec \mu_1,$ we can define a bounded map $L : \K_{\mu_1} \to \K_{\mu_2}$ that acts on test functions as
\begin{equation}
	L |f\rangle_{\mu_1}
		= |f \rangle_{\mu_2},
\end{equation}
so that the image of $L$ is dense in ${\cal K}_{\mu_2}$.
This map $L$ will be used many times in the rest of the paper, as it is the only tool we have for moving between the spaces $\K_{\mu_1}$ and $\K_{\mu_2}.$
In particular, while $L$ acts ``as the identity'' on test functions, other elements of $\K_{\mu_1}$ cannot be identified naturally with elements of $\K_{\mu_2}$; they can only be \textit{mapped} to $\K_{\mu_2}$ using $L$.

One advantage of having introduced $L$ is that it lets us formulate excitability using the ``completed spaces'' $\K_{\mu_1}$ and $\K_{\mu_2}$ instead of the more restrictive space of test functions.
Recall that, by definition, we said that we can excite $\omega_2$ out of $\omega_1$ if for some density matrix $\rho_{\omega_2}$ in $\mathcal H_{\omega_1}$ we have
\begin{equation}
\mathrm{tr}_{\mathcal{H}_{\omega_1}} (\rho_{\omega_2} e^{i\phi[f]}) = \expval{e^{i\phi[f]}}{\omega_2}
\end{equation}
for all test functions $f$.
Using the continuity of the map from real test functions to Weyl operators on $\H_{\omega_j}$, one can define a Weyl operator $e^{i \phi[\psi]}$ acting on $\H_{\omega_j}$ for \textit{any} real element of $\K_{\mu_j}$ --- see appendix \ref{app:review} for details.
Making use of continuity, one can equivalently formulate excitability as the requirement that for any real $|\psi\rangle$ in $\K_{\mu_1},$ one has
\begin{equation}
\mathrm{tr}_{\mathcal{H}_{\omega_1}} (\rho_{\omega_2} e^{i\phi[\psi]}) = \expval{e^{i\phi[L \psi]}}{\omega_2}.
\end{equation}
This notation will appear many times in what follows.

Note that while $L$ is bounded whenever we have $\mu_2 \prec \mu_1$, nothing we have discussed so far implies that it is invertible.
Ultimately, we will actually see that invertibility of $L$ does follow from the excitability relation $\omega_2 \prec \omega_1$, but this will take some doing and relies heavily on results in the rest of the paper.
In the next subsection, we will show at least that when one has $\omega_2 \prec \omega_1$, the map $L$ cannot have a kernel.

Before proceeding to that result, we demonstrate a very useful identity relating $L$ to the operators $R_1$ and $R_2$ that implement the commutator on $\K_{\mu_1}$ and $\K_{\mu_2},$ which are defined via the identity
\begin{equation}
	\langle f | R_j | g\rangle_{\mu_j}
		= \frac{1}{2} \Omega[f^*, g].
\end{equation} 
We have
\begin{equation}
	\langle f | R_1 | g \rangle_{\mu_1}
		= \frac{1}{2} \Omega[f^*, g]
		= \frac{1}{2} \langle f | R_2 | g \rangle_{\mu_2}
		= \frac{1}{2} \langle f | L^{\dagger} R_2 L |g \rangle_{\mu_1}.
\end{equation}
This gives
\begin{equation} \label{eq:R1-from-R2}
	L^{\dagger} R_2 L = R_1,
\end{equation}
which we use many times in the remainder of the paper.
It will also be useful to consider the polar decomposition of $L$, which we write as 
\begin{equation} \label{eq:L-polar}
    L=vQ^{1/2},
\end{equation}
where $Q^{1/2}=\sqrt{L^\dag L}$, and $v$ is a partial isometry such that $v^\dag v$ projects onto the support of $L$, and $v v^{\dagger}$ projects onto the closure of the image of $L$.

\subsection{\texorpdfstring{No kernel for $L$}{No kernel for L}}

So far, we know that when we have $\omega_2 \prec \omega_1,$ the identity map $L | f\rangle_{\mu_1} = |f\rangle_{\mu_2}$ is bounded.
We will now show that when excitability holds, $L$ cannot have a kernel.
This is a striking property of Gaussian states that has no analogue for general states in general quantum field theories, and it will be used in section \ref{sec:other-way} --- after developing the rest of the paper's machinery --- to rule out the possibility of having $\omega_2 \prec \omega_1$ while having $\omega_1 \not\prec \omega_2.$

First, assume that there is a state $|\psi\rangle_{\mu_1}$ in the kernel of $L$; without loss of generality, we will take $|\psi\rangle_{\mu_1}$ to be normalized.
Since $L$ commutes with complex conjugation, we may as well take $|\psi\rangle_{\mu_1}$ to be fixed by the antiunitary complex conjugation map $\Gamma$ on $\K_{\mu_1}.$
For any real number $t,$ we have
\begin{equation}
	\tr_{\H_{\omega_1}}(\rho_{\omega_2} e^{i t \phi[\psi]})
		= \langle \omega_2 | e^{i t \phi[L \psi]} |\omega_2\rangle
		= 1.
\end{equation}
This implies that the support of $\rho_{\omega_2}$ lies within the $+1$ eigenspace of the unitary operator $e^{i t \phi[\psi]}.$
For if we diagonalize the density matrix as $\rho_{\omega_2} = \sum_j p_j |e_j\rangle \langle e_j|,$ then we have
\begin{equation}
	1 = \sum_j p_j \langle e_j | e^{i t \phi[\psi]} |e_j\rangle.
\end{equation}
Unitarity of $e^{i t \phi[\psi]}$ means that all of the expectation values in this expression have norm bounded above by one, and since the $p_j$ coefficients are positive and sum to one, the above equation can only be satisfied if one has $\langle e_j | e^{i t \phi[\psi]} |e_j\rangle = 1$ for each $j,$ which implies $e^{i t \phi[\psi]}|e_j\rangle = |e_j\rangle.$
This establishes the claim that the support of $\rho$ is in the $+1$ eigenspace of each $e^{i t \phi[\psi]}.$

The simultaneous $+1$ eigenspace of each $e^{i t \phi[\psi]},$ for generic real $t,$ is the same as the zero-eigenspace of $\phi[\psi].$
So we know that $\rho_{\omega_2}$ is supported within the zero eigenspace of $\phi[\psi].$
From this we will derive a contradiction, by showing that the zero eigenspace of $\phi[\psi]$ consists only of the zero vector.

To see this, suppose that $|\zeta\rangle \in \H_{\omega_1}$ were in the zero-eigenspace of $\phi[\psi].$
We will decompose the one-particle Hilbert space $\H_{\omega_1}^{(1)}$ into the span of $\phi[\psi] |\omega_1\rangle,$ together with the orthocomplement of this space.
As explained in appendix \ref{app:Haag-duality}, this induces a factorization of $\H_{\omega_1}$ into two tensor factors:
\begin{equation}
	\H_{\omega_1}
		\cong \H_{\psi} \otimes \H_{\perp}.
\end{equation}
The factor $\H_{\psi}$ is built out of states that can be obtained by acting on $|\omega_1\rangle$ with only powers of the field operator $\phi[\psi].$
The other tensor factor is produced by acting on $|\omega_1\rangle$ with field operators that produce one-particle states orthogonal to $\phi[\psi] |\omega_1\rangle.$
Our putative state $|\zeta\rangle$ can be written in terms of this tensor factorization as
\begin{equation}
	|\zeta\rangle
		= \sum_{m,n} c_{m,n} (:\phi[\psi]^m: | \omega_1\rangle) \otimes |e_n\rangle_{\perp},
\end{equation}
where $|e_n\rangle_{\perp}$ is some orthonormal basis for $\H_{\perp}.$
The operator $\phi[\psi]$ acts only on the first tensor factor, and we have
\begin{equation}
	\phi[\psi] |\zeta\rangle
	= \sum_{m,n} c_{m,n} \left( m :\phi[\psi]^{m-1}: | \omega_1\rangle + :\phi[\psi]^{m+1}: | \omega_1 \rangle \right) \otimes |e_n\rangle_{\perp}.
\end{equation}
The condition that $|\zeta\rangle$ is in the kernel of $\phi[\psi]$ gives the recurrence relation
\begin{equation}
	c_{m+1,n} = - \frac{1}{m+1} c_{m-1, n}.
\end{equation}
For fixed $n,$ this can be used to solve for generic $c_{m,n}$ in terms of the two special coefficients $c_{0,n}$ and $c_{1,n},$ as
\begin{equation}
	c_{2m, n}
		= (-1)^m \frac{1}{(2m)!!} c_{0, n}
\end{equation}
and
\begin{equation}
	c_{2m+1, n}
	= (-1)^m \frac{1}{(2m+1)!!} c_{1, n}.
\end{equation}
From this we obtain
\begin{equation}
	|\zeta\rangle
	= \sum_{m,n} (-1)^m \left( \frac{1}{(2m)!!} c_{0,n} :\phi[\psi]^{2m} : | \omega_1\rangle + \frac{1}{(2m+1)!!} c_{1,n} :\phi[\psi]^{2m+1} : | \omega_1\rangle \right) \otimes |e_n\rangle_{\perp}.
\end{equation}
The norm-squared of this state can be computed directly using the fact that the norm-squared of $:\phi[\psi]^k: |\omega_1\rangle$ is $k!.$
One finds
\begin{equation}
	\langle \zeta | \zeta \rangle
		= \sum_{m, n} \left(\frac{(2m)!}{((2m)!!)^2} |c_{0,n}|^2 + \frac{(2m+1)!}{((2m+1)!!)^2} |c_{1,n}|^2\right) 
\end{equation}
The sums over $m$ do not converge,\footnote{For example, Stirling's approximation shows that $(2m!) / ((2m!!)^2)$ scales like $1/\sqrt{m}$ for large $m.$} so existence of $|\zeta\rangle$ is only possible if we have $c_{0,n} = c_{1,n} = 0,$ hence $|\zeta\rangle = 0.$
This establishes that $\phi[\psi]$ cannot have a nontrivial kernel, meaning that the map $L : \K_{\mu_1} \to \K_{\mu_2}$ cannot have a nontrivial kernel, as we claimed above.

To summarize, with the work of this section, we have proven part of the ``necessary'' direction of our main theorem, \ref{thm:general-theorem}.
In particular, that theorem claims that the general conditions for excitability are (i) $\mu_2 \prec \mu_1$, (ii) $\sqrt{X_2} - \sqrt{X_1}$ Hilbert-Schmidt, and (iii) $Q\equiv L^{\dag}L$ has no kernel.
While we have not yet discussed condition (ii), we have already shown that (i) and (iii) are necessary conditions.

\section{Gaussian pure-state excitability}
\label{sec:pure}

Given a free field algebra $\A_0,$ a state $\omega : \A_0 \to \comps$ is said to be pure if, in the GNS representation $\H_{\omega},$ the von Neumann algebra $\A_{\omega}$ has trivial commutant; that is, $\A_{\omega}$ is the full set ${\cal B}(\H_{\omega})$ of all bounded operators on ${\cal H}_{\omega}$.
In this section, we address the question of excitability for two zero-mean Gaussian states $\omega_1$ and $\omega_2$ that are both pure.

This section is essentially a review of the result from \cite{Shale:unitary}, following the arguments due to Wald in \cite{Wald:particle-creation, Wald:s-matrix}.
We make certain modifications to the arguments that allow us to generalize, in the following section, to centrally pure states; also, in section \ref{sec:pure-two-point}, we present a new argument for the equivalence of the conditions from \cite{Shale:unitary, Wald:particle-creation, Wald:s-matrix} with the conditions from Araki and Yamagami in \cite{Araki-Yamagami}.

We begin by assuming $\omega_2 \prec \omega_1$ and deriving a set of necessary conditions; then, at the end of the section, we will turn the argument around and show that these conditions are sufficient.

\subsection{Vectors versus density matrices}

In principle, our assumptions in this section should be that $\omega_1$ and $\omega_2$ are pure states on the free field algebra (which can be taken to be the C$^*$ algebra of Weyl operators), and that there exists a density matrix $\rho_{\omega_2}$ on $\H_{\omega_1}$ that reproduces the correlation functions of $\omega_2.$
In this case, however, it turns out that whenever $\rho_{\omega_2}$ exists, it must actually be rank 1, and so going forward we can write $\rho_{\omega_2} = |\omega_2\rangle \langle \omega_2|,$ with $|\omega_2\rangle$ a vector representative for $\omega_2$ in $\H_{\omega_1}$.\footnote{Note that we are no longer using the ``subscript notation'' $|\omega_2\rangle_1$ from section \ref{sec:excitability}; nowhere in this section will we write down a vector on $\H_{\omega_2},$ so we can simplify write $|\omega_2\rangle$ for the representative in $\H_{\omega_1}$ without confusion.}

To see this, suppose that $\rho_{\omega_2}$ is a nontrivial density matrix, so that we can write it as a convex combination of other density matrices on $\H_{\omega_1}$:
\begin{equation}
	\rho_{\omega_2} = p \sigma + (1 - p) \tau.
\end{equation}
Each of the density matrices $\sigma$ and $\tau$ defines a state on the von Neumann algebra $\A_{\omega_1},$ and induces a state on the C$^*$ subalgebra of Weyl operators.
This allows us to decompose $\omega_2$ as a convex combination of two states on the C$^*$ algebra, but by purity of $\omega_2$, these two states must be equal.
So $\sigma$ and $\tau$ have the same correlation functions on the full Weyl subalgebra, and since the Weyl subalgebra is ultraweakly dense in $\A_{\omega_1},$ we may conclude that $\sigma$ and $\tau$ reproduce the same correlation functions on all of $\A_{\omega_1}.$
Since $\A_{\omega_1}$ is the full algebra of bounded operators $\B(\H_{\omega_1}),$ this implies that $\sigma$ and $\tau$ are equal as density matrices, which is a contradiction.

\subsection{Null state constraints}

Going forward, the idea will be to write down an ansatz for $|\omega_2\rangle$, show that it is uniquely determined, and then to determine the conditions under which the ansatz is well defined.
This basic approach was developed in \cite{Wald:particle-creation, Wald:s-matrix}.
Using the Fock space structure of $\H_{\omega_1}$ explained in appendix \ref{app:fock}, we can decompose the putative state $|\omega_2\rangle$ into its $n$-particle amplitudes,
\begin{equation}
	|\omega_2\rangle
		= \sum_{n=0}^{\infty} |\omega_2^{(n)}\rangle.
\end{equation}
To constrain these amplitudes, we will use that fact that because $\omega_2$ is pure, the Hilbert space $\K_{\mu_2}$ decomposes, as in appendix \ref{app:pure}, into the $\pm i$ eigenspaces of the operator $R_2.$
These two eigenspaces can be mapped into one another via complex conjugation $\ket{\psi}\mapsto\ket{\psi^*}$, and the ``null states'' of $\omega_2$ are the elements of the $-i$ eigenspace.

Concretely, let $L$ be the map from $\K_{\mu_1}$ to $\K_{\mu_2}$ that acts on test functions as
\begin{equation}
	L |f\rangle_{\mu_1} \equiv |f\rangle_{\mu_2}.
    \label{eq:L_again}
\end{equation}
We showed in section \ref{sec:mu-boundedness} that whenever we have $\omega_2 \prec \omega_1,$ this map is bounded.
Moreover, due to the purity of $\omega_1$,  from section \ref{sec:mu-boundedness} we will also have $\omega_1 \prec \omega_2$, and this implies boundedness of $L^{-1}.$

For any state $|\psi\rangle_{\mu_2}$ in the $-i$ eigenspace of $R_2,$ the corresponding operator $\phi[\psi]$ on the GNS space $\H_{\omega_2}$ annihilates the vacuum.
Also, the operator $\phi[L^{-1} \psi],$ which acts on $\H_{\omega_1},$ must annihilate the vector representative of $\omega_2$.
So for every $|\psi\rangle_{\mu_2}$ in the $-i$ eigenspace of $R_2,$ one has
\begin{equation} \label{eq:null-state-constraint}
	\phi[L^{-1} \psi] |\omega_2\rangle = 0.
\end{equation}
Decomposing $\phi[L^{-1} \psi]$ into creation and annihilation operators as in appendix \ref{app:fock}, and separating out the resulting equation into each $n$-particle sector, one finds an infinite set of identities
\begin{align}
	a[L^{-1} \psi] |\omega_2^{(1)}\rangle
		& = 0, \\
	a[L^{-1} \psi] |\omega_2^{(2)}\rangle + a[L^{-1} \psi^*]^{\dagger} |\omega_2^{(0)}\rangle
		& = 0, \\
	a[L^{-1} \psi] |\omega_2^{(3)} \rangle + a[L^{-1} \psi^*]^{\dagger} |\omega_2^{(1)} \rangle
		& = 0, \\
	\dots
		&
\end{align}
It turns out that when $\omega_1$ and $\omega_2$ are pure, these identities completely fix the form of $|\omega_2\rangle.$
We will show this in a few steps.

\subsection{The odd-particle pieces}
\label{sec:pure-odd-particle}

To solve for the state $|\omega_2\rangle$ in the case that $\omega_1$ and $\omega_2$ are both pure, we will begin by showing that equation \eqref{eq:null-state-constraint} implies that $|\omega_2^{(j)}\rangle$ vanishes for odd $j.$

Let us begin with the simplest identity,
\begin{equation} \label{eq:annihilation-vanishing-pure}
	a[L^{-1} \psi] |\omega_2^{(1)}\rangle
	= 0,
\end{equation}
which holds for every $|\psi\rangle_{\mu_2}$ in the $-i$ eigenspace of $R_2.$
Taking the overlap with $|\omega_1\rangle$, we find
\begin{equation}
	0 = \langle a[L^{-1} \psi]^{\dagger} \omega_1 | \omega_2^{(1)} \rangle
		= \langle \phi[L^{-1} \psi^*] \omega_1 | \omega_2^{(1)}\rangle.
\end{equation}
In other words, $|\omega_2^{(1)}\rangle$ is orthogonal to the subspace of the one-particle Hilbert space that can be made by acting with $\phi[L^{-1} \psi^*]$ for any $|\psi\rangle_{\mu_2}$ obeying $(R_2 + i)|\psi\rangle_{\mu_2} = 0.$
We would like to show this actually implies $\ket*{\omega_2^{(1)}}=0$.

As discussed in appendix \ref{app:phase-space}, the map $U : \phi[f] |\omega_1\rangle \mapsto \sqrt{1 - i R_1} |f\rangle_{\mu_1}$ provides a unitary equivalence between $\H_{\omega_1}$ and a subspace of $\K_{\mu_1}.$
In the present, pure setting, this subspace is just the $+i$ eigenspace of $R_1,$ and the map becomes $\phi[f] |\omega_1\rangle \mapsto \sqrt{2} \Pi_{1,i} |f\rangle_{\mu_1}.$
Under this map, our above equation becomes
\begin{equation}\label{eq:density-Linverse-psi*}
	0 = \langle L^{-1} \psi^* | U \omega_2^{(1)}\rangle_{\mu_1},
\end{equation}
where we do not need to write $\Pi_{1,i}$ in the bra because the ket is already within the $+i$ eigenspace of $R_1.$
One has
\begin{equation}
    \begin{aligned}
        0&=\braket*{L^{-1}\psi^*}{U\omega_2^{(1)}}_{\mu_1}\\
        &= -i \braket*{L^{-1}\psi^*}{R_1U\omega_2^{(1)}}_{\mu_1}\\
        &= -i \braket*{\psi^*}{R_2 L U\omega_2^{(1)}}_{\mu_2}\\
        &= i \braket*{R_2 \psi^*}{L U\omega_2^{(1)}}_{\mu_2}\\
        &= \braket*{\psi^*}{LU\omega_2^{(1)}}_{\mu_2}\,,
    \end{aligned}
\end{equation}
where in the third equality we used equation \eqref{eq:R1-from-R2},
that is, $R_1=L^{\dag}R_2L$,
and in the last one $R_2^\dag =-R_2$ and $R_2\ket{\psi^*}_{\mu_2}=i\ket{\psi^*}_{\mu_2}$.
Therefore $LU\ket*{\omega_2^{(1)}}$ has vanishing overlap with all states in the $+i$ eigenspace of $R_2$, and we may conclude
\begin{equation}
	R_2 L U |\omega_2^{(1)}\rangle
		= - i L U |\omega_2^{(1)}\rangle.
\end{equation}
We may now compute
\begin{equation}
	\langle L U \omega_2^{(1)} | L U \omega_2^{(1)}\rangle_{\mu_2}
		= i \langle L U \omega_2^{(1)} | R_2 | L U \omega_2^{(1)}\rangle_{\mu_2}
		= i \langle U \omega_2^{(1)} | R_1 | U \omega_2^{(1)}\rangle_{\mu_1}
		= - \langle U \omega_2^{(1)} | U \omega_2^{(1)}\rangle_{\mu_1}.
\end{equation}
The left-hand side is nonnegative, while the right-hand side is non-positive; so the whole expression must vanish, and we learn $L U |\omega_2^{(1)}\rangle = 0.$
Since $L$ is invertible and $U$ is an isometry, this implies that $|\omega_2^{(1)}\rangle$ must vanish.

Proceeding by induction, we may conclude that $|\omega_2^{(j)}\rangle$ vanishes for all odd $j.$
For example, we have the relation
\begin{equation}
	a[L^{-1} \psi] |\omega_2^{(3)}\rangle + a[L^{-1} \psi^*]^{\dagger} |\omega_2^{(1)}\rangle = 0,
\end{equation}
and since we already know that the one-particle piece vanishes, this is simply
\begin{equation} \label{eq:three-particle-pure}
	a[L^{-1} \psi] |\omega_2^{(3)}\rangle = 0.
\end{equation}
Taking the overlap with a generic two-particle state of the form $|:\phi[g_1] \phi[g_2]: \omega_1\rangle$, we may write this as
\begin{equation}
\label{eq:three-particle-pure2}
	0 = \langle a[L^{-1} \psi]^{\dagger} :\phi[g_1] \phi[g_2]: \omega_1 |\omega_2^{(3)}\rangle \quad \text{whenever} \quad (R_2 + i)|\psi\rangle_{\mu_2} = 0.
\end{equation}
As reviewed in appendix \ref{app:fock}, the $n$-particle spaces of $\H_{\omega_1}$ can be thought of as symmetric tensor powers of the one-particle Hilbert space.
Our calculation in the preceding paragraph shows that states of the form $a[L^{-1} \psi]^{\dagger} |\omega_1\rangle$ are dense in the one-particle Hilbert space.
Obviously the states $:\phi[g]: |\omega_1\rangle$ are dense in this space as well.
So equation \eqref{eq:three-particle-pure2} gives a vanishing overlap of $|\omega_2^{(3)}\rangle$ with a dense set of states in $\H_{\omega_1}^{(3)},$ implying that $|\omega_2^{(3)}\rangle$ vanishes.
The same basic argument holds for arbitrary odd-particle-number amplitudes.

\subsection{The two-particle amplitude}
\label{sec:two-part-ampl}

Now we turn to examine the constraint on the two-point amplitude,
\begin{equation}
	0 = a[L^{-1} \psi] |\omega_2^{(2)}\rangle + a[L^{-1} \psi^*]^{\dagger} |\omega_2^{(0)} \rangle \quad \text{whenever} \quad (R_2 + i)|\psi\rangle_{\mu_2} = 0.
\end{equation}
Since $|\omega_2^{(0)}\rangle$ must be proportional to $|\omega_1\rangle,$ we will momentarily re-normalize the state $|\omega_2\rangle$ so that we have $|\omega_2^{(0)}\rangle = |\omega_1\rangle.$
Of course we will have to change normalization again at the end to guarantee that $|\omega_2\rangle$ actually reproduces the correlation functions of $\omega_2,$ but this will be an easy fix.
Our equation now becomes
\begin{equation} \label{eq:two-particle-pure}
	a[L^{-1} \psi] |\omega_2^{(2)}\rangle = - \phi[L^{-1} \psi] |\omega_1\rangle 
\end{equation}
for every $|\psi \rangle_{\mu_2}$ in the $-i$ eigenspace of $R_2.$
Taking the overlap with a generic one-particle state $\phi[g] |\omega_1\rangle,$ we have
\begin{equation} \label{eq:two-particle-overlap}
	\langle :\phi[L^{-1} \psi^*] \phi[g]: \omega_1 |\omega_2^{(2)}\rangle = - \langle \phi[g] \omega_1 | \phi[L^{-1} \psi] \omega_1\rangle.
\end{equation}
As discussed in the previous subsection, the states $\phi[L^{-1} \psi^*] |\omega_1\rangle$ are dense in the one-particle Hilbert space, so the states $:\phi[L^{-1} \psi^*] \phi[g]: |\omega_1\rangle$ are dense in the two-particle Hilbert space.
Consequently, we see that the above equation completely determines $|\omega_2^{(2)}\rangle.$
To get a more explicit expression, we make use of (i) the interpretation of the two-particle Hilbert space as a symmetric tensor power of the one-particle Hilbert space; and (ii) the unitary equivalence $U$ between the one-particle Hilbert space $\H_{\omega}^{(1)}$ and the $+i$ eigenspace of $R_1$.
Using (i), the above equation can be rewritten as
\begin{equation} 
	\sqrt{2} \langle \phi[L^{-1} \psi^*] \omega_1 \otimes_S \phi[g] \omega_1 | V \omega_2^{(2)}\rangle = - \langle \phi[g] \omega_1 | \phi[L^{-1} \psi] \omega_1\rangle,
\end{equation}
where $V$ is the inverse of the map from equation \eqref{eq:symmetric-power-Fock}.
Using (ii), we can then rewrite this as
\begin{equation} \label{eq:two-particle-overlap-mu}
	\sqrt{2} \langle L^{-1} \psi^* \otimes g | (U \otimes U) V \omega_2^{(2)}\rangle_{\mu_1} = - \langle \Pi_{1, i} g |\Pi_{1, i} L^{-1} \psi \rangle_{\mu_1}.
\end{equation}
(Note that we did not need to explicitly write the symmetrization of the tensor product in the bra on the left-hand side, because taking the inner product with a symmetrized ket automatically projects to the symmetrized subspace.)

There is a conceptually useful way to rewrite equation \eqref{eq:two-particle-overlap-mu}, which is to use an identification between vectors in the tensor product of two copies of the $+i$ eigenspace of $R_1,$ and Hilbert-Schmidt operators that map the $-i$ eigenspace of $R_1$ to the $+i$ eigenspace of $R_1.$
Concretely, given any orthonormal basis $|e_j\rangle_{\mu_1}$  for the $+i$ eigenspace of $R_1,$ the conjugated states $\Gamma |e_j\rangle_{\mu_1}$ form an orthonormal basis for the $-i$ eigenspace.
The Hilbert-Schmidt operators of interest can hence be expanded as $\sum_{j, k} c_{j k} |e_j\rangle_{\mu_1}\langle \Gamma e_k |,$ and can be identified with the normalizable state $\sum_{j k}c_{jk} |e_k \otimes e_j\rangle_{\mu_1}.$
Using this ansatz for $|(U \otimes U) V \omega_2^{(2)} \rangle_{\mu_1}$, we can rewrite equation \eqref{eq:two-particle-overlap-mu} as
\begin{equation}
	\sqrt{2} \sum_{jk} c_{jk} \langle \Pi_{1, i} L^{-1} \psi^* | e_k \rangle_{\mu_1} \langle g | e_j\rangle_{\mu_1} = - \langle g | \Pi_{1,i} L^{-1} \psi\rangle_{\mu_1},
\end{equation}
and using antiunitarity of $\Gamma$ together with the fact that it commutes with $L$, we may rewrite this as
\begin{equation}
\label{eq:state_to_operator}
	\sqrt{2} \sum_{jk} c_{jk} \langle g | e_j\rangle_{\mu_1} \langle \Gamma e_k | \Pi_{1, - i} L^{-1} \psi\rangle_{\mu_1} = - \langle g | \Pi_{1,i} L^{-1} \psi\rangle_{\mu_1},
\end{equation}
So our state $(U \otimes U) V |\omega_2^{(2)}\rangle$ may be identified up to a multiplicative constant with the operator that takes states of the form $\Pi_{1,-i }L^{-1} |\psi\rangle_{\mu_2}$ to $\Pi_{1,i} L^{-1} |\psi\rangle_{\mu_2},$ for any $\psi$ in the $-i$ eigenspace of $R_2.$

Normalizability of the representative’s two-particle amplitude requires this operator to be Hilbert–Schmidt.
To better understand its structure, we note that the projector $\Pi_{1,-i}$ is actually invertible on the domain of states of the form $L^{-1}|\psi\rangle_{\mu_2}.$
Using equation \eqref{eq:R1-from-R2} ($R_1=L^{\dag} R_2 L$), one can write
\begin{align} \label{eq:1+Q-2}
	\begin{split}
		\Pi_{1,-i} L^{-1} |\psi\rangle_{\mu_2}
			& = \frac{i - R_1}{2 i} L^{-1} |\psi\rangle_{\mu_2} \\
			& = \frac{1}{2} L^{-1} |\psi\rangle_{\mu_2}  - \frac{1}{2i} L^{\dagger} R_2 |\psi\rangle_{\mu_2}  \\
			& = \frac{1}{2} L^{-1} |\psi\rangle_{\mu_2}  + \frac{1}{2} L^{\dagger} |\psi\rangle_{\mu_2}  \\
			& = \frac{1 + L^{\dagger} L}{2} L^{-1} |\psi\rangle_{\mu_2} .
	\end{split}
\end{align}
So the operator $((1 + L^{\dagger} L)/2)^{-1}$ provides an inverse for $\Pi_{1,-i}$ on the relevant domain, and the map $\Pi_{1,-i} L^{-1} |\psi\rangle_{\mu_2}  \to \Pi_{1,i} L^{-1} |\psi\rangle_{\mu_2} $ may be written as
\begin{equation}
	\Pi_{1,i} \left( \frac{1 + L^{\dagger} L}{2} \right)^{-1} : \Pi_{1,-i} L^{-1} \K_{\mu_2, -i} \to \Pi_{1,i} L^{-1} \K_{\mu_2, -i},
\end{equation}
where we have introduced the notation $\K_{\mu_2, -i}$ for the $-i$ eigenspace of $R_2,$ in order to make the domain of the operator transparent.

It is a general fact that an operator of the form $A X,$ with $X$ invertible, is Hilbert-Schmidt if and only if $A$ itself is Hilbert-Schmidt.
The operator $((1 + L^{\dagger} L)/2)^{-1}$ is an invertible map from $\Pi_{1,-i} L^{-1} \K_{\mu_2, -i}$ to $L^{-1} \K_{\mu_2, -i}$, so one can drop it from the above expression and conclude that we want is the Hilbert-Schmidt condition for the map
\begin{equation}
	\Pi_{1,i} : L^{-1} \K_{\mu_2, -i} \to \Pi_{1,i} L^{-1} \K_{\mu_2, -i},
\end{equation}
Since $L^{-1}$ is invertible, we may equivalently look for the Hilbert-Schmidt condition on 
\begin{equation}
	\Pi_{1,i} L^{-1} :  \K_{\mu_2, -i} \to \Pi_{1,i} L^{-1} \K_{\mu_2, -i},
\end{equation}
Finally, since $\K_{\mu_2, -i}$ is just the space obtained by acting with $\Pi_{2,-i}$ on $\K_{\mu_2},$ we may equivalently phrase our new necessary condition as the statement that that $\Pi_{1,i} L^{-1} \Pi_{2,-i}$ is Hilbert-Schmidt as a map from $\K_{\mu_2}$ to $\K_{\mu_1}$.

We will now see that this condition is sufficient.

\subsection{Sufficiency}
\label{sec:pure-sufficiency}

So far we have determined that when $\omega_1$ and $\omega_2$ are both pure, and there is a representative for $\omega_2$ in $\H_{\omega_1}$, then (i) the inner product $\mu_2$ must be bounded by the inner product $\mu_1$; and (ii) the operator $\Pi_{1,i} L^{-1} \Pi_{2,-i}$ must be Hilbert-Schmidt as a map from $\K_{\mu_2}$ to $\K_{\mu_1}$.
This last condition, while it sounds technical, is really just the normalizability condition for the two-particle amplitude 
$\ket*{\omega_2^{(2)}}$,
which we constructed in the previous section.
We will now look at the higher-particle amplitudes of the ansatz for $|\omega_2\rangle,$ show that they are all uniquely determined in terms of $\ket*{\omega_2^{(2)}}$, and show that the full ansatz is automatically normalizable as long as $\ket*{\omega_2^{(2)}}$ is normalizable.
Then we will show that this ansatz genuinely reproduces all correlation functions of $\omega_2$; this completes the proof that conditions (i) and (ii) are sufficient, since 
they can be used to construct the state $|\omega_2\rangle$ explicitly.

We have already shown that all odd-particle-number amplitudes for $|\omega_2\rangle$ must vanish.
The even-particle amplitudes satisfy
\begin{equation} \label{eq:even-particle-pure-constraint}
	a[L^{-1} \psi] |\omega_2^{(2 n)}\rangle + a[L^{-1} \psi^*]^{\dagger} |\omega_2^{(2n-2)}\rangle = 0
\end{equation}
for every $|\psi\rangle_{\mu_2}$ in the $-i$ eigenspace of $R_2.$
This equation can be applied inductively --- that is, starting with the $n=2$ case --- to show that all of the $n$-particle amplitudes can be expressed in terms of $\ket*{\omega_2^{(2)}}$ alone.
We give this computation in appendix \ref{app:higher-particle-amplitudes}; the final answer is
\begin{equation}
\label{eq:high_from_low}
	|(U \otimes \dots \otimes U) V \omega_2^{(2n)}\rangle_{\mu_1}
	= \frac{1}{n!} \sqrt{\frac{(2n)!}{2^n}} P_{\text{sym}}|((U \otimes U) V \omega_2^{(2)})^{\otimes n}\rangle_{\mu_1},
\end{equation}
where $V$ is the map from the GNS space $\H_{\omega_1}$ into a symmetric Fock space, and $U$ is the isometry from $\H_{\omega_1}^{(1)}$ into $\K_{\mu_1}.$

Each of these $(2n)$-particle amplitudes, taken individually, is normalizable solely from the assumption that $|\omega_2^{(2)}\rangle$ is normalizable.
But does the assumption that the full state $|\omega_2\rangle$ is normalizable introduce any additional necessary conditions?
To check this, we must compute the sum
\begin{equation}
	\sum_{n}
		\frac{1}{(n!)^2} \frac{(2n)!}{2^n} \langle ((U \otimes U) V \omega_2^{(2)})^{\otimes n} | P_{\text{sym}}|((U \otimes U) V \omega_2^{(2)})^{\otimes n}\rangle_{\mu_1}
\end{equation}
and study the conditions for it to be finite.
This is a simple but somewhat notation-dense combinatorial calculation, so we present it in appendix \ref{app:ansatz-normalizability}.
The end result of that calculation is that no new necessary condition is introduced.

So by studying the ansatz for $|\omega_2\rangle$ and its normalizability, we found a simple necessary condition, which came from normalizability of the two-particle amplitude.\footnote{A natural question to ask is whether non-normalizability of the ansatz, which would forbid $\ket{\omega_2}$ from existing as a state in ${\cal H}_{\omega_1}$, is still compatible with $\langle\omega_2|...|\omega_2\rangle$ defining a valid \textit{weight}.
The answer is no.
For as long as $\langle\omega_2|...|\omega_2\rangle$ spits out finite answers for two-point functions of field operators $\phi[f]\phi[g]$, it must also do so for $\phi[f]\phi[g]-\phi[g]\phi[f]$, which in a free $\ast$-algebra is proportional to the identity.
If $\ket{\omega_2}$ is non-normalizable, then $\langle\omega_2|1|\omega_2\rangle$ diverges.
}
By reversing the logic of this argument, we can actually show that this condition --- together with the first necessary condition $\mu_2 \prec \mu_1$ found in subsection \ref{sec:mu-boundedness}
--- is sufficient.

To do so, we first note that the condition $\mu_2\prec \mu_1$ is sufficient to show that $L^{-1}$ is a bounded operator, which must be true in order for us to make sense of the normalizability on $\Pi_{1,i} L^{-1} \Pi_{2,-i}$.
For suppose that $\omega_1$ and $\omega_2$ are both pure, and that we have $\mu_2 \prec \mu_1.$
Then we automatically have $\mu_1 \prec \mu_2$,\footnote{In fact this only requires assuming that $\omega_1$ is pure.} for in this case $R_1$ is unitary, so for any $\ket{\eta}_{\mu_1} \in {\cal K}_{\mu_1}$ we have
\begin{equation} \label{eq:L-pure-lower-bound}
	\lVert |\eta\rangle \rVert_{\mu_1}
		= \lVert R_1 |\eta\rangle \rVert_{\mu_1}
		= \lVert L^{\dagger} R_2 L |\eta \rangle \rVert_{\mu_1}
		\leq \lVert L^{\dagger} \rVert \lVert R_2 \rVert \lVert L |\eta\rangle \rVert_{\mu_2},
\end{equation}
where in the second step we used the identity $L^{\dagger} R_2 L = R_1$ from equation \eqref{eq:R1-from-R2}.
We can rearrange this into the inequality
\begin{equation}
\label{eq:L-pure-lower-bound-clean}
	\lVert L |\psi\rangle \rVert_{\mu_2} \geq \frac{1}{\lVert L^{\dagger} \rVert \lVert R_2 \rVert} \lVert |\psi\rangle \rVert_{\mu_1},
\end{equation}
so $L$ is lower-bounded, hence invertible.\footnote{Note that invertibility of $L$ is exactly the condition $\mu_1 \prec \mu_2,$ so we do not need to make any extra assumption of this kind.}

Turning now to the normalizability condition, we suppose that $\Pi_{1,i} L^{-1} \Pi_{2,-i}$ is Hilbert-Schmidt as a map from $\K_{\mu_2}$ to $\K_{\mu_1}.$
By repeating the argument of the preceding subsection, we may conclude that the map
\begin{equation}
	\Pi_{1,-i} L^{-1} |f\rangle_{\mu_2}
		\mapsto \Pi_{1,i} L^{-1} |f\rangle_{\mu_2}
\end{equation}
is Hilbert-Schmidt as a map from the $-i$ eigenspace of $R_1$ to the $+i$ eigenspace.
Consequently, applying similar reasoning as we did for equation \eqref{eq:two-particle-overlap-mu}, there exists a state $|\chi\rangle_{\mu_1}$ in two copies of the $+i$ eigenspace that satisfies the equation
\begin{equation}
	\sqrt{2} \langle L^{-1} \psi^* \otimes g | \chi \rangle_{\mu_1} = - \langle \Pi_{1,i} g | \Pi_{1,i} L^{-1} \psi\rangle_{\mu_1}
\end{equation}
for every $\ket{g}_{\mu_2}$ in $\K_{\mu_2}$ and for every $\ket{\psi}_{\mu_2}$ in the $-i$ eigenspace of $R_2.$
The state $|\chi\rangle_{\mu_1}$ can be shown explicitly to be symmetric --- though this involves a computation, see appendix \ref{app:ansatz-symmetry} --- so it is in the image of $(U \otimes U) V,$ and we can \textit{define} the two-particle amplitude $|\omega_2^{(2)}\rangle$ via the equation
\begin{equation}
	|\chi \rangle_{\mu_1} = (U \otimes U) V |\omega_2^{(2)}\rangle.
\end{equation}
This state automatically satisfies the two-particle null-state constraint.
The state built by summing over symmetrized tensor powers of this amplitude, weighted by $\frac{1}{n!} \sqrt{\frac{(2n)!}{2^n}}$ as in equation \eqref{eq:high_from_low}, is normalizable thanks to the calculation from appendix \ref{app:ansatz-normalizability}, and it automatically satisfies all other null-state constraints.
We are left with a state in $\H_{\omega_1}$ that satisfies all the $\omega_2$ null-state constraints; by normalizing it, we end up with a candidate for $|\omega_2\rangle.$

Now all that remains is to show that our ansatz $|\omega_2\rangle$ actually represents the abstract state $\omega_2$ in $\H_{\omega_1}.$
In other words, we must show that what we know about $|\omega_2\rangle$ --- that it is normalized and satisfies all null state constraints --- implies that it reproduces all the $\omega_2$ correlation functions.

To do this, one first shows that the state $|\omega_2\rangle$ is in the domain of all field polynomials $\phi[f_1] \dots \phi[f_m].$
This is a simple calculation but it involves a little bit of unbounded operator terminology, so we present it in appendix \ref{app:ansatz-domain}.
Once one knows this, the goal is to show
\begin{equation}
	\langle \omega_2 | \phi[f_1] \dots \phi[f_m] |\omega_2\rangle
	= \omega_2(\phi[f_1] \dots \phi[f_m]).
\end{equation}
We will now reintroduce a piece of notation we used in section \ref{sec:excitability}: we will pass to the $\omega_2$ GNS representation on the right-hand side, and start distinguishing kets in the two different GNS representations by writing $|\cdot\rangle_{j}.$
So the state $|\omega_2\rangle_{1}$ represents the putative representative of $\omega_2$ in $\H_{\omega_1},$ while $|\omega_2\rangle_{2}$ simply represents the GNS vector in $\H_{\omega_2}.$

The equation we want to show is
\begin{equation}
	\langle \omega_2 | \phi[f_1] \dots \phi[f_m] |\omega_2\rangle_{1}
	= \langle \omega_2 | \phi[f_1] \dots \phi[f_m] |\omega_2\rangle_{2}
\end{equation}
for any test functions $f_1, \dots, f_m.$
Within the space $\K_{\mu_2}$, we can split each $f_j$ as $f_{j, i} + f_{j, -i},$ where $f_{j, \pm i}$ is obtained by projecting $|f_j\rangle_{\mu_2}$ onto the $\pm i$ eigenspace of $R_2.$
We can then write the identity we want to show as
\begin{align} \label{eq:pure-correlators-desired}
	& \langle \omega_2 | (\phi[L^{-1} f_{1,i}] + \phi[L^{-1} f_{1, -i}]) \dots (\phi[L^{-1} f_{m, i}] + \phi[L^{-1} f_{m, -i}]) |\omega_2\rangle_{1} \\
	& \qquad \qquad = \langle \omega_2 | (\phi[f_{1,i}] + \phi[f_{1, -i}]) \dots (\phi[f_{m, i}] + \phi[f_{m, -i}]) |\omega_2\rangle_{2}
\end{align}

Let us start with the case $m=0.$
Since we chose $|\omega_2\rangle_{1}$ to be a normalized state, we have guaranteed that this equation is satisfied.
Proceeding to $m=1,$ for a single field $\phi[f],$ we can use the fact that $|\omega_2\rangle_{1}$ is annihilated by $\phi[L^{-1} f_{-i}]$, and $\langle \omega_2 |_{1}$ is annihilated by $\phi[L^{-1} f_{i}]$.
This means that the one-point functions vanish on both sides of equation \eqref{eq:pure-correlators-desired}.
For $m=2,$ using similar logic, our desired identity (with $f_1 =f$ and $f_2=g$) simplifies to
\begin{equation}
	\langle \omega_2 | \phi[L^{-1} f_{-i}] \phi[L^{-1} g_{i}] |\omega_2\rangle_{1}
	= \langle \omega_2 | \phi[f_{-i}] \phi[g_{i}] |\omega_2\rangle_{2}.
\end{equation}
But we can freely replace each pair of fields in the above expression by its commutator, since $\phi[L^{-1} f_{-i}]$ annihilates $|\omega_2\rangle_{1}$ and $\phi[f_{-i}]$
annihilates $|\omega_2 \rangle_{2}.$
This reduces the equation we are trying to check to
\begin{equation} \label{eq:Gamma-commutator-R}
	\langle \Gamma L^{-1} f_{-i} | R_1 | L^{-1} g_{i} \rangle_{\mu_1}
	= \langle \Gamma f_{-i} | R_2 | g_{i} \rangle_{\mu_2}.
\end{equation}
But this equality follows from the identity $R_1 = L^{\dagger} R_2 L,$ which we verified in equation \eqref{eq:R1-from-R2}.

We can now proceed by induction.
For odd $m,$ our inductive assumption will be that the $k$-point correlators vanish for odd $k < m.$
This allows us to commute terms freely, since replacing two fields by their commutator brings us to an odd-point function of lower degree, which vanishes by assumption.
So we can freely commute all of the $f_{k, -i}$ fields to the right and all of the $f_{k, i}$ fields to the left, leading to zero on both sides of equation \eqref{eq:pure-correlators-desired} in the case of odd $m$.
For the even-point correlators, the identity we want to show is
\begin{equation}
	\langle \omega_2 | \phi[L^{-1} f_{1, -i}] \dots \phi[L^{-1} f_{m, i}] |\omega_2\rangle_{1}
	= \langle \omega_2 | \phi[f_{1, -i}] \dots \phi[f_{m, i}] |\omega_2\rangle_{2}.
\end{equation}
On both sides of the equation, we can commute the $\phi[f_{m, i}]$ through other fields at the cost of introducing a commutator, but each of these is a product of a commutator (which agree on both sides as in equation \eqref{eq:Gamma-commutator-R}) and an even-point function of lower degree (which agree on both sides by the inductive assumption).
So we only need to worry about the term where the $f_{m,i}$ field is moved all the way to the front of both expressions --- but these terms both vanish, since this field annihilates both bras $\langle \omega_2|_{1}$ and $\langle \omega_2|_{2}.$

\subsection{Connection to the two-point function}
\label{sec:pure-two-point}

In this section we have established, following \cite{Wald:particle-creation, Wald:s-matrix}, that when $\omega_1$ and $\omega_2$ are both pure, then $\omega_2$ has a representative in $\H_{\omega_1}$ if and only if (i) the symmetrized inner product $\mu_2$ is bounded by the symmetrized inner product $\mu_1,$ and (ii) within the space $\K_{\mu_1},$ the operator $\Pi_{1,i} L^{-1} \Pi_{2,-i}$ is Hilbert-Schmidt.
To make contact with the theorem we will prove in the general case of a mixed state, it is helpful to rewrite condition (ii) in terms of information about the two-point functions induced by $\omega_1$ and $\omega_2.$

First note that the two-point function induced by $\omega_1$ is represented by a bounded operator on $\K_{\mu_1}$:
\begin{equation}
	\omega_1(\phi[f]^* \phi[g])
		= \langle f | g \rangle_{\mu_1} -  \frac{i}{2}\Omega[f^*, g]
		= \langle f | (1 - i R_1) |g\rangle_{\mu_1}.
\end{equation}
Similarly, for $\omega_2,$ using $R_1 = L^{\dagger} R_2 L,$ we have
\begin{equation}
	\omega_2(\phi[f]^* \phi[g])	
		= \langle f | (1 - i R_2) |g \rangle_{\mu_2}
		= \langle f | (L^{\dagger} L - i R_1) |g \rangle_{\mu_1}.
\end{equation}
We write $Q$ for the operator $L^{\dagger} L$ that represents the $\mu_2$ inner product in $\mu_1$ (see again equation \eqref{eq:L-polar}) and define
\begin{align}
	X_1
		& \equiv 1 - i R_1, \\
	X_2
		& \equiv Q - i R_1.
\end{align}
We will now show that the normalizability criterion that $\Pi_{1,i} L^{-1} \Pi_{2,-i}$ is Hilbert-Schmidt is equivalent to the condition that $X_1-X_2$ is Hilbert-Schmidt.

Writing out the projectors $\Pi_{1,i}$ and $\Pi_{2,-i}$ explicitly in terms of $R_1$ and $R_2,$ we have
\begin{equation}
	\Pi_{1, i} L^{-1} \Pi_{2, -i}
		= \frac{i + R_1}{2i} L^{-1} \frac{i - R_2}{2i}.
\end{equation}
Using again the identity $R_1 = L^{\dagger} R_2 L,$ together with $(R_2)^2 = -1,$ we may rewrite this as
\begin{equation}
	\Pi_{1, i} L^{-1} \Pi_{2, -i}
	= \frac{i L^{-1} + L^{\dagger} R_2}{2i} \frac{i - R_2}{2i}
	= \frac{1}{4} \left(L^{-1} + i L^{-1} R_2 - i L^{\dagger} R_2 - L^{\dagger} \right)
\end{equation}
Since $L$ is invertible, the operator $\Pi_{1, i} L^{-1} \Pi_{2, -i}$ is Hilbert-Schmidt if and only if the operator $4 \Pi_{1, i} L^{-1} \Pi_{2, -i} L$ is Hilbert-Schmidt.
So we may investigate the Hilbert-Schmidt condition on the operator
\begin{equation}
	4 \Pi_{1, i} L^{-1} \Pi_{2, -i} L 
		= 1 + i L^{-1} R_2 L - i R_1 - Q
\end{equation}
Inserting $(L^{\dagger})^{-1} L^{\dagger}$ in the second term, we may rewrite this as
\begin{equation}
	4 \Pi_{1, i} L^{-1} \Pi_{2, -i} L 
	= 1 + i Q^{-1} R_1 - i R_1 - Q
	= (1 - Q) - i (1 - Q^{-1}) R_1.
\end{equation}
This operator is Hilbert-Schmidt if and only if its conjugation by $\Gamma$ is Hilbert-Schmidt.
Because conjugation by $\Gamma$ fixes $Q$ and changes the sign of $i R_1,$ one may add and subtract this expression from its $\Gamma$-conjugate and equivalently check that $Q - 1$ and $(Q^{-1} - 1) R_1$ are Hilbert-Schmidt.
But $R_1$ is invertible, so this is really just the condition that $Q - 1$ and $Q^{-1} - 1$ are Hilbert-Schmidt.
Moreover, since we are assuming that the $\omega_i$ are pure, and $\mu_2\prec \mu_1$, we have that $Q$ is bounded and invertible (see equation \eqref{eq:L-pure-lower-bound-clean} above).
Because $Q$ is invertible, $Q^{-1} - 1$ is Hilbert-Schmidt if and only if $Q - 1 = - Q (Q^{-1} - 1)$ is Hilbert-Schmidt.
So we conclude that the necessary and sufficient condition for excitability is (i) $\mu_2 \prec \mu_1$ and (ii) $Q - 1$, equivalently $X_2 - X_1,$ is Hilbert-Schmidt.
This was claimed in section \ref{sec:summary} as theorem \ref{thm:both-pure}.

To make a connection with the criterion for mixed states claimed in theorem \ref{thm:general-theorem}, we will now show that for pure states, the condition that $X_2 - X_1$ is Hilbert-Schmidt is the same as the condition that $\sqrt{X_2} - \sqrt{X_1}$ is Hilbert-Schmidt.
The second statement clearly implies the first, since the Hilbert-Schmidt property is preserved under multiplication by bounded operators, and one has
\begin{equation}
	X_2 - X_1 \propto (\sqrt{X_2} - \sqrt{X_1}) (\sqrt{X_2} + \sqrt{X_1}) + (\sqrt{X_2} + \sqrt{X_1}) (\sqrt{X_2} - \sqrt{X_1}).
\end{equation}
To see that the first condition implies the second, we will make our first use of a technique that will be employed many times in the remainder of the paper, which is \textit{Birman-Solomyak theory}.
The basic idea of Birman Solomyak theory is to tell you when, given a function $f$ and a pair of positive operators $A$ and $B,$ one conclude that $f(A) - f(B)$ is in the $p$-Schatten-class whenever $A - B$ is in that class.
The key theory was developed by Birman and Solomyak in the 1960s; see \cite{Skripka:book} for a comprehensive textbook account.
For our purposes, we only need the Hilbert-Schmidt version of the theory, for which one can find nice, simple proofs in \cite{Bach:revisited}.

This is the form of the theorem we will use, though it admits generalizations:
\begin{theorem}[Birman-Solomyak]
	Let $A$ and $B$ be bounded, self-adjoint operators on a Hilbert space $\H$.
	Suppose that $f : \reals \to \reals$ is Lipschitz on the union of the spectra of $A$ and $B$.
	Then whenever $A - B$ is Hilbert-Schmidt, so is $f(A) - f(B).$
\end{theorem}
A function is said to be ``Lipschitz'' on a domain if all of its tangent lines have uniformly bounded slope.
Concretely, if one has
\begin{equation}
	\sup_{x, y} \frac{|f(x) - f(y)|}{|x - y|} < \infty.
\end{equation}
Intuitively --- though this formula is not literally correct unless $A$ and $B$ commute --- one should imagine an expression like
\begin{equation}
	f(A) - f(B) \sim \frac{f(A) - f(B)}{A - B} (A - B),
\end{equation}
with the Lipschitz condition guaranteeing that the first factor is bounded, so whenever $A-B$ is Hilbert-Schmidt, so is $f(A) - f(B)$.

In particular, the square root function is Lipschitz on any domain of the positive real numbers that is bounded away from zero.
Since $Q$ is positive and invertible, its spectrum is contained in such a domain.
So using the Birman-Solomyak theorem, and the fact that we are assuming $Q - 1$ is Hilbert-Schmidt, we may conclude that $Q^{1/2} - 1$ is also Hilbert-Schmidt.
Then we look at the expression
\begin{equation}
	\sqrt{X_2} - \sqrt{X_1}
		= \sqrt{Q - i R_1} - \sqrt{1 - i R_1}
		= \sqrt{Q^{1/2} (1 - i Q^{-1/2} L^{\dagger} R_2 L Q^{-1/2}) Q^{1/2}} - \sqrt{1 - i R_1}.
\end{equation}
Let us write the polar decomposition of $L$ as $L = v Q^{1/2}.$
Because $L$ is invertible, the operator $v$ is unitary.
We may rewrite the above as
\begin{equation}
	\sqrt{X_2} - \sqrt{X_1}
	= \sqrt{Q^{1/2} v^{\dagger} (1 - i R_2) v Q^{1/2}} - \sqrt{1 - i R_1}.
    \label{eq:x2x1_toomanyQs}
\end{equation}
In order to simplify the first factor on the right hand side of equation \eqref{eq:x2x1_toomanyQs}, we can consider the expression
\begin{equation}
	\sqrt{Q^{1/2} v^{\dagger} (1 - i R_2) v Q^{1/2}}
		= \left|\sqrt{v^{\dagger}(1 - i R_2) v} Q^{1/2}\right|.
\end{equation}
The absolute value function is Lipschitz.
We also know that
\begin{equation}
	\sqrt{v^{\dagger}(1 - i R_2) v} Q^{1/2} - \sqrt{v^{\dagger}(1 - i R_2) v}
		= \sqrt{v^{\dagger}(1 - i R_2) v} (Q^{1/2} - 1)
\end{equation}
is Hilbert-Schmidt, since $Q^{1/2} - 1$ is Hilbert-Schmidt.
So applying the Birman-Solomyak theorem again, with the absolute value function, we may conclude
\begin{equation}
	\sqrt{Q^{1/2} v^{\dagger} (1 - i R_2) v Q^{1/2}} - \sqrt{v^{\dagger} (1 - i R_2) v}\quad \text{is Hilbert-Schmidt}.
\end{equation}
Now revisiting equation \eqref{eq:x2x1_toomanyQs},
this gives the Hilbert-Schmidt equivalence\footnote{Given two bounded operators $A,B$, we say that $A$ is Hilbert–Schmidt equivalent to $B$, and write $A\sim_{\mathrm{H.S.}}\!B$, if $A-B$ is Hilbert–Schmidt.
Note that $\sim_{\mathrm{H.S.}}$ defines an equivalence relation. Moreover, if $A\sim_{\mathrm{H.S.}}\!B$, then $AO\sim_{\mathrm{H.S.}}\!BO$ for any bounded operator $O$. }

\begin{equation}
	\sqrt{X_2} - \sqrt{X_1}
	\sim_{\text{H.S.}} \sqrt{v^{\dagger} (1 - i R_2) v} - \sqrt{1 - i R_1}
	= v^{\dagger} \sqrt{1 - i R_2} v - \sqrt{1 - i R_1}.
\end{equation}
Because $R_1$ and $R_2$ both square to minus the identity, the operator $1 - i R_j$ is simply $2 \Pi_{j, i}.$
The square root of an orthogonal projector is itself, so we can compute the above explicitly as
\begin{align}
	\begin{split}
	\sqrt{X_2} - \sqrt{X_1}
	& \sim_{\text{H.S.}} \sqrt{2} (v^{\dagger} \Pi_{2, i} v - \Pi_{1, i}) \\
	& = \frac{1}{\sqrt{2}} (v^{\dagger} (1 - i R_2) v - (1 - i R_1)) \\
	& = - i \frac{1}{\sqrt{2}} (v^{\dagger} R_2 v - R_1).
	\end{split}
\end{align}
Finally, because we have $L = v Q^{1/2}$ and $Q^{1/2} - 1$ is Hilbert-Schmidt, we have $v \sim_{\text{H.S.}} L,$ and we obtain
\begin{align}
	\begin{split}
		\sqrt{X_2} - \sqrt{X_1}
		& \sim_{\text{H.S.}} - i \frac{1}{\sqrt{2}} (L^{\dagger} R_2 L - R_1)
		= 0,
	\end{split}
\end{align}
where we have used the identity $L^{\dagger} R_2 L = R_1.$
This completes the proof that when $\omega_1$ and $\omega_2$ are both pure, the condition that $X_2 - X_1$ be Hilbert-Schmidt implies that $\sqrt{X_2} - \sqrt{X_1}$ is Hilbert-Schmidt.

Note that we could not have applied the Birman-Solomyak theorem directly to $X_2 - X_1$ using the square root function, since the operators $X_2$ and $X_1$ can contain zero in their spectra, and the square root function is not Lipschitz on domains containing zero.
We give an example in appendix \ref{app:counterexample-1} of a setting where $\omega_1$ is pure, $\omega_2$ is not pure, and one has $X_2 - X_1$ Hilbert-Schmidt but not $\sqrt{X_2} - \sqrt{X_1}$ Hilbert-Schmidt.
In this case $\omega_2$ cannot be excited out of $\omega_1,$ which is consistent with the theorem we prove for factorial states in the next section.

\section{Gaussian excitability via canonical purifications}
\label{sec:passing-to-canonical}

In the preceding section, we proved theorem \ref{thm:both-pure} --- that for zero-mean Gaussian states, when $\omega_1$ and $\omega_2$ are both pure, the necessary and sufficient criteria for $\omega_2 \prec \omega_1$ are (i) $\mu_2 \prec \mu_1$ and (ii) $Q - 1$ Hilbert-Schmidt.
We also showed that in this case, $\mu_2 \prec \mu_1$ automatically implies $\mu_1 \prec \mu_2,$ so one does not need to check that result separately.

The main theorem of the paper, \ref{thm:general-theorem}, claims that the general conditions for excitability are (i) $\mu_2 \prec \mu_1$, (ii) $\sqrt{X_2} - \sqrt{X_1}$ Hilbert-Schmidt, and (iii) $Q$ has no kernel.
To prepare for proving this general theorem, this section provides a proof of theorem \ref{thm:passing-to-canonical}, which establishes equivalence of these conditions with certain conditions on the canonical purification.

Namely, we prove the equivalence of the following conditions:
\begin{itemize}
	\item $\mu_2 \prec \mu_1,$ and $\sqrt{X_2} - \sqrt{X_1}$ Hilbert-Schmidt, and $Q$ has no kernel;
	\item $\hat{\mu}_2 \prec \hat{\mu}_1$, and $\hat{Q} - 1$ Hilbert-Schmidt, and $\hat{Q}$ has no kernel.
\end{itemize}
In fact, we remark that the equivalence of the other conditions could be proved without making any assumptions about the kernels of $Q$ or $\hat{Q}$, since the Hilbert-Schmidt conditions imply that $Q$ and $\hat{Q}$ have finite-dimensional kernels, and these finite-dimensional kernels can be treated separately without any real complications.
But on physical grounds there is no need to prove this more general statement, since we established in section \ref{sec:general-necessary} that the absence of a kernel is a necessary condition for excitability.

At the end of the section, we explain how our results imply theorem \ref{thm:mixed-from-factor}, which states that when $\omega_1$ is a factor state, the excitability statement $\omega_2 \prec \omega_1$ is equivalent to the statements $\mu_2 \prec \mu_1$ with $\sqrt{X_2} - \sqrt{X_1}$ Hilbert-Schmidt.
While this may seem a little puzzling --- as no explicit mention is made in theorem \ref{thm:mixed-from-factor} of the necessary condition that $Q$ have no kernel --- we explain in section \ref{sec:factorial-consequences} why the conditions of theorem \ref{thm:mixed-from-factor} automatically imply that $Q$ has no kernel, without needing to check this condition separately.

\subsection{Canonical purifications of Gaussian states}
\label{sec:gaussian-canonical-formulas}

We begin by discussing general properties of the canonical purification $\hat{\omega}$ of a Gaussian state $\omega$.
In particular, we explain that the canonical purification of a Gaussian state is itself Gaussian.

Let $\A_0$ be a free-field $*$-algebra, and let $\omega : \A_0 \to \comps$ be a Gaussian state.
As in \cite{Sorce:paper1}, the canonical purification of $\omega$ is defined as the state
\begin{equation}
	\hat{\omega} : \A_0 \otimes \A_0^{\text{op}} \to \comps
\end{equation}
given by
\begin{equation}
	\hat{\omega}(a \otimes b)
		\equiv \langle a^{\dagger} \omega | J_{\omega} | b^{\dagger} \omega \rangle,
\end{equation}
with $J_{\omega}$ the modular conjugation of $|\omega\rangle$ with respect to the von Neumann algebra $\A_{\omega}$ in the GNS representation $\H_{\omega}.$
See remark \ref{rem:omega_hat} above for the basic intuition behind this definition.

As in appendix \ref{app:review}, the free field algebra $\A_0$ is generated by symbols $\phi[f],$ labeled by smearing functions $f \in C^{\infty}_0(\mathcal{M}),$ and subject to linearity constraints as well as the canonical commutation relation
\begin{equation}
	[\phi[f], \phi[g]] = - i \Omega[f, g].
\end{equation}
The algebraic tensor product $\A_0 \otimes \A_0^{\text{op}}$ can be thought of as a free field algebra generated by elements $\phi[f \oplus g],$ with $f \oplus g$ in the direct sum $C^{\infty}_0(\mathcal{M}) \oplus C^{\infty}_0(\mathcal{M}),$ subject to the commutation relation
\begin{equation}
	[\phi[f \oplus g], \phi[u \oplus v]] = - i \Omega[f, g] + i \Omega[u, v].
\end{equation}
This is accomplished by mapping
\begin{equation}
	\phi[f \oplus g] \mapsto \phi[f] \otimes 1 + 1 \otimes \phi[g] \in \A_0 \otimes \A_0^{\text{op}}.
\end{equation}
With this interpretation of $\A_0 \otimes \A_0^{\text{op}}$ as a free field algebra, we claim that $\hat{\omega}$ is itself a Gaussian state.

To see this, we note that every correlation function of $\hat{\omega}$ can be computed by taking derivatives of the expression
\begin{equation}
	\langle \omega | e^{i \phi[f]} J_{\omega} e^{-i \phi[g]} |\omega\rangle.
\end{equation}
In the GNS representation of $\hat{\omega},$ this will correspond to the correlator
\begin{equation}
	\langle \omega | e^{i \phi[f]} J_{\omega} e^{-i \phi[g]} |\omega\rangle
	= \langle \hat{\omega} | e^{i \phi[f \oplus g]} |\hat{\omega}\rangle.
\end{equation}
Matching to the equations from appendix \ref{app:Weyl}, we see that $\hat{\omega}$ is Gaussian if there is a semi- inner product $\hat{\mu}$ on $C^{\infty}_0(\mathcal{M}) \oplus C^{\infty}_0(\mathcal{M})$ satisfying
\begin{equation} \label{eq:tentative-muhat}
	\langle \omega | e^{i \phi[f]} J_{\omega} e^{-i \phi[g]} |\omega\rangle
	= e^{- \langle f \oplus g | f \oplus g\rangle_{\hat{\mu}}/2}.
\end{equation}
We will now verify that this is the case.

In appendix \ref{app:1P-modular-ops}, we computed the action of the modular conjugation $J_{\omega}$ on states of the form $e^{- i \phi[g]} |\omega\rangle$ in terms of the one-particle modular operator $j.$
We can do this by making use of a relationship, introduced in appendix \ref{app:generalized-Weyl}, between field-like operators $\chi_{\alpha}$ (which include ordinary field operators $\phi[g]$ with $g$ real) and elements $\ket{\alpha}$ of the one-particle Hilbert space (on which $j$ acts).
This relationship obeys the property $\phi[g]=\chi_{\phi[g] |\omega\rangle}$ when $g$ is real.
The antilinear operator $J_{\omega}$ acts on $e^{- i \phi[g]} |\omega\rangle$ by sending it to $e^{i \chi_{j \phi[g]|\omega\rangle}}\ket{\omega}.$
So we have
\begin{equation}
	\langle \omega | e^{i \phi[f]} J_{\omega} e^{-i \phi[g]} |\omega\rangle
		= \langle \omega | e^{i \chi_{\phi[f] |\omega\rangle}} e^{i \chi_{j \phi[g] |\omega\rangle}} |\omega\rangle.
        \label{eq:below-chi-explanation}
\end{equation}
Using equation \eqref{eq:chi-commutators} from appendix \ref{app:one-particle-modular}, one can compute this further as
\begin{equation}
	\langle \omega | e^{i \phi[f]} J_{\omega} e^{-i \phi[g]} |\omega\rangle
	= \langle \omega | e^{i \chi_{\phi[f] |\omega\rangle + j \phi[g] |\omega\rangle}} |\omega\rangle e^{- i \text{Im} \langle \phi[f] \omega | j \phi[g] \omega\rangle}.
\end{equation}
Now using that $\omega$ is a Gaussian state, we apply equation \eqref{eq:generalized-Weyl-eval} to rewrite this as
\begin{equation}
	\langle \omega | e^{i \phi[f]} J_{\omega} e^{-i \phi[g]} |\omega\rangle
	= e^{- \langle (\phi[f] + j \phi[g]) \omega | (\phi[f] + j \phi[g]) \omega \rangle/2} e^{- i \text{Im} \langle \phi[f] \omega | j \phi[g] \omega\rangle},
\end{equation}
and combining terms (and using antiunitarity of $j$) gives
\begin{equation}
	\langle \omega | e^{i \phi[f]} J_{\omega} e^{-i \phi[g]} |\omega\rangle
	= \exp\left[- \frac{1}{2} \left( \langle \phi[f] \omega | \phi[f] \omega\rangle + \langle \phi[g] \omega | \phi[g] \omega\rangle + 2 \langle \phi[f] \omega | j \phi[g] \omega \rangle \right)\right].
\end{equation}
This expression can be rewritten in terms of the $\mu$ inner product from appendix \ref{app:phase-space} by using the isometry $\phi[f] |\omega\rangle \mapsto \sqrt{1 - i R} |f\rangle_{\mu},$ and recalling from appendix \ref{app:non-cs} that under this isometry, the modular conjugation $j$ maps to the projection of $\Gamma$ away from the $\pm i$ eigenspace of $R$.
This gives
\begin{align}
	\begin{split} 
	& \langle \omega | e^{i \phi[f]} J_{\omega} e^{-i \phi[g]} |\omega\rangle \\
	& \quad = \exp\left[- \frac{1}{2} \left( \langle f | (1 - i R) |f \rangle_{\mu} + \langle g | (1 - i R) | g\rangle_{\mu} + 2 \langle f | \sqrt{1 - i R} \Gamma (1 - \Pi_i - \Pi_{-i}) \sqrt{1-iR} | g \rangle_{\mu} \right)\right].
	\end{split}
\end{align}
Since $R$ satisfies $\Gamma R \Gamma = R$ and $R^{\dagger} = - R,$ since $\Gamma$ is antiunitary, and since $f$ and $g$ are real, one has $\langle f | R | f \rangle = \langle g | R | g \rangle = 0.$
So this expression simplifies to
\begin{align}
	\begin{split} 
		& \langle \omega | e^{i \phi[f]} J_{\omega} e^{-i \phi[g]} |\omega\rangle \\
		& \qquad = \exp\left[- \frac{1}{2} \left( \langle f |f \rangle_{\mu} + \langle g | g\rangle_{\mu} + 2 \langle f | \sqrt{1 - i R} \Gamma (1 - \Pi_i - \Pi_{-i}) \sqrt{1-i R}| g \rangle_{\mu} \right)\right].
	\end{split}
\end{align}
One can then use that $\Gamma$ conjugates $\sqrt{1 - i R}$ to $\sqrt{1 + i R}$ to simplify the expression further as
\begin{align}
	\begin{split} 
		& \langle \omega | e^{i \phi[f]} J_{\omega} e^{-i \phi[g]} |\omega\rangle \\
		& \qquad = \exp\left[- \frac{1}{2} \left( \langle f |f \rangle_{\mu} + \langle g | g\rangle_{\mu} + 2 \langle f | \sqrt{1 + R^2}| g \rangle_{\mu} \right)\right].
	\end{split}
\end{align}
Finally, once again using that $f$ and $g$ are real, that $\Gamma$ is antiunitary, and that conjugation by $\Gamma$ preserves $R$, one can verify $\langle f | \sqrt{1 + R^2} g \rangle = \langle g | \sqrt{1 + R^2} f\rangle$, yielding the final expression
\begin{align}
	\begin{split} 
		& \langle \omega | e^{i \phi[f]} J_{\omega} e^{-i \phi[g]} |\omega\rangle \\
		& \qquad = \exp\left[- \frac{1}{2} \left( \langle f |f \rangle_{\mu} + \langle g | g\rangle_{\mu} + \langle f | \sqrt{1 + R^2}| g \rangle_{\mu} + \langle g | \sqrt{1 + R^2}| f \rangle_{\mu}\right)\right].
	\end{split}
    \label{eq:hat_omega_gaussian}
\end{align}
Consequently,
we see that $\hat{\omega}$ is a Gaussian state, with corresponding $\hat{\mu}$ inner product given by
\begin{equation} \label{eq:muhat}
	\langle u \oplus v | f \oplus g \rangle_{\hat{\mu}}
		= \langle u | f \rangle_{\mu} + \langle v | g \rangle_{\mu} + \langle u | \sqrt{1 + R^2} |g \rangle_{\mu} + \langle v | \sqrt{1 + R^2} | f\rangle_{\mu}.
\end{equation}
This is exactly the formula that can be found in e.g. \cite[equation 3.9]{Araki-Yamagami} for the ``doubling'' of a Gaussian state.

It is important to note that on the space $C^{\infty}_{0}(\M) \oplus C_0^{\infty}(\M)$, the ``purified'' inner product $\hat{\mu}$ is distinct from the ``unentangled'' inner product $\mu \oplus \mu.$
One does have the domination condition $\hat{\mu} \prec \mu \oplus \mu,$ which is obvious from the fact that one can rewrite equation \eqref{eq:muhat} to express $\hat{\mu}$ as a block matrix acting on $\K_{\mu} \oplus \K_{\mu}$.
But $\hat{\mu}$ and $\mu \oplus \mu$ will not always be topologically equivalent; in particular, if $|\eta\rangle_{\mu}$ is in the kernel of $R$, then it is easy to see that $|\eta\rangle \oplus (-1)|\eta\rangle$ will be null for $\hat{\mu}.$
Remember that when we write $\K_{\hat{\mu}}$ below, we are talking about the space obtained by taking a quotient by all null states of $\hat{\mu},$ followed by a topological completion.

\subsection{Domination conditions on symmetrized inner products}
\label{sec:domination}

Towards a proof of theorem \ref{thm:passing-to-canonical},
we now prove that for Gaussian states $\omega_1$ and $\omega_2,$ we have the following implications:
\begin{equation} \label{eq:hat-dominance-implies}
	(\hat{\mu}_1 \prec \hat{\mu}_1) \Rightarrow (\mu_2 \prec \mu_1)
\end{equation}
and
\begin{equation} \label{eq:regular-dominance-implies}
	\left(\mu_2 \prec \mu_1 \text{ and } \sqrt{X_2} - \sqrt{X_1} \text{ Hilbert-Schmidt and } \ker{Q} = 0 \right) \Rightarrow \left( \hat{\mu}_2 \prec \hat{\mu}_1 \right).
\end{equation}

The implication in equation \eqref{eq:hat-dominance-implies} is easy to show.
For if there is a constant $C$ with
\begin{equation}
	\langle f \oplus g | f \oplus g \rangle_{\hat{\mu}_2}
		\leq C \langle f \oplus g | f \oplus g \rangle_{\hat{\mu}_1},
\end{equation}
then taking $g = 0$ gives
\begin{equation}
	\langle f | f \rangle_{\mu_2} \leq C \langle f | f \rangle_{\mu_1}.
\end{equation}

For the other direction, suppose that we have $\mu_2 \prec \mu_1$.
We will not use the other assumptions until absolutely necessary; they will come in at the last step.
Consider a sequence (or net) $f_{n} \oplus g_n$ that converges to zero with respect to the $\hat{\mu}_1$ inner product.
We must show that it also converges to zero in the $\hat{\mu}_2$ inner product.

First, we note that we can write
\begin{align}
	\begin{split}
	\langle f_n \oplus g_n | f_n \oplus g_n \rangle_{\hat{\mu}_1}
		& = \frac{1}{2} \left[ \langle f_n + g_n | 1 +  {\textstyle \sqrt{1 + R_1^2}} |f_n + g_n \rangle_{\mu_1} \right] \\
		& \qquad	+ \frac{1}{2} \left[ \langle f_n - g_n | 1 - {\textstyle \sqrt{1 + R_1^2}} |f_n - g_n \rangle_{\mu_1} \right].
	\end{split}
\end{align}
Each term on the right-hand side is nonnegative, so the whole expression goes to zero if and only if each of these terms goes to zero.
Since $1 + \sqrt{1+ R_1^2}$ is invertible, the first term goes to zero if and only if $|f_n + g_n\rangle_{\mu_1}$ goes to zero.
This obviously implies $|f_n + g_n \rangle_{\mu_2} \to 0,$ since we have assumed $\mu_2 \prec \mu_1.$

The second term is a bit more complicated; it converges to zero if and only if we have
\begin{equation}
	\sqrt{1 - \textstyle{\sqrt{1 + R_1^2}}} |f_n - g_n\rangle_{\mu_1} \to 0.
\end{equation}
If we remember again that $1+\sqrt{1+R_1^2}$ is an invertible operator, however, then this becomes simpler; we multiply by the square root of this invertible operator and find the equivalent condition
\begin{equation}
	|R_1| |f_n - g_n\rangle_{\mu_1} \to 0.
\end{equation}
So to complete the proof, we must show the implication
\begin{equation}
	\left( |R_1| |f_n - g_n\rangle_{\mu_1} \to 0 \right) \Rightarrow \left( |R_2| |f_n - g_n\rangle_{\mu_2} \to 0 \right).
\end{equation}
But using $R_1 = L^{\dagger} R_2 L$ from equation \eqref{eq:R1-from-R2}, and the fact that $L$ acts as the identity on test functions, we find
\begin{equation} \label{eq:Ldag-R2-convergence}
	\langle f_n - g_n | |R_1|^2 |f_n - g_n \rangle_{\mu_1}
		= \langle L^{\dagger} R_2 (f_n - g_n) | L^{\dagger} R_2 (f_n - g_n) \rangle_{\mu_1}.
\end{equation}

If $\mu_2$ and $\mu_1$ are mutually dominated by one another, then $L$ is an invertible operator, and one can conclude from this equation that one has $|R_2| |f_n - g_n \rangle_{\mu_2} \to 0.$
This is where the assumptions on $Q$ and $\sqrt{X_2} - \sqrt{X_1}$ come in.
Since $\sqrt{X_2} - \sqrt{X_1}$ is Hilbert-Schmidt, by multiplying on either side by $\sqrt{X_2} + \sqrt{X_1}$ and taking linear combinations, one may conclude that the operator $X_2 - X_1 = Q - 1$ is Hilbert-Schmidt.
Since $Q - 1$ is Hilbert-Schmidt, the operator $Q$ must have a finite number of eigenvalues in any subset of the real line that is bounded away from one.
In particular, $Q$ must have a finite kernel, and there must be a gap between its kernel and its first nonzero eigenvalue.
When we put in the additional assumption that $Q$ has no kernel, we learn that $Q$ is an invertible map; since $Q$ is defined by $Q = L^{\dagger} L,$ we conclude that $L$ is invertible.
So from the convergence of equation \eqref{eq:Ldag-R2-convergence}, we may conclude
\begin{equation}
	\langle R_2 (f_n - g_n) | R_2 (f_n - g_n) \rangle_{\mu_2} \to 0,
\end{equation}
as desired.

\subsection{Hilbert-Schmidt conditions on two-point operators}

Now we prove the full equivalence statement of theorem \ref{thm:passing-to-canonical}, namely
\begin{align}
	\begin{split}
	& \left(\mu_2 \prec \mu_1, \sqrt{X_2} - \sqrt{X_1} \text{ Hilbert-Schmidt, and } \ker{Q} = 0\right) \\
		& \qquad \Leftrightarrow \left(\hat{\mu}_2 \prec \hat{\mu}_1, \hat{X}_2 - \hat{X}_1 \text{ Hilbert-Schmidt, and } \ker{\hat{Q}}=0\right).
	\end{split}
    \label{eq:iff-lift-to-purify}
\end{align}
To accomplish this, we will need expressions that relate the operators $X_1, X_2$ to the operators $\hat{X}_1, \hat{X}_2.$
We have already shown that the conditions on the left imply $\hat{\mu}_2 \prec \hat{\mu}_1,$ and that the conditions on the right imply $\mu_2 \prec \mu_1.$
So we can henceforth assume domination of inner products no matter which side of the implication we are trying to prove.

\subsubsection{\texorpdfstring{Expressions for $\hat{Q}$}{Expressions for Q}}

As a first step towards expressing the Hilbert-Schmidt condition on $\hat{X}_2-\hat{X}_1= \hat{Q}-1$ as a Hilbert-Schmidt condition on the original space,
we first generate explicit expressions for the operator $\hat{Q}$, which is abstractly defined by
\begin{equation}
	\langle x \oplus y | \hat{Q} | f \oplus g \rangle_{\hat{\mu}_1}
		= \langle x \oplus y | f \oplus g \rangle_{\hat{\mu}_2}\,.
        \label{eq:hatQ-def}
\end{equation}

We begin by rewriting the $\hat{\mu}_1, \hat{\mu}_2$ inner products in an instructive way.
We can think of the $\hat{\mu}_j$ inner product as being represented by a block matrix within the ``unentangled'' inner product $\mu_j \oplus \mu_j$:
\begin{equation}
	\langle x \oplus y | f \oplus g \rangle_{\hat{\mu}_j}
		= \begin{pmatrix} \langle x | & \langle y | \end{pmatrix}
		\begin{pmatrix} 1 & \sqrt{1 + R_j^2} \\ \sqrt{1 + R_j^2} & 1 \end{pmatrix}
		\begin{pmatrix} |f\rangle \\ |g\rangle \end{pmatrix}_{\mu_j \oplus \mu_j}.
        \label{eq:mu-hat-to-mumu}
\end{equation}
Call this matrix $F_j$.
We can obtain formulas for operators with respect to the $\hat{\mu}_j$ inner product by first writing them with respect to the $\mu_j \oplus \mu_j$ inner product, then acting on the left with $F_j^{-1}.$
We will proceed to do this for the operator $\hat{Q}$.

Applying equation \eqref{eq:mu-hat-to-mumu} to the right hand side of equation \eqref{eq:hatQ-def}, we have
\begin{equation}
	\langle x \oplus y | \hat{Q} | f \oplus g \rangle_{\hat{\mu}_1}
		= \begin{pmatrix} \langle x | & \langle y | \end{pmatrix}
		\begin{pmatrix} 1 & \sqrt{1 + R_2^2} \\ \sqrt{1 + R_2^2} & 1 \end{pmatrix}
		\begin{pmatrix} |f\rangle \\ |g\rangle \end{pmatrix}_{\mu_2 \oplus \mu_2},
\end{equation}
and since $L : \K_{\mu_1} \to \K_{\mu_2}$ is the operator that converts between the $\mu_2$ and $\mu_1$ inner products, we can write this as
\begin{equation}
	\langle x \oplus y | \hat{Q} | f \oplus g \rangle_{\hat{\mu}_1}
	= \begin{pmatrix} \langle x | & \langle y | \end{pmatrix}
	\begin{pmatrix} Q & L^{\dagger} \sqrt{1 + R_2^2} L \\ L^{\dagger} \sqrt{1 + R_2^2} L & Q \end{pmatrix}
	\begin{pmatrix} |f\rangle \\ |g\rangle \end{pmatrix}_{\mu_1 \oplus \mu_1},
\end{equation}
where as usual we have written $Q = L^{\dagger} L.$
We can simplify our notation slightly by writing
\begin{align}
	W_1
		& = \textstyle{\sqrt{1 + R_1^2}}, \\
	W_2
		& = L^{\dagger} \textstyle{\sqrt{1 + R_2^2}} L,
\end{align}
in which case one has
\begin{equation}
	\langle x \oplus y | \hat{Q} | f \oplus g \rangle_{\hat{\mu}_1}
	= \begin{pmatrix} \langle x | & \langle y | \end{pmatrix}
	\begin{pmatrix} Q &  W_2 \\ W_2 & Q \end{pmatrix}
	\begin{pmatrix} |f\rangle \\ |g\rangle \end{pmatrix}_{\mu_1 \oplus \mu_1}.
    \label{eq:Qhat_in_muplusmu}
\end{equation}
One can easily compute $F_1^{-1}$ as
\begin{equation}
	F_1^{-1}
		= \frac{1}{1-W_1^2}
		\begin{pmatrix}
			1 & -W_1 \\ -W_1 & 1	
			\end{pmatrix},
\end{equation}
thereby obtaining an explicit expression for $\hat{Q}$ as
\begin{equation}
	\hat{Q} = F_1^{-1} \begin{pmatrix} Q & W_2 \\ W_2 & Q \end{pmatrix}
		= \frac{1}{1-W_1^2} \begin{pmatrix} (Q - W_1 W_2) & \quad (W_2 - W_1 Q)  \\ (W_2 - W_1 Q)  & \quad (Q - W_1 W_2) \end{pmatrix}.
    \label{eq:Qhat-as-matrix}
\end{equation}
In looking at this expression, one might be concerned about the fact that $1 - W_1^2 = - R_1^2$ is not invertible as a bounded operator.
It is invertible as an \textit{unbounded} operator on the image of $R_1^2,$ and if $R_1^2$ has a kernel, then this image will not even be dense in Hilbert space.
Nevertheless, we are guaranteed that the above equation gives a valid expression for $\hat{Q}$ ``whenever it makes sense,'' i.e., when it acts on vectors such that $(1 - W_1^2)^{-1}$ ends up being applied within its domain.
Moreover, using a decomposition of $\hat{Q}$ discussed in the next paragraph, we will show that this expression actually does makes sense on a dense subspace of vectors in ${\cal K}_{\hat{\mu}_1}$.

To better understand the action of $\hat{Q}$ in equation \eqref{eq:Qhat-as-matrix}, it is convenient to
split the Hilbert space into symmetric and antisymmetric subspaces.
On these subspaces, equation \eqref{eq:Qhat-as-matrix} becomes
\begin{equation}
	\hat{Q} |f \oplus f \rangle
		= \frac{1}{1+W_1} \begin{pmatrix} Q+ W_2 & 0 \\ 0 & Q+W_2 \end{pmatrix} \begin{pmatrix} |f\rangle \\ |f\rangle \end{pmatrix}
\end{equation}
and 
\begin{equation}
	\hat{Q} |f \oplus - f \rangle
	= \frac{1}{1-W_1} \begin{pmatrix} Q-W_2 & 0 \\ 0 & Q-W_2 \end{pmatrix} \begin{pmatrix} |f\rangle \\ |\text{$-f$}\rangle \end{pmatrix}.
\end{equation}
From these expressions, we can see that the symmetric and antisymmetric subspaces are preserved by $\hat{Q}.$
Moreover, the structure of $\hat{\mu}$ allows us to create unitary identifications that map these subspaces of $\K_{\hat{\mu}_1}$ into corresponding subspaces of $\K_{\mu_1},$ which will allow us to turn statements about $\hat{Q}$ into statements about operators on $\K_{\mu_1}.$
Concretely, from equation \eqref{eq:muhat}, one has
\begin{equation}
	\langle f \oplus f | f \oplus f \rangle_{\hat{\mu}_1}
			= 2 \langle f | (1 + W_1) |f\rangle_{\mu_1}
\end{equation}
and
\begin{equation}
	\langle f \oplus -f | f \oplus -f \rangle_{\hat{\mu}_1}
		= 2 \langle f | (1 - W_1) |f\rangle_{\mu_1}.
\end{equation}
The latter set of equations tells us that the map
\begin{equation}
	|f\rangle
		\mapsto \frac{1}{\sqrt{2}} \left|\frac{1}{\sqrt{1+W_1}} f \oplus \frac{1}{\sqrt{1+W_1}} f \right\rangle
\end{equation}
is a unitary map from $\K_{\mu_1}$ to the symmetric subspace of $\K_{\hat{\mu}_1},$ while 
\begin{equation}
	|f\rangle
	\mapsto \frac{1}{\sqrt{2}} \left|\frac{1}{\sqrt{1-W_1}} f \oplus - \frac{1}{\sqrt{1-W_1}} f \right\rangle
\end{equation}
is a unitary map from the support\footnote{The unitary map is only defined on the support of $R_1,$ since $\sqrt{1 - W_1}$ is only invertible on a dense subspace of this support. In principle one might think that our map is only an isometry, i.e., that it is not surjective, since it can only produce vectors of the form $|f \oplus -f\rangle$ with $f$ in the support of $R_1.$
But as mentioned previously, if $f$ is in the kernel of $R_1,$ then $|f \oplus -f\rangle$ is a null state of $\hat{\mu}_1.$} of $R_1$ to the antisymmetric subspace of $\K_{\hat{\mu}_1}.$
Conjugating $\hat{Q}$ by these unitary maps, one can think of $\hat{Q}$ as being made up of two blocks --- a ``symmetric block,'' which acts on $\K_{\mu_1}$ by
\begin{equation} \label{eq:Q-sym}
	\hat{Q}_{\text{sym}} |f \rangle
		= \frac{1}{\sqrt{1+W_1}} (Q + W_2) \frac{1}{\sqrt{1+W_1}} |f\rangle,
\end{equation}
and an ``antisymmetric block,'' which acts on the support of $R_1$ by
\begin{equation} \label{eq:Q-antisym}
	\hat{Q}_{\text{antisym}} |f \rangle
	= \frac{1}{\sqrt{1-W_1}} (Q - W_2) \frac{1}{\sqrt{1-W_1}} |f\rangle.
\end{equation}

Taken together, equations \eqref{eq:Q-sym} and  \eqref{eq:Q-antisym} are clearly a simplification as compared to equation \eqref{eq:Qhat-as-matrix}, and this is the form of $\hat{Q}$ that we will use in the following subsection to make contact with the Hilbert-Schmidt condition on $\sqrt{X_2}-\sqrt{X_1}$.
We are also now able to justify our claim above that the expression \eqref{eq:Qhat-as-matrix} for $\hat{Q}$ ``makes sense'' on a dense subspace of ${\cal K}_{\hat{\mu}_1}$.
In particular, we can show that the expression \eqref{eq:Q-sym} for $\hat{Q}_{\text{sym}}$ is densely defined on ${\cal K}_{\mu_1}$, and the expression \eqref{eq:Q-antisym} for $\hat{Q}_{\text{antisym}}$ is densely defined on the support of $R_1$.\footnote{Note that since we already know that $\hat{Q}$ is bounded, the operators $\hat{Q}_{\text{sym}}$ and $\hat{Q}_{\text{antisym}}$ will automatically be bounded wherever they are defined.}

The operator $\hat{Q}_{\text{sym}}$ is densely defined because $\sqrt{1 + W_1}$ is boundedly invertible; indeed, for this reason, equation \eqref{eq:Q-sym} can be used as an expression for $\hat{Q}_{\text{sym}}$ on all of $\K_{\mu_1}.$
On the other hand, $\sqrt{1 - W_1}$ is not boundedly invertible, so $\hat{Q}_{\text{antisym}}$ is only defined a priori on vectors $|f\rangle$ such that (i) $|f\rangle$ is in the image of $\sqrt{1 - W_1}$, and (ii) $(Q - W_2) \frac{1}{\sqrt{1 - W_1}} |f\rangle$ is in the image of $\sqrt{1 - W_1}.$
To show that these vectors are dense in the support of $R_1$, we will show that condition (i) above implies condition (ii), so that \eqref{eq:Q-antisym}  actually defines $\hat{Q}_{\text{antisym}}$ on the full image of $\sqrt{1 - W_1} = \sqrt{1 - \sqrt{1 + R_1^2}},$ which is dense within the support of $R_1.$

For any $|f\rangle$ in the image of $\sqrt{1 - W_1}$, clearly $(Q - W_2) \frac{1}{\sqrt{1 - W_1}} \ket{f}$  is in the image of $Q-W_2$, and hence in the image of $\sqrt{Q-W_2}$.
So we can accomplish our goal by showing that the image of $\sqrt{Q - W_2}$ is contained within that of $\sqrt{1 - W_1}$.
This follows from a standard result known as ``Douglas's lemma'', once we demonstrate the existence of a constant $C$ such that on the support of $R_1,$ one has 
\begin{equation} \label{eq:Q-W2-ineq}
    Q - W_2 \leq C (1 - W_1).
\end{equation}

To prove an inequality like \eqref{eq:Q-W2-ineq}, we note that the function $2 (1 - \sqrt{1 + t^2})$ exceeds $|t|^2$ for $t$ in the range $[-1,1]i,$ so to prove \eqref{eq:Q-W2-ineq}, it suffices to find $C'$ satisfying the inequality
\begin{equation}
    Q - W_2 \leq C' R_1^{\dagger} R_1.
\end{equation}
Writing $L = v Q^{1/2}$ and applying the identity $R_1 = L^{\dagger} R_2 L$ from from equation \eqref{eq:R1-from-R2} to the right hand side, we see that the inequality we are trying to prove is
\begin{equation} \label{ineq:well-def-penultimate}
	Q^{1/2} \left(v^{\dagger} v - v^{\dagger} {\textstyle \sqrt{1 + R_2^2}} v \right) Q^{1/2} \leq C' Q^{1/2} (v^{\dagger} R_2^{\dagger} v) Q (v^{\dagger} R_2 v) Q^{1/2}.
\end{equation}
The trick we will use is to note that regardless of which direction we are trying to show in theorem \ref{thm:passing-to-canonical}, we are always guaranteed that $Q$ is invertible.
For the ``implies'' direction of equation \eqref{eq:iff-lift-to-purify}, the assumption that $\sqrt{X_2} - \sqrt{X_1}$ is Hilbert-Schmidt guarantees that $Q-1$ is Hilbert-Schmidt, so $Q$ is invertible on its support, and the assumption $\ker{Q} = 0$ means that $Q$ is invertible.
Conversely, if we were instead making the assumption that $\hat{X}_2 - \hat{X}_1$ were Hilbert-Schmidt with $\ker{\hat{Q}} = 0,$ then restricting to the subspace $|f \oplus 0\rangle$ in equation \eqref{eq:Qhat_in_muplusmu} would give us that $Q-1$ is Hilbert-Schmidt, and for any $|z\rangle_{\mu_1}$ in the kernel of $Q$ we would have $|z \oplus 0\rangle_{\hat{\mu}_1}$ in the kernel of $\hat{Q},$\footnote{Here we are again using equation \eqref{eq:Qhat_in_muplusmu}, together with the fact that $Q\ket{z}_{\mu_1}=0$ implies $L \ket{z}_{\mu_1} = 0$ and hence $W_2\ket{z}_{\mu_1}=0$.} so we can conclude $\ker{Q} = 0.$
Taken together, these results tell us that $Q$ is invertible.
So in particular $Q^{1/2}$ is invertible, and to show inequality \eqref{ineq:well-def-penultimate}, it suffices to show
\begin{equation}
	v^{\dagger} (1 - {\textstyle\sqrt{1 + R_2^2}}) v \leq C' v^{\dagger} (R_2^{\dagger} v Q v^{\dagger} R_2) v.
\end{equation}
But since $Q$ is invertible, there is a positive constant $\alpha$ with $v Q v^{\dagger} \geq \alpha v v^{\dagger} = \alpha,$ and one finds that putting $C' = \alpha^{-1},$ we have
\begin{equation}
	\alpha^{-1} v^{\dagger} (R_2^{\dagger} v Q v^{\dagger} R_2) v
		\geq v^{\dagger} R_2^{\dagger} R_2 v
		\geq v^{\dagger} (1 - {\textstyle\sqrt{1+R_2^2}}) v,
\end{equation}
where the last inequality simply follows from the fact that $R_2$ has spectrum contained between $-i$ and $i,$ and one has $|z|^2 \geq 1 - \sqrt{1+z^2}$ on this range.

\subsubsection{First Hilbert-Schmidt conditions}

Putting all this together, we have demonstrated that equations \eqref{eq:Q-sym} and \eqref{eq:Q-antisym} tell us everything we need to know about $\hat{Q}$ by breaking it into two pieces $\hat{Q}_{\text{sym}}$ and $\hat{Q}_{\text{antisym}}$.
In particular, $\hat{Q} - 1$ will be Hilbert-Schmidt if and only if $\hat{Q}_{\text{sym}} - 1$ and $\hat{Q}_{\text{antisym}} - 1$ are Hilbert-Schmidt.
So instead of checking that $\hat{Q} - 1$ is Hilbert-Schmidt on $\K_{\hat{\mu}_1},$ we may equivalently check that
\begin{itemize}
	\item $\frac{1}{\sqrt{1+W_1}} (Q - 1 + W_2 - W_1) \frac{1}{\sqrt{1+W_1}}$ is Hilbert-Schmidt on $\K_{\mu_1}$, and
	\item $\frac{1}{\sqrt{1-W_1}} (Q - 1 - W_2 + W_1) \frac{1}{\sqrt{1-W_1}}$ is Hilbert-Schmidt on the support of $R_1.$
\end{itemize}
Since the operator $\sqrt{1+W_1}$ is invertible, we may ignore it in the first expression, making our conditions
\begin{itemize}
	\item $Q - 1 + W_2 - W_1$ is Hilbert-Schmidt on $\K_{\mu_1}$, and
	\item $\frac{1}{\sqrt{1-W_1}} (Q - 1 - W_2 + W_1) \frac{1}{\sqrt{1-W_1}}$ is Hilbert-Schmidt on the support of $R_1.$
\end{itemize}
Even though $\frac{1}{\sqrt{1-W_1}}$ is not bounded, we still know that multiplying a Hilbert-Schmidt operator on the left and right by $\sqrt{1-W_1}$ produces a Hilbert-Schmidt operator.
From this we can conclude that Hilbert-Schmidtness of $\hat{Q} - 1$ \textit{implies} the conditions
\begin{itemize}
	\item $(Q - 1) + (W_2 - W_1)$ is Hilbert-Schmidt on $\K_{\mu_1}$, and
	\item $(Q - 1) - (W_2 - W_1)$ is Hilbert-Schmidt on the support of $R_1.$
\end{itemize}
A simple calculation shows that $(Q-1)-(W_2-W_1)$ actually vanishes on the kernel of $R_1$,\footnote{The identity $R_1 = L^{\dagger} R_2 L$, together with the fact that $L$ has dense image, implies that $L$ maps the kernel of $R_1$ to the kernel of $R_2.$
So the operator $W_2 = L^{\dagger} \sqrt{1 + R_2^2} L$, restricted to the kernel of $R_1,$ is simply $L^{\dagger} \sqrt{1 + 0} L = Q.$
On the kernel of $R_1$, we also obviously have $W_1 = \sqrt{1 + R_1^2} = 1.$
It follows that on the kernel of $R_1$ we have $(Q - 1) - (W_2 - W_1) = (Q - W_2) - (1 - W_1) = 0.$}
so the second condition trivially extends to all of $\mathcal{K}_{\mu_1}$.
Overall, we simply find that the Hilbert-Schmidt property for $\hat{Q} - 1$ implies the conditions
\begin{itemize}
	\item $Q - 1$ is Hilbert-Schmidt on $\K_{\mu_1}$, and
	\item $W_2 - W_1$ is Hilbert-Schmidt on $\K_{\mu_1}.$
\end{itemize}

These two conditions are \textit{necessary} for $\hat{Q} - 1$ to be Hilbert-Schmidt.
Next we will show that these conditions are actually \textit{sufficient}; then in the following subsection, we will show that they are equivalent to the single condition that $\sqrt{X_2} - \sqrt{X_1}$ be Hilbert-Schmidt, which finally proves theorem \ref{thm:passing-to-canonical}.

Let us begin with sufficiency.
The conditions that $Q - 1$ and $W_2 - W_1$ be Hilbert-Schmidt obviously imply that $Q - 1 + W_2 - W_1$ is Hilbert-Schmidt.
All we need to show is that nothing was lost in dropping the factors of $\frac{1}{\sqrt{1-W_1}}$ in the above expression; i.e., we must show that whenever $Q-1$ and $W_2 - W_1$ are Hilbert-Schmidt, so is $\frac{1}{\sqrt{1-W_1}} (Q - 1 - W_2 + W_1) \frac{1}{\sqrt{1-W_1}}.$
We will first use our usual trick of left- and right-multiplying by the invertible operator $\frac{1}{\sqrt{1+W_1}},$ so that we are trying to demonstrate the Hilbert-Schmidt condition on
\begin{equation}
	\frac{1}{|R_1|} ((Q - 1) - (W_2 - W_1)) \frac{1}{|R_1|},
\end{equation}
defined on the support of $R_1.$
If we write a polar decomposition $R_1 = u |R_1|,$ then $u$ is a partial isometry from the support of $R_1$ to the image of $R_1$; but since $R_1$ is anti-Hermitian, these spaces are the same, and $u$ is simply a unitary operator from the support of $R_1$ to itself.
Taking adjoints one can write $R_1 = - |R_1| u^{\dagger},$ and inverting gives
\begin{equation}
	\frac{1}{R_1} = - u \frac{1}{|R_1|}.
\end{equation}
So there is an invertible operator within the support of $R_1$ that one can use to exchange $\frac{1}{|R_1|}$ for $\frac{1}{R_1},$ meaning that what we are looking for is the Hilbert-Schmidt condition on
\begin{equation}
	\frac{1}{R_1} ((Q - 1) - (W_2 - W_1)) \frac{1}{R_1}.
\end{equation}
From the relation $Q^{1/2} v^{\dagger} R_2 v Q^{1/2} = R_1,$ one finds that the inverse of $R_1$ can be written
\begin{equation}
	\frac{1}{R_1} = Q^{-1/2} \frac{1}{v^{\dagger} R_2 v} Q^{-1/2}.
\end{equation}
Grouping $Q - W_2 = Q^{1/2} v^{\dagger} \left(1 - \sqrt{1 + R_2^2} \right) v Q^{1/2}$ in the above expression, one finds that we are looking for the Hilbert-Schmidt condition on
\begin{equation}
	Q^{-1/2} \frac{1}{v^{\dagger} R_2 v} v^{\dagger} \left(1 - \textstyle{\sqrt{1 + R_2^2}}\right) v \frac{1}{v^{\dagger} R_2 v} Q^{-1/2} - \frac{1}{R_1} (1 - W_1) \frac{1}{R_1}.
\end{equation}
Using the identity $v v^{\dagger} = 1,$ we can move $v^{\dagger} (\cdot) v$ inside the square root to obtain the equivalent expression
\begin{equation} \label{eq:penpenultimate-WW}
	Q^{-1/2} \frac{1}{v^{\dagger} R_2 v} \left(1 - \textstyle{\sqrt{1 +  (v^{\dagger} R_2 v)^2} }\right) \frac{1}{v^{\dagger} R_2 v} Q^{-1/2} - \frac{1}{R_1} (1 - W_1) \frac{1}{R_1}.
\end{equation}
Finally, as in the preceding subsection, we know from our assumptions that $Q$ is invertible.
So one may apply the Birman-Solomyak theorem from section \ref{sec:pure-two-point} to conclude that $Q^{-1/2} - 1$ is Hilbert-Schmidt, and expression \eqref{eq:penpenultimate-WW} is Hilbert-Schmidt equivalent to
\begin{equation} \label{eq:penultimate-WW}
	\frac{1}{v^{\dagger} R_2 v} \left(1 - \sqrt{1 +  (v^{\dagger} R_2 v)^2} \right)  \frac{1}{v^{\dagger} R_2 v} - \frac{1}{R_1} \left(1 - \sqrt{1 + R_1^2}  \right) \frac{1}{R_1}.
\end{equation}
So far we have only used the assumption that $Q-1$ is Hilbert-Schmidt.
We will need to use the assumption that $W_2 - W_1$ is Hilbert-Schmidt to conclude that expression \eqref{eq:penultimate-WW} is Hilbert-Schmidt.
The key is that we can write
\begin{equation}
	W_2 - W_1 = Q^{1/2} \sqrt{1 + (v^{\dagger} R_2 v)^2} Q^{1/2} - \sqrt{1 + R_1^2},
\end{equation}
and the same kind of Birman-Solomyak trick we just used tells us that $Q^{1/2} - 1$ is Hilbert-Schmidt, so we have
\begin{equation} \label{eq:explicit-DW-HS}
	\sqrt{1 + (v^{\dagger} R_2 v)^2} - \sqrt{1 + R_1^2} \quad \text{ Hilbert-Schmidt}.
\end{equation}

To relate this expression to the one in equation \eqref{eq:penultimate-WW}, we use the Birman-Solomyak theorem a second time.
The function
\begin{equation}
	f(t) = - \frac{1}{1+t}
\end{equation}
is Lipschitz on any domain of $t$ bounded away from $t=-1.$
One also has
\begin{equation}
	f({\textstyle{\sqrt{1 + R_1^2}}})
		= \frac{1 - \sqrt{1+R_1^2}}{R_1^2},
\end{equation}
and similarly for $f(\sqrt{1+(v^{\dagger} R_2 v)^2}).$
So from equation \eqref{eq:explicit-DW-HS} and the Birman-Solomyak theorem, we may conclude
\begin{equation}
	\frac{1}{v^{\dagger} R_2 v} \left(1 - \sqrt{1 +  (v^{\dagger} R_2 v)^2} \right) \frac{1}{v^{\dagger} R_2 v} - \frac{1}{R_1} \left(1 - \sqrt{1 + R_1^2} \right) \frac{1}{R_1} \quad \text{ Hilbert-Schmidt},
\end{equation}
which was the desired condition in equation \eqref{eq:penultimate-WW}.

\subsubsection{\texorpdfstring{Connection to $\sqrt{X_2} - \sqrt{X_1}$}{Connection to sqrtX2 - sqrtX1}}

So far we have shown equivalence of the following two conditions:
\begin{itemize}
	\item $\hat{\mu}_2 \prec \hat{\mu}_1,$ $\ker{\hat{Q}} = 0,$ and $\hat{Q} - 1$ Hilbert-Schmidt.
	\item $\mu_2 \prec \mu_1$, $\ker{Q} = 0,$ and both $Q-1$ and $\sqrt{1 + (v^{\dagger} R_2 v)^2} - \sqrt{1 + R_1^2}$ Hilbert-Schmidt.
\end{itemize}
To prove theorem \ref{thm:passing-to-canonical}, it remains to show that in the second bullet point, the assumption that this pair of operators is Hilbert-Schmidt is the same as the assumption that $\sqrt{X_2} - \sqrt{X_1}$ is Hilbert-Schmidt.

First, let us assume that both operators $Q-1$ and $\sqrt{1 + (v^{\dagger} R_2 v)^2} - \sqrt{1 + R_1^2}$ are Hilbert-Schmidt.
We wish to show the Hilbert-Schmidt condition on
\begin{equation}
	\sqrt{X_2} - \sqrt{X_1}
		= \sqrt{Q - i R_1} - \sqrt{1 - i R_1}.
\end{equation}
Using the identity $R_1 = L^{\dagger} R_2 L$ from equation \eqref{eq:R1-from-R2}, and making use of the polar decomposition $L = v Q^{1/2},$ we can rewrite $\sqrt{X_2} - \sqrt{X_1}$ as
\begin{equation}
	\sqrt{X_2} - \sqrt{X_1}
	= \sqrt{Q^{1/2} (1 - i v^{\dagger} R_2 v) Q^{1/2}} - \sqrt{1 - i R_1}.
\end{equation}
Let us simplify the first term on the right-hand side of this equation.
As we have explained several times already, when $Q - 1$ is Hilbert-Schmidt, so is $Q^{1/2} - 1.$
Applying the Birman-Solomyak theorem to the absolute value function as in section \ref{sec:pure-two-point}, one has
\begin{equation}
	\sqrt{Q^{1/2} (1 - i v^{\dagger} R_2 v) Q^{1/2}} - \sqrt{1 - i v^{\dagger} R_2 v} \quad \text{ Hilbert-Schmidt},
\end{equation}
so we have
\begin{equation}
	\sqrt{X_2} - \sqrt{X_1}
	\sim_{\text{H.S.}} \sqrt{1 - i v^{\dagger} R_2 v} - \sqrt{1 - i R_1}.
\end{equation}
We must now figure out how to ``square'' $iv^{\dag}R_2 v$ and $i R_1$ in this expression, in order to use our assumption that $\sqrt{1 + (v^{\dagger} R_2 v)^2} - \sqrt{1 + R_1^2}$ is Hilbert-Schmidt.
Usefully,
The operators $\sqrt{1 + i R_1} + \sqrt{1 - i R_1}$ and $\sqrt{1 + i v^{\dagger} R_2 v} + \sqrt{1 - i v^{\dagger} R_2 v}$ are both invertible.
So to check if $\sqrt{X_2} - \sqrt{X_1}$ is Hilbert-Schmidt, we can equivalently check the Hilbert-Schmidt condition on
\begin{equation}
	\left( \sqrt{1 + i v^{\dagger} R_2 v} + \sqrt{1 - i v^{\dagger} R_2 v} \right) \left( \sqrt{1 - i v^{\dagger} R_2 v} - \sqrt{1 - i R_1}\right) \left( \sqrt{1 + i R_1} + \sqrt{1 - i R_1}\right),
\end{equation}
which simplifies to
\begin{align} \label{eq:some-annoying-thing}
	\begin{split}
	& \left( 1 - i v^{\dagger} R_2 v + \sqrt{1 + (v^{\dagger} R_2 v)^2} \right) \left( \sqrt{1 + i R_1} + \sqrt{1 - i R_1}\right) \\
		& \qquad - \left( \sqrt{1 + i v^{\dagger} R_2 v} + \sqrt{1 - i v^{\dagger} R_2 v} \right) \left(1 - i R_1 + {\textstyle \sqrt{1 + R_1^2}} \right).
	\end{split}
\end{align}
Note that as above, since $Q^{1/2} - 1$ is Hilbert-Schmidt, the operator $v^{\dagger} R_2 v$ is Hilbert-Schmidt equivalent to $Q^{1/2} v^{\dagger} R_2 v Q^{1/2}=R_1$.
If we also use the assumption that $\sqrt{1 + (v^{\dagger} R_2 v)^2}$ is Hilbert-Schmidt equivalent to $\sqrt{1 + R_1^2},$ we find that the whole expression in \eqref{eq:some-annoying-thing} is Hilbert-Schmidt equivalent to
\begin{align}
	\begin{split}
		& \left( 1 - i R_1 + {\textstyle \sqrt{1 + R_1^2}} \right) \left( \sqrt{1 + i R_1} + \sqrt{1 - i R_1}\right) \\
		& \qquad - \left( \sqrt{1 + i v^{\dagger} R_2 v} + \sqrt{1 - i v^{\dagger} R_2 v} \right) \left(1 - i R_1 + {\textstyle \sqrt{1 + R_1^2}} \right).
	\end{split}
\end{align}
Finally, note that the related expression
\begin{align}
	\begin{split}
		& \left( 1 - i R_1 + {\textstyle \sqrt{1 + R_1^2} } \right) \left( \sqrt{1 + i R_1} + \sqrt{1 - i R_1}\right) \\
		& \qquad - \left( \sqrt{1 + i R_1} + \sqrt{1 - i R_1} \right) \left(1 - i R_1 + {\textstyle \sqrt{1 + R_1^2}} \right),
	\end{split}
\end{align}
vanishes; hence, in order to show $\sqrt{X_2}-\sqrt{X_1}$ is Hilbert-Schmidt, it suffices to show that
$\sqrt{1 + i v^{\dagger} R_2 v} + \sqrt{1 - i v^{\dagger} R_2 v}$ is Hilbert-Schmidt equivalent to $\sqrt{1 + i R_1} + \sqrt{1 - i R_1}.$
This is true because one can obtain $\sqrt{1 + i R_1} + \sqrt{1 - i R_1}$ by applying the Lipschitz function
\begin{equation}
	f(t) = \sqrt{2} \sqrt{1+t}
\end{equation}
to the operator $\sqrt{1 + R_1^2}$, and one can obtain $\sqrt{1 + i v^{\dagger} R_2 v} + \sqrt{1 - i v^{\dagger} R_2 v}$ by applying the same function to $\sqrt{1 + (v^{\dagger} R_2 v)^2},$ which we already know is Hilbert-Schmidt equivalent to $\sqrt{1 + R_1^2}$.

For the other direction, we assume $\sqrt{X_2} - \sqrt{X_1}$ is Hilbert-Schmidt, and we wish to show that $Q-1$ and $\sqrt{1 + (v^{\dagger} R_2 v)^2} - \sqrt{1 + R_1^2}$ are Hilbert-Schmidt.
By left-multiplying and right-multiplying $\sqrt{X_2} - \sqrt{X_1}$ by $\sqrt{X_2} + \sqrt{X_1}$, and taking linear combinations, one easily finds the Hilbert-Schmidt condition on $X_2 - X_1 = Q - 1.$
So it remains to show that $\sqrt{1 + (v^{\dagger} R_2 v)^2} - \sqrt{1 + R_1^2}$ is Hilbert-Schmidt.
But once we know that $Q - 1$ is Hilbert-Schmidt, we know that $Q^{1/2} - 1$ is Hilbert-Schmidt, so write
\begin{equation}
	\sqrt{X_2} - \sqrt{X_1}
		= \sqrt{Q^{1/2} (1 - i v^{\dagger} R_2 v) Q^{1/2}} - \sqrt{1 - i R_1}
\end{equation}
and apply the Birman-Solomyak theorem with the absolute value function to conclude
\begin{equation}
	\sqrt{X_2} - \sqrt{X_1}
	\sim_{\text{H.S.}} \sqrt{1 - i v^{\dagger} R_2 v} - \sqrt{1 - i R_1}.
\end{equation}
Since we started with the assumption that $\sqrt{X_2} - \sqrt{X_1}$ was Hilbert-Schmidt, we have learned that the right-hand side above is Hilbert-Schmidt.
Conjugating by the complex conjugation operator $\Gamma$, which is antilinear and preserves $R_2$ and $R_1$ we learn that we have both
\begin{equation}
	\sqrt{1 - i v^{\dagger} R_2 v} - \sqrt{1 - i R_1} \quad \text{ Hilbert-Schmidt}
\end{equation}
\textit{and}
\begin{equation}
	\sqrt{1 + i v^{\dagger} R_2 v} - \sqrt{1 + i R_1} \quad \text{ Hilbert-Schmidt}.
\end{equation}
Adding these expressions, we learn
\begin{equation}
	\left( \sqrt{1 + i v^{\dagger} R_2 v} + \sqrt{1 - i v^{\dagger} R_2 v} \right) - \left(\sqrt{1+ i R_1} - \sqrt{1 - i R_1}\right) \quad \text{ Hilbert-Schmidt}.
\end{equation}
Applying the Lipschitz function $f(t) = (t^2 - 2)/2,$ one concludes
\begin{equation}
	\sqrt{1 + (v^{\dagger} R_2 v)^2} - \sqrt{1 + R_1^2} \quad \text{ Hilbert-Schmidt},
\end{equation}
as desired.

\subsection{Consequences for factorial states}
\label{sec:factorial-consequences}

Now we have proven theorem \ref{thm:passing-to-canonical}.
In practice, this theorem will be used to reduce the problem of excitability for general mixed states, to the simpler problem of excitability for purified states.
Indeed, with theorem \ref{thm:passing-to-canonical} in hand, we may now prove theorem \ref{thm:mixed-from-factor}, which addresses mixed-state excitability $\omega_2 \prec \omega_1$ subject to the assumption that $\omega_1$ is a factorial state.
Relaxing the factorial assumption is the task of section \ref{sec:centrally-pure}.

Theorem \ref{thm:mixed-from-factor} says that if $\omega_1$ is factorial, then we have $\omega_2 \prec \omega_1$ if and only if we have $\mu_2 \prec \mu_1$ and $\sqrt{X_2} - \sqrt{X_1}$ is Hilbert-Schmidt.
Notice that there is no condition checking that $Q$ is kernel-free.
The reason is that whenever we have $\mu_2 \prec \mu_1,$ we may use the identity $L^{\dagger} R_2 L = R_1$ from equation \eqref{eq:R1-from-R2}.
This equation implies that anything in the kernel of $L$ must also be in the kernel of $R_1.$
So when $\omega_1$ is factorial, meaning $R_1$ has trivial kernel, we may conclude that $L$ and $Q = L^{\dagger} L$ have trivial kernel.
So when $\omega_1$ is factorial, the triviality of the kernel of $Q$ does not need to be checked separately from the condition $\mu_2 \prec \mu_1.$

Now we prove the direction of theorem \ref{thm:mixed-from-factor} that starts with the assumption $\omega_2 \prec \omega_1.$
Since $\omega_1$ is factorial, proposition \ref{prop:two-way-excitability} tells us that we automatically have $\omega_1 \prec \omega_2.$
This leads to an isomorphism of the von Neumann algebras $\A_{\omega_1}$ and $\A_{\omega_2}$, so $\omega_2$ is factorial as well.
Since the canonical purification of a factorial state is pure, this means that both $\hat{\omega}_1$ and $\hat{\omega}_2$ are pure.
The general theorem \ref{thm:canonical-checking} tells us that $\omega_2 \prec \omega_1$ implies $\hat{\omega}_2 \prec \hat{\omega}_1,$ and since both of these states are pure, we learn from theorem \ref{thm:both-pure} that we must have $\hat{\mu}_2 \prec \hat{\mu}_1$, with $\hat{Q} - 1$ Hilbert-Schmidt.
Since $\hat{\omega}_1$ is pure we automatically get $\ker{\hat{Q}} = 0,$ and theorem \ref{thm:passing-to-canonical} then tells us that we have $\mu_2 \prec \mu_1$, and $\sqrt{X_2} - \sqrt{X_1}$ Hilbert-Schmidt.

In the other direction, assume that $\omega_1$ is factorial, and that we have $\mu_2 \prec \mu_1$, together with $\sqrt{X_2} - \sqrt{X_1}$ Hilbert-Schmidt.
As discussed above, $\ker{Q}$ is automatically trivial, so theorem \ref{thm:passing-to-canonical} gives $\hat{\mu}_2 \prec \hat{\mu}_1,$ with $\hat{Q} - 1$ Hilbert-Schmidt, and with the kernel of $\hat{Q}$ trivial.
We want to use these facts to reach the conclusion $\hat{\omega}_2 \prec \hat{\omega}_1,$ from which we can conclude $\omega_2 \prec \omega_1$ by theorem \ref{thm:canonical-checking}.

If we knew that $\hat{\omega}_1$ and $\hat{\omega}_2$ were both pure, then we could apply theorem \ref{thm:both-pure}, and we would be done.
But so far, we only know that $\omega_1$ is factorial, so that $\hat{\omega}_1$ is pure --- we don't know anything about $\hat{\omega}_2.$
Let us show that the assumptions we are making actually imply that $\hat{\omega}_2$ is pure.

Since we know that $\hat{Q} - 1$ is Hilbert-Schmidt and $\ker{\hat{Q}}$ is trivial, we know that $\hat{Q}$ is invertible,\footnote{We have used this argument before, but we remind the reader that since $\hat{Q} - 1$ is Hilbert-Schmidt, $\hat{Q}$ must be invertible away from its kernel, since it cannot have infinitely many eigenvalues piling up near zero. If we also know that the kernel is trivial, then $\hat{Q}$ is simply invertible.} so we may conclude that $\hat{L}$ is invertible.
Using the identity $\hat{L}^{\dagger} \hat{R}_2 \hat{L} = \hat{R}_1$, we conclude that $\hat{R}_2$ has trivial kernel, since if we had
\begin{equation}
    \hat{R}_2 |\psi\rangle = 0,
\end{equation}
then this would imply
\begin{equation}
    \hat{R}_1 L^{-1} |\psi\rangle = 0,
\end{equation}
but this implies $|\psi\rangle = 0,$ since $\hat{R}_1$ has trivial kernel.
So $\hat{R}_2$ has trivial kernel, and $\A_{\hat{\omega}_2}$ has trivial center.
Since $\hat{\omega}_2$ is centrally pure, we may conclude from the fact that it is factorial that $\hat{\omega}_2$ must be genuinely pure.
This completes the proof, since now we may apply theorem \ref{thm:both-pure} to conclude $\hat{\omega}_2 \prec \hat{\omega}_1,$ hence $\omega_2 \prec \omega_1.$

\subsection{A comment on exciting from pure states}
\label{sec:exciting-from-pure}

When $\omega_1$ is pure and we have $\mu_2 \prec \mu_1$, the condition that $\sqrt{X_2} - \sqrt{X_1}$ is Hilbert-Schmidt can be expressed as a condition directly on the operator $X_2 - X_1 = Q - 1.$
This is the content of theorem \ref{thm:mixed-from-pure}.
Here we provide a proof.
We already know, from earlier in this section, that the condition that $\sqrt{X_2} - \sqrt{X_1}$ is Hilbert-Schmidt is equivalent to the conditions
\begin{equation}
	Q - 1 \quad \text{ Hilbert-Schmidt}
\end{equation}
and
\begin{equation}
	\sqrt{1 + (v^{\dagger} R_2 v)^2} - \sqrt{1 + R_1^2} \quad \text{ Hilbert-Schmidt},
\end{equation}
with $L = v Q^{1/2}.$
When $\omega_1$ is pure, $R_1$ squares to $-1$, so the second condition above is simply
\begin{equation}
	\sqrt{1 + (v^{\dagger} R_2 v)^2} \quad \text{ Hilbert-Schmidt},
\end{equation}
or equivalently
\begin{equation}
	1 + (v^{\dagger} R_2 v)^2 \quad \text{ trace-class}.
\end{equation}
Moreover, the identity $L^{\dagger} R_2 L = R_1$ implies --- as in equation \eqref{eq:L-pure-lower-bound} --- that $L$ is invertible, so we can write $v^{\dagger} R_2 v = Q^{-1/2} R_1 Q^{-1/2},$ and the condition becomes
\begin{equation}
	1 + (Q^{-1/2} R_1 Q^{-1/2})^2 \quad \text{ trace-class}.
\end{equation}
Finally, since $Q$ is invertible, we can left and right multiply by $Q^{1/2}$ to obtain the condition
\begin{equation} \label{eq:Q-Qinv-expression}
	Q + R_1 Q^{-1} R_1 \quad \text{ trace-class}.
\end{equation}
Summing up, we have found that this condition, together with $Q - 1$ Hilbert-Schmidt, is the same (in the case of $\omega_1$ pure) as the condition that $\sqrt{X_2} - \sqrt{X_1}$ is Hilbert-Schmidt.
This already gives part of theorem \ref{thm:mixed-from-pure}.
To finish proving that theorem, we need to show that these two conditions together are equivalent to the statement that the ``diagonal elements'' in the block decomposition of $Q-1$ are trace-class.

Writing 
\begin{equation}
	Q = \begin{pmatrix} a & b \\ b^{\dagger} & c \end{pmatrix},
\end{equation}
we are interested in when $a-1$ and $c-1$ are trace-class.
Here, the block decomposition is defined by noting $R_1$ has eigenvalues $\pm i,$ and letting the first row correspond to the $+i$ eigenspace and the second row correspond to the $-i$ eigenspace. 
The invertibility of $Q$ means that it is strictly positive, with $Q \geq \lambda > 0,$ and this implies important properties for the blocks. In particular, $a$ and $c$ are strictly positive, with $a \geq \lambda, c \geq \lambda.$
Invertibility of $Q$ also implies that the operators $y$ and $y'$ defined by
\begin{equation}
\label{eq:y-yp-definition}
    y\equiv c - b^{\dagger} a^{-1} b\qquad y'\equiv a - b c^{-1} b^{\dagger}
\end{equation}
are invertible.
This is because $y$ and $y'$ arise as diagonal elements after conjugating \(Q\) by the invertible operators
\begin{equation}
\begin{pmatrix} 1 & 0 \\ - b^{\dagger} a^{-1} & 1 \end{pmatrix}
\quad \text{and} \quad
\begin{pmatrix} 1 & - b c^{-1} \\ 0 & 1 \end{pmatrix}, 
\end{equation}
respectively.

The relation between the blocks of $Q$ and the operator  $Q+R_1Q^{-1}R_1$ can now be computed by writing the inverse of $Q$ as\footnote{To verify this, check the identity $$\frac{1}{a - b c^{-1} b^{\dagger}} = a^{-1} + a^{-1} b \frac{1}{c - b^{\dagger} a^{-1} b} b^{\dagger} a^{-1},$$ then substitute this into the top-left block.}
\begin{equation}
	Q^{-1}
		= \begin{pmatrix}
			 y'^{-1} &  - a^{-1} b y^{-1}  \\
			- y^{-1}b^{\dagger} a^{-1} 
			& y^{-1}
			\end{pmatrix},
\end{equation}
and since in this block decomposition we simply have $R_1 = \begin{pmatrix} i & 0 \\ 0 & -i \end{pmatrix},$ we can explicitly compute
\begin{equation}
	Q + R_1 Q^{-1} R_1
		= \begin{pmatrix} a & b \\ b^{\dagger} & c \end{pmatrix}
			- \begin{pmatrix}
			 y'^{-1} &   a^{-1} b y^{-1}  \\
			 y^{-1}b^{\dagger} a^{-1} 
			& y^{-1}
			\end{pmatrix}.
\end{equation}

First we assume that $\sqrt{X_2} - \sqrt{X_1}$ is Hilbert-Schmidt, so that our goal is to prove that $a-1$ and $c-1$ are trace-class.
In this case we have $Q-1$ Hilbert-Schmidt, so by acting with block projections (which in particular are bounded operators), we know that that $a-1, b,$ and $c-1$ are Hilbert-Schmidt.
To upgrade the Hilbert-Schmidt conditions on $a-1$ and $c-1$ to trace-class conditions, we will need to use that $Q + R_1 Q^{-1} R_1$ is trace-class.
In particular, this means that the diagonal blocks of $Q + R_1 Q^{-1} R_1,$ which are $a-y'^{-1}$ and $c-y^{-1}$, are trace-class.
Multiplying $a-y'^{-1}$ on the left by $a^{-1}$ and on the right by $y'$ --- and performing an analogous manipulation for the second diagonal block --- one finds that $y'-a^{-1}$ and $y-c^{-1}$ are trace class.
Since $b$ is Hilbert-Schmidt, the operators $b c^{-1} b^{\dagger}$ and $b^{\dagger} a^{-1} b$ appearing in the definitions of $y$ and $y'$ are trace-class,\footnote{The product of two Hilbert-Schmidt operators is trace-class, and the product of a Hilbert-Schmidt operator with a bounded operator is Hilbert-Schmidt, so if $b$ is Hilbert-Schmidt, then $b T b^{\dagger}$ is trace-class for any bounded $T$.} and this implies
\begin{equation}
	a - a^{-1} \quad \text{ trace-class}
    \label{eq:a-ainv-traceclass}
\end{equation}
and
\begin{equation}
	c - c^{-1} \quad \text{ trace-class}.
\end{equation}
Multiplying by $a$ in one expression and by $c$ in the other, one finds
\begin{equation}
	a^2 - 1 \quad \text{ trace-class}
\end{equation}
and
\begin{equation}
	c^2 - 1 \quad \text{ trace-class}.
\end{equation}

We are almost at what we wish to conclude, which is that $a-1$ and $c-1$ are trace-class.
At this point, one can apply a version of the Birman-Solomyak theorem to reach the desired conclusion.
Unfortunately, this is complicated to explain, because for trace-class operators, the Birman-Solomyak theorem is expressed in terms of ``Besov functions,'' not Lipschitz ones.
Luckily, there is also a more direct path to the desired result, which is to use the following expression for the square root of a positive operator:\footnote{For a systematic discussion of how such integral expressions are derived, see \cite{sorce-blog-Pick}.}
\begin{equation}
	\sqrt{A} 
		= \frac{1}{\pi} \int_{0}^{\infty} dt\, t^{-1/2} (1 + t A^{-1})^{-1}.
\end{equation}
So putting $A = a^2,$ one finds
\begin{align}
	\begin{split}
	a - 1
		& = \frac{1}{\pi} \int_{0}^{\infty} dt\, t^{-1/2} \left[(1+t a^{-2})^{-1} - (1+t)^{-1} \right] \\
		& = \frac{1}{\pi} \int_{0}^{\infty} dt\, t^{-1/2} (1+t a^{-2})^{-1} \left[ (1 + t) - (1 + t a^{-2}) \right] (1+t)^{-1} \\
		& = \frac{1}{\pi} \int_{0}^{\infty} dt\, t^{1/2} (1+t a^{-2})^{-1} \left[ 1 - a^{-2} \right] (1+t)^{-1}.
	\end{split}
\end{align}
Taking trace norms $\norm{\cdot}_1$ on both sides and using a Hölder-type inequality to pull $\norm{(1-ta^{-2})^{-1}}_{\infty}$ out of the full trace norm on the right-hand side, one gets
\begin{align}
	\begin{split}
		\lVert a - 1 \rVert_1
		& \leq \frac{1}{\pi} \lVert 1 - a^{-2} \rVert_1 \int_{0}^{\infty} dt\, \frac{t^{1/2}}{1+t} \frac{1}{1 + t \lVert a^{-2} \rVert_{\infty}}.
	\end{split}
\end{align}
The integral converges, and the remaining trace norm on the right-hand side is finite, since it can be obtained by multiplying the finite-trace operator $a^2 - 1$ by the bounded operator $a^{-2}.$
We conclude that $a-1$ is trace-class; a similar argument gives that $c-1$ is trace-class, and we have proved the part of theorem \ref{thm:mixed-from-pure} that says ``if $\omega_2 \prec \omega_1$ and $\omega_1$ is pure, then the diagonal blocks of $Q-1$ are trace-class.''

In the other direction, suppose that $a-1$ and $c-1$ are trace-class.
Our goal is to show that $Q-1$ is Hilbert-Schmidt, and to show the condition in \eqref{eq:Q-Qinv-expression} that $Q + R_1 Q^{-1} R_1$ is trace-class.
To prove the Hilbert-Schmidt condition on $Q-1$, we note that the inequality $X_2=Q - i R_1 \geq 0$ gives
\begin{equation}
	\begin{pmatrix}
		a + 1 & b \\
		b^{\dagger} & c - 1
	\end{pmatrix}
		\geq 0,
\end{equation}
and conjugation by note minus $\begin{pmatrix} 1 & 0 \\ -b^{\dagger} (1 + a)^{-1} & 1 \end{pmatrix}$ gives
\begin{equation}
	c - 1 \geq b^{\dagger} (1+a)^{-1} b \geq \frac{1}{\lVert 1 + a \rVert} b^{\dagger} b
\end{equation}
upon restricting the positivity constraint to one of the diagonal blocks.
Since $c-1$ is trace-class, we learn that $b^{\dagger} b$ is trace-class, hence $b$ is Hilbert-Schmidt --- we did not need to put this in as an assumption, and we learn that $Q-1$ is Hilbert-Schmidt, since $a-1$ and $c-1$ are Hilbert-Schmidt by virtue of being trace-class.

Now we proceed to prove the condition in equation \eqref{eq:Q-Qinv-expression} that $Q + R_1 Q^{-1} R_1$ is trace-class.
Since $a-1$ is trace-class, we can multiply by $a^{-1}$ to learn that $a^{-1} - 1$ is trace-class, and together with $b$ being Hilbert-Schmidt, this implies that the operator
\begin{equation}
	y'-a^{-1}
		= (a - 1) - b c^{-1} b^{\dagger} - (a^{-1} - 1)
\end{equation}
is trace-class, and similarly for $y-c^{-1}$.
Multiplying by bounded operators, as above equation \eqref{eq:a-ainv-traceclass}, we learn that the diagonal blocks $a-y'^{-1}$ and $c-y^{-1}$ of $Q + R_1 Q^{-1} R_1$ are trace-class.

We will complete our proof by showing that the off-diagonal block 
\begin{equation}
    b - a^{-1} b y^{-1}
\end{equation} 
is  trace-class as well.
Note that 
\begin{equation}
    y-1=c - b^{\dagger} a^{-1} b - 1
\end{equation}
is trace-class (because $c-1$ is trace-class and $b$ is Hilbert-Schmidt), so multiplying by $y^{-1}$ we find $y^{-1}-1$ is trace-class, and since $a^{-1} -1$ is trace-class, we find
\begin{equation}
	b - a^{-1} b y^{-1}
		\sim_{\text{trace-class}} (b - b) = 0.
\end{equation}

So the diagonal and off-diagonal blocks of $Q + R_1 Q^{-1} R_1$ are trace-class, and hence the full operator is trace-class.
This completes the proof of the part of theorem \ref{thm:mixed-from-pure} that says, ``if $\omega_1$ is pure and $\mu_2 \prec \mu_1,$ with the diagonal blocks of $Q-1$ being trace-class, then we have $\omega_2 \prec \omega_1$.''

\section{Centrally pure Gaussian states}
\label{sec:centrally-pure}

So far, we have proved theorem \ref{thm:mixed-from-factor}.
This gives a general criterion for excitability $\omega_2 \prec \omega_1$ when $\omega_1$ is a factorial state.
The key insight was that when $\omega_1$ is factorial, the canonical purification $\hat{\omega}_1$ is pure, so as explained in section \ref{sec:factorial-consequences}, one can apply the theorem that says pure states can be excited out of one another if and only if one has $\hat{\mu}_2 \prec \hat{\mu}_1$ and the operator $\hat{Q} - 1$ is Hilbert-Schmidt.

For a non-factorial state, the canonical purification is not pure in the technical sense; instead, it is only centrally pure.
Crucially, we can still prove the paper's main theorem --- theorem \ref{thm:general-theorem} --- by using the simple logic of section \ref{sec:factorial-consequences}, as long as the excitability conditions for  centrally pure states mimic those of pure states.
So to complete the proof, we will establish the appropriate excitability criteria for centrally pure states.

In this section, we will take $\omega_1$ and $\omega_2$ to be centrally pure states.\footnote{In proving the main theorem of the paper, these will be the canonical purifications $\hat{\omega}_1$ and $\hat{\omega}_2$, but for present section it is best not to muddy the notation by putting hats everywhere.}
As in remark \ref{rem:centrally-pure-def}, a centrally pure state $\omega$ is characterized by the inclusion ${\cal A}_{\omega}'\subseteq {\cal A}_{\omega}$, which expresses the absence of new ``purifying'' degrees of freedom introduced to ${\cal A}_0$ by the commutant in the GNS representation ${\cal H}_{\omega}$.
In the particular case of centrally pure Gaussian states, central purity is equivalent to the statement that $R_1$ and $R_2$ have spectra $\pm i$ and zero.
We will prove theorem \ref{thm:c-pure-from-c-pure}, which states that one has $\omega_2 \prec \omega_1$ if and only if $\mu_2 \prec \mu_1$, $\ker{Q} = 0,$ and the operator $Q - 1$ is Hilbert-Schmidt.

The idea of the proof is to split $\omega_1$ into a ``pure piece'' and an ``abelian piece.''
In particular, while the $\pm i$ eigenspaces of $R_1$ generate a subsystem for which $\omega_1$ is pure (in the strict, factorial sense), the kernel of $R_1$ generates a subsystem for which $\omega_1$ is classical (in the sense of describing a state on an abelian algebra).
One can then break the criterion $\omega_2 \prec \omega_1$ into these two pieces, and apply theorem \ref{thm:mixed-from-pure} --- which tells us when a general mixed state can be excited from a pure state --- together with a new theorem that will tell us when excitability is possible on an abelian algebra.
Putting these together will yield a proof of theorem \ref{thm:general-theorem}.

\subsection{Structure of centrally pure states}
\label{sec:centrally-pure-structure}

We begin by explaining how a centrally pure Gaussian state $\omega$ can be split into a ``pure piece'' and an ``abelian piece'' by appealing to the structure of the one-particle Hilbert space.

Let $\omega$ be a centrally pure state. In appendix \ref{app:pure} we show that, within the corresponding $\K_\mu$ space, the anti-self-adjoint operator $R$ squares to $-1$ on the orthocomplement of its kernel.
It follows that its spectrum is simply $\{+i, 0, -i\}$.
Let $\Pi_{i}, \Pi_{0},$ and $\Pi_{-i}$ be the projectors on the respective eigenspaces, so that one has 
\begin{equation}
	1 = \Pi_i + \Pi_0 + \Pi_{-i}
\end{equation}
and 
\begin{equation}
    R=i(\Pi_{i}-\Pi_{-i})\,.
\end{equation}

The one-particle Hilbert space $\H_{\omega}^{(1)}$ is, as usual, unitarily equivalent to the orthocomplement of $(R = - i)$ via the map $\phi[f] |\omega\rangle \mapsto \sqrt{1 - i R}|f\rangle_{\mu}.$
In the present setting, this is simply
\begin{equation}
	\phi[f] |\omega\rangle
		\mapsto \sqrt{2} \Pi_{i}|f\rangle_{\mu} + \Pi_0 |f\rangle_{\mu}.
\end{equation}
The $+i$ eigenstates of $R$ and the kernel of $R$ are mutually orthogonal.
From this, we see that the one-particle Hilbert space $\H_{\omega}^{(1)}$ splits into two orthogonal pieces --- a piece generated by the $+i$ eigenstates of $R$,
\begin{equation}
	\H_{\omega}^{(1, i)}
		\equiv \bar{\{\phi[\psi] |\omega\rangle \text{ such that } R |\psi\rangle_{\mu} = i |\psi\rangle_{\mu}\}},
\end{equation}
and a piece generated by the kernel of $R$,
\begin{equation}
	\H_{\omega}^{(1, 0)}
	\equiv \bar{\{\phi[\psi] |\omega\rangle \text{ such that } R |\psi\rangle_{\mu} = 0\}}.
\end{equation}

In appendix \ref{app:Haag-duality}, we explain that for any orthodecomposition of $\H_{\omega}^{(1)}$ into subspaces, one gets a corresponding tensor product decomposition of the full Hilbert space $\H_{\omega}.$
In the present setting, this is obtained by taking a generic $n$-particle state, $:\phi[f_1] \dots \phi[f_n]: |\omega\rangle,$ decomposing each $f_j$ in terms of eigenstates of $R$ as
\begin{equation}
	f_j = f_{j, i} + f_{j, 0} + f_{j, -i},
\end{equation}
and writing 
\begin{equation}
	:\phi[f_1] \dots \phi[f_n]: |\omega\rangle
		= :\phi[f_{1, i} + f_{1, 0}] \dots \phi[f_{n, i} + f_{n, 0}]: |\omega\rangle.
\end{equation}
Then one notices that due to orthogonality of $f_{j, i}$ and $f_{k, 0},$ the normal ordering splits up cleanly between $+i$ terms and kernel terms; i.e., from the definition of normal ordering in equation \eqref{eq:normal-ordering} of the appendix, one has
\begin{align}
	\begin{split}
		& :\phi[f_1] \dots \phi[f_n]: |\omega\rangle \\
		& \qquad = \sum_{p=0}^{n} \sum_{j_1 < \dots < j_p} 
		\left( :\phi[f_{1, i}] \dots \hat{\phi[f_{j_1, i}]} \dots \hat{\phi[f_{j_p, i}]} \dots \phi[f_{n, i}]: \right) \left( : \phi[f_{j_1, 0}] \dots \phi[f_{j_p, 0}] :\right)  |\omega\rangle.
	\end{split}
\end{align}
When taking overlaps of such states, the $R=+i$ terms and the $R=0$ terms have zero overlap with one another, so all overlaps factorize between these two different pieces.
Consequently, one can introduce the subspaces
\begin{align}
	\H_{\omega, i}
		& \equiv \text{ states obtained by acting on $\omega$ with normal-ordered $R=+i$ fields}, \\
	\H_{\omega, 0}
		& \equiv \text{ states obtained by acting on $\omega$ with normal-ordered $R=0$ fields},
\end{align}
and the map
\begin{align} \label{eq:i-0-tensor-map}
	\begin{split}
	& \left( :\phi[f_{1, i}] \dots \phi[f_{m, i}]:\right)
		\left(:\phi[g_{1, 0}] \dots \phi[g_{n, 0}]: \right) |\omega\rangle \\
		& \qquad \mapsto \left( :\phi[f_{1, i}] \dots \phi[f_{m, i}]: |\omega\rangle \right) \otimes \left( :\phi[g_{1,0}] \dots \phi[g_{n, 0}]: |\omega\rangle \right)
	\end{split}
\end{align}
furnishes a unitary equivalence between $\H_{\omega}$ and $\H_{\omega, i} \otimes \H_{\omega, 0}.$

A very nice feature of this tensor factorization is that the von Neumann algebra $\A_{\omega}$ also factorizes.
In particular, we claim the identity
\begin{equation} \label{eq:c-pure-vN-reduction}
	\A_{\omega}
		\cong \B(\H_{\omega,i}) \otimes \mathcal{Z}_{\omega, 0}.
\end{equation}
There is a lot of structure in this statement; it says that $\A_{\omega}$ is generated by the \textit{full} algebra of bounded operators on $\H_{\omega, i}$ together with an \textit{abelian} algebra $\mathcal{Z}_{\omega, 0}$ acting on $\H_{\omega,0}.$
Moreover, given this decomposition, the central purity relation $\A_{\omega}' \subseteq \A_{\omega}$ will clearly imply a maximality condition $\Z_{\omega,0}' = \Z_{\omega, 0}$, which in turn guarantees that the center of $\A_\omega$ can be identified with $1\otimes\mathcal{Z}_{\omega,0}$ under the unitary equivalence \eqref{eq:i-0-tensor-map}.

We dedicate the remainder of this subsection to showing the decomposition in equation \eqref{eq:c-pure-vN-reduction}.
To accomplish this, we consider a generic real test function $f,$ and decompose it within $\K_{\mu}$ as
\begin{equation}
	|f\rangle_{\mu} = |f_i + f_{-i}\rangle_{\mu} + |f_0\rangle_{\mu}.
\end{equation}
Because $R$ commutes with complex conjugation, each of the states in the above expression is real, so one can consider the Weyl operators $e^{i \phi[f_i + f_{-i}]}$ and $e^{i \phi[f_0]}.$
Because these correspond to eigenstates of the $R$ operators in distinct eigenspaces, the operators $\phi[f_i + f_{-i}]$ and $\phi[f_0]$ commute, so one has
\begin{equation}
	e^{i \phi[f]}
		= e^{i \phi[f_{i} + f_{-i}]} e^{i \phi[f_0]}.
    \label{eq:weyls-factorize}
\end{equation}

Since $\A_{\omega}$ is the von Neumann algebra generated by all Weyl operators $e^{i \phi[f]},$ we learn that we can decompose it as
\begin{equation} \label{eq:c-pure-omega-add}
	\A_{\omega}
		= \A_{\omega, \pm i} \vee \A_{\omega, 0},
\end{equation}
with
\begin{equation} \label{eq:A-pmi}
	\A_{\omega, \pm i} = \langle e^{i \phi[f_i + f_{-i}]}\rangle
\end{equation}
and
\begin{equation}
	\A_{\omega, 0} = \langle e^{i \phi[f_0]}\rangle.
\end{equation}
Since the commutator operator $R$ annihilates $|f_0\rangle_{\mu},$ the algebra $\A_{\omega,0}$ is clearly abelian.

It is clear that under the unitary equivalence in equation \eqref{eq:i-0-tensor-map}, the decomposition from equation \eqref{eq:c-pure-omega-add} becomes a tensor product. 
So all that remains is to show that under this map, $\A_{\omega, \pm i}$ becomes the full algebra of bounded operators on $\H_{\omega, i}$.\footnote{Recall from the comment below equation \eqref{eq:c-pure-vN-reduction} that this automatically implies $\A_{\omega, 0}$ is the center of $\A_{\omega}$.}
To show this, all we will need is that the commutant of $\A_{\omega, \pm i}$ acts trivially on $\H_{\omega, i}.$
This can be accomplished using the commutator computations from appendix \ref{app:Haag-duality}.
In the language of appendices \ref{app:generalized-Weyl}-\ref{app:non-cs}, one associates, with any subset $X$ of the one-particle Hilbert space, the von Neumann algebra
\begin{equation}
	\A(X)
	= \langle e^{i \chi_{\alpha}} \, | \, |\alpha\rangle \in X\rangle.
\end{equation}
The $\chi$ operators, introduced above equation \eqref{eq:below-chi-explanation},
obey the following property for any $f_{\pm i}$ in the $\pm i$ eigenspaces of $R$:
\begin{equation}
\chi_{\phi[f_i]\ket{\omega}}=\phi[f_i+f_{-i}].
\end{equation}
Hence, by consulting equation \eqref{eq:A-pmi}, we see that $\A_{\omega, \pm i}$ is simply $\A(X_{\pm i})$, where the spaces $X_{\pm i}$ are defined via
\begin{equation}
	X_{\pm i}
		= \bar{\{\phi[\psi] |\omega\rangle \text{ such that } R |\psi\rangle_{\mu} = i |\psi\rangle_{\mu} \} }.
        \label{eq:X-pm-subspace}
\end{equation}
Here, as usual, we are generalizing from test functions $f$ to general elements $\ket{\psi}_{\mu}$ of ${\cal K}_{\mu}$, so in particular $X_{\pm i}$ can be thought of as the real subspace of $\H_{\omega}^{(1)}$ obtained by acting with $U^{\dagger}$ on the $+i$ eigenspace of $R$, where $U$ is the isometry from $\H_{\omega}^{(1)}$ into $\K_{\mu}.$

From appendix \ref{app:Haag-duality}, the commutant is
\begin{equation}
	\A(X_{\pm i})'
		= \A(X_{\pm i}'),
\end{equation}
with
\begin{equation}
	X_{\pm i}'
		= \{|\beta\rangle \in \H_{\omega}^{(1)} \text{ such that } \text{Im}\langle \beta | \alpha \rangle = 0 \text{ for all } |\alpha\rangle \in X_{\pm i}\}.
\end{equation}
So $X_{\pm i}'$ consists of the one-particle states $|\beta\rangle$ with
\begin{equation}
	\text{Im}\langle \beta | \phi[\psi] |\omega\rangle = 0
		\quad \text{ for all } \quad R |\psi\rangle_{\mu} = i |\psi\rangle_{\mu}.
\end{equation}
Acting on both the bra and the ket with $U$, we see that this is just the set of states $|\beta\rangle$ such that $U |\beta\rangle$ has vanishing imaginary overlap with the $+i$ eigenspace of $R$.
But the $+i$ eigenspace of $R$ is a complex vector space; so if a state has vanishing imaginary overlap that space, it is orthogonal to that space.
Consequently, we see that $X_{\pm i}$ is the set of states $|\beta\rangle$ such that $U |\beta\rangle$ is orthogonal to the space $R=+i.$
Since the image of $U$ is automatically orthogonal to the space $R=-i$ by construction, we conclude that $U |\beta\rangle$ must be in the kernel of $R$, since this is the only space left in $\K_{\mu}.$
Consequently, we learn
\begin{equation}
	X_{\pm i}'
		= \bar{\{\phi[\psi] |\omega\rangle \text{ such that } R |\psi\rangle_{\mu} = 0 \} }.
\end{equation}
So $\A_{\omega, \pm i}' = \A(X_{\pm i})'=\A(X_{\pm i}')$ is exactly the algebra $\A_{\omega, 0},$ which acts trivially on the tensor factor $\H_{\omega,i}$ within the decomposition $\H_{\omega} = \H_{\omega, i} \otimes \H_{\omega, 0}.$
This completes the proof of the claimed decomposition in equation \eqref{eq:c-pure-vN-reduction}.

\subsection{Breaking excitability into tensor factors}
\label{sec:excitability-breakdown}

Now suppose we have states $\omega_1$ and $\omega_2,$ and assume that $\omega_1$ is centrally pure.
For the moment, we will not assume that $\omega_2$ is centrally pure.
By the results of the preceding subsection, one has a decomposition
\begin{equation}
	\H_{\omega_1}
		\cong \H_{\omega_1, i} \otimes \H_{\omega_1, 0}
\end{equation}
with a corresponding decomposition of the von Neumann algebra
\begin{equation}
	\A_{\omega_1}
	\cong \B(\H_{\omega_1, i}) \otimes \Z_{\omega_1, 0}.
\end{equation}
We wish to use this decomposition to simplify our understanding of the excitability relation $\omega_2 \prec \omega_1.$

Our starting point is proposition \ref{prop:alpha-definition} --- see also remark \ref{rem:alpha-bounded} --- which in the present setting tells us that $\omega_2 \prec \omega_1$ is equivalent to the condition that the ``identity map'' on Weyl operators extend to a normal homomorphism from $\A_{\omega_1}$ to $\A_{\omega_2}.$
I.e., if one considers the map $\alpha$ that takes $e^{i \phi[f]}$ acting on $\H_{\omega_1}$ to $e^{i \phi[f]}$ acting on $\H_{\omega_2},$ it is necessary and sufficient that this map be a $*$-homomorphism that is continuous with respect to the ultraweak topology.

Suppose first that one has $\omega_2 \prec \omega_1,$ so we know that $\alpha$ is a normal homomorphism.
We also know, from section \ref{sec:mu-boundedness}, that $\omega_2 \prec \omega_1$ implies $\mu_2 \prec \mu_1.$
Consequently, the map $L |f\rangle_{\mu_1} = |f \rangle_{\mu_2}$ is bounded.
It follows that the map $\alpha$, defined on test functions by
\begin{equation}
	\alpha(e^{i \phi[f]}) = e^{i \phi[f]},
\end{equation}
uniquely extends by ultraweak continuity to generic real elements $|\psi\rangle_{\mu_1}$ in $\K_{\mu_1}$ as
\begin{equation} \label{eq:alpha-L}
	\alpha(e^{i \phi[\psi]}) = e^{i \phi[L \psi]}.
\end{equation}
In particular, the algebra $\A_{\omega_1, \pm i}$ from the preceding section is generated by operators $e^{i \phi[\psi]}$ where $|\psi\rangle_{\mu}$ is in the support of $R$, while the algebra $\A_{\omega_1, 0}$ is generated by operators $e^{i \phi[\psi]}$ where $|\psi\rangle_{\mu}$ is in the kernel of $R.$
Consequently, $\alpha$ has two restrictions:
\begin{align}
	\alpha_{\pm i} \label{eq:alpha-i}
		& : \A_{\omega_1, \pm i} \to \A_{\omega_2}, \\
	\alpha_{0} \label{eq:alpha-0}
		& : \A_{\omega_1, 0} \to \A_{\omega_2},
\end{align}
each of which is a normal homomorphism, and each of which can be written explicitly in terms of equation \eqref{eq:alpha-L}.

Going in the other direction, we now claim that in order to prove that $\alpha$ extends to a normal homomorphism, it is sufficient to check that the maps $\alpha_{\pm i}$ and $\alpha_0$ --- which initially are defined only on Weyl operators --- extend to normal homomorphisms within the appropriate subalgebras.\footnote{This fact, and many pieces of the argument below, were explained to us by Lauritz van Luijk.}
This will give us a way to break the excitability condition $\omega_2 \prec \omega_1$ into two pieces, one corresponding to $\A_{\omega_1, \pm i}$ and one corresponding to $\A_{\omega_1, 0}.$

To prove this claim, we start from the necessary condition $\mu_2 \prec \mu_1,$ so that we can make sense of the bounded ``identity map'' $L : \K_{\mu_1} \to \K_{\mu_2}.$
Then equation \eqref{eq:alpha-L} gives maps $\alpha_{\pm i}$ and $\alpha_0$ that are defined on operators of the form $e^{i \phi[\psi]}$ within $\A_{\omega_1, \pm i}$ and $\A_{\omega_1, 0}.$
If we suppose further that each of these extends to a normal homomorphism, then we get maps as in equations \eqref{eq:alpha-i} and \eqref{eq:alpha-0}, and after rewriting $\H_{\omega_1} \cong \H_{\omega_1, i} \otimes \H_{\omega_1, 0},$ we can rewrite these as maps
\begin{align}
	\alpha_{\pm i}
		& : \B(\H_{\omega_1,i}) \to \A_{\omega_2}, \\
	\alpha_{0}
		& : \Z_{\omega_1,0} \to \A_{\omega_2}.
\end{align}
The map $\alpha_{\pm i}$ is actually quite constrained thanks to the observation that $\B(\H_{\omega_1, i})$ is a factor algebra.
It is a basic fact about von Neumann algebras that any normal homomorphism out of a factor must be an isomorphism onto its image.\footnote{To see this, note that since $\alpha_{\pm i}$ is a normal homomorphism, the only way it could fail to be an isomorphism onto its image would be if it failed to be injective.
But the kernel of $\alpha_{\pm i}$ is an ultraweakly closed, two-sided ideal, and factors admit no such nontrivial ideals --- see e.g. \cite[proposition II.3.12]{Takeaski:I}.
So the kernel of $\alpha_{\pm i}$ is either everything --- impossible, since we have $\alpha_{\pm i}(1) = 1$ --- or nothing, i.e., $\alpha_{\pm i}$ is injective.}
Consequently, the image
\begin{equation}
	\alpha_{\pm i}(\B(\H_{\omega_1, i})) \subseteq \A_{\omega_2}
\end{equation}
is an isomorphic copy of $\B(\H_{\omega_1, i}).$
This is not only a factor, it is a \textit{type I factor}.
As explained e.g. in \cite[section 7.1]{Sorce:types}, this means that there is a tensor product decomposition of $\H_{\omega_2}$ as
\begin{equation}
	\H_{\omega_2}
	\cong \K \otimes \mathcal{F}
\end{equation}
such that we have
\begin{equation}
	\alpha_{\pm i}(\B(\H_{\omega_1, i}))
	\cong \B(\K) \otimes 1_{\mathcal{F}}.
\end{equation}
The image of $\alpha_0,$ i.e. the algebra
\begin{equation}
	\alpha_0(\Z_{\omega_1, 0}) \subseteq \A_{\omega_2},
\end{equation}
necessarily commutes with the image of $\B(\H_{\omega_1,i}),$ thanks to the homomorphism property of $\alpha.$
So with respect to the above decomposition, $\alpha_0(\Z_{\omega_1, 0})$ commutes with $\B(\K),$ and therefore acts entirely on the tensor factor $\mathcal{F}$ within $\H_{\omega_2}.$
This implies that after writing $\H_{\omega_1}$ as $\H_{\omega_1, i} \otimes \H_{\omega_1, 0},$ and after writing $\H_{\omega_2}$ as $\K \otimes \mathcal{F},$ we have the decomposition $\alpha = \alpha_{\pm i} \otimes \alpha_0.$
A tensor product of normal homomorphisms is obviously a normal homomorphism; this completes the proof of the claim.

How will this discussion help us prove the excitability criteria for centrally pure states as in theorem \ref{thm:c-pure-from-c-pure}?
So far we have only assumed that $\omega_1$ is centrally pure; now we also assume $\omega_2$ is centrally pure, so we have decompositions
\begin{equation}
	\H_{\omega_2}
		\cong \H_{\omega_2, i} \otimes \H_{\omega_2, 0}
\end{equation}
and
\begin{equation} \label{eq:omega-2-central-decomp}
	\A_{\omega_2}
		\cong \B(\H_{\omega_2, i}) \otimes \Z_{\omega_2, 0}.
\end{equation}

A priori, the image of the restriction $\alpha_0$ was just some complicated subalgebra of $\A_{\omega_2}.$
But in fact, the image of $\alpha_0$ must lie within $\Z_{\omega_2, 0}.$
This is because $\alpha_0$ is defined to send elements $e^{i \phi[\psi]}$ with $R_1 |\psi\rangle_{\mu} = 0$ to elements $e^{i \phi[L \psi]}$ acting on $\H_{\omega_2}.$
But the identity $L^{\dagger} R_2 L=R_1$ tells us that $L$ maps the kernel of $R_1$ to the kernel of $R_2,$ which means that $\alpha$ maps $\Z_{\omega_1, 0}$ to a subalgebra of $\Z_{\omega_2, 0}.$
Hence, crucially, we can convert the task of understanding whether $\alpha_0$ extends to a normal homomorphism, to the task of establishing excitability conditions for Gaussian states on abelian algebras as in theorem \ref{thm:abelian}.
This will be the task of the following subsection.

We remark that no such simple classification is possible for the image of $\alpha_{\pm i}$ --- this will generally map $\B(\H_{\omega_1,i})$ to a subalgebra that is entangled between the two factors on the right-hand side of \eqref{eq:omega-2-central-decomp}.
But the restriction of $\omega_1$ to $\B(\H_{\omega_1, i})$ is pure, so it is possible to deal with continuity of $\alpha_{\pm i}$ using the general condition for exciting a mixed state out of a pure state as in theorem \ref{thm:mixed-from-pure}.

To summarize, excitability of centrally pure states breaks down as
\begin{center}
\begin{tikzpicture}[>=stealth, line width=0.7pt]

\node (A1) at (1,3.3) {$\mathcal A_{\omega_1}$};
\node (cong1) at (1.8,3.3) {$\cong$};
\node (BHtop) at (2.9,3.3) {$\mathcal A_{\omega_1,\pm i}$};
\node (Ztop)  at (5.2,3.3) {$\mathcal A_{\omega_1,0}$};
\node at (4.0,3.3) {$\;\;\otimes$};

\node (A2) at (1,0.8) {$\mathcal A_{\omega_2}$};
\node (cong2) at (1.8,0.8) {$\cong$};
\node (BHbot) at (2.9,0.8) {$\mathcal A_{\omega_2,\pm i}$};
\node (Zbot)  at (5.1,0.8) {$\;\mathcal A_{\omega_2,0}$};
\node at (4,0.8) {$\;\;\otimes$};

\draw[->] (Ztop.south) -- (Ztop.south |- Zbot.north);
\node at (5.5,2.1) {$\;\alpha_0$};

\draw[->] (A1.south) -- (A1.south |- A2.north);
\node at (0.65,2.1) {$\;\alpha$};

\coordinate (P) at (3.8,2.2);

\draw (BHtop.south) -- (P);

\draw[->] (P) -- (BHbot.north);
\draw[->] (P) -- (Zbot.north);

\node at (3.1,2.1) {$\alpha_{\pm i}$};

\end{tikzpicture}
\end{center}
In section \ref{sec:general-proof} we will put together theorems \ref{thm:mixed-from-pure} and \ref{thm:abelian} to analyze the general excitability problem.

\subsection{The abelian case}
\label{sec:abelian}

Now we turn to the excitability question for a free field algebra where the commutator vanishes, i.e., where we have $\Omega = 0$.
We have some basic field operators $\phi[f]$ that all mutually commute with one another forming an algebra $\A_0$, and a pair of zero-mean Gaussian states $\tau_1$ and $\tau_2$, where we are using the symbols $\tau_j$ for these abelian states so as not to confuse them with the states $\omega_j$ from the preceding subsection.

Our goal is to prove theorem \ref{thm:abelian}, which says that we have $\tau_2 \prec \tau_1$ if and only if we have (i) $\mu_2 \prec \mu_1,$ (ii) $\ker{Q} = 0$, and (iii) $Q - 1$ Hilbert-Schmidt. Recall that the necessity of conditions (i) and (ii) was already established in Section \ref{sec:general-necessary}.

First we note that since $\A_0$ is abelian, the state $|\tau_j \rangle$ is automatically cyclic \textit{and separating} for $\A_{\tau_j}$ within $\H_{\tau_j}.$\footnote{Cyclicity is always true in the GNS construction; and since separating for $\A_{\tau_j}$ is the same as cyclic for $\A_{\tau_j}',$ and we have $\A_{\tau_j}' \supseteq \A_{\tau_j}$ in the abelian case, the separating property follows as well.}
Consequently, whenever we have $\tau_2 \prec \tau_1,$ we can conclude as in section \ref{sec:mutual-excitability} that by \cite[corollary 5.24]{Stratila:book}, there is in fact a \textit{vector} representative for $\tau_2$ within $\H_{\tau_1}.$
As in the pure-state case of section \ref{sec:pure}, we will find that well-definedness of the vector representative implies that $Q-1$ is Hilbert-Schmidt.

Unlike in the pure-state case, there will be a great deal of freedom in choosing our vector representative, for we could always act on it with any unitary in $\A_{\tau_1}$ and get an equally good representative.
There is, however, a canonical representative furnished by Araki's theory of the natural cone \cite{Araki:natural-1}.
Theorem 7.1 of that paper implies the existence of a preferred representative $|\tau_2^{\natural}\rangle$.
This vector has the special property that it is in the set
\begin{equation}
	\mathfrak{P}_{\natural} = \bar{\{\Delta_{\tau_1}^{1/4} P |\tau_1\rangle \text{ such that } P \in \A_{\tau_1} \text{ and } P \geq 0\}}.
\end{equation}
The abelian property of $\A_{\tau_1}$ automatically implies that $\Delta_{\tau_1}$ is the identity,\footnote{This follows from the fact that for any algebra $\M$ and cyclic-separating state $|\psi\rangle,$ the adjoint of the Tomita operator $S_{\psi, \M}^{\dagger}$ can be identified with the Tomita operator for the commutant, $S_{\psi, \M'}.$ In the abelian case, one has $\M \subseteq \M',$ so for $a \in \M$ we have
$$\Delta_{\psi, \M} a |\psi\rangle = S_{\psi, \M'} S_{\psi, \M} a |\psi\rangle = S_{\psi, \M'} a^{\dagger} |\psi\rangle = a |\psi\rangle.$$}
which gives us
\begin{equation}
	\mathfrak{P}_{\natural} = \bar{\{P |\tau_1\rangle \text{ such that } P \in \A_{\tau_1} \text{ and } P \geq 0\}}.
    \label{eq:natural-cone-no-delta}
\end{equation}
A simple argument --- see e.g. \cite[proposition 10.8]{Stratila:book} --- tells us that every vector in this set can be written directly as $P |\tau_1\rangle$ where $P$ is a positive operator that is \textit{affiliated} with $\A_{\tau_1}$.
I.e., to remove the closure in equation \eqref{eq:natural-cone-no-delta}, we need only accept that $P$ may need to be unbounded.

So to sum up, when we assume $\tau_2 \prec \tau_1,$ we know there is a positive operator $P_{\tau_2}$ affiliated with $\A_{\tau_1}$ such that
\begin{equation}
	| \tau_2^{\natural} \rangle
		\equiv P_{\tau_2} |\tau_1 \rangle
\end{equation}
is a vector representative for $\tau_2.$
The game will be to solve for this operator in terms of the operator $Q$ implementing the $\mu_2$ inner product in the $\mathcal{K}_{\mu_1}$ space.
In the present setting, since the commutator vanishes, we are effectively in the ``non-quantum regime'' where everything is governed by classical probability theory; $\tau_1$ and $\tau_2$ can be thought of as infinite-dimensional Gaussian measures, and finding $P_{\tau_2}$ amounts to finding the ``Radon-Nikodym derivatives'' that convert between these measures.\footnote{This perspective comes from \cite{Segal:distributions}, from which our present technique is essentially adapted by working directly in Hilbert space instead of on certain abstract structures built over Hilbert space.}

We begin by choosing an orthonormal basis $|z_j\rangle_{\mu_1}$ for $\K_{\mu_1}$.
We are free to take each $|z_j\rangle_{\mu_1}$ to be real (i.e., unchanged by complex conjugation).
The way we will build up the state $|\tau_2^{\natural}\rangle$ is by considering the sequence of von Neumann algebras $\A_{\tau_1, (m)}$ that are generated by the field operators $\phi[z_1], \dots, \phi[z_m].$
We will consider also the sequence of subspaces $\H_{\tau_1, (m)} \subseteq \H_{\tau_1}$ that are built by acting on $|\tau_1\rangle$ only with elements of $\A_{\tau_1, (m)}$.\footnote{Note that these are completely different from the $n$-particle subspaces $\H_{\tau_1}^{(n)}$, since each $\H_{\tau_1, (m)}$ will contain states of arbitrary particle number.}
We will find the correct sequence of natural cone states $|(\tau_2^{\natural})_{(m)}\rangle$ within each of these subspaces that necessarily converge to $|\tau_2^{\natural}\rangle,$ and prove theorem \ref{thm:abelian} by identifying a Hilbert-Schmidt criterion for convergence.

What do we mean by the state $|(\tau_2^{\natural})_{(m)}\rangle$?
This is a state in the subspace $\H_{\tau_1, (m)}$ that is in the natural cone for the subalgebra $\A_{\tau_1, (m)}.$
I.e., it is a state that reproduces all the correct correlators for $\tau_2$ within the subalgebra $\A_{\tau_1, (m)},$ and that also can be written as
\begin{equation} \label{eq:tau-natural-m}
	|(\tau_2^{\natural})_{(m)} \rangle
		= P_{\tau_2, (m)} |\tau_1\rangle
\end{equation}
for some positive operator $P_{\tau_2, (m)}$ affiliated with $\A_{\tau_1, (m)}.$
But since this latter algebra is the von Neumann algebra generated by a finite set of commuting, independent Hermitian operators, we know that $P_{\tau_2, (m)}$ can be written in terms of some positive function $p_m(t_1, \dots, t_m)$ on $\reals^m$, i.e., we have
\begin{equation} \label{eq:m-natural-p}
	|(\tau_2^{\natural})_{(m)} \rangle
		= p_m(\phi[z_1], \dots, \phi[z_m]) |\tau_1\rangle.
\end{equation}

To express the function $p_m$ explicitly, i.e. in terms of $Q$, we just need to start computing correlation functions.
Correlation functions in the state $|(\tau_2^{\natural})_{(m)} \rangle$ are naturally expressed in terms of $Q$ by writing
\begin{equation}
\begin{aligned}
    \langle (\tau_2^{\natural})_{(m)} | e^{i \phi[t_1 z_1 + \dots + t_m z_m]} |(\tau_2^{\natural})_{(m)}\rangle &=
    \langle \tau_2 | e^{i \phi[t_1 L z_1 + \dots + t_m L z_m]} | \tau_2 \rangle\\
    &=
    e^{-\frac{1}{2}
    \langle t_1 L z_1 + \dots + t_m L z_m | t_1 L z_1 + \dots + t_m L z_m \rangle_{\mu_2}}\\
    &=
    e^{-\frac{1}{2}
\langle t_1 z_1 + \dots + t_m z_m |Q| t_1  z_1 + \dots + t_m  z_m \rangle_{\mu_1}},
\end{aligned}
\label{eq:tau-natural-Weyl-expectation}
\end{equation}
then taking derivatives with respect to the $t_i$ parameters.
This expression is massaged in appendix \ref{app:pm-proof} to yield the following expression for $p_m$:
\begin{equation} \label{eq:pm_main}
	p_m(x_1, \dots, x_m) = \frac{1}{(\det Q_m)^{1/4}} \exp\left[ \frac{1}{4}\sum_{j,k=1}^m x_j\bra{z_j}1 - Q_m^{-1}\ket{z_k}_{\mu_1}\!x_k\, \right],
\end{equation}
where $Q_m$ is the restriction of $Q$ to the subspace of $\K_{\mu_1}$ spanned by $\ket{z_1}_{\mu_1}$ through $\ket{z_m}_{\mu_1}$. 
This gives a complete characterization for the natural cone representatives of the restriction of $\tau_2$ to each $\A_{\tau_1, (m)},$ via the formula \eqref{eq:tau-natural-m}.

We now need to show that (i) in order to have $\tau_2 \prec \tau_1,$ it is necessary that we have convergence $|(\tau_2^{\natural})_{(m)}\rangle \to |\tau_2^{\natural}\rangle,$ and (ii) this gives the ``necessity'' part of theorem \ref{thm:abelian}.
After we accomplish this, it will be easy to turn around and show sufficiency.

\subsubsection{Necessity}
Recall our key formulas,
\begin{equation}
	|\tau_2^{\natural}\rangle = P_{\tau_2} |\tau_1\rangle
\end{equation}
and
\begin{equation}
	|(\tau_2^{\natural})_{(m)}\rangle
		= P_{\tau_2, (m)} |\tau_1\rangle ,
\end{equation}
where $P_{\tau_2}$  and  $P_{\tau_2, (m)}$ are positive operators affiliated with $\A_{\tau_1}$ and $\A_{\tau_1, (m)}$ respectively.
Our first goal is to show that whenever we have $\tau_2 \prec \tau_1$,  we have the convergence $|(\tau_2^{\natural })_{(m)}\rangle \to |\tau_2^{\natural}\rangle.$
We have
\begin{equation}
	\norm*{\ket*{\tau_2^{\natural}} - \ket*{(\tau_2^{\natural})_{(m)}}}^2
		= 2 - \langle (\tau_2^{\natural})_{(m)} |\tau_2^{\natural}\rangle - \langle \tau_2^{\natural} | (\tau_2^{\natural})_{(m)} \rangle,
\end{equation}
so it is sufficient to show the convergence
\begin{equation}
	\langle \tau_2^{\natural} | (\tau_2^{\natural})_{(m)} \rangle
    =
    \langle P_{\tau_2} \tau_1 | P_{\tau_2,{(m)}} \tau_1 \rangle
    \to 1.
\end{equation}
We will accomplish this in three steps.
First, we will show that the overlaps are positive.
Then, we will show they are bounded above by one, and finally we will show that they are bounded below by one in the limit of $m$ going to infinity.

The overlaps are positive because they can be expressed as overlaps of vectors in the natural cone of $\ket{\tau_1}$ for the subalgebra $\mathcal{A}_{\tau_1,(m)}$.\footnote{It is a basic property of the natural cone --- see e.g. \cite[sections 10.23-10.25]{Stratila:book} --- that the overlaps are positive.}
In particular, we have
\begin{equation}
\bra*{P_{\tau_2} \tau_1}\ket*
{P_{\tau_2,{(m)}} \tau_1} 
=
\bra*{\Pi_{(m)}P_{\tau_2}\Pi_{(m)}\tau_1}\ket*{P_{\tau_2,(m)}\tau_1},
\label{eq:first-natural-cone-overlap}
\end{equation}
where we define $\Pi_{(m)}$ as projector onto $\H_{\tau_1, (m)}.$
On the right-hand side of the above equation, the ket is in the natural cone by construction, 
and we now argue that the bra is also in the natural cone.
Clearly, $\Pi_{(m)}P_{\tau_2}\Pi_{(m)}$ is positive; the nontrivial step is to show it is affiliated to $\mathcal{A}_{\tau_1,(m)}$.
To show this, we note that since $\Pi_{(m)}$ projects onto an invariant subspace of $\A_{\tau_1, (m)}$, it must live in the commutant of that algebra.
Moreover, since $P_{\tau_2}$ is affiliated with $\A_{\tau_1},$ and this is an abelian algebra that contains $\A_{\tau_1, (m)},$ the operator $P_{\tau_2}$ also formally commutes with $\A_{\tau_1, (m)}.$
Therefore we may conclude that the operator $\Pi_{(m)}P_{\tau_2}\Pi_{(m)}$ formally commutes with $\A_{\tau_1, (m)},$ and hence is affiliated to its commutant $\A_{\tau_1, (m)}'$.
Since $\A_{\tau_1, (m)}'= \A_{\tau_1, (m)}$ within the the GNS representation $\H_{\tau_1, (m)}$,\footnote{For any state $\sigma$ on an abelian $\ast$-algebra $\Z_0$, and $\Z$ the associated von Neumann algebra, one has $\Z' = \Z$ within the GNS representation for $\sigma.$ 
This follows from the cyclic-separating property of $|\sigma\rangle$ for $\Z$, which implies the modular conjugation identity $J_{\sigma} \Z J_{\sigma} = \Z'$, together with the fact that in the abelian case, one has $J_{\sigma} z |\sigma \rangle = z^{\dagger} |\sigma\rangle,$ hence
$$
	J_{\sigma} z_1 J_{\sigma} z_2 |\sigma \rangle
		= z_2 z_1^{\dagger} |\sigma\rangle
		= z_1^{\dagger} z_2 |\sigma\rangle,
$$
from which we may conclude $J_{\sigma} z_1 J_{\sigma} = z_1^{\dagger}$, hence $J_\sigma \Z J_\sigma = \mathcal Z$.} and since $\Pi_{(m)}P_{\tau_2}\Pi_{(m)}$ can be interpreted as an operator on $\H_{\tau_1, (m)}$ by virtue of the conjugation by projectors, it follows that $\Pi_{(m)}P_{\tau_2}\Pi_{(m)}$ is also affiliated with $\A_{\tau_1, (m)}$.

We conclude that the $\bra*{ \tau_2^{\natural}}\ket*{(\tau_2^{\natural})_{(m)}}$ overlaps are positive.
Let us now show they are bounded above by 1.
This is easy; by positivity, we may consider the absolute value of the overlap, and then apply the Cauchy–Schwarz inequality. Using also that both the bra and the ket are normalized by construction, we obtain $\langle \tau_2^{\natural} | (\tau_2^{\natural})_{(m)} \rangle \leq 1.$

To finish, we will lower-bound the sequence of overlaps by a sequence that converges to one, using the following chain of reasoning:
\begin{equation}
\label{eq:lower-bounding-sequence}
\begin{aligned}
        \langle P_{\tau_2} \tau_1 | P_{\tau_2,(m)} \tau_1 \rangle &= \braket{P_{\tau_2}\tau_1}{\Pi_{(m)}P_{\tau_2}\tau_1} + \braket{P_{\tau_2}\tau_1}{(P_{\tau_2,(m)}-\Pi_{(m)}P_{\tau_2})\tau_1}\\
        &=\braket{P_{\tau_2}\tau_1}{\Pi_{(m)}P_{\tau_2}\tau_1} + \braket{\Pi_{(m)}P_{\tau_2}\Pi_{(m)}\tau_1}{(P_{\tau_2,(m)}-\Pi_{(m)}P_{\tau_2}\Pi_{(m)})\tau_1}\\
        &\geq \braket{P_{{\tau_2}}\tau_1}{\Pi_{(m)}P_{\tau_2}\tau_1}.
\end{aligned}
\end{equation}
In order to get the last inequality, we argue in appendix \ref{app:lower-bounding-sequence} that the second summand above is a natural cone overlap, and hence must be positive.
Clearly, the last expression in the above equation converges to one as $m$ goes to infinity, so this concludes the proof.

Taking stock, we have shown that when we have $\tau_2 \prec \tau_1,$ it is necessary that we have $|(\tau_2^{\natural})_{(m)}\rangle \to |\tau_2^{\natural}\rangle.$
In particular, the sequence $\ket*{(\tau_2^{\natural})_{(m)}}$ must be Cauchy.
Hence, for any $\epsilon>0$, there must be an integer $N_\epsilon$ such that for any $m\geq n >N_\epsilon$, we have 
\begin{equation}
\label{eq:CauchyTau}
	\norm*{|(\tau_2^{\natural})_{(m)}\rangle - |(\tau_2^{\natural})_{(n)}\rangle}^{\!2}  < \epsilon\,.
\end{equation}
We will now show that this Cauchy property implies $Q-1$ can be obtained as the limit of a sequence of Hilbert-Schmidt operators, 
where the limit is taken with respect to the Hilbert-Schmidt norm; this will imply that $Q-1$ must itself be Hilbert-Schmidt.

To begin, we expand the left-hand side of the above expression as
\begin{equation}
	\norm*{|(\tau_2^{\natural})_{(m)}\rangle - |(\tau_2^{\natural})_{(n)}\rangle}^{\!2}  =  2-2\braket*{(\tau_2^{\natural})_{(m)}}{(\tau_2^{\natural})_{(n)}},
\end{equation}
where we used that each vector in the sequence is normalized to one (by construction), and that the overlaps of these vectors are real (because they are obtained by acting on $|\tau_1\rangle$ with positive, commuting operators).
The next step is to express the above quantity, and hence the Cauchy property \eqref{eq:CauchyTau}, in terms of $Q$ using the relation
\begin{equation}
	|(\tau_2^{\natural})_{(m)}\rangle
		= p_m(\phi[z_1], \dots, \phi[z_m]) |\tau_1\rangle,
\end{equation}
together with equation \eqref{eq:pm_main}.
This is a straightforward but somewhat lengthy computation, which we defer to appendix \ref{app:overlap-taus}.
The final result is the following.
Let $E_m$ denote the projector onto the subspace of $\mathcal K_{\mu_1}$ spanned by $|z_1\rangle_{\mu_1}, \dots, |z_m\rangle_{\mu_1},$ and define 
\begin{equation}
    Q_{n,\perp}=Q_n+E_m-E_n, 
\end{equation} 
which acts as $Q_{n}$ on the subspace $E_n\mathcal K_{\mu_1}$ and as the projector $E_m$ on its orthogonal complement (recall that we took $m\ge n$).
Then
\begin{equation}
\label{eq:tauOverlaps}
 \braket*{(\tau_2^{\natural})_{(m)}}{(\tau_2^{\natural})_{(n)}} =  \left(\prod_{j=1}^{m} \frac{\lambda^{(m,n)}_j + 1/\lambda^{(m,n)}_j}{2}\right)^{-1/2},
\end{equation}
where $\lambda_j^{(m,n)}$ are the eigenvalues of $Q_{n,\perp}^{-1/4}Q_m^{1/2}Q_{n,\perp}^{-1/4}$.\footnote{Note that $Q$ has no kernel by the results of section \ref{sec:general-necessary}, and since the subspace spanned by $\ket{z_j}_{\mu_1}$ for $j=1,\dots n$ is finite-dimensional, the operator $Q_n$ is always invertible. Consequently, also $Q_{n,\perp}$ is invertible.}

The Cauchy property \eqref{eq:CauchyTau} may now be expressed in terms of $Q$. Namely, for large enough $m$ and $n$, we must have
\begin{equation}
	\prod_{j=1}^{m} \frac{\lambda^{(m,n)}_j + 1/\lambda^{(m,n)}_j}{2}
		< \left(\frac{2}{2-\epsilon}\right)^2.
\end{equation}
As shown in appendix \ref{app:why-q-is-invertible}, $Q$ is invertible under the present set of assumptions.
From this, in appendix \ref{app:bound-lambdas}, we can lower bound the left-hand side as
\begin{equation} \label{eq:HS-mu}
	\prod_{j=1}^{m} \frac{\lambda^{(m,n)}_j + 1/\lambda^{(m,n)}_j}{2}
		\geq 1 + C \sum_{j=1}^{m} (\lambda^{(m,n)}_j - 1)^2,
\end{equation}
where $C$ is a positive constant independent of $m$ and $n$.
We recognize the quantity multiplying $C$ as the square of the Hilbert-Schmidt norm of $Q_{n,\perp}^{-1/4}Q_m^{1/2}Q_{n,\perp}^{-1/4}-1$, so that putting together the previous two equations, we have
\begin{equation}
	\norm{Q_{n, \perp}^{-1/4} Q_{m}^{1/2} Q_{n, \perp}^{-1/4} - 1}_2 \leq \epsilon,
\end{equation}
where we redefined $\epsilon$ to match $[(2/(2-\epsilon))^2 - 1]^{1/2}/C^{1/2}$.
Now using the H\"older-type inequality $\norm{ABC}_2\le \norm{A}\norm{B}_2 \norm{C}$, we obtain
\begin{align}
	\begin{split}
	\left\lVert Q_m^{1/2} - Q_{n,\perp}^{1/2} \right\rVert_2
		& = \left\lVert Q_{n,\perp}^{1/4} \left( Q_{n,\perp}^{-1/4} Q_m^{1/2} Q_{n,\perp}^{-1/4} - 1 \right) Q_{n,\perp}^{1/4}  \right\rVert_2 \leq \left\lVert Q_{n, \perp}^{1/4} \right\rVert ^2 \epsilon\,.
	\end{split}
\end{align}
Finally, applying the triangle inequality as follows (and another Hölder inequality),
\begin{equation}
\begin{aligned}
	\norm{Q_m - Q_{n,\perp}}_2
		&= \norm{Q_m^{1/2} (Q_m^{1/2} - Q_{n,\perp}^{1/2}) + (Q_m^{1/2} - Q_{n,\perp}^{1/2}) Q_{n,\perp}^{1/2}}_2\\
        &\le \norm{Q_m^{1/2} (Q_m^{1/2} - Q_{n,\perp}^{1/2})}_2 + \norm{(Q_m^{1/2} - Q_{n,\perp}^{1/2}) Q_{n,\perp}^{1/2}}_2\\
        &\le \left(\norm{Q_m^{1/2}} + \norm{Q_{n,\perp}^{1/2}}\right) \norm{Q_m^{1/2} - Q_{n,\perp}^{1/2}}_2,
\end{aligned}
\end{equation}
one finds
\begin{equation}
	\left\lVert Q_m - Q_{n,\perp} \right\rVert_2
		\leq \epsilon,
\end{equation}
where once more we reabsorbed all positive constants in the infinitesimal parameter $\epsilon$.\footnote{In particular, here we redefined $\epsilon$ to match $2 v \epsilon$, where, consistently with the notation used in appendix \ref{app:bound-lambdas}, $v$ is the maximum of the operator norms of the $Q_m$'s and the $Q_{m,\perp}$'s.}
Recalling that $Q_{n,\perp} = Q_n+E_m-E_n$, we note that the above condition reads
\begin{equation}
	\left\lVert (Q_m-E_m) - (Q_{n}-E_n) \right\rVert_2
		\leq \epsilon,
\end{equation}
so the sequence $Q_n- E_n$ is Cauchy in the Hilbert-Schmidt norm, with limit $Q - 1$.
Each element of this sequence is trivially Hilbert-Schmidt, by virtue of acting on a finite dimensional space, and since the Hilbert-Schmidt operators are complete in the Hilbert-Schmidt topology, we conclude that $Q - 1$ must be Hilbert-Schmidt, as desired.

The argument we have just undertaken is rather complicated, so we pause to sum up.
By assuming $\tau_2 \prec \tau_1$ in the abelian case, we concluded that there must exist a natural cone representative $|\tau_2^{\natural}\rangle$ for $\tau_2$.
Assuming the necessary condition that $Q$ has no kernel, we wrote a sequence $|(\tau_2^{\natural})_{(m)}\rangle \to |\tau_2^{\natural}\rangle,$ and showed that in order for this sequence to converge, it is necessary that $Q - 1$ be Hilbert-Schmidt.
This completes the ``necessity'' proof --- we have shown that when we have $\tau_2 \prec \tau_1,$ we must have (i) $\mu_2 \prec \mu_1,$ (ii) $\ker{Q} = 0$, and (iii) $Q - 1$ Hilbert-Schmidt.

\subsubsection{Sufficiency}\label{sec:suff-abelian}
The ``sufficiency'' direction is extremely easy now that we have already done all of the above work.
Our task is to construct the representative $\ket*{\tau_2^{\natural}}$ by hand, as a limit of states $\ket*{(\tau_2^{\natural})_{(m)}}$, using the starting assumption that $Q - 1$ is Hilbert-Schmidt.
If we know that $Q - 1$ is Hilbert-Schmidt, then we know that $Q$ can be decomposed into a discrete\footnote{Hilbert-Schmidt operators are compact, and compact operators have discrete spectra.} sum of eigenvectors, $Q = \sum_j q_j |e_j \rangle \langle e_j|.$
We take these vectors $|e_j\rangle$ to be a special case of the generic orthonormal basis $z_j$ used above.
For finite $m$, we simply define the state
\begin{equation}
	|(\tau_2^{\natural})_{(m)}\rangle
		\equiv p_m(\phi[e_1], \dots, \phi[e_m])
\end{equation}
with 
\begin{equation}
	p_m(x_1, \dots, x_m)
		= \frac{1}{(\prod_{j=1}^{m} q_j)^{1/4}} \exp[ \frac{1}{4} \sum_{j} (1 - q_j^{-1}) x_j^2]
\end{equation}
as inspired by equation \eqref{eq:pm_main}.
Repeating the above analysis (see equation \eqref{eq:tauOverlaps}), one computes\footnote{To facilitate the comparison with equation \eqref{eq:tauOverlaps} above, note that $\lambda_j^{(m,n)}=1$ for $j=1,\ldots,n$ while $\lambda_j^{(m,n)}=q_j^{1/2}$ for $j=n+1,\ldots,m$.}
\begin{align}
	\begin{split} 
		& \left\lVert |(\tau_2^{\natural})_{(m)}\rangle - |(\tau_2^{\natural})_{(n)}\rangle \right\rVert^2
		= 2 \left(1 - \left(\prod_{j=n+1}^{m} \frac{q_j^{1/2} + q_j^{-1/2}}{2}\right)^{-1/2}\right).
	\end{split} 
\end{align}

To show that this has the Cauchy property, we need only show that the product converges in the limit $m \to \infty.$
This is the case if and only if we have convergence of the infinite sum
\begin{equation}
	\sum_{j=n+1}^{\infty} \log\left( \frac{q_j^{1/2} + q_j^{-1/2}}{2} \right).
\end{equation}
Since $Q-1$ is Hilbert-Schmidt, the eigenvalues $q_j$ must converge to $1$ in this limit.
A similar argument as given above then tells us that (for sufficiently large $n$) there is a constant $C'$ such that the above sum is upper bounded by $C' \sum_{j} (q_j - 1)^2$, which converges because $Q - 1$ is Hilbert-Schmidt.
So we learn that $|(\tau_2^{\natural})_{(m)}\rangle$ is a Cauchy sequence, so it converges; we call its endpoint $|\tau_2^{\natural}\rangle.$

All that remains is to show that this vector, which we have constructed as the endpoint of a particular sequence, is actually a good representative for $\tau_2$ in $\H_{\tau_1}.$
By construction, each of the $|(\tau_2^{\natural})_{(m)}\rangle$ vectors produces the correct correlation functions of $\tau_2$ on the algebra $\A_{\tau_1, (m)}$ that is generated by the operators $\phi[e_1]$ through $\phi[e_m].$
So $|\tau_2^{\natural}\rangle$ produces the correct correlation functions on the algebra
\begin{equation}
	(\cup_m \A_{\tau_1, (m)})''.
\end{equation}
This algebra is in fact equal to all of $\A_{\tau_1}.$
To see this, note that $\A_{\tau_1}$ is generated by the Weyl operators $e^{i \phi[\psi]}$ for generic real $|\psi\rangle_{\mu_1} \in \K_{\mu_1}.$
But writing $|\psi\rangle_{\mu_1} = \sum_{j=1}^{\infty} c_j |e_j\rangle_{\mu_1},$ we have
\begin{equation}
e^{\,i \phi[\sum_{j=1}^{N} c_j e_j]} \;\xrightarrow{\;\text{strongly}\;}\; e^{\,i \phi[\psi]}
\end{equation}
and this limit is obviously in the von Neumann algebra generated by $\cup_m \A_{\tau_1, (m)}.$

\subsection{The general theorem}
\label{sec:general-proof}

Now we put together theorems \ref{thm:mixed-from-pure} and \ref{thm:abelian} --- which deal with the mixed-from-pure case and the abelian case, respectively --- to prove the centrally pure excitability result, i.e. theorem \ref{thm:c-pure-from-c-pure}.
As explained at the beginning of the section, this result (together with the simple logic of section \ref{sec:factorial-consequences}) automatically implies the general excitability theorem,  theorem \ref{thm:general-theorem}, which is the main result of the paper.

To prove theorem \ref{thm:c-pure-from-c-pure}, we make use of the decomposition from section \ref{sec:excitability-breakdown}, which tells us that for centrally pure states, we have $\omega_2 \prec \omega_1$ if and only if the following ``restriction maps'' are ultraweakly continuous: the map $\alpha_{\pm i}:\A_{\omega_1, \pm i}\to {\cal A}_{\omega_2}$ with
\begin{equation}
	\alpha_{\pm i}(e^{i \phi[\psi]}) = e^{i \phi[L \psi]} \quad \text{with} \quad  |\psi\rangle_{\mu_1} \in \Pi_{1,\pm i}{\cal K}_{\mu_1} \text{  real}
    \label{eq:pure-piece-alpha}
\end{equation}
and the map 
$\alpha_{0}:\A_{\omega_1, 0}\to {\cal A}_{\omega_2}$ with
\begin{equation}
	\alpha_{0}(e^{i \phi[\psi]}) = e^{i \phi[L \psi]} \quad \text{with} \quad  |\psi\rangle_{\mu_1} \in \Pi_{0}{\cal K}_{\mu_1} \text{  real},
\end{equation}
where we write $\Pi_{1, i}, \Pi_{1,-i},$ and $\Pi_{1,0}$ for the projections onto the eigenspaces of $R_1$, as well as $\Pi_{1,\pm i}\equiv\Pi_{1, i}+\Pi_{1,-i}$.
As mentioned at the end of section \ref{sec:excitability-breakdown}, the key point is that the restriction of $\omega_1$ to $\A_{\omega_1, \pm i}$ is pure, and the restriction of $\omega_1$ to $\A_{\omega_1,0}$ is abelian, and so we will be able to address the continuity of $\alpha_{\pm i}$ and $\alpha_0$ by appealing to our existing excitability theorems.
The physical picture is that the full excitability relation $\omega_2 \prec \omega_1$ breaks down into an excitability relation for the ``pure piece'' of $\omega_2$ and an excitability relation for the ``abelian piece'' of $\omega_2$.

To address the ``pure'' part of the question, the game is to find a constraint on $Q$ which expresses ultraweak continuity of the map $\alpha_{\pm i}$ in equation \eqref{eq:pure-piece-alpha}.
This is equivalent to an excitability relation between the appropriate restrictions of $\omega_1$ and $\omega_2$ to the $\ast$-algebra generated by $\phi[\psi]$ with $|\psi\rangle_{\mu_1}$ in the support of $R_1$,\footnote{To see why this is the relevant excitability problem, call this $\ast$-algebra ${\cal A}_{0,\pm i}$.
Recall from section \ref{sec:one-way-excitability} that when we have excitability, the two algebras related by a normal homomorphism are the von Neumann completions of the original $\ast$-algebra with respect to the two states of interest.
Clearly ${\cal A}_{\omega_1,\pm i}$ is the von Neumann completion of  ${\cal A}_{0,\pm i}$ with respect to the state $\omega_1$.
Meanwhile, if the map $\alpha_{\pm i}$ is ultraweakly continuous, its image will be the von Neumann completion of ${\cal A}_{0,\pm i}$ with respect to $\omega_2$, where the latter is defined on ${\cal A}_{0,\pm i}$ by pulling back along $L$.} for which $\omega_1$ is pure.
Clearly, for the pure excitability relation, the relevant operator is not the full operator $Q = L^{\dagger} L$, but rather $\Pi_{1,\pm i}Q\Pi_{1,\pm i}$. In particular, by theorem \ref{thm:mixed-from-pure}, once we have assumed $\mu_2 \prec \mu_1,$ the remaining necessary and sufficient conditions for the ``pure piece'' of the excitability relation $\omega_2 \prec \omega_1$ are\footnote{Recall that, for factorial and hence for pure states, the dominance of the $\mu$ inner products together with the Hilbert-Schmidt condition on the difference of the square roots of the $X$'s automatically imply that $Q$ has trivial kernel, hence we do not write this condition explicitly.}
\begin{equation} \label{eq:pure-conditions}
	\Pi_{1,i} (Q - 1) \Pi_{1,i}
	\quad \text{ and } \quad \Pi_{1,-i} (Q - 1) \Pi_{1,-i}
	\quad \text{trace-class}.
\end{equation}
Similarly, By theorem \ref{thm:abelian}, once we have assumed $\mu_2 \prec \mu_1,$  the remaining necessary and sufficient conditions for the ``abelian piece'' are 
\begin{equation} \label{eq:abelian-conditions}
	\Pi_{1,0} (Q - 1) \Pi_{1,0} \;\; \text{Hilbert-Schmidt}\quad \text{ and } \quad
	\ker(\Pi_{1,0} Q \Pi_{1,0}) \cap \ker{R_1} = 0.
\end{equation}

To prove theorem \ref{thm:c-pure-from-c-pure}, we must show that the conditions in equations \eqref{eq:pure-conditions} and \eqref{eq:abelian-conditions} are equivalent to the pair of requirements that (i) $\ker{Q}$ is trivial and (ii) $Q - 1$ is Hilbert-Schmidt.
Note that in either direction of the proof, we may assume that $Q$ has no kernel.
In one direction this is an assumption; in the other direction we only assume that the kernel of $\Pi_{1,0}Q\Pi_{1,0}$ is trivial within $\ker{R_1}$, but the identity $L^{\dagger} R_2 L = R_1$ tells us that the kernel of $Q$ is contained in the kernel of $R_1,$ so any vector in the kernel of $Q$ would also be in $\ker(\Pi_{1,0} Q \Pi_{1,0}) \cap \ker{R_1},$ which has been assumed to be trivial.

We may also assume in either case that we have $\ker(\Pi_{1,0} Q \Pi_{1,0}) \cap \ker{R_1} = 0$.
For while in one direction this is an assumption, in the other direction we assume $\ker{Q} = 0$ and $Q-1$ Hilbert-Schmidt.
From this, there must be some constant $u$ with $Q \geq u > 0,$ and one therefore has $\Pi_{1,0} Q \Pi_{1,0} \geq u > 0$ as an operator on the kernel of $R_1.$

It will be essential, in proving the theorem, to use the fact that $\omega_1$ and $\omega_2$ are centrally pure and hence that the $R_j$ operators have spectrum $-i,0,i$.
Let $\tilde{\Pi}_{2, i}, \tilde{\Pi}_{2, -i},$ and $\tilde{\Pi}_{2, 0}$ be the projectors onto the eigenspaces of $R_2$, where we use tildes to emphasize that these operators act on $\K_{\mu_2}$.
With the polar decomposition $L = v Q^{1/2},$ we define operators on $\K_{\mu_1}$ by
\begin{align}
	\begin{split}
	\Pi_{2, i}
		& = v^{\dagger} \tilde{\Pi}_{2, i} v, \\
	\Pi_{2, -i}
		& = v^{\dagger} \tilde{\Pi}_{2, -i} v, \\
	\Pi_{2, 0}
		& = v^{\dagger} \tilde{\Pi}_{2, 0} v.
	\end{split}
\end{align}
The fact that $L$ has a dense image gives the identity $v v^{\dagger} = 1,$ while the fact that it has trivial kernel gives the identity $v^\dagger  v = 1.$
Together, these identities tell us that the $\Pi_{2, \dots}$ operators are actually projectors on $\K_{\mu_1},$ that they are mutually orthogonal, and that they sum to the identity.
In practice, these projectors will be used to express $Q$ in various block decompositions.

To simplify these blocks, it will convenient to use the identity $L^{\dagger} R_2 L = R_1$, which can be expanded as
\begin{equation} \label{eq:centrally-pure-identity}
	Q^{1/2} (\Pi_{2,i} - \Pi_{2, -i}) Q^{1/2}
		= \Pi_{1,i} - \Pi_{1,-i}.
\end{equation}
In particular, we can write a block decomposition of $Q^{1/2}$ via the formula
\begin{equation}
	Q^{1/2}
		= (\Pi_{1,i} + \Pi_{1,-i} + \Pi_{1,0}) Q^{1/2} (\Pi_{2,i} + \Pi_{2,-i} + \Pi_{2,0}),
\end{equation}
so that
\begin{equation} \label{eq:Q-half-blocks}
	Q^{1/2}
	=
	\begin{pmatrix}
		a & b & c \\ j & k & l \\ x & y & z
	\end{pmatrix},
\end{equation}
where, for instance, \begin{equation}
a=\Pi_{1,i}Q^{1/2}\Pi_{2,i}\qquad
    b=\Pi_{1,i}Q^{1/2}\Pi_{2,-i}.
\end{equation}
Equation \eqref{eq:centrally-pure-identity} for $L^{\dagger} R_2 L = R_1$ now becomes the single matrix identity
\begin{equation}\label{eq:matrix-identity}
	\begin{pmatrix}
		(a a^{\dagger} - b b^{\dagger}) & (a j^{\dagger} - b k^{\dagger}) & (a x^{\dagger} - b y^{\dagger}) \\
		(j a^{\dagger} - k b^{\dagger}) & (j j^{\dagger} - kk^{\dagger}) & (j x^{\dagger} - k y^{\dagger}) \\
		(x a^{\dagger} - y b^{\dagger}) & (x j^{\dagger} - y k^{\dagger}) & (x x^{\dagger} - y y^{\dagger})
	\end{pmatrix}
	=
	\begin{pmatrix}
		1 & 0 & 0 \\ 0 & -1 & 0 \\ 0 & 0 & 0
	\end{pmatrix},
\end{equation}
expressed in the $R_1$ block decomposition; for example,  the relation $aa^{\dag}-bb^{\dag}=1$ comes from acting on \eqref{eq:centrally-pure-identity} with $\Pi_{1,i}$ on either side.
Note that within the $R_1$ block decomposition for $Q$, we have the formula
\begin{equation}
	Q
	=
	\begin{pmatrix}
		(a a^{\dagger} + b b^{\dagger} + c c^{\dagger})
			& (a j^{\dagger} + b k^{\dagger} + c l^{\dagger})
			& (a x^{\dagger} + b y^{\dagger} + c z^{\dagger}) \\
		(j a^{\dagger} + k b^{\dagger} + l c^{\dagger})
			& (j j^{\dagger} + k k^{\dagger} + l l^{\dagger})
			& (j x^{\dagger} + k y^{\dagger} + l z^{\dagger}) \\
		(x a^{\dagger} + y b^{\dagger} + z c^{\dagger})
			& (x j^{\dagger} + y k^{\dagger} + z l^{\dagger})
			& (x x^{\dagger} + y y^{\dagger} + z z^{\dagger})
	\end{pmatrix}.
\end{equation}
where, again, the top-left component is given by $\Pi_{1,i}Q\Pi_{1,i}$, and in particular we performed the following manipulation:
\begin{equation}
    \Pi_{1,i}Q\Pi_{1,i}
    =
    \Pi_{1,i}Q^{1/2}(\Pi_{2,i}+ \Pi_{2,-i}+\Pi_{2,0})Q^{1/2}\Pi_{1,i}
    = aa^{\dag} + bb^{\dag}+c c^{\dag}\,,
\end{equation}
and similarly for the other blocks.
We can simplify $Q$ using some of the identities in equation \eqref{eq:matrix-identity}, as follows:
\textsc{\begin{equation} \label{eq:Q-block-matrix}
	Q
	=
	\begin{pmatrix}
		(2 b b^{\dagger} + c c^{\dagger} + 1)
		& (2 a j^{\dagger} + c l^{\dagger})
		& (2 b y^{\dagger} + c z^{\dagger}) \\
		(2 j a^{\dagger} + l c^{\dagger})
		& (2 j j^{\dagger} + l l^{\dagger} + 1)
		& (2 j x^{\dagger} + l z^{\dagger}) \\
		(2 y b^{\dagger} + z c^{\dagger})
		& (2 x j^{\dagger} + z l^{\dagger})
		& (2 x x^{\dagger} + z z^{\dagger})
	\end{pmatrix}.
\end{equation}}

Suppose now that we are making the first set of assumptions, i.e., those encompassed by equations \eqref{eq:pure-conditions} and \eqref{eq:abelian-conditions}.
Equation \eqref{eq:pure-conditions} tells us that we have
\begin{equation}
	2 b b^{\dagger} + c c^{\dagger} \quad \text{ and } \quad 2 j j^{\dagger} + l l^{\dagger} \quad \text{trace-class}.
\end{equation}
All terms in these expressions are positive, so we learn that $bb^{\dagger}, cc^{\dagger}, jj^{\dagger},$ and $ll^{\dagger}$ are trace-class.
Equivalently, the operators $b, c, j,$ and $l$ are Hilbert-Schmidt.
Equation \eqref{eq:abelian-conditions} tells us that $2 x x^{\dagger} + z z^{\dagger} - 1$ is Hilbert-Schmidt.
If we write
\begin{equation} \label{eq:Q-minus-1-block}
	Q - 1
	=
	\begin{pmatrix}
		(2 b b^{\dagger} + c c^{\dagger})
		& (2 a j^{\dagger} + c l^{\dagger})
		& (2 b y^{\dagger} + c z^{\dagger}) \\
		(2 j a^{\dagger} + l c^{\dagger})
		& (2 j j^{\dagger} + l l^{\dagger})
		& (2 j x^{\dagger} + l z^{\dagger}) \\
		(2 y b^{\dagger} + z c^{\dagger})
		& (2 x j^{\dagger} + z l^{\dagger})
		& (2 x x^{\dagger} + z z^{\dagger} - 1)
	\end{pmatrix},
\end{equation}
then we see that $Q-1$ is Hilbert-Schmidt, since the bottom-right block is Hilbert-Schmidt by assumption, and every other term in every other block contains one of the Hilbert-Schmidt operators $b, c, j,$ or $l.$

In the other direction, suppose that we are assuming $Q-1$ is Hilbert-Schmidt and that $\ker Q =0$.
As discussed above, these conditions imply that $\ker{(\Pi_{1,0} Q \Pi_{1,0})}$ is trivial within $\ker{R_1}$.
This obviously gives $\Pi_{0} (Q - 1) \Pi_0$ Hilbert-Schmidt, so we just need to establish equation \eqref{eq:pure-conditions}.
In this setting, we will actually find it easier to prove the other condition from the statement of theorem \ref{thm:mixed-from-pure}, i.e., to prove that one has
\begin{equation} \label{eq:Q-block-TC-desired}
	Q_{\pm i} + R_1 Q_{\pm i}^{-1} R_1 \quad \text{trace-class,}
\end{equation}
where $Q_{\pm i}$ is the projection of $Q$ to the support of $R_1.$
Technically, we must also show that $\Pi_{\pm i} (Q - 1) \Pi_{\pm i}$ is Hilbert-Schmidt, but this follows from the same reasoning that we just used for $\Pi_{0} (Q - 1) \Pi_0$.

{We will prove equation \eqref{eq:Q-block-TC-desired} by writing $Q$ in a block decomposition for $R_1$, and constraining these blocks via the following intermediate relation:
\begin{equation}
    \sqrt{Q - i R_1} - \sqrt{1 - i R_1} \quad \text{Hilbert-Schmidt.}
    \label{eq:intermediate-hs}
\end{equation}
Of course, this is \textit{not} the $\sqrt{X_2}-\sqrt{X_1}$ relation that arises from applying the standard criterion in theorem \ref{thm:mixed-from-pure} to the case at hand; that would instead read 
\begin{equation}
    \sqrt{Q_{\pm i} -i R_1}-
    \sqrt{\Pi_{\pm i}-i R_1} \quad \textrm{ Hilbert-Schmidt,}
    \label{eq:right-track}
\end{equation}
and this is difficult to prove directly.
For now, let us prove equation \eqref{eq:intermediate-hs}.
First, as discussed above, we know under the present set of assumptions that $Q$ is positive and bounded away from zero, i.e., $Q$ is invertible.
From this and the fact that $Q - 1$ is Hilbert-Schmidt, we can use the Birman-Solomyak theorem as in section \ref{sec:pure-two-point} to conclude that $Q^{1/2} - 1$ is also Hilbert-Schmidt. Equivalently $Q^{1/2}\sim_{\text{H.S.}} 1$.
Let us write
\begin{equation}
	\sqrt{Q - i R_1}
		= \sqrt{Q^{1/2} (1 - i v^{\dagger} R_2 v) Q^{1/2}}.
\end{equation}
As in section \ref{sec:pure-two-point}, applying the Birman-Solomyak theorem to the absolute value function allows us to remove the $Q^{1/2}$ factors:\footnote{\label{foot:abs-val-BS}Let us recall this result, since it is used many times in this section.
Given a positive, bounded operator $O$, and a second operator $M$ with $M\sim_{H.S.} 1$, one clearly has $\sqrt{O}\sim_{\mathrm{H.S.}}\sqrt{O}M$, and  the Birman-Solomyak theorem applied to the absolute value function then leads to $\sqrt{O}=|\sqrt{O}| \sim_{\mathrm{H.S.}}|\sqrt{O}M|=\sqrt{M^\dag O M}$. To get equation \eqref{eq:used-abs-value-bs} we can set $O=1-iv^{\dag}R_2 v$ and $M=Q^{1/2}$.}
\begin{equation}
	\sqrt{Q - i R_1}
	\sim_{\text{H.S.}} \sqrt{1 - i v^{\dagger} R_2 v}.
    \label{eq:used-abs-value-bs}
\end{equation}
Using $v^{\dagger} R_2 v = i (\Pi_{2,i} - \Pi_{2,-i}),$ we obtain
\begin{equation}
	\sqrt{Q - i R_1}
	\sim_{\text{H.S.}} \sqrt{2 \Pi_{2,i} + \Pi_{2,0}}
	= \sqrt{2} \Pi_{2,i} + \Pi_{2,0}.
    \label{eq:expand-as-projector}
\end{equation}
This gives
\begin{equation} \label{eq:Q-root-cent-pure}
	\sqrt{Q - i R_1}
		- \sqrt{1 - i R_1}
		\sim_{\text{H.S.}}
		\sqrt{2} (\Pi_{2,i} - \Pi_{1,i}) + (\Pi_{2,0} - \Pi_{1,0}),
\end{equation}
where we have expanded $\sqrt{1-iR_1}=\sqrt{2\Pi_{1,i}+\Pi_{1,0}}$ as in the last step of equation \eqref{eq:expand-as-projector}.
To show the desired relation \eqref{eq:intermediate-hs} we now need only massage the combination of projectors on the right-hand side into something Hilbert-Schmidt.

To do so, we will make use of the identity
\begin{equation}
	R_1
		= L^{\dagger} R_2 L
		= i Q^{1/2} (\Pi_{2,i} - \Pi_{2, -i}) Q^{1/2}.
\end{equation}
If we conjugate by $Q^{-1/2},$ square the expression, then conjugate by $Q^{1/2},$ we obtain
\begin{equation}
	R_1 Q^{-1} R_1
		= - Q^{1/2} (1 - \Pi_{2, 0}) Q^{1/2},
\end{equation}
or
\begin{equation} \label{eq:cent-pure-QQ-identity}
	Q + R_1 Q^{-1} R_1
	= Q^{1/2} \Pi_{2, 0} Q^{1/2}.
\end{equation}
Using $Q^{1/2} \sim_{\mathrm{H.S.}\!}1$, we learn that the right-hand side is Hilbert-Schmidt equivalent to $\Pi_{2,0}.$
Similarly, using $Q \sim_{\mathrm{H.S.}\!}1$ and $Q^{-1} \sim_{\mathrm{H.S.}\!}1$, we obtain that the left-hand side is Hilbert-Schmidt equivalent to $1+R_1^2=\Pi_{1,0}$. We conclude the Hilbert-Schmidt equivalence
\begin{equation}
	\Pi_{2,0}
		\sim_{\text{H.S.}} \Pi_{1, 0}. 
\end{equation}
It follows that we can subtract arbitrary multiples of $\Pi_{2,0} - \Pi_{1,0}$ from the right-hand side of equation \eqref{eq:Q-root-cent-pure}, obtaining
\begin{align}
	\begin{split}
	\sqrt{Q - i R_1} - \sqrt{1 - i R_1}
		& \sim_{\text{H.S.}} \sqrt{2} (\Pi_{2,i} + \frac{1}{2} \Pi_{2,0} - \Pi_{1,i} - \frac{1}{2}\Pi_{1,0}) \\
		& = \frac{1}{\sqrt{2}} (v^{\dagger} (1 - i R_2) v - (1 - i R_1)).
	\end{split}
\end{align}
We can freely conjugate the first term by $Q^{1/2}$ without changing the Hilbert-Schmidt property --- since we have $Q^{1/2}\sim_{\mathrm{H.S}}1$  --- and rearranging terms gives
\begin{align}
	\begin{split}
		\sqrt{Q - i R_1} - \sqrt{1 - i R_1}
		\sim_{\text{H.S}} \frac{1}{\sqrt{2}} (Q - 1) - i \frac{1}{\sqrt{2}} (L^{\dagger} R_2 L - R_1).
	\end{split}
\end{align}
The first term is Hilbert-Schmidt by assumption, and the second term vanishes, so we have $\sqrt{Q - i R_1} - \sqrt{1 - i R_1}$ Hilbert-Schmidt (equation \eqref{eq:intermediate-hs}).

We are now prepared to prove equation \eqref{eq:Q-block-TC-desired}.
We decompose} $Q$ in a 2-by-2 block decomposition according to the support of $R_1$ and the kernel of $R_1.$
Note that this is similar in spirit to the $R_1$ block decomposition of the matrices in equations \eqref{eq:matrix-identity}-\eqref{eq:Q-minus-1-block}, but now we group together the $+i$ and $-i$ blocks as follows:
\begin{equation}
\begin{aligned}
    	Q = \;
        \Pi_{1,\pm i}Q\Pi_{1,\pm i} + \Pi_{1,\pm i}&Q\Pi_{1,0} \\
        +\Pi_{1,0} Q \Pi_{1,\pm i} + \Pi_{1,0} &Q \Pi_{1,0}.
\end{aligned}
\end{equation}
We write this as
\begin{equation}
	Q
		= \begin{pmatrix}
			A & B \\ B^{\dagger} & C
		\end{pmatrix}
        =
        \begin{pmatrix}
			Q_{\pm i} & B \\ B^{\dagger} & C
		\end{pmatrix}.
\end{equation}
As in section \ref{sec:exciting-from-pure}, we also define the combination
\begin{equation}
    Y\equiv A - B C^{-1} B^{\dagger}\,,
\end{equation}
which naturally arises if we conjugate $Q$ by the matrix 
\begin{equation}
	G
		= \begin{pmatrix}
			1 & - B C^{-1} \\ 0 & 1
		\end{pmatrix}
\end{equation}
to obtain
\begin{equation}
	G Q G^{\dagger}
		= \begin{pmatrix} Y & 0 \\ 0 & C \end{pmatrix}.
\end{equation}
The trick we will use to prove equation \eqref{eq:Q-block-TC-desired}, namely that $A+R_1 A^{-1}R_1$ is trace-class, is to instead prove the equivalent statement that $Y+R_1 Y^{-1}R_1$ is trace-class.
This is an equivalent statement, because $Q - 1$ being Hilbert-Schmidt implies that the block $B$ is Hilbert-Schmidt, so $BC^{-1}B^{\dag}=A-Y$ is trace-class; hence, we have $Y\sim_{T.C.}A$ (and $Y^{-1}\sim_{T.C.}A^{-1}$).

To constrain $Y$ and $R_1$, we begin by writing
\begin{equation}
	G (Q - i R_1) G^{\dagger}
		=
			\begin{pmatrix} Y - i R_1& 0 \\ 0 & C \end{pmatrix},
\end{equation}
where we have used the easy-to-check identity $G R_1 G^{\dagger} = R_1$.
Now, since $B$ is Hilbert-Schmidt, with $BC^{-1}\sim_{H.S.}B$, one has that $G - 1$ is Hilbert-Schmidt.
Consequently, the Birman-Solomyak theorem applied to the absolute value function gives\footnote{See footnote \ref{foot:abs-val-BS} with $O \equiv Q-iR_1$ and $M\equiv G^\dag$.}
\begin{equation}
	\sqrt{Q - i R_1}
		\sim_{\text{H.S.}}
			\sqrt{G (Q - i R_1) G^{\dagger}}
			= \begin{pmatrix}
				\sqrt{Y- i R_1} & 0 \\ 0 & \sqrt{C}
			\end{pmatrix}.
\end{equation}

We may now apply equation \eqref{eq:intermediate-hs}, namely our observation that $\sqrt{Q - i R_1} - \sqrt{1 - i R_1}$ is Hilbert-Schmidt, and we learn
\begin{equation}\label{eq:intermediate-HS}
	\sqrt{Y- i R_1}
		- \sqrt{\Pi_{1, \pm i} - i R_1} \quad \text{Hilbert-Schmidt}.
\end{equation}
Since we know that $A - 1$ is Hilbert-Schmidt, and $B$ is Hilbert-Schmidt, we have $Y-1$  Hilbert-Schmidt. Applying again the Birman-Solomyak theorem with the square root to the first summand above then leads to\footnote{See again footnote \ref{foot:abs-val-BS}, with $O\equiv Y-iR_1$ and $M\equiv Y^{-1/2}$.}
\begin{equation}
	\sqrt{\Pi_{1,\pm i}  - i Y^{-1/2} R_1 Y^{-1/2}}
	- \sqrt{\Pi_{1, \pm i} - i R_1} \quad \text{Hilbert-Schmidt}.
\end{equation}

Using Birman-Solomyak once more, this time with the function $t \mapsto \sqrt{1 - (t^2 - 1)^2},$ one finds
\begin{equation}
	\sqrt{\Pi_{1,\pm i} + \left( Y^{-1/2} R_1 Y^{-1/2} \right)^2}
	- \sqrt{\Pi_{1, \pm i}  + R_1^2} \quad \text{Hilbert-Schmidt},
\end{equation}
where the second term vanishes and so we are left with
\begin{equation}
	\Pi_{1,\pm i}  + \left( Y^{-1/2} R_1 Y^{-1/2} \right)^2 \quad
	\text{trace-class}.
\end{equation}
Multiplying on both sides by $Y^{1/2}$ we obtain $Y + R_1 Y^{-1} R_1$ trace-class, and this completes the proof of equation \eqref{eq:Q-block-TC-desired}, and hence also of theorem \ref{thm:c-pure-from-c-pure}.

\section{Gaussian excitability implies quasiequivalence}
\label{sec:other-way}

In the previous sections, we proved the main theorem of the paper, theorem \ref{thm:general-theorem}.
This theorem says that for zero-mean Gaussian states $\omega_1$ and $\omega_2,$ the excitability relation $\omega_2 \prec \omega_1$ is equivalent to
\begin{enumerate}[(i)]
		\item $\mu_2 \prec \mu_1$,
		\item $\sqrt{X_2} - \sqrt{X_1}$ is Hilbert-Schmidt, and
		\item $Q$ has trivial kernel.
\end{enumerate}

In this section, we use this result to prove theorem \ref{thm:always-quasi}, which states that whenever we have $\omega_2 \prec \omega_1$ for zero-mean Gaussian states, we automatically have $\omega_1 \prec \omega_2.$
This should be contrasted with the discussion of one-way excitability in section \ref{sec:one-way-excitability}.
There, it was shown that for any non-factorial state $\omega_1$, there exists a state $\omega_2$ with $\omega_2 \prec \omega_1$ and $\omega_1 \not\prec \omega_2.$
The key point of the present section is that this situation cannot occur when both states are zero-mean Gaussian states.

The main observation is that condition (ii) above, which states that $\sqrt{X_2} - \sqrt{X_1}$ is Hilbert-Schmidt, necessarily implies that $Q-1$ is Hilbert-Schmidt as well, since $Q - 1 = X_2 - X_1$ can be obtained by left- and right-multiplying the Hilbert-Schmidt operator $\sqrt{X_2} - \sqrt{X_1}$ by the bounded operator $\sqrt{X_2} + \sqrt{X_1}.$
So whenever we have $\omega_2 \prec \omega_1,$ we necessarily have that $Q-1$ is Hilbert-Schmidt.
This means that $Q$ must be invertible away from its kernel, since it cannot have infinitely many eigenvalues accumulating near zero.
But since $\omega_2 \prec \omega_1$ implies condition (iii) above, which is that $Q$ has trivial kernel, we learn that $Q$ is always invertible when we have $\omega_2 \prec \omega_1.$

Invertibility of $Q$ is the statement that $\mu_1$ and $\mu_2$ induce equivalent inner products on the space of smearing functions.
So we automatically have $\mu_1 \prec \mu_2,$ and if we define an operator $Q'$ on $\K_{\mu_2}$ that reproduces the $\mu_1$ inner product, we know that $Q'$ has no kernel.
Consequently, the versions of conditions (i) and (iii) for the relation $\omega_1 \prec \omega_2$ are satisfied.
To fully prove $\omega_1 \prec \omega_2,$ it remains only to check the appropriate version of condition (ii) above, 
i.e., we must show the property
\begin{equation}
    \sqrt{Q' - i R_2} - \sqrt{1 - i R_2} \quad \text{Hilbert-Schmidt on $\K_{\mu_2}$}.
\end{equation}
Letting $L$ be the ``identity map'' on test functions from $\K_{\mu_1}$ to $\K_{\mu_2},$ we have
\begin{equation}
    Q' = (L^{-1})^{\dagger} L^{-1},
\end{equation}
and as in section \ref{sec:general-necessary}, one also has
\begin{equation}
    R_2 = (L^{-1})^{\dagger} R_1 L^{-1}.
\end{equation}
So what we are really trying to establish is the Hilbert-Schmidt condition on
\begin{equation}\label{eq:alternate-AY-condition}
    \sqrt{(L^{-1})^{\dagger} (1 - i R_1) L^{-1}} - \sqrt{1 - i (L^{-1})^{\dagger} R_1 L^{-1}}.
\end{equation}

Writing the polar decomposition $L = v Q^{1/2}$, we may rewrite the operator we are studying as
\begin{equation}
    \sqrt{v Q^{-1/2} (1 - i R_1) Q^{-1/2} v^{\dagger}} - \sqrt{1 - i v Q^{-1/2} R_1 Q^{-1/2} v^{\dagger}},
\end{equation}
or, pulling out the $v$'s using $v^\dag v=1$,
\begin{equation}
    v \left( \sqrt{Q^{-1/2} (1 - i R_1) Q^{-1/2}} - \sqrt{Q^{-1/2} (Q - i R_1 ) Q^{-1/2}} \right) v^{\dagger}.
\end{equation}
Since $Q-1$ is Hilbert-Schmidt, one can apply the Birman-Solomyak theory of section \ref{sec:pure-two-point} to conclude that $Q^{-1/2} - 1$ and $Q^{1/2} - 1$ are Hilbert-Schmidt, and also that
\begin{equation}
    v \left( \sqrt{Q^{-1/2} (1 - i R_1) Q^{-1/2}}\right)v^\dag \sim_{\mathrm{H.S.}}\sqrt{1-iR_1},
\end{equation}
where we have used again the argument in footnote \ref{foot:abs-val-BS},
and similarly
\begin{equation}
    v\left(\sqrt{Q^{-1/2} (Q - i R_1 ) Q^{-1/2}} \right) v^{\dagger}\sim_{\mathrm{H.S.}}\sqrt{Q-iR_1}.
\end{equation}
The difference of the operators on the right-hand sides of these two equations is simply $\sqrt{X_2}-\sqrt{X_1}$, which is Hilbert-Schmidt by condition (ii) above, since we are assuming $\omega_2 \prec \omega_1$.
Consequently, the difference of the left-hand sides --- that is, the operator in equation \eqref{eq:alternate-AY-condition} --- must also be Hilbert-Schmidt, and this is what we were trying to show to reach the conclusion $\omega_1\prec \omega_2$.

\section{Discussion}
\label{sec:discussion}

The main point of this paper was to give a careful treatment of the ``excitability question,'' which is a physical cousin of the mathematical ``quasiequivalence question'' that has been discussed by algebraic quantum field theorists.
We established general algebraic criteria for excitability, and explained how it is possible in a general theory to have an excitability relation $\omega_2 \prec \omega_1$ without having the converse $\omega_1 \prec \omega_2.$
In a lengthy proof that involved several nested layers of argument, we then established theorem \ref{thm:general-theorem}, which gives general criteria for excitability in terms of a ``finite-trace condition'' when both $\omega_1$ and $\omega_2$ are zero-mean Gaussian states in a free theory.
We also established several other theorems --- theorems \ref{thm:both-pure}, \ref{thm:mixed-from-pure}, and \ref{thm:c-pure-from-c-pure} --- in which we provided alternate, simplified criteria that one can check in special cases where $\omega_1$ or $\omega_2$ has additional structure.

A surprising-to-us consequence of our analysis is that for zero-mean Gaussian states, one-way excitability always implies quasiequivalence.
This did not have to be true --- indeed, in section \ref{sec:excitability}, we gave examples of algebraic states in which one-way excitability is strict.
The additional assumption that the excited state is Gaussian, however, kills these examples.

A few immediate applications of our analysis spring to mind.
We summarize them below.
\begin{enumerate}[(i)]
	\item There are settings in which one would like to construct, within a given sector, a set of correlation functions that have useful features.
	One example is provided by the geometric modular flow conjecture from \cite{Jensen:JSS}.
	The analysis of the present paper allows this question to be answered, at least in principle, in free field theory.
    A forthcoming paper by one of us with Gautam Satishchandran will explore this point further \cite{Satishchandran:uniqueness}.
	
	\item Related to the above, since our analysis works for any free field theory --- including generalized free fields! --- it can be used to study sector equivalence within the large-N limit of AdS/CFT, which has many interesting implications as discussed in e.g. \cite{Leutheusser:HSMT, Witten:largeN, Leutheusser:subalgebra}.
	It may also be possible to apply these ideas to the semiclassical limit of closed-universe quantum gravity as formulated in \cite{Witten:background}.
	
	\item Strict one-way excitability is intimately related to the existence of a center for the von Neumann algebra associated with an algebraic state.
	In Klein-Gordon theory, one knows thanks to \cite{Verch:hadamard} that such a situation does not arise in physical (i.e., ``Hadamard'') states.
	Even if it did, our analysis would show that one-way excitability cannot happen for Gaussian states.
	
	However, there are important field theories in which a center for the von Neumann algebra is unavoidable: these are theories with higher-form symmetries \cite{Gaiotto:higher-form}, where topological operators of codimension greater than one will necessarily appear in the center of any local algebra.
	It would be interesting to study whether one-way excitability in such a setting reveals anything about the structure of the symmetry.
	
	\item Much of the actual technical material of this paper was concerned with nitty-gritty details of free, Gaussian states.
	The simplicity of this theory allowed us to argue at a high level of rigor; on the other hand, even in this simple setting, the arguments are quite complicated!
	It would be very interesting to understand if the algebraic approach to excitability lets one make interesting statements about when excitability is possible in interacting theories, even if one must sacrifice rigor in the process.
\end{enumerate}

\acknowledgments{
We gratefully acknowledge stimulating conversations with Marc Klinger, Roberto Longo, Lauritz van Luijk, Gautam Satishchandran, and Bob Wald.
We especially thank Lauritz van Luijk for sharing his expansive knowledge of von Neumann algebras to help us simplify several key arguments in this paper.
JC acknowledges the support of the Natural Sciences and Engineering Research Council of Canada
through Vanier Canada Graduate Scholarships [Funding Reference Number: CGV–192707]. FC acknowledges support from a Graduate Program Fellowship from the Department of Physics at the University of California, Davis.}

\appendix 

\section{(Generalized) free field theory}
\label{app:review}

This appendix provides an overview of the algebraic theory of free quantum fields.
The framework presented here encompasses both Klein-Gordon theory and generalized free field theory.
Our presentation differs somewhat from the classic literature, and some of our proof techniques are new.
The earlier reviews from which we extrapolated much of this material are \cite{Wald:QFTCS} and \cite{Hollands:review}.

\subsection{Key points}

\begin{itemize}
	\item A generalized free field theory starts with a $*$-algebra $\A_0$ generated by fields $\phi[f],$ where $f \in C^{\infty}_0$ is a compactly supported smooth function.
    The fields are subject to linearity relations and a c-number commutator
	\begin{equation}
		[\phi[f], \phi[g]] = - i \Omega[f, g].
	\end{equation}
	\item A state $\omega : \A_0 \to \comps$ is said to be Gaussian if the connected correlators
		\begin{equation}
			\omega_n^c(\phi[f_1] \dots \phi[f_n])
			= (-i)^n \left.\frac{\del^n}{\del t_1 \dots \del t_n} \log\omega(e^{i t_1 \phi[f_1]} \dots e^{i t_n \phi[f_n]})\right
			|_{t_1=\dots=t_n=0}
		\end{equation}
	vanish for $n \geq 3.$
	For such a state, one defines the ``vacuum-subtracted'' field $\tilde{\phi}[f] = \phi[f] - \omega(\phi[f]),$ for which the odd-point correlators vanish and one has
	\begin{equation}
		\omega(\tilde{\phi}[f_1] \dots \tilde{\phi}[f_{2n}])
		= \sum_{\substack{\text{ordered pairings} \\ P \text{ of } \{1, \dots, 2n\}}} \; \prod_{(j,k) \in P} \omega(\tilde{\phi}[f_j] \tilde{\phi}[f_k]).
        \label{eq:wick-formula}
	\end{equation}
	\item For any one-point function satisfying hermiticity $\omega(\phi[f])^* = \omega(\phi[f]^*),$ and for any two-point function respecting the canonical commutation relation and satisfying the positivity constraint
	\begin{equation}
		\sum_{j, k} c_j^* c_k \omega(\tilde{\phi}[f_j^*] \tilde{\phi}[f_k]) \geq 0,
	\end{equation}
	one obtains a unique Gaussian state on $\A_0$ via equation \eqref{eq:wick-formula}.
	\item Normal-ordering with respect to a state $\omega$ is defined by 
	\begin{equation}
		:\tilde{\phi}[f_1] \dots \tilde{\phi}[f_n]:
		\equiv \sum_{P}
		\prod_{j \notin P} \tilde{\phi}[f_j] \prod_{(k, \ell) \in P} \left[-\omega(\tilde{\phi}[f_k] \tilde{\phi}[f_\ell])\right].
	\end{equation}
	In the GNS space $\H_{\omega},$ one writes $\H_{\omega}^{(n)}$ for the Hilbert space obtained by acting on $|\omega\rangle$ with normal-ordered monomials of length $n.$
	This is called the ``$n$-particle space,'' and subspaces with different particle number are mutually orthogonal.
	\item The map
	\begin{equation}
		:\tilde{\phi}[g_1] \dots \tilde{\phi}[g_n]: |\omega\rangle
			\mapsto \sqrt{n!} |\tilde{\phi}[g_1] \omega \otimes_S \dots \otimes_S\tilde{\phi}[g_n] \omega\rangle
	\end{equation}
	gives a unitary equivalence between $\H_{\omega}^{(n)}$ and the $n$-th symmetric tensor power of $\H_{\omega}^{(1)}.$
	\item The field operator $\tilde{\phi}[f]$ can be written in terms of ``creation and annihilation operators'' as $\tilde{\phi}[f] = a[f] + a[f^*]^{\dagger},$ where one has
	\begin{equation}
		a[f^*]^{\dagger} :\tilde{\phi}[g_1] \dots \tilde{\phi}[g_n]: |\omega\rangle
		= : \tilde{\phi}[f] \tilde{\phi}[g_1] \dots \tilde{\phi}[g_n]: |\omega\rangle
	\end{equation}
	and
	\begin{equation}
		a[f] :\tilde{\phi}[g_1] \dots \tilde{\phi}[g_n]: |\omega\rangle
			= \sum_{j} \omega(\tilde{\phi}[f] \tilde{\phi}[g_j]) \left( :\tilde{\phi}[g_1] \dots \hat{\tilde{\phi}[f_j]} \dots \tilde{\phi}[g_n]: |\omega\rangle \right),
	\end{equation}
	with the hat denoting that the corresponding term is removed from the expression.
	\item After taking ``closures of operators,'' the operator we have called $a[f]^{\dagger}$ is genuinely the adjoint of $a[f]$; consequently, the adjoint of $\phi[f]$ is $\phi[f^*].$
	This means that for real $f,$ the operator $\phi[f]$ is self-adjoint and one can define the unitary ``Weyl operator'' $e^{i \phi[f]}.$
	\item Given a Gaussian state $\omega,$ one defines a ``symmetrized inner product'' on the space of test functions by
	\begin{equation}
		\langle f | g \rangle_{\mu}
			= \frac{\omega(\tilde{\phi}[f^*] \tilde{\phi}[g] + \tilde{\phi}[g] \tilde{\phi}[f^*])}{2}.
	\end{equation}
	The Hilbert space made out of $C^{\infty}_0$ equipped with this inner product is called $\K_{\mu}.$
	On this space, there is a bounded operator $R$ representing the commutator,
	\begin{equation}
		\langle f | R | g\rangle_{\mu} = \frac{1}{2} \Omega[f^*, g].
	\end{equation}
	The complex conjugation $\Gamma | f \rangle_{\mu} = |f^* \rangle_{\mu}$ is antiunitary, and one has
	\begin{equation}
		R^{\dagger} = - R
	\end{equation}
	and
	\begin{equation}
		\Gamma R \Gamma = R,
	\end{equation}
	with the spectrum of $R$ contained between $-i$ and $i.$
	When $\omega$ is pure, one has $R^2 = -1$ so that the only eigenvalues of $R$ are $\pm i$.
	When $\omega$ is centrally pure, the only eigenvalues of $R$ are $\pm i$ and zero.
	The two-point function is represented on $\K_{\mu}$ by
	\begin{equation}
		\omega(\tilde{\phi}[f^*] \tilde{\phi}[g])
			= \langle f | (1 - i R) |g \rangle_{\mu}.
	\end{equation}
	Hence the map
	\begin{equation}
		U \left( \tilde{\phi}[f] |\omega\rangle \right)
			= \sqrt{1 - i R} |f\rangle_{\mu}
	\end{equation}
	furnishes a unitary equivalence from $\H_{\omega}^{(1)}$ onto the orthocomplement of the $-i$ eigenspace of $R$.
	\item The von Neumann algebra $\A_{\omega}$ generated by $\A_0$ within the GNS space of a Gaussian state may equivalently be taken to be (i) the smallest von Neumann algebra such that each $\phi[f]$ is affiliated, or (ii) the smallest von Neumann algebra containing all Weyl operators $e^{i \phi[f]}$ with $f$ real.
	\item 
	The map $|f\rangle_{\mu} \to e^{i \phi[f]}$ is continuous with respect to the $\mu$-topology on the domain and the strong topology on the image, which allows one to define operators $\phi[\psi]$ for any $|\psi\rangle_{\mu}$ in $\K_{\mu}.$
	\item Weyl operators satisfy
	\begin{equation}
		\langle \omega | e^{i \tilde{\phi}[f]} |\omega\rangle = e^{- \langle f | f \rangle_{\mu}/2},
	\end{equation}
	or
	\begin{equation}
		\langle \omega | e^{i \phi[f]} |\omega\rangle = e^{- \langle f | f \rangle_{\mu}/2} e^{i \omega(\phi[f])}.
	\end{equation}
	Given any state $\ket{\psi}\in {\cal H}_{\omega}$ whose expectation values have this form, i.e., 
	\begin{equation}
		\langle \psi | e^{i\phi[f]} |\psi\rangle
			= e^{- \langle f | f \rangle_{\nu}/2} e^{i \psi(\phi[f])} \quad \text{for all Weyl operators},
	\end{equation}
	with $\psi(\phi[f])^* = \psi(\phi[f^*])$ and with $\nu$ a real inner product on the space of smearing functions, one can conclude that $|\psi\rangle$ induces a Gaussian state on the $*$-algebra $\A_0$.
	If one has a functional of this form but does not know that it comes from a ket $|\psi\rangle,$ one can still conclude that it induces a Gaussian state provided one checks the additional boundedness condition
	\begin{equation}
		\frac{i}{2} \Omega[f^*, f] \leq \langle f | f \rangle_{\nu}
	\end{equation}
	for general complex $f,$ or equivalently
	\begin{equation}
		|\Omega[u, v]| \leq 2 \lVert u \rVert_{\nu} \lVert v \rVert_{\nu}
	\end{equation}
	for general real $u$ and $v.$
\end{itemize}

\subsection{Basics}
\label{app:star-algebra}
	
From the algebraic perspective, a (generalized) free field theory is a collection of fields on spacetime that satisfy a ``c-number'' commutation relation.
One starts with a spacetime $\M$ and considers the collection $C^{\infty}_0(\M)$ of ``test functions'' --- smooth, compactly supported functions from $\M$ to $\mathbb{C}.$
To each test function $f$, one associates an object $\phi[f],$ which is supposed to represent the formal expression
\begin{equation}
	\phi[f] \equiv \int d^{D} x\, \sqrt{|g|} \phi(x) f(x),
\end{equation}
with $\phi(x)$ the generalized free field.

The fact that the field theory is ``free'' is encoded in the statement that the commutator $[\phi(x), \phi(y)]$ is proportional to the identity.
Formally, one writes
\begin{equation} \label{eq:appendix-formal-commutator}
	[\phi(x), \phi(y)] = - i \Omega(x, y),
\end{equation}
with $\Omega$ an antisymmetric distribution on two copies of $\M$.
The choice of $\Omega$ is part of the definition of the theory.
By smearing, one translates equation \eqref{eq:appendix-formal-commutator} into an algebraic equation
\begin{equation} \label{eq:appendix-distribution-commutator}
	[\phi[f], \phi[g]]
		= - i \Omega[f, g].
\end{equation}

The \textit{abstract $*$-algebra} of the generalized free field theory is the collection of polynomials made out of smeared fields $\phi[f],$ subject to the commutation relation \eqref{eq:appendix-distribution-commutator}, also subject to a linearity condition $\phi[\alpha f + g] = \alpha \phi[f] + \phi[g],$ and finally endowed with an adjoint operation that corresponds to complex conjugation:
\begin{equation}
	\phi[f]^* \equiv \phi[f^*].
\end{equation} 
This $*$-algebra is denoted $\A_0,$ or sometimes $\A_0(\M).$

For a generalized free field theory, the data we have discussed so far is all that goes into the definition of the theory.
In Klein-Gordon theory, one has additional data supplied by the equation of motion.
If $\M$ is a globally hyperbolic spacetime, then the Klein-Gordon equation $(\Box - m^2) \phi = 0$ admits fundamental ``advanced'' and ``retarded'' propagators \cite{Friedlander:book}, which are kernels $A$ and $R$ such that one has
\begin{equation}
	(\Box - m^2) (A f) = f
\end{equation}
with $Af$ vanishing to the future of the support of $f,$ and similarly for $Rf$ but with $Rf$ vanishing to the past of the support of $f.$
The ``advanced-minus-retarded'' kernel $E = A - R$ maps test functions to solutions:
\begin{equation}
    (\Box - m^2) (E f) = 0.
    \label{eq:advanced-minus-retarded}
\end{equation}

The first way that the equations of motion show up in algebraic Klein-Gordon theory is through the canonical commutation relation.
In particular, the commutator from equation \eqref{eq:appendix-distribution-commutator} is determined by setting
\begin{equation}\label{eq:Omega-KG}
	\Omega(x, y) = E(x, y),
\end{equation}
which comes from canonical quantization of the symplectic form on the Klein-Gordon phase space.
By abuse of notation we will often refer to $\Omega$ as the ``symplectic form'' even in the context of generalized free field theory, where $\Omega$ is actually just an antisymmetric distribution on two copies of $\M$.

Another important way in which the equations of motion show up is that in algebraic Klein-Gordon theory, one imposes $(\Box - m^2) \phi = 0$ on the abstract free field via integration by parts.
This amounts to subjecting the $*$-algebra $\A_0$ to the algebraic equivalence relation
\begin{equation}
	\phi[(\Box - m^2) f] = 0,
\end{equation}
or, equivalently, quotienting the space of test functions by $\text{im}(\Box - m^2)$.
Effectively, one is smearing the free field $\phi(x)$ against elements of the space 
\begin{equation}
    \mathscr{S} \equiv C^{\infty}_0(\M) / \text{im}(\Box - m^2).
    \label{eq:phase-space-of-kg}
\end{equation}

The choice of symbol ``$\mathscr{S}$'' is deliberate; one can actually show that $\mathscr{S}$ is isomorphic to the phase space of \textit{solutions} of Klein-Gordon theory that are smooth and that have compact support on any spatial slice.
The isomorphism is simply the map $[f] \mapsto E f.$
That this map is well defined follows from the compatibility of $E$ with the equations of motion, i.e., $E (\Box - m^2) f = 0.$
That the map is an isomorphism follows from the facts --- see \cite[lemma 3.2.1]{Wald:QFTCS} --- that $f \mapsto E f$ is surjective onto phase space, and that the kernel of this map is \textit{exactly} the image of the equations of motion.

\subsection{Gaussian states and normal ordering}
\label{app:normal-ordering}

Given an abstract $*$-algebra $\A_0$, like the $*$-algebra of free fields discussed in the previous subsection, one constructs Hilbert space representations via the GNS construction.
To do this, one specifies a linear functional $\omega : \A_0 \to \comps$ that satisfies the positivity condition
\begin{equation}\label{eq:state-positivity}
	\omega(a^* a) \geq 0.
\end{equation}
This functional is supposed to represent the set of correlation functions for an abstract quantum state.
The GNS construction allows one to realize these correlation functions on a Hilbert space.
To do this, one starts with the abstract vector space
\begin{equation}
	V_{\omega} = \{|a \omega\rangle\, |\, a \in \A_0\},
    \label{eq:V-omega-space}
\end{equation}
and endows this space with the inner product
\begin{equation}\label{eq:GNS-inn-prod}
	\braket{a \omega}{b \omega} = \omega(a^* b).
\end{equation}
Taking the quotient by the space of null states, then taking a Hilbert space completion, one obtains a Hilbert space $\H_{\omega}$ whose elements are (Cauchy sequences of) equivalence classes $[| a\omega\rangle].$
There is rarely any confusion in simply writing $|a \omega\rangle$ for the corresponding equivalence class.
The Hilbert space $\H_{\omega}$ carries an action of the $*$-algebra $\A_0$ via the formula
\begin{equation}
	b |a \omega\rangle = |(b a) \omega\rangle.
\end{equation}
This tells us how each element $b$ of the $*$-algebra should act on states of the form $a |\omega\rangle$ with $a \in \A_0.$
However, this action will generally be unbounded, and cannot be completed continuously to an action on all of $\H_{\omega}.$
We say that the operators in $\A_0$ are represented on $\H_{\omega}$ as unbounded operators with a common dense domain $\{a |\omega\rangle\}.$
For more on unbounded operators, especially as they appear in the GNS construction, see \cite[appendix A]{Sorce:paper1}.

Given a (generalized) free field algebra, there are some nice simplifications due to the fact that the canonical commutation relation is proportional to the identity.
This algebraic feature means that free field algebras support \textit{Gaussian} states, where all $n$-point functions are determined from the $1$- and $2$-point functions.
They are defined by starting with the formal expression for connected $n$-point functions,
\begin{equation}\label{eq:conn-n-point}
	\omega_n^c(\phi[f_1] \dots \phi[f_n])
		= (-i)^n \left.\frac{\del^n}{\del t_1 \dots \del t_n} \log\omega(e^{i t_1 \phi[f_1]} \dots e^{i t_n \phi[f_n]})\right
		|_{t_1=\dots=t_n=0},
\end{equation}
and demanding that these vanish for $n \geq 3$.
More explicitly, for a Gaussian state with vanishing one-point function, all odd-point correlators vanish and the even-point correlators are given by taking Wick contractions:
\begin{equation} \label{eq:appendix-even-point}
	\omega(\phi[f_1] \dots \phi[f_{2n}])
		= \sum_{\substack{\text{ordered pairings} \\ P \text{ of } \{1, \dots, 2n\}}}\; \prod_{(j,k) \in P} \omega(\phi[f_j] \phi[f_k]).
\end{equation}
It turns out that this simple formula is also useful for Gaussian states with non-vanishing one-point function.
For if the one-point function is nonvanishing, one can define a new field by
\begin{equation} \label{eq:appendix-vev}
	\tilde{\phi}[f] \equiv \phi[f] - \omega(\phi[f]),
\end{equation}
and for $n \geq 2,$ one has (from the generating functional expression \eqref{eq:conn-n-point})
\begin{equation}
	\omega_n^c(\phi[f_1] \dots \phi[f_n]) = \omega_n^c(\tilde{\phi}[f_1] \dots \tilde{\phi}[f_n]).
\end{equation}
This means that if $\omega$ a Gaussian state for the $\phi$ fields, it is also a Gaussian state with respect to the ``vacuum-subtracted'' fields. 
But $\omega$ has vanishing one-point function with respect to $\tilde \phi$, and hence the odd-point functions of $\tilde \phi$ vanish, while the even-point functions are given by the analogue of equation \eqref{eq:appendix-even-point}, with $\phi$ replaced by the vacuum-subtracted field $\tilde \phi$.

To summarize, Gaussian states are determined by their 1- and 2-point functions.
Before discussing the implications of this for the GNS construction, we will address the question of which choices of 1- and 2-point functions yield well-defined Gaussian states.
For any one-point function $\omega(\phi[f])$ with $\omega(\phi[f])^* = \omega(\phi[f^*]),$ and for any two-point function respecting the canonical commutation relation,
\begin{equation}
	\omega(\phi[f] \phi[g]) - \omega(\phi[g] \phi[f])
		= - i \Omega[f, g],
\end{equation}
one can use equation \eqref{eq:appendix-even-point} with $\phi$ replaced by $\tilde{\phi}$ to \textit{define} a consistent functional on the $*$-algebra $\A_0.$\footnote{\label{foot:EoM-state} In the case of Klein-Gordon theory, one must also require that the one- and two-point functions respect the equation of motion. That is, for a functional of $f$ to consistently define a candidate one-point function of a state in Klein–Gordon theory, it must vanish whenever $f=(\square -m^2)g$ for some test function $g$. A similar constraint must be imposed on the two-point function.}
In order to show that this functional really corresponds to the correlation functions of a quantum state, we need to check the positivity condition \eqref{eq:state-positivity}, where now $a$ is a general polynomial in the smeared fields $\phi[f]$.
Luckily, if one has positivity of the two-point function for vacuum-subtracted fields, namely
\begin{equation} \label{eq:two-point-positivity}
	\sum_{j, k} c_j^* c_k\, \omega(\tilde{\phi}[f_j^*] \tilde{\phi}[f_k]) \geq 0,
\end{equation}
then the corresponding Gaussian functional defined by \eqref{eq:appendix-even-point} is guaranteed to be positive.
We will prove this at the end of this subsection after developing normal-ordering technology.

We now return to our discussion of the GNS construction, and how it simplifies for Gaussian GNS states.
Recall that the pre-GNS space $V_{\omega}$ was defined as the set of states obtained by acting on $|\omega\rangle$ with linear combinations of field monomials $\phi[f_1] \dots \phi[f_n]$, for any integer $n$.
For Gaussian states, it is extremely useful to define ``Wick monomials'' or ``normal-ordered monomials'' $:\tilde{\phi}[f_1] \dots \tilde{\phi}[f_n]:,$ defined with respect to the chosen state $\omega.$
One defines
\begin{equation} \label{eq:normal-ordering}
	:\tilde{\phi}[f_1] \dots \tilde{\phi}[f_n]:\;\,
		\equiv \sum_{P}
			\prod_{j \notin P} \tilde{\phi}[f_j] \prod_{(k, \ell) \in P} \left[-\omega(\tilde{\phi}[f_k] \tilde{\phi}[f_\ell])\right],
\end{equation}
where the sum is taken over all possible ways $P$ of picking an even number of elements from the monomial and dividing them up into ordered pairs that respect the original ordering.
It is easy to check, by induction, the dual relation
\begin{equation}
	\tilde{\phi}[f_1] \dots \tilde{\phi}[f_n]
		= \sum_{P}
		: (\prod_{j \notin P} \tilde{\phi}[f_j]): \prod_{(k, \ell) \in P} \omega(\tilde{\phi}[f_k] \tilde{\phi}[f_\ell]).
\end{equation}
This means that every polynomial in the fields can be written as a linear combination of Wick polynomials, and consequently that the pre-GNS space $V_{\omega}$ can be obtained by acting on $|\omega\rangle$ with Wick polynomials instead of ordinary field polynomials.

The main utility of normal-ordered monomials is that they satisfy a useful ``orthogonality relation'' with respect to the state $\omega.$
The starting point for this is the general identity
\begin{equation} \label{eq:pre-creation-annihilation}
	\tilde{\phi}[f_1] (:\tilde{\phi}[f_2] \dots \tilde{\phi}[f_n]:)
		= (:\tilde{\phi}[f_1] \dots \tilde{\phi}[f_n]:) + \sum_{j=2}^{n} \omega(\tilde{\phi}[f_1] \tilde{\phi}[f_j]) \left( :\tilde{\phi}[f_2] \dots \hat{\tilde{\phi}[f_j]} \dots \tilde{\phi}[f_n]: \right),
\end{equation}
where the hat in the second term indicates that the corresponding element $\tilde{\phi}[f_j]$ should be removed from the expression.
Equation \eqref{eq:pre-creation-annihilation} is easily verified by direct calculation from the definition \eqref{eq:normal-ordering}.
Using this formula, and proceeding by induction, it is straightforward to compute the following:
\begin{equation} \label{eq:normal-orthogonality}
	\omega\left((:\tilde{\phi}[f_1] \dots\tilde{\phi}[f_m]:)^* :\tilde{\phi}[g_1] \dots \tilde{\phi}[g_n]:\right)
		= \delta_{m, n} \sum_{\pi \in S_m} \omega(\tilde{\phi}[f_1^*] \tilde{\phi}[g_{\pi(1)}]) \dots \omega(\tilde{\phi}[f_m^*] \tilde{\phi}[g_{\pi(m)}]).
\end{equation}
In the GNS language, this is written
\begin{equation} \label{eq:GNS-normal-overlaps}
	\langle  :\tilde{\phi}[f_1] \dots\tilde{\phi}[f_m]: \omega | :\tilde{\phi}[g_1] \dots \tilde{\phi}[g_n]: \omega\rangle
	= \delta_{m, n} \sum_{\pi \in S_m} \langle \tilde{\phi}[f_1] \omega | \tilde{\phi}[g_{\pi(1)}] \omega\rangle \dots \langle \tilde{\phi}[f_m] \omega | \tilde{\phi}[g_{\pi(m)}] \omega\rangle.
\end{equation}
This immediately tells us that when $\omega$ is Gaussian, the GNS space splits into mutually orthogonal ``$n$-particle subspaces'' obtained by acting on $|\omega\rangle$ with one-point-subtracted Wick monomials of order $n$:
\begin{equation}
    \H_{\omega}^{(n)}
		\equiv \bar{V_{\omega}^{(n)}/\text{null states}}
    \label{eq:n-particle-hibert-space}
\end{equation}
where in analogy to equation \eqref{eq:V-omega-space}, we have defined
\begin{equation}
	V_{\omega}^{(n)}
		\equiv \{:\tilde{\phi}[f_1] \dots \tilde{\phi}[f_n]:|\omega\rangle\}.
\end{equation}

We will explore this structure more in the next subsection; first, however, we use normal ordering to prove the claim that $\omega$ is positive on $\A_0$ whenever equation \eqref{eq:two-point-positivity} holds.

The point is that even if $\omega$ were not positive, equation \eqref{eq:normal-orthogonality} would still break the $*$-algebra $\A_0$ into ``$n$-particle pieces'' that are mutually orthogonal with respect to $\omega.$
Consequently, positivity of $\omega$ need only be checked on each $n$-particle space. 
Equation \eqref{eq:two-point-positivity} is the assumption that $\omega$ is positive on the $1$-particle space.
For $n > 1,$ we have
\begin{equation} \label{eq:n-Wick-overlap}
	\omega((:\tilde{\phi}[f_1] \dots \tilde{\phi}[f_n]:)^* :\tilde{\phi}[g_1] \dots \tilde{\phi}[g_n]:)
		= \sum_{\pi \in S_n} \omega(\tilde{\phi}[f_1]^* \tilde{\phi}[g_{\pi(1)}]) \dots \omega(\tilde{\phi}[f_n]^* \tilde{\phi}[g_{\pi(n)}]).
\end{equation}
Denote by $\H_{\omega}^{(1)}$ the space we get by completing the space $\{\tilde{\phi}[f]|\omega\rangle\}$ as in equation \eqref{eq:n-particle-hibert-space}.
Equation \eqref{eq:n-Wick-overlap} can be rewritten in terms of inner products in the $n$-th tensor power of this space as
\begin{equation}
\label{eq:check-positivity-omega}
	\omega((:\tilde{\phi}[f_1] \dots \tilde{\phi}[f_n]:)^* :\tilde{\phi}[g_1] \dots \tilde{\phi}[g_n]:)
	= \sum_{\pi \in S_n} \langle \tilde{\phi}[f_1] \omega \otimes \dots \otimes \tilde{\phi}[f_n] \omega | \hat{\pi} | \tilde{\phi}[g_1] \omega \otimes \dots \tilde{\phi}[g_n] \omega \rangle.
\end{equation}
So the $\omega$ inner product on $V_{\omega}^{(n)}$ is simply implemented by the operator $\sum_{\pi} \hat{\pi}$ on the $n$-fold tensor product of $\H_{\omega}^{(1)}.$
This is (proportional to) the orthogonal projection onto the space of symmetric tensors, so it is positive, and hence induces a positive semidefinite inner product.

\subsection{The Fock structure of Gaussian representations}
\label{app:fock}

In the previous subsection, we showed that the GNS representation of a Gaussian state $\omega$ decomposes into $n$-particle sectors,
\begin{equation}
	\H_{\omega}
		= \oplus_n \H_{\omega}^{(n)},
\end{equation}
where $\H_{\omega}^{(n)}$ is defined by
\begin{equation}
	\H_{\omega}^{(n)}
		= \bar{\{:\tilde{\phi}[f_1] \dots \tilde{\phi}[f_n]:|\omega\rangle\}}
\end{equation}
and
\begin{equation}
	\tilde{\phi}[f]
		= \phi[f] - \omega(\phi[f]).
\end{equation}
We now connect this result with the more familiar notion of a Fock space generated by the repeated action of creation operators.
Equation \eqref{eq:pre-creation-annihilation} tells us how $\tilde{\phi}[g]$ acts on the $n$-particle subspace --- we have
\begin{align}
	\begin{split}
	 \tilde{\phi}[g] \left( :\tilde{\phi}[f_1] \dots \tilde{\phi}[f_n]: |\omega\rangle \right) = \;&:\tilde{\phi}[g] \tilde{\phi}[f_1] \dots \tilde{\phi}[f_n]: |\omega\rangle\\
	& 
	+ \sum_{j=1}^{n} \omega(\tilde{\phi}[g] \tilde{\phi}[f_j]) :\tilde{\phi}[f_1] \dots \hat{\tilde{\phi}[f_j]} \dots \tilde{\phi}[f_n]: |\omega\rangle.
	\end{split}
\end{align}
The left-hand side shows $\tilde{\phi}[g]$ acting on an element of $\H_{\omega}^{(n)}.$
The first term on the right-hand side is in $\H_{\omega}^{(n+1)},$ and the second is in $\H_{\omega}^{(n-1)}.$
So we see that the field operator splits into a piece that adds a particle, and a piece that removes a particle.
These are the creation and annihilation operators of free field theory.
We write
\begin{equation} \label{eq:annihilation-definition}
	a[g] :\tilde{\phi}[f_1] \dots \tilde{\phi}[f_n]: |\omega\rangle
		= \sum_{j=1}^{n} \omega(\tilde{\phi}[g] \tilde{\phi}[f_j]) :\tilde{\phi}[f_1] \dots \hat{\tilde{\phi}[f_j]} \dots \tilde{\phi}[f_n]: |\omega\rangle
\end{equation}
and
\begin{equation} \label{eq:creation-definition}
	a[g^*]^{\dagger} :\tilde{\phi}[f_1] \dots \tilde{\phi}[f_n]: |\omega\rangle
		= :\tilde{\phi}[g] \tilde{\phi}[f_1] \dots \tilde{\phi}[f_n]: |\omega\rangle.
\end{equation}
Then we have $\tilde{\phi}[g] = a[g] + a[g^*]^{\dagger}.$

The notation here may seem a little confusing.
Why did we call the creation operator $a[g^*]^{\dagger}$ instead of $a[g]^{\dagger}$?
The reason is that we want $a[g]^{\dagger}$ to be the adjoint of $a[g]$, and one can easily check that in order to have
\begin{equation} \label{eq:creation-annihilation-adjoint}
	\langle \psi | a[g] \chi \rangle = \langle a[g]^{\dagger} \psi | \chi \rangle
\end{equation}
for $|\psi\rangle$ and $|\chi\rangle$ finite-particle vectors, one must define the creation operator as in equation \eqref{eq:creation-definition}.

There is another subtlety we must address, which is that verifying equation \eqref{eq:creation-annihilation-adjoint} is not actually enough to show that $a[g]^{\dagger}$ is genuinely the adjoint of $a[g].$
The reason for this is that the operators are unbounded.
Equations \eqref{eq:annihilation-definition} and \eqref{eq:creation-definition} define the creation and annihilation operators on $V_{\omega},$ which is a dense subspace of $\H_{\omega}.$
But these operators are not bounded, so they cannot be extended to all of $\H_{\omega}.$
What equation \eqref{eq:creation-annihilation-adjoint} tells us is that the \textit{genuine} adjoint of $a[g]$ is an \textit{extension} of the creation operator ``$a[g]^{\dagger}$'' --- i.e., these operators agree on all vectors where both are defined.
Usefully, this tells us that the adjoint of $a[g]$ is densely defined (since the creation operator is densely defined), and consequently $a[g]$ is a ``closable operator'' --- see the review in \cite[section 2.3]{Sorce:modular} for details.
A similar argument tells us that the creation operator is a closable operator.
By abuse of notation, we can use the symbols $a[g]$ and $a[g]^{\dagger}$ to denote the closures of the operators defined by equations \eqref{eq:annihilation-definition} and \eqref{eq:creation-definition}.
Under this convention, it can be shown that these operators are genuinely adjoints of one another, though the proof of this statement is a little outside our present scope, so we present it in appendix \ref{app:ladder-adjoints}.
An interesting consequence of this statement is that the field operator $\tilde{\phi}[f],$ upon taking closures, has adjoint $\tilde{\phi}[f^*].$
Consequently, the field operator $\tilde{\phi}[f]$ is self-adjoint for real $f,$ which means it is a genuine physical observable.

Note that by equation \eqref{eq:creation-definition}, the vectors in $V_{\omega}^{(n)}$ have the form of a ``vacuum'' $|\omega\rangle$ acted upon by $n$ commuting creation operators.
This points to the existence of a canonical isomorphism between $\H_{\omega}^{(n)}$ and the $n$-th symmetric tensor power $(\H_{\omega}^{(1)})^{\otimes_S n}.$
Indeed, one has
\begin{align}
	\begin{split}
	&
		\langle \tilde{\phi}[f_1] \omega \otimes_S \dots \otimes_S \tilde{\phi}[f_n] \omega |
		\tilde{\phi}[g_1] \omega \otimes_S \dots \otimes_S \tilde{\phi}[g_n] \omega \rangle \\
		& \qquad = \frac{1}{n!} \sum_{\pi \in S_n} \langle \tilde{\phi}[f_1] \omega | \tilde{\phi}[g_{\pi(1)}] \rangle \dots \langle \tilde{\phi}[f_n] \omega | \tilde{\phi}[g_{\pi(n)}] \rangle,
	\end{split} 
\end{align}
and comparing this with equation \eqref{eq:GNS-normal-overlaps} immediately reveals that the map
\begin{equation} \label{eq:symmetric-power-Fock}
	 |	\tilde{\phi}[g_1] \omega \otimes_S \dots \otimes_S \tilde{\phi}[g_n] \omega \rangle
		\mapsto \frac{1}{\sqrt{n!}}:\tilde{\phi}[g_1] \dots \tilde{\phi}[g_n]: |\omega\rangle
\end{equation}
furnishes a unitary equivalence between $(\H_{\omega}^{(1)})^{\otimes_S n}$ and $\H_{\omega}^{(n)}.$
A nice consequence of this is that if one picks an orthonormal basis $\{\tilde{\phi}[f_j] |\omega\rangle\}$ for $\H_{\omega}^{(1)}$ --- which one can always do because vectors of the form $\tilde{\phi}[g] |\omega\rangle$ are dense in $\H_{\omega}^{(1)}$ --- then the vectors
\begin{equation}
	:\tilde{\phi}[f_1]^{p_1} \dots \tilde{\phi}[f_k]^{p_k}: |\omega\rangle,
\end{equation}
with $p_1 + \dots + p_k = n,$ form an orthogonal basis for $\H_{\omega}^{(n)}.$

\subsection{Quantizing phase space}
\label{app:phase-space}

The ``Fock structure'' described above tells us that when studying a Gaussian state of a free field theory, all of the data of the GNS representation is really contained within the one-particle Hilbert space $\H_{\omega}^{(1)}.$\footnote{The intuition behind this result comes from remembering that the GNS inner product is inherited from the Gaussian state $\omega$ and hence is fully characterized by the two-point functions of vacuum-subtracted fields. These, in turn, determine (and modulo null states, are determined by) the inner product on $\H_\omega^{(1)}$.}
There is another way of arriving at this Hilbert space that is very useful for certain proofs.
This is to obtain $\H_{\omega}^{(1)}$ as a subspace of an enlarged Hilbert space that is obtained by ``quantizing the phase space of the theory,'' in other words by promoting the phase space to a Hilbert space.

One begins with the space of complex smearing functions $C^{\infty}_0(\M)$.
This generates the $*$-algebra $\A_0,$ on which one defines a Gaussian state $\omega$.

The GNS space $\H_{\omega}$ is obtained by using $\omega$ to induce an inner product on $\A_0,$ then taking the quotient by null states and performing a Hilbert space completion.

A different Hilbert space is produced by symmetrizing $\omega$ to obtain an inner product directly on $C^{\infty}_0(\M).$
This inner product, written $\mu,$ is defined by
\begin{equation}
	\langle f | g\rangle_{\mu}
		= \frac{\omega(\tilde{\phi}[f^*] \tilde{\phi}[g]) + \omega(\tilde{\phi}[g] \tilde{\phi}[f^*])}{2}.
\end{equation}
This is clearly positive semidefinite because it is related to $\omega$ by using the commutator:
\begin{equation} \label{eq:mu-omega-Omega}
	\langle f | g \rangle_{\mu}
		= \omega(\tilde{\phi}[f^*] \tilde{\phi}[g]) + \frac{i}{2} \Omega[f^*, g].
\end{equation}
In general, the $\mu$ inner product is intimately tied to the symplectic form $\Omega$.
A first hint of this comes from using the Cauchy-Schwarz inequality for $\omega$  to get the inequality
\begin{equation}
	\Big| \langle f | g \rangle_{\mu} - \frac{i}{2} \Omega[f^*, g] \Big|
		\leq \sqrt{\omega(\tilde{\phi}[f^*] \tilde{\phi}[f])} \sqrt{\omega(\tilde{\phi}[g^*] \tilde{\phi}[g])}.
\end{equation}
If $f$ and $g$ are real, then the first term on the left-hand side is real, while the second is imaginary, which lets us evaluate the absolute value of the complex number on the left; on the right-hand side, we use the identity $\omega(\phi[f]^2) = \langle f | f \rangle_{\mu}$ for $f$ real.
Putting these together, we find that for real $f$ and $g,$ one has the inequality
\begin{equation}
	|\langle f | g \rangle_{\mu}|^2 + \frac{1}{4} |\Omega[f, g]|^2
	\leq \langle f | f \rangle_{\mu} \langle g | g \rangle_{\mu}.
\end{equation}
In particular, this gives
\begin{equation} \label{eq:mu-Omega-inequality}
	\langle f | f \rangle_{\mu} \langle g | g \rangle_{\mu}
		\geq \frac{1}{4} |\Omega[f, g]|^2, \qquad f, g \text{ real},
\end{equation}
which is a very useful inequality relating $\Omega$ and $\mu.$
In appendix \ref{app:pure}, we will return to this inequality and show that it is saturated in the case that $\omega$ is pure, meaning that for any $f$, there is a choice of $g$ for which equality holds.

By the ``phase space'' of the theory, we will mean the vector space $\mathscr{S}_{\mu}$ obtained by endowing $C^{\infty}_0(\M)$ with the $\mu$ inner product, then taking the quotient by null states:
\begin{equation}
    \mathscr{S}_{\mu}=C^{\infty}_0(\M)/\ker\mu\,.
\end{equation}
The reason we call this a ``phase space'' is that in the case of Klein-Gordon theory, one has $\ker \mu = \ker \Omega$ --- as we will now demonstrate --- which gives $\mathscr{S}_{\mu} = \mathscr{S},$ the usual Klein-Gordon phase space introduced in equation \eqref{eq:phase-space-of-kg}.
Note that in any free theory, for $u, v$ real, one can easily show
\begin{equation} \label{eq:mu-complex-splitting}
	\langle u + i v | u + i v\rangle_{\mu}
		= \langle u | u \rangle_{\mu} + \langle v | v \rangle_{\mu},
\end{equation}
which implies that if $u + i v$ is null for $\mu,$ then its real and imaginary parts are null as well.
Inequality \eqref{eq:mu-Omega-inequality} then implies that $u$ and $v$, hence $u + i v$, are degeneracies of the symplectic form.
This gives the general inclusion $\ker \mu \subseteq \ker \Omega,$ and our task is to show that in Klein-Gordon theory, the converse inclusion also holds.

In the case of Klein-Gordon theory, the symplectic form is $\Omega = E,$ with $E$ the advanced-minus-retarded propagator.
As discussed above, the only degeneracies of $\Omega$ consist of test functions that are in the image of the equations of motion, i.e., we have $\ker \Omega=\im(\Box - m^2)$.
But in Klein-Gordon theory, the $*$-algebra $\A_0$ involves a quotient by the images of the equations of motion, so everything in the image of the equations of motion is a degeneracy not just of $\Omega,$ but also of any algebraic state $\omega$ (see footnote \ref{foot:EoM-state}).
By equation \eqref{eq:mu-omega-Omega}, this implies that every element of $\ker \Omega$ is itself a degeneracy of the inner product $\mu.$

Note that in generalized free field theories, the quotient by $\text{ker}\mu$ might not kill all of the degeneracies of the symplectic form.
In this case, the space $\mathscr{S}_{\mu}$ is $\mu$-dependent.

Now that we know $\mu$ is a positive definite inner product on the phase space $\mathscr{S}_{\mu},$ we can define its Hilbert space completion $\K_{\mu}.$
The utility of having constructed $\K_{\mu}$ is that it retains some information pertaining to the state $\omega$ that is discarded in the construction of $\H_{\omega}.$
In particular, if there exist test functions $f$ with\footnote{Note this cannot happen if the state $\omega$ is faithful, but we will see later that many states of interest in free field theory are not faithful; in particular, no pure state is faithful.}
\begin{equation}\label{eq:null-state}
	\omega(\tilde{\phi}[f^*] \tilde{\phi}[f]) = 0
\end{equation}
but for which $f$ is not a degeneracy of $\mu$, then there will be a state in $\K_{\mu}$ representing $f$, while no such state is present in $\H_{\omega}$.
In this sense, null states for $\omega$ can be kept alive within the space $\K_{\mu}.$

Of course, some information about $\omega$ is discarded in the construction of $\K_{\mu}$, but the point is that this information will also be discarded in the construction of $\H_\omega$, because every null state for $\mu$ is also a null state for $\omega$, by equation \eqref{eq:mu-omega-Omega} together with the inclusion $\ker\mu \subseteq \ker \Omega$.
Hence, $\K_{\mu}$ knows ``strictly more'' about $\omega$ than $\H_\omega^{(1)}$ does.
We call the null states of $\omega$ that are not null states for $\mu$ ``one-particle null states.''

This last observation makes it sound like we should be able to identify the one-particle Hilbert space, $\H_{\omega}^{(1)},$ as a subspace of $\K_{\mu}.$
Indeed, this is the case.
To see this, we first define an operator $R$ on $\K_{\mu}$ that represents the symplectic form:\footnote{This operator is called $J$ in \cite{Wald:QFTCS}, but we use $R$ as we wish to reserve $J$ for the modular conjugation operator.}
\begin{equation}
	\langle f | R | g\rangle_{\mu}
		= \frac{1}{2} \Omega[f^*, g].
        \label{eq:def-of-R}
\end{equation}
The fact that such an operator exists is guaranteed by inequality \eqref{eq:mu-Omega-inequality}.
That equation tells us that $\Omega$ is bounded-in-$\mu$ when viewed as a quadratic form, and the Riesz lemma therefore tells us that it can be represented as a bounded operator.
The operator $R$ has several nice properties, all of which are inherited from $\Omega.$
The first is that it takes real states to real states.
A useful way to write this is to define the complex conjugation operator $\Gamma : \K_{\mu} \to \K_{\mu}$ by
\begin{equation}
	\Gamma |f\rangle_{\mu} = |f^* \rangle_{\mu}.
\end{equation}
This is an antilinear operator, and from the definition of $\mu$ one can easily see that it is actually antiunitary.
By a ``real state'' in $\K_{\mu}$ we mean one that is fixed by $\Gamma,$ and one easily computes
\begin{align}
	\begin{split}
	\langle f | \Gamma R \Gamma | g \rangle_{\mu}
		& = \langle R g^* | f^* \rangle_{\mu} \\
		& = \bar{\langle f^* | R g^* \rangle_{\mu}} \\
		& = \frac{1}{2} \bar{\Omega[f, g^*]} \\
		& = \frac{1}{2} \Omega[f^*, g] \\
		& = \langle f | R | g \rangle_{\mu}.
	\end{split}
\end{align}
This gives us the identity
\begin{equation}
	\Gamma R \Gamma = R,
\end{equation}
which confirms that $R$ takes real states to real states.

Another useful property of $R$ is that it is antihermitian, $R^{\dagger} = - R,$ which follows from the calculation
\begin{align}
	\begin{split}
		\langle f | R | g \rangle_{\mu}
			& = \frac{1}{2} \Omega[f^*, g] \\
			& = - \frac{1}{2} \Omega[g, f^*] \\
			& = - \frac{1}{2} \bar{\Omega[g^*, f]} \\
			& = - \bar{\langle g | R f\rangle}_{\mu} \\
			& = - \langle R f | g \rangle_{\mu}.
	\end{split}
\end{align}
This tells us that the spectrum of $R$ is purely imaginary.
By inequality \eqref{eq:mu-Omega-inequality}, we also see that the norm of $R$ is upper bounded by one, so in fact the spectrum of $R$ must lie between $-i$ and $i.$

Now that we know some details about $R$, we are ready to figure out where the one-particle Hilbert space lives within the completed phase space $\K_{\mu}.$
The key point is that the ``one-particle null states'' of $\omega$ are exactly the $(-i)$ eigenvalues of $R$.
To see this, one rewrites equation \eqref{eq:mu-omega-Omega} using the operator $R$:
\begin{equation}
	\omega(\tilde{\phi}[f^*] \tilde{\phi}[g])
		= \langle f | (1 - i R) | g \rangle_{\mu}.
\end{equation}
The state $\tilde{\phi}[g] |\omega\rangle$ is null in $\H_{\omega}^{(1)}$ if and only if both sides of this equation vanish independently of $f$; equivalently, one must have $R |g\rangle_{\mu} = - i |g \rangle_{\mu}.$

This makes it sound like the one-particle Hilbert space $\H_{\omega}^{(1)}$ should live in $\K_{\mu}$ as the orthocomplement of the $-i$ eigenspace of $R$.
To make this concrete, one defines a map from $\H_{\omega}^{(1)}$ to $\K_{\mu}$ via the formula
\begin{equation} \label{eq:app-U-def}
	U \left(\tilde{\phi}[f] |\omega\rangle\right) = \sqrt{1 - i R} |f\rangle_{\mu}.
\end{equation}
By the above discussion, this is an isometry.
Its image is the image of $\sqrt{1 - i R},$ which is the orthocomplement of the $-i$ eigenspace of $R.$
We therefore learn that $\H_{\omega}^{(1)}$ can be mapped unitarily to this orthocomplement within $\K_{\mu}.$

\subsection{von Neumann algebras generated by the fields}
\label{app:vN-generating}

Given a $*$-algebra $\A_0$ and a GNS representation $\H_{\omega},$ there are a few ways to ``complete'' $\A_0$ into a von Neumann algebra $\A_{\omega}$ acting on $\H_{\omega}.$
One natural way is to declare that $\A_{\omega}$ is the smallest von Neumann algebra such that all of the operators coming from $\A_0$ are ``affiliated'' to $\A_0.$
For more details on affiliation and this construction, see \cite[appendix A]{Sorce:paper1} and the discussion below. 
For other choices, see \cite{Driessler:unbounded-to-vN, Buchholz:unbounded-to-vN}.

Usefully, when $\A_0$ is a free field algebra and $\omega$ is a Gaussian state, there is really only one canonical choice of von Neumann algebra in the GNS representation $\H_{\omega},$ and it corresponds to the choice $\A_{\omega}$ described in the preceding paragraph.
The main claim of this subsection is in the free field setting, $\A_{\omega}$ can be described in any of three equivalent ways:
\begin{itemize}
    \item the smallest von Neumann algebra such that each $\phi[f]$ is affiliated, called ${\cal A}_{\omega}$,
    \item the smallest von Neumann algebra containing all bounded functions of $\phi[f]$ with $f$ real, called $\A_{\text{self-adjoint}}$,
    \item the smallest von Neumann algebra containing all Weyl operators $e^{i \phi[f]}$ with $f$ real.
\end{itemize}
Note that for the second and third bullet points to make sense, $\phi[f]$ must have self-adjoint closure when $f$ is real, but this is exactly what we showed in appendix \ref{app:ladder-adjoints}.
We also comment that for the purposes of this discussion, the distinction between $\phi[f]$ and $\tilde{\phi}[f]$ is immaterial --- they generate the same von Neumann algebras, since they only differ by a multiple of the identity.

different choices one might make for $\A_{\omega}$ all give rise to the same algebra.
The main reason for this is that, as discussed in appendix \ref{app:ladder-adjoints}, the free field operator $\phi[f]$ has self-adjoint closure whenever $f$ is real.
For any self-adjoint operator $T$, one can act on $T$ with bounded functions to produce bounded operators.
A natural candidate for $\A_{\omega}$ is the smallest von Neumann algebra containing all bounded functions of $\phi[f]$ for $f$ real.
One might worry that this is smaller than the algebra discussed above, since it seems to only know about the self-adjoint elements of $\A_0,$ but we will see in this section that the two definitions coincide.
We will also see that it is equivalent to the smallest von Neumann algebra containing all of the ``Weyl operators'' $e^{i \phi[f]},$ where $f$ is real.
For the purposes of this discussion, the distinction between $\phi[f]$ and $\tilde{\phi}[f]$ is immaterial --- they generate the same von Neumann algebras, since they only differ by a multiple of the identity.

We begin by relating the algebras in the first two bullet points.

For a general $*$-algebra $\A_0$ and state $\omega,$ one can consider as in \cite[appendix A.4]{Sorce:paper1} the polar decomposition of any $a \in \A_0$ as it acts on $\H_{\omega}$:
\begin{equation}
	a = V_a |a|.
\end{equation}
The minimal affiliated algebra $\A_{\omega}$ is the smallest von Neumann algebra containing all partial isometries $V_a$ and all bounded functions of the positive operators $|a|.$
For the self-adjoint operator $\phi[f],$ the polar decomposition is of the form
\begin{equation}
	\phi[f] = V_f |\phi[f]|,
\end{equation}
where $V_f$ is simply equal to $+1$ on the portion of Hilbert space where $\phi[f]$ has positive spectrum, equal to $-1$ on the portion of Hilbert space where $\phi[f]$ has negative spectrum, and equal to $0$ on the portion of Hilbert space where $\phi[f]$ is zero.
The operator $V_f$ is self-adjoint, so the von Neumann algebra it generates includes the projectors onto the $+1$, $0,$ and $-1$ eigenspaces of $V_f$,\footnote{For example, the $+1$ eigenspace projector $\Pi_+$ is obtained by acting on $V_f$ with the indicator function that is equal to one at $+1$ and equal to zero everywhere else on the real line.} which are the projections onto the positive and negative parts of the spectrum of $\phi[f].$
By combining these projectors with bounded functions of $|\phi[f]|$, one can obtain any bounded function of $\phi[f].$\footnote{If $\psi$ is a bounded function on the spectrum of $\phi[f]$, with $\psi_{\pm}$ its restrictions to the positive and negative portions of the spectrum, then one has
\begin{equation}
	\psi(\phi[f]) = \Pi_+ \psi_+(|\phi[f]|) + \Pi_- \psi_-(-|\phi[f]|) + \psi(0) \Pi_0,
\end{equation}
with $\Pi_{\pm}$ the positive/negative spectral projections of $\phi[f]$ and $\Pi_0$ the projector onto the kernel of $\phi[f].$}
This guarantees that the von Neumann algebra generated by the self-adjoint elements $\phi[f]$ is contained in the minimal affiliated algebra $\A_{\omega}$:
\begin{equation}
	\A_{\text{self-adjoint}} \subseteq \A_{\omega}.
\end{equation}

To see the converse, it is easiest to show the equivalent statement about commutants,
\begin{equation}
	\A_{\text{self-adjoint}}' \subseteq \A_{\omega}'.
\end{equation}
In \cite[appendix A]{Sorce:paper1}, it was shown that the commutant of $\A_{\omega}$ is the set of all bounded operators that commute with all operators in $\A_0$ and all adjoints of operators in $\A_0.$\footnote{While the abstract $\ast$-algebra is closed under the abstract  ``adjoint'' operation, this operation need not coincide with the Hilbert space adjoint in the GNS representation --- due to domain issues --- so $\A_0$ need not be closed under adjoints in this sense. In the free field case, this distinction is not important, since we already showed that $\phi[f]{\dagger}$ is the closure of $\phi[f^*].$}
When one talks about a bounded operator $T$ commuting with an unbounded operator $a,$ one means that $T$ preserves the domain of $a,$ and that one has $T a = a T$ on that domain.
So we are now aiming to show that if a bounded operator commutes with all of the self-adjoint operators $\phi[f]$ with $f$ real, then it commutes with all operators in $\A_0$ and their adjoints.

To show that a bounded operator $T$ commutes with an unbounded operator $a,$ it suffices to show commutativity on a ``core'' for $a.$
A core for an unbounded operator is a subspace of the domain such that the full action of $a$ can be determined by taking limits from this subspace.
By construction, all of the operators coming from $\A_0$ are initially defined on the GNS domain $V_{\omega},$ then ``closed'' by taking limits.
So all of these operators have $V_{\omega}$ as a core.

Now, suppose that $T$ formally commutes with all $\phi[f]$ for $f$ real.
Consider a product $\phi[f_1] \dots \phi[f_n],$ with each $f_j$ real.
Suppose that $|\psi\rangle$ is in the GNS domain $V_{\omega}$.
Then because $T$ formally commutes with $\phi[f_n],$ we know that $T|\psi\rangle$ is in the domain of $\phi[f_n],$ and satisfies
\begin{equation}
	\phi[f_n] T |\psi\rangle
		= T \phi[f_n] |\psi\rangle.
\end{equation}
But $\phi[f_n] |\psi\rangle$ is in the domain of $\phi[f_{n-1}];$ since $T$ formally commutes with $\phi[f_{n-1}],$ the vector $T \phi[f_n] |\psi\rangle$ is in the domain of $\phi[f_{n-1}],$ and we have
\begin{equation}
	\phi[f_{n-1}] \phi[f_n] T |\psi\rangle
	= T \phi[f_{n-1}] \phi[f_n] |\psi\rangle.
\end{equation}
Clearly we can proceed by induction to learn that $T|\psi\rangle$ is in the domain of $\phi[f_1] \dots \phi[f_n]$ and that $T$ commutes with this field monomial on the GNS domain.
The same obviously holds for any linear combination of field monomials, and since the field is linear ($\phi[f + i g] = \phi[f] + i \phi[g]$) an arbitrary field polynomial can be written as a linear combination of field monomials with purely real smearing functions.
It follows that in the case of the free field algebra, we have $\A_{\text{self-adjoint}}' \subseteq \A_{\omega}',$ and the von Neumann algebra generated by the self-adjoint operators $\phi[f]$ with $f$ real is exactly the same as the minimal von Neumann algebra for which the full $*$-algebra $\A_0$ is affiliated.

Now we turn our attention to the Weyl operators.
For $f$ real, the operator $\phi[f]$ is self-adjoint, so one can consider its exponential $e^{i \phi[f]}.$
This is called a \textit{Weyl operator} for the free field.
They can be useful to consider for certain calculations as the Weyl operators are bounded operators, while the field operators $\phi[f]$ are unbounded.
In fact, some approaches to free field theory take the Weyl operators as the core objects of the theory, and start with an abstract C$^*$ algebra of free fields.

The set of Weyl operators $e^{i \phi[f]}$ generates the exact same von Neumann algebra as the set of self-adjoint field operators $\phi[f].$
Since $e^{i \phi[f]}$ is a bounded function of $\phi[f],$ it is obviously contained in the von Neumann algebra generated by $\phi[f].$
Conversely, by the spectral theorem, one has for any vectors $|\psi\rangle$ and $|\chi\rangle$ a measure $\nu_{\psi, \eta}$ on the real line such that the formula
\begin{equation}
	\langle \psi | \chi(\phi[f]) | \eta\rangle
		= \int d\nu_{\psi,\eta}(t) \chi(t)
\end{equation}
holds for any bounded function $\chi$ of the real line.
If we take $\chi$ to be smooth and compactly supported, then we can write it in terms of its Fourier transform to obtain
\begin{equation}
	\langle \psi | \chi(\phi[f]) | \eta\rangle
	= \int d\nu_{\psi,\eta}(t) \int \frac{dp}{2\pi} \hat{\chi}(p) e^{i p t}.
\end{equation}
When the integrand is replaced by its absolute value it remains integrable, so one can switch the order of integration to obtain
\begin{equation}
	\langle \psi | \chi(\phi[f]) | \eta\rangle
	= \int \frac{dp}{2\pi} \hat{\chi}(p) \int d\nu_{\psi,\eta}(t) e^{i p t}
	= \int \frac{dp}{2\pi} \hat{\chi}(p) \langle \psi | e^{i p \phi[f]} |\eta\rangle.
\end{equation}
This tells us that the equation
\begin{equation}
	\chi(\phi[f])
		= \int \frac{dp}{2 \pi} \hat{\chi}(p) e^{i p \phi[f]}
\end{equation}
holds when the integral on the right-hand side is interpreted in the sense of weak convergence.
Consequently, each operator $\chi(\phi[f])$ is contained in the von Neumann algebra generated by the operators $e^{i p \phi[f]}$ with varying $p$.

While we only showed this for the smooth, compactly supported functions $\chi,$ it is generally true that these operators generate all bounded functions of $\phi[f].$
For if $\theta$ is a bounded function on the real line, one has
\begin{equation}
	|\langle \psi | \theta(\phi[f]) | \eta \rangle
		- \langle \psi | \chi(\phi[f]) |\eta\rangle |
		\leq \int d\nu_{\psi, \eta}(t) |\theta(t) - \chi(t)|,
\end{equation}
and on any compact subset of the real line one can approximate $\theta$ arbitrarily well by smooth compactly supported functions $\chi$ in $L^1(\nu_{\psi,\eta})$, since such functions are dense in any $L^1$ space.
The measure $\nu_{\psi,\eta}$ is finite on the full real line, so the right-hand side of the above inequality can be made arbitrarily small by approximating $\theta$ using $\chi$ within a large compact set, and choosing the compact set large enough that the measure is small outside of that set.

Finally, let us comment that we can enlarge the generating set for ${\cal A}_{\omega}$ from Weyl operators of real test functions, to Weyl operators associated with any real element $|\psi\rangle_{\mu}$ of the completed phase space $\K_{\mu}.$ 
These operators, denoted (formally, for now) by $e^{i \tilde{\phi}[\psi]}$, make sense because the map from real test functions $f$ to Weyl operators $e^{i \tilde{\phi}[f]}$ is continuous with respect to the $\mu$-topology in the domain and the strong operator topology in the range, which we will show momentarily. 
Moreover, since von Neumann algebras are closed under strong limits, the von Neumann algebra generated by $e^{i\tilde \phi[\psi]},$ for real $\ket \psi _\mu\in \K_{\mu}$, clearly coincides with that generated using only the (non-completed) space of test functions.

Let us show the continuity of the map $f\mapsto e^{i \tilde{\phi}[f]}$ claimed above. 
For any real sequence (or net) $f_n \to 0$ in the $\mu$ norm, one has
\begin{equation}
	e^{i \tilde{\phi}[f_n]} e^{i \tilde{\phi}[g]} |\omega\rangle
	= e^{i \tilde{\phi}[f_n + g]} e^{\frac{i}{2} \Omega[f_n, g]} |\omega\rangle,
\end{equation}
hence
\begin{equation}
\begin{aligned}
    \big\lVert e^{i \tilde{\phi}[f_n]} e^{i \tilde{\phi}[g]} |\omega\rangle - e^{i \tilde{\phi}[g]} |\omega\rangle\big\rVert^2
	&= 2 - e^{i \Omega[f_n, g]} \langle \omega | e^{i \tilde{\phi}[f_n]} | \omega \rangle - e^{- i \Omega[f_n, g]} \langle \omega | e^{- i \tilde{\phi}[f_n]} | \omega \rangle\\
    &=2 - \big( e^{2 i \langle f_n | R | g \rangle_{\mu}} + e^{- 2 i \langle f_n | R | g \rangle_{\mu}} \big) e^{- \langle f_n | f_n\rangle_{\mu}/2} \to 0,
\end{aligned}
\end{equation}
where in the second step we used a power series for the Weyl operator together with the Gaussianity of $\omega$ (see near equation \eqref{eq:gaussian-expectation} below for details).
So $e^{i \tilde{\phi}[f_n]}$ converges to the identity when acting on a dense set of states;\footnote{Note that the states $e^{i \tilde{\phi}[g]} |\omega\rangle$ have dense span in $\H_{\omega},$ since the Weyl operators generate the same von Neumann algebra as the field operators, and the field operators acting on $|\omega\rangle$ have dense span.
\label{foot:field-op-for-K-mu}} since every operator in this sequence has norm one, one can easily show using the triangle inequality that $e^{i \tilde{\phi}[f_n]}$ converges to the identity when acting on any state.
The map from real test functions to Weyl operators is therefore strongly continuous; one can use this map to define a Weyl operator $e^{i \tilde{\phi}[\psi]}$ for any real $|\psi\rangle_{\mu}$.

Now that we have constructed  these generalized Weyl operators, we can use them, as suggested by our notation so far, to define operators $\tilde \phi[\psi]$ acting on $\H_\omega$ for any $\ket{\psi}_\mu$ in the completed phase space $\K_\mu$, not just for test functions.
This is because the arguments above guarantee that the map $t\to e^{i\tilde{\phi}[t\psi]}$ is a strongly continuous one-parameter unitary group for any real $\ket{\psi}_\mu$ in $\K_\mu$. 
By Stone's theorem, such a group admits a unique self-adjoint generator, which we define to be $\tilde \phi[\psi]$.
Field operators $\phi[\psi]$ for non-real $\psi$ can be obtained by using this procedure to obtain the real and imaginary parts independently.

\subsection{Further comments on Weyl operators}
\label{app:Weyl}

We dedicate this subsection to a discussion of the following identity:
\begin{equation}
\label{eq:gaussian-expectation}
	\langle \omega | e^{i \tilde{\phi}[f]} |\omega\rangle
		= e^{- \langle f | f \rangle_{\mu} / 2},
\end{equation}
which holds for any Gaussian state $\omega$ and Weyl operator $e^{i \tilde{\phi}[f]}$ with $f$ real, and which moreover can be used to answer the question of which vectors $|\psi\rangle$ in $\H_{\omega}$ induce Gaussian states on $\A_0.$

As indicated at the end of the preceding subsection, equation \eqref{eq:gaussian-expectation} can be verified with a power-series representation for the Weyl operator.
In particular, even though the operator $\tilde{\phi}[f]$ is unbounded, the expression
\begin{equation}
	\langle \omega | e^{i \tilde{\phi}[f]} |\omega\rangle
		= \sum_n \frac{i^n}{n!} \langle \omega | \tilde{\phi}[f]^n |\omega\rangle
\end{equation}
is valid so long as the power series converges.
Using the Gaussian form of the correlators, one computes this as
\begin{align}
	\begin{split}
	\langle \omega | e^{i \tilde{\phi}[f]} |\omega\rangle
		& = \sum_k \frac{i^{2k}}{(2k)!} \langle \omega | \tilde{\phi}[f]^{2k} |\omega\rangle \\
		& = \sum_k \frac{(-1)^k}{(2k)!} (\langle f | f \rangle_{\mu})^k \times \text{\# ways of pairing up $2k$ elements}\\
		& = \sum_k \frac{(-1)^k}{(2k)!} (\langle f | f \rangle_{\mu})^k \frac{(2k)!}{k! 2^k} \\
		& = \sum_k \frac{1}{k!}  \left(- \frac{1}{2} \langle f | f \rangle_{\mu} \right)^k \\
		& = e^{- \langle f | f \rangle_{\mu}/2}.
	\end{split}
\end{align}
If we had not shifted the Weyl operators by the one-point function, then we would instead have
\begin{align}
	\begin{split}
		\langle \omega | e^{i \phi[f]} |\omega\rangle
		& = \langle \omega | e^{i \tilde{\phi}[f]} |\omega\rangle e^{i \omega(\phi[f])} \\
		& = e^{- \langle f | f \rangle_{\mu}/2} e^{i \omega(\phi[f])}.
	\end{split}
\end{align}

In fact, \textit{any} state that has Weyl expectation values of this form will be Gaussian!
That is, if a state $|\psi\rangle \in \H_{\omega}$ satisfies
\begin{align} \label{eq:app-psi-gaussian}
	\begin{split}
		\langle \psi | e^{i \phi[f]} |\psi\rangle
			& = e^{- \langle f | f \rangle_{\nu}/2} e^{i \psi(\phi[f])},
	\end{split}
\end{align}
where $\nu$ is an inner product on real test functions and $\psi(\phi[f])$ satisfies $\psi(\phi[f])^* = \psi(\phi[f^*])$, then $|\psi\rangle$ automatically induces a Gaussian state on the $*$-algebra $\A_0.$
The functional on the $*$-algebra is induced by
\begin{align}\label{eq:n-pt-from-Weyl}
	\begin{split}
		\psi(\phi[f_1] \dots \phi[f_n])
		& \equiv (-i)^n \left. \frac{\del^n}{\del t_1 \dots \del t_n} \langle \psi | e^{i t_1\phi[ f_1]} \dots e^{i t_n\phi[ f_n]}|\psi\rangle \right|_{t_1 = \dots = t_n = 0}.
	\end{split}
\end{align}
As in appendix \ref{app:normal-ordering}, the connected $n$-point correlators are determined by
\begin{equation}
	\psi_n^c(\phi[f_1] \dots \phi[f_n])
		= (-i)^n \left.\frac{\del^n}{\del t_1 \dots \del t_n} \log\langle \psi | e^{i t_1 \phi[f_1]} \dots e^{i t_n \phi[f_n]}) | \psi \rangle \right
	|_{t_1=\dots=t_n=0},
\end{equation}
and using the Baker-Campbell-Hausdorff formula this can be computed as
\begin{align}
	\begin{split}
		& \psi_n^c(\phi[f_1] \dots \phi[f_n]) \\
		& \quad = (-i)^n \left.\frac{\del^n}{\del t_1 \dots \del t_n} \log\left( \langle \psi | e^{i \phi[\sum_j t_j f_j]} | \psi \rangle
		e^{\frac{i}{2} \sum_{j < k} t_j t_k \Omega[f_j, f_k]}\right) \right|_{t_1=\dots=t_n=0} \\
		& \quad = (-i)^n \frac{\del^n}{\del t_1 \dots \del t_n} \bigg(- \frac{1}{2} \langle \sum_j t_j f_j| \sum_k t_k f_k\rangle_{\nu} + i \psi(\phi[\sum_j t_j f_j]) + \frac{i}{2} \sum_{j < k} t_j t_k \Omega[f_j, f_k]\bigg) \bigg|_{t_1=\dots=t_n=0},
	\end{split}
\end{align}
and this vanishes for $n > 2,$ so $\psi$ is Gaussian.
Moreover, since we assumed that there is a vector $|\psi\rangle \in \H_{\omega}$ satisfying equation \eqref{eq:app-psi-gaussian}, positivity is automatic, and hence we may conclude that $\psi$ is a Gaussian \textit{state}.

What if we did not assume the existence of a vector representative for $\psi$?
In particular, we could have only assumed the existence of an abstract functional $\psi$ satisfying
\begin{align}\label{eq:app-psi-gaussian-abstract}
	\begin{split}
		\psi(e^{i \phi[f]})
		& = e^{- \langle f | f \rangle_{\nu}/2} e^{i \psi(\phi[f])}.
	\end{split}
\end{align}
In this case, the above calculation shows that $\psi$ defines a Gaussian functional on $\A_0,$ but it does not explicitly show that it is a state, i.e., that $\psi$ is positive.
As in appendix \ref{app:normal-ordering}, to check positivity of $\psi,$ it suffices to check positivity of the vacuum-subtracted two point function, i.e., to check
\begin{equation} \label{eq:app-psi-fs}
	\psi\Big(\sum_{j,k} c_j^* c_k (\phi[f_j^*] - \psi(\phi[f_j^*])) (\phi[f_k] - \psi(\phi[f_k])) \Big) \geq 0.
\end{equation}
In fact, it suffices to check this in the case that the test functions $f_j$ are all real (but with coefficients $c_j$ still allowed to be complex).
For if we write $f_j$ in terms of real and imaginary parts as $f_j = u_j + i v_j$, then the condition we are trying to check becomes
\begin{equation}
		 \psi\Big(\sum_{j,k} c_j^* c_k \tilde{\phi}[u_j - i v_j] \tilde{\phi}[u_k + i v_k] \Big)  \geq 0,
\end{equation}
where now the tildes denote the subtraction of the $\psi$ one-point function.
Writing $u = \sum_j c_j u_j$ and $v = \sum_j c_j v_j,$ this is
\begin{equation}\label{eq:pos-psi-real-smear}
    \psi(\tilde{\phi}[u^*] \tilde{\phi}[u]) + \psi(\tilde{\phi}[v^*] \tilde{\phi}[v]) - 2\Im (\psi(\tilde\phi[u^*]\tilde \phi[v])) \geq 0.
\end{equation}
Assuming positivity of $\psi$ on complex linear combinations of real test functions, the first two terms are non-negative. The same assumption also ensures that $\psi$ satisfies the Cauchy--Schwarz inequality on this domain, allowing us to bound the last term in such a way that \eqref{eq:pos-psi-real-smear} follows immediately.

Now, returning to equation \eqref{eq:app-psi-fs} in the case where each function $f_j$ is real, we expand this expression as
\begin{equation}
	\psi\Big(\sum_{j,k} c_j^* c_k \phi[f_j] \phi[f_k] \Big) \geq \Big| \sum_j c_j \psi(\phi[f_j]) \Big|^2.
\end{equation}
Applying the abstract-functional version of \eqref{eq:n-pt-from-Weyl} together with the usual identity \eqref{eq:app-psi-gaussian-abstract} to the left-hand side of the above equation, this can be rewritten as 
\begin{align}
	\begin{split}
		\Big| \sum_j c_j \psi(\phi[f_j]) \Big|^2
			& \leq \sum_{j,k} c_j^* c_k (-1)  \frac{\del^2}{\del r \del s} e^{- \langle r f_j + s f_k | r f_j + s f_k\rangle_{\nu}/2} e^{i \psi(\phi[r f_j + s f_k])} e^{\frac{i}{2} r s \Omega[f_j, f_k]} \Big|_{r=s=0} \\
			& = \frac{1}{2}\sum_{j,k} c_j^* c_k \left( \langle f_j | f_k \rangle_{\nu} + \langle f_k | f_j \rangle_{\nu} + 2 \psi(\phi[f_j]) \psi(\phi[f_k]) - i \Omega[f_j, f_k] \right).
	\end{split}
\end{align}
Canceling the left-hand side with the middle term on the right, the desired inequality becomes
\begin{align}
	\begin{split}
		0
			& \leq \sum_{j,k} c_j^* c_k \left( \langle f_j | f_k \rangle_{\nu} + \langle f_k | f_j \rangle_{\nu} - i \Omega[f_j, f_k] \right).
	\end{split}
\end{align}
If we assume that $\nu$ is an inner product on the space of real test functions, i.e., that it is symmetric on the space of real test functions, then this becomes
\begin{align}
	\begin{split}
		0
		& \leq \sum_{j,k} c_j^* c_k \left(2 \langle f_j | f_k \rangle_{\nu} - i \Omega[f_j, f_k] \right).
	\end{split}
\end{align}
Finally, putting $f = \sum_j c_j f_j,$ this expression becomes
\begin{align}
	\begin{split}
		\langle f | f \rangle_{\nu}
			\geq \frac{i}{2} \Omega[f^*, f].
	\end{split}
\end{align}
So we see that the positivity requirement reduces to a simple boundedness condition for $\Omega$ in terms of $\nu.$
This boundedness condition is actually the same as inequality \eqref{eq:mu-Omega-inequality}; for putting $f = u + i v,$ the above inequality is
\begin{equation} 
	\langle u | u \rangle_{\nu} + \langle v | v \rangle_{\nu}
		\geq - \Omega[u, v],
\end{equation}
and since $\Omega$ is antisymmetric and this must hold for all $u$ and $v,$ we may simplify the inequality as
\begin{equation} \label{eq:app-Omega-sum-inequality}
	\langle u | u \rangle_{\nu} + \langle v | v \rangle_{\nu}
	\geq |\Omega[u,v]|.
\end{equation}
Inequality \eqref{eq:mu-Omega-inequality} can be rewritten as
\begin{equation}
	2 \sqrt{\langle u | u \rangle_{\nu}} \sqrt{\langle v | v \rangle_{\nu}} \geq |\Omega[u, v]|.
\end{equation}
We can obtain this from inequality \eqref{eq:app-Omega-sum-inequality} by substituting
\begin{equation}
    u \mapsto u \sqrt{\frac{\lVert v \rVert_{\nu}}{\lVert u \rVert_{\nu}}}, \quad \quad 
    v
	\mapsto v \sqrt{\frac{\lVert u \rVert_{\nu}}{\lVert v \rVert_{\nu}}}.
\end{equation}
The other direction is obvious, since one always has $(\lVert u \rVert_{\nu})^2 + (\lVert v \rVert_{\nu})^2 \geq 2 \lVert u \rVert_{\nu} \lVert v \rVert_{\nu}.$
So we conclude that positivity of $\psi$ is exactly the requirement that inequality \eqref{eq:mu-Omega-inequality} hold as a relation between $\nu$ and $\Omega.$

\subsection{More on pure states}
\label{app:pure}

A state $\omega$ on a $\ast$-algebra  $\A_0$ is said to be pure if it cannot be decomposed as a nontrivial convex combination of other states,
\begin{equation}
	\omega = p \omega_1 + (1-p) \omega_2.
\end{equation}
In the case of a free field algebra, an equivalent condition --- see for example\footnote{Note that in \cite{Sorce:paper1} one showed that purity was equivalent to vanishing of the \textit{weak} commutant, while the claim here is about the \textit{strong} commutant. But by examining the definitions of that paper it is very easy to see that in the present case, because all field operators have normal closures, the weak and strong commutants coincide.} \cite[appendix A.7]{Sorce:paper1} --- is that the von Neumann algebra generated by $\A_0$ on $\H_{\omega}$ has trivial commutant.
We claimed in appendix \ref{app:phase-space} that when $\omega$ is pure, the inequality
\begin{equation}
	\langle f | f \rangle_{\mu} \langle g | g \rangle_{\mu} \geq \frac{1}{4} |\Omega(f, g)|^2, \qquad f, g \text{ real}
    \label{eq:omega-mu-inequality}
\end{equation}
is saturated in the sense that for any real $f,$ there exists a real $g$ such that the above is an equality.
First, we show the slightly weaker statement
\begin{equation} \label{eq:supremum-saturation}
	\langle f | f \rangle_{\mu} = \sup_{g \text{ real, nonzero}} \frac{|\Omega(f, g)|^2}{4 \langle g | g \rangle_{\mu}},
\end{equation}
then we will show that the supremum is saturated; this automatically implies inequality \eqref{eq:omega-mu-inequality} is also saturated.
The main utility of exploring these inequalities is to reveal that the operator $R$ defined in appendix \ref{app:phase-space} satisfies $R^2 = -1,$ along with a closely related condition for ``centrally pure'' states.

Suppose, toward contradiction, that equation \eqref{eq:supremum-saturation} does \textit{not} hold.
Then there exists some real $f$ such that
\begin{equation}
	\langle f | f \rangle_{\mu} \gneq \sup_{g \text{ real, nonzero}} \frac{|\Omega(f, g)|^2}{4 \langle g | g \rangle_{\mu}}.
\end{equation}
One has
\begin{align}
	\begin{split}
	\omega(\tilde{\phi}[- i f - g] \tilde{\phi}[i f - g])
		& = \langle i f - g | i f - g \rangle_{\mu} - \frac{i}{2} \Omega(- i f - g, i f - g) \\
		& = \langle f | f \rangle_{\mu} + \langle g | g \rangle_{\mu} + \Omega(f, g),
	\end{split}
\end{align}
hence
\begin{align}
	\begin{split}
	\inf_{g \text{ s.t. $\langle g | g \rangle_{\mu} = \langle f | f \rangle_{\mu}$}} \omega(\tilde{\phi}[- i f - g] \tilde{\phi}[i f - g])
		& = 2 \langle f | f \rangle_{\mu} - \sup_{g \text{ s.t. $\langle g | g \rangle_{\mu} = \langle f | f \rangle_{\mu}$}} |\Omega(f, g)| \\
		& \gneq 0.
	\end{split}
\end{align}
This tells us that within the one-particle Hilbert space $\H_{\omega}^{(1)},$ the state $i \tilde{\phi}[f]|\omega\rangle$ cannot be approximated arbitrarily well by states of the form $\tilde{\phi}[g]|\omega\rangle$ with $g$ real and,
consequently, that the real subspace of $\H_{\omega}^{(1)}$ is not dense in $\H_{\omega}^{(1)}.$
This means that if we treat $\H_{\omega}^{(1)}$ as a real Hilbert space with inner product
\begin{equation}
	(\tilde{\phi}[x] \omega, \tilde{\phi}[y] \omega)
		= \text{Re}\braket*{\tilde{\phi}[x] \omega}{\tilde{\phi}[y] \omega},
\end{equation}
then there exists a state $|\psi\rangle \in \H_{\omega}^{(1)}$ that is real-orthogonal to the subspace spanned by $\tilde{\phi}[g] |\omega\rangle$ with $g$ real.
As explained in appendix \ref{app:generalized-Weyl}, one can create a nontrivial self-adjoint operator $\chi_{i \psi}$ that commutes with every $\tilde{\phi}[g].$\footnote{In the language of that section (see equation \eqref{eq:chi-commutator}), 
\begin{equation}
    [\chi_{i \psi},\tilde{\phi}[g]]
    =
    [\chi_{i \psi},\chi_{\phi[g]\ket{\omega}}]
    =
    - 2i \text{Re}\langle \psi | \phi[g]\omega\rangle
\end{equation}
for all real $g$, and this vanishes under the current definition of $\psi$.}
Consequently, bounded functions of $\chi_{i \psi}$ live in the commutant of the von Neumann algebra generated by the free field, meaning that $\omega$ cannot be pure.

We now know that equation \eqref{eq:supremum-saturation} holds for pure states.
This in fact means that the operator $R$ defined in appendix \ref{app:phase-space} satisfies $R^2 = -1.$
To see this, one rewrites the equation in terms of $R$ as
\begin{equation}
	\langle f | f \rangle_{\mu}
		= \sup_{g \text{ real, nonzero}} \frac{|\langle f | R g \rangle_{\mu}|^2}{\langle g | g \rangle_{\mu}}.
\end{equation}
Hence we have
\begin{equation}
	\langle f | f \rangle_{\mu}
		= \sup_{g \text{ real} \,\langle g | g \rangle_{\mu} = 1} |\langle f | R g \rangle_{\mu}|^2,
\end{equation}
since (using Cauchy-Schwarz) saturation is achieved by $|g\rangle_{\mu} = R^{\dagger} |f\rangle_{\mu} / \lVert R^{\dagger} |f\rangle_{\mu}\rVert.$
This tells us that $R^{\dagger}$ is an isometry on the real subspace of $\K_{\mu}$; since $R^{\dagger}$ takes real solutions to real solutions and $\mu$ splits as in equation \eqref{eq:mu-complex-splitting}, the operator $R^{\dagger}$ is an isometry on all of $\K_{\mu}.$
This gives $R R^{\dagger} = 1,$ and using $R^{\dagger} = -R$ we conclude $R^2 = -1,$ as desired.
Consequently, we actually also learn that $|g\rangle_{\mu} = R^{\dagger} |f\rangle_{\mu} / \lVert R^{\dagger} |f\rangle_{\mu} \rVert$ saturates the supremum in equation \eqref{eq:supremum-saturation}, as claimed above.
Note that when we have $R^2 = -1,$ the only points in the spectrum of $R$ are $\pm i.$
Consequently, the isometric embedding of $\H_{\omega}^{(1)}$ into $\K_{\mu}$ from appendix \ref{app:phase-space} is actually a unitary equivalence between $\H_{\omega}^{(1)}$ and the $+i$ eigenspace of $R$.

Another interesting class of states are those that are ``centrally pure.''
This means that while the von Neumann algebra $\A_{\omega}$ generated by $\A_0$ has a commutant, the commutant is contained entirely within $\A_{\omega},$ i.e., the commutant is the center.
In this case, by its defining equation \eqref{eq:def-of-R}, $R$ will have a kernel, but we will now show that we have $R^2 = -1$ on the orthocomplement of the kernel.

We note that for real $u,$ the condition $R^2 |u\rangle_{\mu} = - |u\rangle_{\mu}$ is equivalent to the statement that $\phi[u + i R u]$ annihilates $|\omega\rangle$ in the GNS Hilbert space, for this occurs if and only if $\ket{u + i R u}_\mu$ lies in the $(-i)$ eigenspace of $R$, and for real $\ket u_\mu$ this is equivalent to $R^2\ket u_\mu=-\ket u_\mu$.
So we consider a real state $|u\rangle_{\mu}$ in the orthocomplement of the kernel of $R$, and we define
\begin{equation}\label{eq:null-psi}
    |\psi\rangle = \phi[u + i R u] |\omega\rangle,
\end{equation}
with the goal of showing $|\psi\rangle = 0.$
Using the isometry $U$ from equation \eqref{eq:app-U-def}, we can compute the norm of $\ket{\psi}$ as
\begin{equation}
	\langle \psi | \psi\rangle
		= \langle \phi[u + i R u]\,\omega | \psi \rangle
		= \langle u + i R u | \sqrt{1 - i R} U \psi\rangle_{\mu}.
\end{equation}
Since $\ket u_\mu$ lies in the orthocomplement of the kernel of $R$, in order to show $\ket{\psi}=0$, it suffices show that $\ket{U\psi}_{\mu}$ lies in the kernel of $R$.

The strategy is to use a duality relation for commutants of free quantum fields, first proved in \cite{Araki:lattice}.
For completeness, we provide proofs of all the relevant facts in appendix \ref{app:one-particle-modular}, following \cite[section 6.5]{Longo:lectures}.
The idea is that one can compute, using again the isometry \eqref{eq:app-U-def}, for any real $f,$ the overlap
\begin{equation}
	\langle \phi[f] \omega | \psi \rangle
		= \langle f | (1 - i R) |u + i R u\rangle_{\mu}
		= \langle f | u \rangle_{\mu} + \langle f | R^2 | u\rangle_{\mu},
\end{equation}
and because $R$ takes real states to real states, this expression is real.
Consequently, if we denote by $X$ the (closure of the) set of states $\phi[f]|\omega\rangle$ with $f$ real, then the state $|\psi\rangle$ has real overlaps with all the states in $X$, and hence lies in the set $X'$ defined in appendix \ref{app:cs-subspace}.
But from that same appendix, one can actually show that, for a centrally pure state, $X'\subseteq X$.\footnote{Recall from appendix \ref{app:generalized-Weyl} that $\A_\omega = \A(X)$. The same appendix gives $\A(X') \subseteq \A(X)' = \A_\omega'$, and by central purity $\A_\omega' \subseteq \A_\omega = \A(X)$, so that $\A(X') \subseteq \A(X)$. It follows that $\A(X') = \A(X') \cap \A(X)$. The final equation of appendix \ref{app:Haag-duality} then implies $\A(X') = \A(X' \cap X)$, and the discussion at the end of appendix \ref{app:cs-subspace} yields $X' = X' \cap X$, hence $X' \subseteq X$.}
In particular, it follows that $|\psi\rangle$ is in $X$, meaning there is a sequence of real test functions $g_n$ with
\begin{equation}
	|\psi\rangle = \lim_n \phi[g_n] |\omega\rangle.
    \label{eq:psi-in-X}
\end{equation}

This reality property will help us show that $\ket{U\psi}_{\mu}$ is in the kernel of $R$ because it allows us to strategically pull out factors of the operator $\Gamma$ which flips the sign of $iR$ via $\Gamma(iR)\Gamma=-iR$.
In particular, using the isometry $U$ from appendix \ref{app:phase-space}, we can translate equation \eqref{eq:psi-in-X} to the space $\K_{\mu}$ as 
\begin{equation}
	|U \psi\rangle_{\mu}
		= \lim_n \sqrt{1 - i R} |g_n\rangle_{\mu}.
\end{equation}
Because each $g_n$ is real, we may freely insert a $\Gamma$ and move it through $\sqrt{1 - i R}$ to obtain
\begin{equation}
	\ket{U\psi}_\mu
		= \lim_n \Gamma \sqrt{1 + i R} |g_n \rangle_{\mu}
		= \lim_n \Gamma \sqrt{\frac{1+ i R}{1 - i R}} U\phi[g_n]|\omega\rangle
		= \Gamma \sqrt{\frac{1+ i R}{1 - i R}} \ket{U\psi}_\mu.
\end{equation}
On the other hand, from the explicit form of $|\psi\rangle$ given in \eqref{eq:null-psi}, we have
\begin{equation}\label{eq:U-psi-intermediate}
	\ket{U\psi}_\mu
		= \sqrt{1 - i R} |u + i R u\rangle_{\mu}
		= \sqrt{1 - i R} (1 + i R) |u\rangle_{\mu}
		= \sqrt{1 + i R} \sqrt{1 + R^2} |u\rangle_{\mu}.
\end{equation}
Since $u$ is real and $R$ preserves reality, we obtain
\begin{equation}
	\ket{U\psi}_\mu = \sqrt{1 + i R} \Gamma \sqrt{1 + R^2} |u\rangle_{\mu} = \Gamma \sqrt{1 - i R} \sqrt{1 + R^2} |u\rangle_{\mu}
	= \Gamma \sqrt{\frac{1 - iR}{1 + i R}} \ket{U\psi}_\mu,
\end{equation}
where in the last equality we inserted  a factor of $\sqrt{1+iR}^{-1}\sqrt{1+iR}$ and used \eqref{eq:U-psi-intermediate}.
Putting this relation together with the previous one, we find
\begin{equation}
	\sqrt{\frac{1 - iR}{1 + i R}} \ket{U\psi}_\mu
		= \sqrt{\frac{1 + iR}{1 - i R}} \ket{U\psi}_\mu,
\end{equation}
from which we may freely conclude that $\ket{U\psi}_\mu$ is in the kernel of $R$.

This completes the claim $R^2 |u\rangle_{\mu} = - |u\rangle_{\mu}$ for $|u\rangle_{\mu}$ in the orthocomplement of the kernel of $R.$
Consequently, for centrally pure states, $R$ has a discrete spectrum and its only possible eigenvalues are $+i, 0,$ and $-i.$

\section{Adjoints of ladder operators}
\label{app:ladder-adjoints}

This appendix fills in the details of a claim made in appendix \ref{app:fock}.
Namely, given a Gaussian state $\omega,$ we defined the creation and annihilation operators
\begin{equation}
\label{eq:phi-is-a-adag}
	\tilde{\phi}[f] = a[f] + a[f^*]^{\dagger},
\end{equation}
and claimed that upon taking closures, $a[f]^{\dagger}$ is genuinely the adjoint of $a[f]$.\footnote{We again emphasize that this implies essential self-adjointness for the field operators $\tilde \phi[f]$ for $f$ real.}

To prove this, we recall the definition of the adjoint of an unbounded operator --- again, see \cite[section 2.3]{Sorce:modular} for a review.
If $T$ is an unbounded operator defined on a dense domain $D \subseteq \H$, then the adjoint $T^{\dagger}$ is defined on all vectors $|\psi\rangle$ such that the functional
\begin{equation}
	|\chi \rangle \mapsto \langle \psi | T \chi \rangle
\end{equation}
is bounded for $|\chi\rangle \in D.$
The Riesz lemma then guarantees that there exists a vector $T^{\dagger} |\psi\rangle$ satisfying
\begin{equation}
	\langle T^{\dagger} \psi | \chi \rangle = \langle \psi | T \chi\rangle,
\end{equation}
and density of $D$ tells us that $T^{\dagger} |\psi\rangle$ is unique.
This defines the action of $T^{\dagger}$ on its domain.

Now, in appendix \ref{app:fock}, we already explained that the adjoint of $a[f]$ is an extension of the creation operator $a[f]^{\dagger}$, which can be schematically written
\begin{equation}
    (a[f])^{\dag}\supseteq \text{``}a[f]^{\dagger}\text{''}.
\end{equation}
Consequently, $a[f]^{\dagger}$ agrees with the adjoint of $a[f]$ on its domain.
To show equality, it therefore suffices to show that every vector in the domain of the adjoint of $a[f]$ is in the domain of (the closure of) the creation operator.
To do this, it is important to know that the creation operators are bounded on any $n$-particle subspace.
As explained in appendix \ref{app:fock}, one can choose an orthonormal basis $\tilde{\phi}[f_j]|\omega\rangle$ for the one-particle Hilbert space, and then the vectors
\begin{equation}
	:\tilde{\phi}[f_1]^{p_1} \dots \tilde{\phi}[f_k]^{p_k} : |\omega\rangle
\end{equation}
with $p_1 + \dots + p_k = n$ will form an orthogonal basis for $\H_{\omega}^{(n)}.$
One then computes (using equation \eqref{eq:GNS-normal-overlaps})
\begin{equation} \label{eq:n-particle-basis-norms}
	\lVert :\tilde{\phi}[f_1]^{p_1} \dots \tilde{\phi}[f_k]^{p_k} : |\omega\rangle \rVert^2
		= p_1! \dots p_k!.
\end{equation}
For a generic test function $g,$ we may decompose it in the one-particle Hilbert space as $\tilde{\phi}[g] |\omega\rangle = \sum_{j} c_j \tilde{\phi}[f_j] |\omega\rangle.$
Using the equivalence from appendix \ref{app:fock}, which relates the $n$-particle space to the $n$-th symmetric tensor power of the one-particle Hilbert space, one can compute
\begin{equation}
	a[g^*]^{\dagger} :\tilde{\phi}[f_1]^{p_1} \dots \tilde{\phi}[f_k]^{p_k} : |\omega\rangle
		= \sum_j c_j :\tilde{\phi}[f_1]^{p_1} \dots \tilde{\phi}[f_j]^{p_j+1} \dots \tilde{\phi}[f_k]^{p_k} :|\omega\rangle.
\end{equation}
Combining with equation \eqref{eq:n-particle-basis-norms}, one finds
\begin{align}
	\begin{split}
		\frac{\left\lVert a[g^*]^{\dagger} :\tilde{\phi}[f_1]^{p_1} \dots \tilde{\phi}[f_k]^{p_k} : |\omega\rangle \right\rVert^2}{\left\lVert :\tilde{\phi}[f_1]^{p_1} \dots \tilde{\phi}[f_k]^{p_k} : |\omega\rangle \right\rVert^2}
		& = \sum_j |c_j|^2 (p_j+1) \\
		& \leq (n+1) \lVert \phi[g] |\omega \rangle \rVert^2,
	\end{split}
\end{align}
where in the last step we have upper bounded $p_j$ by $n$.
This shows that $a[g^*]^{\dagger}$ is bounded on the $n$-particle Hilbert space by $\sqrt{n+1} \lVert \phi[g] |\omega\rangle \rVert.$
A similar calculation shows that $a[g^*]$ is bounded on this space by $\sqrt{n} \lVert \phi[g] |\omega \rangle \rVert.$
So the creation and annihilation operators are bounded on $n$-particle vectors; consequently, in extending to closures, we may as well take these operators to be defined on every $n$-particle Hilbert space and not just on the GNS domain $V_{\omega}.$

Now we return to our main goal, which is to show that every vector in the domain of the adjoint of $a[f]$ is in the domain of the closure of $a[f]^{\dagger}.$
Let $|\psi\rangle$ be in the domain of the adjoint of $a[f].$
This means that the map 
\begin{equation}
	|\chi\rangle \mapsto \langle \psi | a[f] \chi \rangle
    \label{eq:adjoint-domain-map}
\end{equation}
is bounded for every finite-particle vector $|\chi\rangle.$
The vector $|\psi\rangle$ can be split into its $n$-particle amplitudes,
\begin{equation}
	|\psi\rangle = \oplus_n |\psi^{(n)}\rangle.
\end{equation}
Clearly we have
\begin{equation}
	|\psi\rangle = \lim_N \oplus_{n=0}^{N} |\psi^{(n)}\rangle.
\end{equation}
To show that $|\psi\rangle$ is in the domain of the (closure of the) creation operator $a[f]^{\dagger},$ it suffices to show that the sequence
\begin{equation} \label{eq:app-image-sequence}
	\lim_N \oplus_{n=0}^{N} a[f]^{\dagger} |\psi^{(n)}\rangle
\end{equation}
converges.\footnote{This just comes from the definition of a closed operator --- the domain of the closure of $a[f]^{\dagger}$ is the set of vectors that can be reached by limits of sequences in $V_{\omega}$ such that the ``image sequence'' obtained by acting with $a[f]^{\dagger}$ also converges.}
To see that the sequence \eqref{eq:app-image-sequence} converges, we use the fact that by boundedness of the map \eqref{eq:adjoint-domain-map} there is a constant $C$ with
\begin{equation}
	|\langle \psi | a[f] \chi \rangle|
		\leq C \lVert \chi \rVert
\end{equation}
for every $|\chi\rangle$ in the domain of $a[f].$
Next we note that the vector $\oplus_{n=0}^{N} a[f]^{\dagger} |\psi^{(n)}\rangle$ is itself in the domain of $a[f].$
So putting this vector in the role of $\chi$, we have
\begin{equation}
	\bigg| \sum_{n=0}^{N} \langle \psi | a[f] a[f]^{\dagger} \psi^{(n)}\rangle \bigg|
		\leq C \,\big\lVert\!\oplus_{n=0}^{N} a[f]^{\dagger} |\psi^{(n)}\rangle \big\rVert.
\end{equation}
The state $a[f] a[f]^{\dagger} |\psi^{(n)}\rangle$ is in the $n$-particle Hilbert space, so we are free to replace $\bra{\psi}$ by $\bra*{\psi^{(n)}}$ on the left-hand side, giving
\begin{align}
	\big\lVert \oplus_{n=0}^{N}  a[f]^{\dagger} |\psi^{(n)}\rangle \big\rVert
	\leq C 
\end{align}
The left-hand side is increasing in $N$, and bounded, so it must be a convergent sequence.
Orthogonality of the $n$-particle Hilbert spaces then tells us that the sequence of \textit{vectors}
\begin{equation}
	a[f]^{\dagger} \oplus_{n=0}^{N} |\psi^{(n)}\rangle
\end{equation}
is Cauchy, and hence convergent, which is what we wanted to show.

\section{One-particle modular theory}
\label{app:one-particle-modular}

The Tomita operator $S = J \Delta^{1/2}$ associated with a Gaussian state has an analogue, $s = j \delta^{1/2},$ acting only on the one-particle space $\H_{\omega}^{(1)}.$
These operators are best understood by writing the $n$-particle Fock spaces explicitly as $n$-th symmetric tensor powers of $\H_{\omega}^{(1)}$ via the identification of equation \eqref{eq:symmetric-power-Fock}:
\begin{equation}
	: \phi[f_1] \dots \phi[f_n]: |\omega\rangle \leftrightarrow \sqrt{n!} | \phi[f_1] \omega \rangle \otimes_S \dots \otimes_S | \phi[f_n] \omega \rangle.
\end{equation}
On the GNS domain, the connection between $S$ and $s$ is via the formula
\begin{equation}
	S \left( | \phi[f_1] \omega \rangle \otimes_S \dots \otimes_S | \phi[f_n] \omega \rangle\right)
	= | s \phi[f_1] \omega \rangle \otimes_S \dots \otimes_S | s \phi[f_n] \omega \rangle,
\end{equation}
and similarly for the connection between $J$ and $j,$ and the connection between $\Delta^{1/2}$ and $\delta^{1/2}.$
By mapping modular operators from the full GNS space $\H_{\omega}$ to the one-particle space $\H_{\omega}^{(1)},$ one obtains many simplifications that can be used, for example, to prove a general form of Haag duality.
This appendix elaborates that theory and its main results.
We mostly follow \cite{Figliolini:tomita} and \cite[section 6.5]{Longo:lectures}, though note that the Weyl operators there are not equivalent to our operators $e^{i \chi_\psi}$ below.\footnote{They work with the anti-self-adjoint operators $a_{\psi} - a_{\psi}^{\dagger}$, instead of the self-adjoint operators $a_{\psi} + a_{\psi}^{\dagger}.$}

To simplify the presentation in this appendix, we drop the ``tilde'' notation $\tilde{\phi}$ that was used in the previous two appendices; we will write $\phi$ for the vacuum-subtracted field.
For zero-mean Gaussian states, like those considered in the main text, there is no difference between these two conventions.

\subsection{Key points}
\begin{itemize}
	\item Let $\omega$ be a Gaussian state with one-particle Hilbert space $\H_{\omega}^{(1)}.$
	For any state $|\psi\rangle$ in the one-particle Hilbert space, one can introduce creation and annihilation operators $a_{\psi}$ and $a_{\psi}^{\dagger}$, together with a self-adjoint operator
	\begin{equation}
		\chi_{\psi} = a_{\psi} + a_{\psi}^{\dagger}.
	\end{equation}
	For $|\psi\rangle = \phi[f] |\omega\rangle,$ one has $a_{\phi[f] |\omega\rangle} = a[f^*].$
	We have $\chi_{\phi[f] |\omega\rangle} = \phi[f]$ if and only if $f$ is real.
	
	\item To any closed, real subspace $X$ of the one-particle Hilbert space, one can construct a von Neumann algebra $\A(X)$ from the operators $\chi_{\psi}$ for $|\psi\rangle \in X.$
	The closed, real subspace $X'$ is defined as the set of all one-particle vectors that have purely real overlap with all vectors in $X$.
	One has $\A(X') \subseteq \A(X)$ trivially, and a more complicated argument gives $\A(X') = \A(X),$ which is a form of Haag duality.
	
	\item The complex spaces $X \cap i X$ and $X' + i X'$ are orthocomplements of one another.
	
	\item $|\omega\rangle$ is cyclic for $\A(X)$ if and only if $X + i X$ is dense, and it is separating for $\A(X)$ if and only if $X \cap i X$ contains only the zero vector.
	
	\item In the cyclic-separating case, one can define a ``one-particle Tomita operator'' on $X + i X$ via $s |x + i y \rangle = |x - i y\rangle$.
	This is a closable operator with polar decomposition $s = j \delta^{1/2}.$
	It is related to the full Tomita operator $S$ in the Fock space picture by the formula
	\begin{equation}
		S |\psi_1 \otimes \dots \otimes \psi_n\rangle
			= |s \psi_1 \otimes \dots \otimes s \psi_n\rangle,
	\end{equation}
	with similar formulas holding for $J, j$ and $\Delta, \delta.$
	This same formula holds in the non-cyclic-separating setting, though the argument is more involved.
	
	\item Under the isometry $U : \H_{\omega}^{(1)} \to \K_{\mu_1}$ that takes $\phi[f] |\omega\rangle$ to $\sqrt{1 - i R} |f\rangle_{\mu},$ the operator $j$ conjugates to the complex conjugation $\Gamma$ projected away from $R = \pm i.$
	The operator $\delta^{1/2}$ conjugates to $\sqrt{\frac{1 + i R}{1 - i R}}$ projected away from $R = \pm i.$
\end{itemize}

\subsection{Generalizations of field operators}
\label{app:generalized-Weyl}

To start, we note that in the main text, we are concerned with the von Neumann algebra generated by the essentially self-adjoint operators $\phi[f]$, with $f$ a real test function.
It will be very helpful in what follows to consider a more general class of von Neumann algebras, since they can arise when taking commutants of the von Neumann algebra we actually care about.
We begin by thinking of a general field operator in terms of its creation and annihilation operators,
\begin{equation}
	\phi[f] = a[f] + a[f^*]^{\dagger},
\end{equation}
recalling the definitions
\begin{equation}
	a[f] :\phi[g_1] \dots \phi[g_n]: | \omega\rangle
	= \sum_{j} \langle \phi[f^*]\omega | \phi[g_j]\omega \rangle :\phi[g_1] \dots \hat{\phi[g_{j}]} \dots \phi[g_n] : |\omega\rangle
\end{equation}
and
\begin{equation}
	a[f^*]^{\dagger} :\phi[g_1] \dots \phi[g_n]: | \omega\rangle
	= \;\;: \phi[f] \phi[g_1] \dots \phi[g_n] : |\omega\rangle.
\end{equation}
In terms of symmetric tensor powers and one-particle quantities, we can rewrite this productively as
\begin{align}
	\begin{split}
	& a[f] \left(|\phi[g_1] \omega \rangle \otimes_S \dots \otimes_S \phi[g_n] \omega\rangle \right) \\
	& \qquad \qquad = \sqrt{n} \sum_{j} \langle \phi[f^*] \omega | \phi[g_j] \omega \rangle \left(|\phi[g_1] \omega \rangle \otimes_S \dots \hat{|\phi[g_j] \omega \rangle} \otimes_S \dots \otimes_S \phi[g_n] \omega\rangle \right)
	\end{split}
\end{align}
and
\begin{align}
	\begin{split}
	& a[f^*]^{\dagger} \left(|\phi[g_1] \omega \rangle \otimes_S \dots \otimes_S \phi[g_n] \omega\rangle \right) \\
	& \qquad \qquad = \sqrt{n+1} \left(|\phi[f] \omega \rangle \otimes_S |\phi[g_1] \omega \rangle \otimes_S \dots \otimes_S \phi[g_n] \omega\rangle \right).
	\end{split}
\end{align}
We note that from these two definitions, it is clear that the details of the test function $f$ are not so important.
The operator $a[f]$ only depends on the \textit{state} $\phi[f^*] |\omega\rangle,$ while the operator $a[f^*]^{\dagger}$ only depends on the state $\phi[f] |\omega\rangle.$
For any state $|\psi\rangle$ in the one-particle Hilbert space, it therefore makes sense to define unbounded operators on any Hilbert space of finite-particle-number states, that we call $a_{\psi}$ and $a_{\psi}^{\dagger}$, defined by
\begin{equation}
	a_{\psi} \left(|v_1 \rangle \otimes_S \dots \otimes_S |v_n \rangle \right)
	= \sqrt{n} \sum_{j} \langle \psi | v_j \rangle \left(|v_1 \rangle \otimes_S \dots \hat{|v_j \rangle} \otimes_S \dots \otimes_S v_n\rangle \right)
\end{equation}
and
\begin{equation}
	a_{\psi}^{\dag} \left(|v_1 \rangle \otimes_S \dots \otimes_S |v_n \rangle \right)
	= \sqrt{n+1} \left(|\psi \rangle \otimes_S |v_1 \rangle \otimes_S \dots \otimes_S | v_n \rangle \right).
\end{equation}
Clearly for a test function $f$ we have $a[f^*] = a_{\phi[f] |\omega\rangle},$\footnote{Technically $a_{\phi[f] |\omega\rangle}$ is defined on a larger domain than $a[f^*]$, but the closures of these operators coincide.}
and an identical argument to that given in appendix \ref{app:ladder-adjoints} tells us that the operator we have been calling $a_{\psi}^{\dagger}$ actually is, after taking closures, the adjoint of the operator $a_{\psi}.$

So now, in addition to the field operators $\phi[f],$ we have a whole host of additional operators based on elements of the one-particle Hilbert space.
We write their self-adjoint combinations as
\begin{equation}
	\chi_{\psi}
	= a_{\psi} + a_{\psi}^{\dagger}.
\end{equation}
For $\psi$ of the form $\phi[f] |\omega\rangle,$ we have
\begin{equation}
	\chi_{\phi[f] |\omega\rangle}
	= a[f^*] + a[f^*]^{\dagger}
	= (a[f^*] - a[f]) + \phi[f].
\end{equation}
So we only have $\chi_{\phi[f] |\omega\rangle} = \phi[f]$ when $f$ is real; otherwise, we should not expect our ``$\chi$'' operators to have anything to do with the field operators $\phi.$
Nevertheless, because each $\chi$ operator is essentially self-adjoint, we can proceed exactly as in appendix \ref{app:vN-generating} to associate, with any subset $X$ of the one-particle Hilbert space, the von Neumann algebra
\begin{equation}
	\A(X)
	= \big\langle e^{i \chi_{\psi}} \, \big| \, |\psi\rangle \in X\big\rangle
	= \big\langle \,\text{bounded functions of } \chi_{\psi}\,\big|\, |\psi\rangle \in X\,\big\rangle. 
\end{equation}
Moreover, just as for ordinary Weyl operators, a simple expansion of $e^{i \chi_{\psi}}$ in terms of a power series gives
\begin{equation} \label{eq:generalized-Weyl-eval}
	\langle \omega | e^{i \chi_{\psi}} |\omega\rangle
		= e^{- \langle \psi | \psi \rangle/2}.
\end{equation}

Note that for complex numbers $\alpha$ and $\beta,$ one has
\begin{equation}\label{eq:chi-linear-comb}
	\chi_{\alpha \psi_1 + \beta \psi_2}
	= \alpha^* a_{\psi_1} + \beta^* a_{\psi_2}
	+ \alpha a_{\psi_1}^{\dagger} + \beta a_{\psi_2}^{\dagger},
\end{equation}
where this equality holds on the shared domain of finite-particle states.
So in the case where $\alpha$ and $\beta$ are real, we have $\chi_{\alpha \psi_1 + \beta \psi_2} = \alpha \chi_{\psi_1} + \beta \chi_{\psi_2}.$
Consequently, in terms of the generated von Neumann algebra $\A(X)$, we may as well take $X$ to be a real subspace, i.e., given a generic subset $X$ with no vector space structure, we can replace it with its real-linear span without changing $\A(X)$.

By repeating the argument given in section \ref{app:vN-generating}, which showed strong continuity of the map $f \mapsto e^{i \phi[f]}$ with $f$ a real element of $\K_{\mu},$ one can see that the map from one-particle states to generalized Weyl operators, $|\psi\rangle \to e^{i \chi_{\psi}},$ is strongly continuous.
Consequently, in defining the algebra $\A(X),$ there is no problem in replacing $X$ with its closure $\bar{X}$.
So without loss of generality, we will henceforth assume that $X$ is a closed, real  subspace of $\H_\omega^{(1)}$.

The full von Neumann algebra $\A_{\omega}$ associated with a GNS representation is obtained by taking $X$ to be the set of states $\phi[f] |\omega\rangle$ with $f$ real.
This is notably different from the von Neumann algebra that would be obtained by taking $X$ to be the full complex space $\H_{\omega}^{(1)}.$
In the latter case, as we will show momentarily, one always obtains $\A(\H_{\omega}^{(1)}) = \B(\H_{\omega}),$ the full algebra of bounded operators on $\H_{\omega}.$
This happens even when $\omega$ is mixed, while in that case the algebra 
\begin{equation}
    \A_{\omega} = \A(\{\phi[f] |\omega\rangle \,|\, f \text{ real}\})
\end{equation}
is a strict subalgebra of $\B(\H_{\omega}).$

We will show the above claim as a special case of a more general statement, which is that whenever $X$ is a \textit{complex} subspace of ${\cal H}_{\omega}^{(1)}$, and we write
\begin{equation}
    {\cal H}_{\omega}^{(1)}=X\oplus X^\perp,
\end{equation}
then one has (i) $\H_\omega \cong  \H_X \otimes \H_{X^\perp}$, where $\H_X$ is the Fock space of $X$:
\begin{equation}
    \H_X=\bigoplus_{n=0}^\infty X^{\otimes_S n},
\end{equation}
and (ii) $\A(X) \cong \B(\H_X)\otimes 1_{X^\perp}$.

To see point (i), note that one can map $\H_X \otimes \H_{X^\perp}$ to the ordinary Fock space $\H_\omega$ via the map
\begin{align}\label{eq:split-H}
	\begin{split}
		&	|\psi_1 \otimes_S \dots \otimes_S \psi_p\rangle \otimes
	|\psi^\perp_1 \otimes_S \dots \otimes_S \psi^\perp_q\rangle \\
	& \qquad \mapsto \sqrt{\frac{(p+q)!}{p! q!}} |\psi_1 \otimes_S \dots \otimes_S \psi_p \otimes_S \psi^\perp_1 \otimes_S \dots \otimes_S \psi^\perp_q \rangle.
	\end{split}
\end{align}
Unitarity follows from the mutual orthogonality of $X$ and $X^\perp$, which means that even when overlaps are taken on the right-hand side, there is no contribution from permutations that ``mix'' these terms. 

To see point (ii), note that complex linearity of $X$ guarantees that for any $\ket{\psi}$ in $X$, $i\ket{\psi}$ is also in $X$, so that in particular the algebra $\A(X)$ contains (in the affiliated sense) both $\chi_\psi$ and $\chi_{i\psi}$, and hence also $a_\psi^\dagger$ and $a_\psi$.
It follows that any bounded operator that commutes with $\A(X)$ within $\H_X$ would have to commute with all creation and annihilation operators for that Fock space, but this implies that the operator must be proportional to the identity.
It follows that, under the isomorphism of point (i), one has $\A(X)'\subseteq 1_X\otimes \B(\H_{X^{\perp}})$.
The reverse inclusion is immediate, and taking commutants gives claim (ii).

\subsection{Cyclic and separating}
\label{app:cs-subspace}

We are now ready to address an important question for the construction of modular theory: under what conditions on a closed, real subspace $X$ can we conclude that $|\omega\rangle$ is cyclic or separating for the von Neumann algebra $\A(X)$?

First we show that cyclicity is equivalent to the statement that the complex-linear span of $X$ is dense in $\H_{\omega}^{(1)}.$
By taking derivatives of operators like $e^{i \chi_{\psi}},$ one can obtain, within the set $\bar{\A(X) |\omega\rangle},$ any vector of the form
\begin{equation}
	\chi_{\psi_1} \dots \chi_{\psi_n} |\omega\rangle \qquad |\psi_{j}\rangle \in X,
\end{equation}
and such vectors are dense within $\bar{\A(X) |\omega\rangle}.$
By normal-ordering the $\chi$ operators with an analog of equation \eqref{eq:normal-ordering}, one sees that, under the unitary equivalence \eqref{eq:symmetric-power-Fock}, the set of all linear combinations of such states is exactly the set of all linear combinations of states in $(X + i X) \otimes_S\dots \otimes_S (X + i X)$.
The closure of this set is equal to the full symmetric Fock space if and only if $X + i X$ is dense in $\H_{\omega}^{(1)}.$
Since the symmetric Fock space of $\H_\omega^{(1)}$ is canonically isomorphic to $\H_\omega$ through equation \eqref{eq:symmetric-power-Fock}, the claim is proved.

Now we show that $|\omega\rangle$ is separating for $\A(X)$ if and only if one has $X \cap i X = \{0\}$.\footnote{Recall we have assumed (without loss of generality) that $X$ is a real vector space, so the zero vector is always contained in both $X$ and $iX$.}
For the easy direction, we assume the existence of a nonzero vector $|\psi\rangle$ in $X \cap i X.$
Note that using equation \eqref{eq:chi-linear-comb}, we have that $\chi_{i \psi} - i \chi_{\psi}$ is pure annihilation, and hence
\begin{equation}
	(\chi_{i \psi} - i \chi_{\psi}) |\omega\rangle = 0.
\end{equation}
But the operator on the left-hand side does not vanish on arbitrary $n$-particle states, so $|\omega\rangle$ is not separating for $\A(X)$.\footnote{A subtlety here is that the operator $\chi_{i \psi} - i \chi_{\psi}$ is unbounded, so it is not actually in the algebra $\A(X)$. Instead it is affiliated with $\A(X)$. Its adjoint, $\chi_{i \psi} + i \chi_{\psi}$, is also affiliated with $\A(X)$. If $|\omega\rangle$ were separating for $\A(X$), then as in e.g. \cite[appendix B]{Sorce:modular}, one would have $(\chi_{i \psi} - i \chi_{\psi}) |\omega\rangle$ in the domain of the Tomita operator $S$ with image $(\chi_{i \psi} + i \chi_{\psi}) |\omega\rangle.$ But then the Tomita operator would take the zero vector to a nonzero vector, which is a contradiction.}

To show the other direction --- that $X \cap i X = \{0\}$ implies that $|\omega\rangle$ is separating for $\A(X)$ --- we will check the equivalent condition that $|\omega\rangle$ is cyclic for the commutant $\A(X)'$.
Towards a characterization of $\A(X)'$, note that the commutation relation $[a_{\psi_1}, a_{\psi_2}^{\dagger}] = \langle \psi_1 | \psi_2\rangle$ implies
\begin{equation}
\label{eq:chi-commutator}
    [\chi_{\psi_1},\chi_{\psi_2}]=2 i\, \text{Im} \left(\langle \psi_1|\psi_2\rangle \right)\,,
\end{equation}
and from this one computes
\begin{equation} \label{eq:chi-commutators}
	e^{i \chi_{\psi_1}} e^{i \chi_{\psi_2}}
	= e^{i \chi_{\psi_1 + \psi_2}} e^{- i\, \text{Im}(\langle \psi_1 | \psi_2\rangle)}
	= e^{i \chi_{\psi_2}} e^{i \chi_{\psi_1}} e^{- 2 i\, \text{Im}(\langle \psi_1 | \psi_2\rangle)}.
\end{equation}
In particular, this means that for any one-particle vector $|\eta \rangle$ that has purely real overlap with all elements of $X$, the operator $e^{i \chi_{\eta}}$ is in $\A(X)'.$
This encourages us to define the space
\begin{equation} \label{eq:X-prime-definition}
	X'
	= \{ |\eta\rangle \in \H_{\omega}^{(1)} \,|\, \text{Im}\langle \psi | \eta \rangle = 0 \text{ for all } |\psi\rangle \in X\}.
\end{equation}
We clearly have $\A(X)' \supseteq \A(X')$; in a later section we will show equality, which is a form of Haag duality, first demonstrated by Araki in \cite{Araki:lattice}.
The spaces $X$ and $X'$ have nice relationships to one another, primarily due to the fact that passing from $X$ to $X'$ is the same as multiplying by $i$ and then taking the orthocomplement with respect to the real inner product
\begin{equation}
	(\psi, \eta) = \text{Re}\langle\psi | \eta\rangle,
    \label{eq:real-inner-product}
\end{equation}
so one has $X' = (i X)^{\perp_{\reals}}.$
This implies the identity $X'' = \bar{X}$.
Additionally, one can now show that we have $X \cap i X = \{0\}$ if and only if $X' + i X'$ is dense.
This fact becomes trivial upon recognizing that $X \cap i X$ and $\bar{X' +i X'}$ are both complex Hilbert spaces, and they are easily seen to be one another's orthocomplements with respect to the original, complex inner product. 

So far, we have shown that whenever we have $X \cap i X = \{0\},$ we have $X' + i X'$ dense, so we may conclude that $|\omega\rangle$ is cyclic for $\A(X')$.
Given the inclusion $\A(X') \subseteq \A(X)',$ clearly any cyclic vector for $\A(X')$ is also cyclic for $\A(X)'.$
This tells us that when we have $X \cap i X = \{0\},$ the vector $|\omega\rangle$ is cyclic for $\A(X)',$ hence separating for $\A(X).$

Another useful statement is that for two closed, real subspaces $X, Y \subseteq \H_{\omega}^{(1)},$ one has $\A(X) = \A(Y)$ if and only if one has $X = Y.$
The ``if'' direction is obvious.
For the ``only if'' direction, one first observes the identity
\begin{equation}\label{eq:A-sum-is-union-A}
	\A(\bar{X+Y}) = \A(X) \vee \A(Y),
\end{equation}
which follows from $\A(\bar{X+Y}) = \A(X+Y)$ and the fact that every operator $e^{i \chi_{\psi_X + \psi_Y}}$ can be written as a complex multiple of the product $e^{i \chi_{\psi_X}} e^{i \chi_{\psi_Y}}$.
Now, assuming toward contradiction that we have $\A(X) = \A(Y)$ but $X \neq Y,$ we have
\begin{equation}
	\A(X) = \A(X) \vee \A(Y) = \A(\bar{X+Y})
\end{equation}
but $X \subsetneq \bar{X+Y}.$
Taking commutants, one finds
\begin{equation}
	\A(X)' = \A(\bar{X+Y})',
\end{equation}
but $(X+Y)' \subsetneq X'.$
Choosing $|\psi\rangle$ in $X'$ that is not in $(X+Y)',$ one finds that $e^{i \chi_{\psi}}$ will be in $\A(X)'$ but not in $\A(\bar{X+Y})',$ which is a contradiction.

\subsection{Modular operators}
\label{app:1P-modular-ops}

We now know that for a closed, real subspace $X \subseteq \H_{\omega}^{(1)},$ the conditions for $|\omega\rangle$ to be cyclic-separating for $\A(X)$ are
\begin{align}
	\bar{X + i X}
	& = \H_{\omega}^{(1)}, \\
	X \cap i X & = \{0\}.
\end{align}
In this setting, we can define a Tomita operator $S$ as the closure of the usual formula
\begin{equation}
	S a |\omega\rangle = a^{\dagger} |\omega\rangle, \qquad a \in \A(X).
\end{equation}
By acting with $S$ on the state $|\psi\rangle = \chi_{\psi} |\omega\rangle,$ and using the fact that the self-adjoint operator $\chi_{\psi}$ is affiliated to $\A(X)$ when $|\psi\rangle$ is in $X$, one finds
\begin{equation}
	S |\psi\rangle = |\psi\rangle, \qquad |\psi\rangle \in X.
\end{equation}
By antilinearity of $S$, we have an obvious action of $X$ on the domain $X + i X$ as
\begin{equation}
	S (|\psi\rangle + i |\eta\rangle)
	= |\psi\rangle - i |\eta\rangle, \qquad |\psi\rangle, |\eta\rangle \in X.
\end{equation}

As is usually the case with Gaussian states, many interesting questions about the global GNS space $\H_{\omega}$ can be translated into simpler questions about the one-particle space $\H_{\omega}^{(1)}.$
This encourages us to ask how much information about the Tomita operator $S$ can be recovered from its action on $\H_{\omega}^{(1)}.$
To pursue this line of investigation, we \textit{define} an operator $s$ on the domain $X + i X$ by
\begin{equation}\label{eq:def-s}
	s (|\psi\rangle + i |\eta\rangle)
	= |\psi\rangle - i |\eta\rangle, \qquad |\psi\rangle, |\eta\rangle \in X.
\end{equation}
As is typically useful in studying Tomita operators (see e.g. \cite{Witten:notes} or \cite[appendix B]{Sorce:modular}), one defines another operator $s'$ with domain $X' + i X'$ via the formula
\begin{equation}
	s' (|\psi'\rangle + i |\eta'\rangle)
	= |\psi'\rangle - i |\eta'\rangle, \qquad |\psi'\rangle, |\eta'\rangle \in X'.
\end{equation}
By computing overlaps, one finds that $s^{\dagger}$ is an extension of $s'$ and $(s')^{\dagger}$ is an extension of $s$.\footnote{The relevant calculation is
	\begin{align*}
		\langle \psi' + i \eta' | s (\psi + i \eta)\rangle
		& = \langle \psi' + i \eta' | \psi - i \eta\rangle \\
		& = \langle \psi' | \psi \rangle - i \langle \eta' | \psi\rangle - i \langle \psi' | \eta \rangle - \langle \eta' | \eta\rangle \\
		& = \langle \psi | \psi' \rangle - i \langle \psi | \eta'\rangle - i \langle \eta | \psi' \rangle - \langle \eta | \eta'\rangle \\
		& = \langle \psi + i \eta | \psi' - i \eta' \rangle \\
		& = \langle \psi + i \eta | s' (\psi' + i \eta') \rangle,
	\end{align*}
    which shows $s'\subseteq s^\dag$. Swapping $s'$ with $s$ gives $s\subseteq (s')^\dag$.}
This shows that $s$ has densely defined adjoint, and hence is closable.  Swapping the primes shows that also $s'$ is closable.
By abuse of notation, one writes $s$ and $s'$ for the closures.

One then writes the polar decompositions
\begin{align}
	s
	& = j \delta^{1/2}, \\
	s'
	& = j' (\delta')^{1/2}.
\end{align}
Because $s$ and $s'$ square to the identity on their domains, each is invertible, which means that $j$ and $j'$ are antiunitary.
Standard manipulations --- see reviews in \cite{Witten:notes} or \cite[appendix B]{Sorce:modular} --- then give the identities
\begin{align}
	j
	& = j', \\
	j^2
	& = 1, \\
	\delta'
	& = \delta^{-1}, \\
	j \delta j
	& = \delta^{-1}.
\end{align}

We now draw a sharper connection between $S$ and $s,$ and between the polar decompositions $S = J \Delta^{1/2}$ and $s = j \delta^{1/2}.$
For any $|\psi_1\rangle, \dots |\psi_n\rangle$ in $X$, one has
\begin{equation}
	S \chi_{\psi_1} \dots \chi_{\psi_n} |\omega\rangle
	= \chi_{\psi_n} \dots \chi_{\psi_1} |\omega\rangle.
\end{equation}
By acting on normal-ordered products of $\chi$ operators, one concludes the formula
\begin{equation}
	S |\psi_1 \otimes_S \dots \otimes_S \psi_n\rangle
	= |\psi_1 \otimes_S \dots \otimes_S \psi_n \rangle.
\end{equation}
Extending by antilinearity, and taking the vectors $|\psi_j\rangle$ to be in $X + iX$ instead of simply $X$, we find formula
\begin{equation}
	S |\psi_1 \otimes_S \dots \otimes_S \psi_n\rangle
	= |s \psi_1 \otimes_S \dots \otimes_S s \psi_n \rangle,
\end{equation}
and by taking limits one may conclude that the same formula holds for any $|\psi_j\rangle$ in the domain of $s.$
Using $\Delta = S^{\dagger} S$ and $\delta = s^{\dagger} s$ one may also derive
\begin{equation}
	\Delta^{1/2} |\psi_1 \otimes_S \dots \otimes_S \psi_n\rangle
	= |\delta^{1/2} \psi_1 \otimes_S \dots \otimes_S \delta^{1/2} \psi_n \rangle
\end{equation}
on the same domain, and by appealing to polar decompositions one finds
\begin{equation} \label{eq:J-j-formula}
	J |\psi_1 \otimes_S \dots \otimes_S \psi_n\rangle
	= |j \psi_1 \otimes_S \dots \otimes_S j \psi_n \rangle,
\end{equation}
where this relation holds for \textit{any} vector on Hilbert space by boundedness of $J$ and $j.$

To close this subsection, we give an explicit formula for $j$ and $\delta^{1/2}$ in the case of physical interest, namely when $X$ is the set $\{\phi[f] |\omega\rangle \, |\, f \text{ real}\}$.
In this case one can take advantage of the isometry $U$ from appendix \ref{app:phase-space}, which maps from $\H_{\omega}^{(1)}$ to $\K_{\mu}$ via the formula
\begin{equation}
	U \phi[f] |\omega\rangle
	= \sqrt{1 - i R} |f\rangle_{\mu}.
\end{equation}
The image of $U$ is the orthocomplement of the $-i$ eigenspace of $R$.
In the case of present interest, where $|\omega\rangle$ is cyclic and separating for $\A(X)$, the operator $R$ actually has no $-i$ eigenvalues; for these would give rise to null states for $|\omega\rangle,$ which would prevent it from being separating.

Now, for general complex $f,$ one has
\begin{align}
	\begin{split}
		U s U^{\dagger} \sqrt{1 - i R} |f\rangle_{\mu}
		& = U s \phi[f] |\omega\rangle \\
		& = U \phi[f^*] |\omega\rangle \\
		& = \sqrt{1 - i R} \Gamma |f\rangle_{\mu} \\
		& = \Gamma \sqrt{1 + i R} |f\rangle_{\mu}.
	\end{split}
\end{align}
So for general $|\eta\rangle_{\mu}$ in the image of $\sqrt{1 - i R},$ 
 one has\footnote{A priori, this equation only tells us that $UsU^{\dag}$ is defined on all vectors for which $\Gamma \sqrt{\frac{1+iR}{1 - i R}}$ is defined, but not vice versa.
However, every vector $\eta$ in the domain of $UsU^{\dag}$ must have  $U^{\dag}\ket{\eta}$ in the domain of $s$ --- that is, in the GNS domain. The image of the GNS domain under $U$ is the image of $\sqrt{1-iR}$.
Since $UsU^{\dag}$ and $\Gamma \sqrt{\frac{1+iR}{1 - i R}}$ share the same dense domain, they are equal as unbounded operators, as is needed for the polar decomposition arguments to follow.}
\begin{equation}
	U s U^{\dagger} |\eta \rangle_{\mu}
	= \Gamma \sqrt{\frac{1+iR}{1 - i R}} |\eta\rangle_{\mu}.
\end{equation}
By writing $s = j \delta^{1/2}$ and using $U^{\dagger} U = 1,$ one has
\begin{equation}\label{eq:sep-pol-dec}
	(U j U^{\dagger})(U \delta^{1/2} U^{\dagger}) |\eta \rangle_{\mu}
	= \Gamma \sqrt{\frac{1+iR}{1 - i R}} |\eta\rangle_{\mu}.
\end{equation}
Since in the present separating setting  $R$ has no $-i$ eigenvalues, and since $R$ commutes with $\Gamma,$ we also know that $R$ has no $+i$ eigenvalues.
Therefore the positive operator $\sqrt{(1+i R)/(1 - iR)}$ has no kernel.
The operator $\Gamma$ is antiunitary.
So by uniqueness of the polar decomposition, we may conclude
\begin{equation}
	U j U^{\dagger}
	= \Gamma
\end{equation}
and
\begin{equation}
	U \delta^{1/2} U^{\dagger}
	= \sqrt{\frac{1 + i R}{1 - i R}}.
\end{equation}

\subsection{Haag duality}
\label{app:Haag-duality}

We now apply the modular theory toolkit that we have developed so far to discuss Haag duality.
In appendix \ref{app:cs-subspace}, we proved
\begin{equation}
	\A(X') \subseteq \A(X)'.
\end{equation}
We now prove equality.
Without loss of generality we take $X$ to be closed and real-linear.

In the case where $|\omega\rangle$ is cyclic and separating for $X$, the proof is very simple.
Using equation \eqref{eq:J-j-formula}, one may easily compute
\begin{equation}
	J \chi_{\psi} J
	= \chi_{j \psi}.
\end{equation}
Since $J$ conjugates $\A(X)$ to $\A(X)'$ in the cyclic-separating case, we find
\begin{equation}
	\A(X)'
	= J \A(X) J
	= \A(j X).
\end{equation}
So all that remains is to show the inclusion $j X \subseteq X'.$
To see this, fix $|\psi\rangle, |\eta\rangle \in X,$ and write, using by definition \eqref{eq:def-s} that $s$ acts like the identity on $X$,
\begin{equation}
	\langle \psi + \eta | j | \psi + \eta \rangle
	= \langle \psi + \eta | \delta^{1/2} |\psi + \eta\rangle.
\end{equation}
This is real by positivity of the operator $\delta^{1/2}.$
On the other hand, expanding the left-hand side and using antiunitarity of $j$, one has
\begin{equation}
	\langle \psi | j | \psi \rangle + \langle \eta | j | \eta \rangle
	+ 2 \langle \psi | j | \eta \rangle \in \reals.
\end{equation}
The first two terms are also real by the same argument, so one finds that $j |\eta\rangle$ must have real overlap with every $|\psi\rangle$ in $X$.
This gives $j X \subseteq X'.$
We conclude $\A(X)' \subseteq \A(X'),$ but we already had $\A(X') \subseteq \A(X)',$ so this completes the proof that
\begin{equation}\label{eq:Haag-cyc-sep}
    \A(X)' = \A(X') \quad \mathrm{if} \; \ket{\omega}\;\mathrm{is}\;\mathrm{cyclic}\;\mathrm{and}\;\mathrm{separating}\;\mathrm{for}\;\A(X).
\end{equation}

In the case where $|\omega\rangle$ is not cyclic and separating for $X$, one employs a reduction technique to reduce to the cyclic-separating case.\footnote{We thank Roberto Longo for explaining this argument to us over email.}
We write\footnote{Note that we are using the standard complex orthocomplement rather than the real $\perp_{\mathbb{R}}$ mentioned above.}
\begin{align}
	A 
	& = X \cap i X \\
	B
	& = (X + i X)^{\perp},
\end{align}
and we decompose the one-particle Hilbert space as
\begin{equation}
	\H_{\omega}^{(1)}
		= A \oplus B \oplus C,
\end{equation}
where $C$ is just $(A \oplus B)^{\perp}.$
The idea is that $A$ contains the ``non-separating'' data of $\A(X)$, while $B$ contains the ``non-cyclic'' data.
Given any decomposition of $\H_{\omega}^{(1)}$ into mutually orthogonal subspaces, there is a natural decomposition of $\H_{\omega}$ into an associated tensor product.
The way this works is that we define
\begin{equation}
	\H_{A}
		\equiv \bigoplus_n A^{\otimes_S n},
\end{equation}
and similarly for $\H_B$ and $\H_C.$
One can map $\H_A \otimes \H_B \otimes \H_C$ to the ordinary Fock space $\H_{\omega}$ via a generalization of the map in equation \eqref{eq:split-H} to the case of three subspaces.

Decomposing $X$ into $A$ and its orthocomplement $A^{\perp}\cap X$, we have, using \eqref{eq:A-sum-is-union-A},
\begin{equation}
	\A(X) = \A(A + (A^{\perp} \cap X)) = \A(A) \vee \A(A^{\perp} \cap X).
\end{equation}
Taking commutants gives
\begin{equation}
	\A(X)' = \A(A)' \cap \A(A^{\perp} \cap X)'.
\end{equation}
Because $A = X \cap i X$ is a complex subspace, we may use the result discussed near equation \eqref{eq:split-H} to conclude that $\mathcal A(A)$ is the full algebra of bounded operators on $\H_A$.
Consequently, under the identification of $\H_{\omega}$ with $\H_{A} \otimes \H_{B} \otimes \H_C,$ the commutant $\A(A)'$ is simply the algebra of bounded operators on $\H_B\otimes \H_C.$
We may therefore write
\begin{equation}
	\A(X)' = (1_A \otimes \B(\H_B) \otimes \B(\H_C)) \cap \A(A^{\perp} \cap X)'.
\end{equation}

We can get a little extra mileage by explicitly writing $A^{\perp} = B \oplus C,$ and remembering that all of the vectors in $B$ are orthogonal to $X$ by definition.
So we have $A^{\perp} \cap X = C\cap X,$ and this gives
\begin{equation}
	\A(X)' = (1_A \otimes \B(\H_B) \otimes \B(\H_C)) \cap \A(C \cap X)'.
\end{equation}
Since \textit{every} bounded operator on $\H_B$ is in the commutant of $\A(C \cap X),$ we can decompose this further as
\begin{equation}
	\A(X)' = \left[1_A \otimes \B(\H_B) \otimes 1_C \right] \vee \left[ (1_A \otimes 1_B \otimes \B(\H_C)) \cap \A(C \cap X)' \right].
\end{equation}
The second algebra in this expression is simply the commutant of $\A(C \cap X)$ considered as an algebra acting on $\H_C$.
But, using the results from section \ref{app:cs-subspace}, one can see that $\ket \omega$ is both cyclic\footnote{To show that $\A(C\cap X)\ket{\omega}$ is dense in $\H_C$ we note that 
\begin{equation}
C \cap X + i (C \cap X) = C \cap X + (i C) \cap (i X) = C \cap X + C \cap i X = C \cap (X + i X) = C.
\end{equation}} and separating\footnote{The state $|\omega\rangle$ is separating for this algebra $\A(C \cap X)$ because we have
\begin{equation}
	(C \cap X) \cap i (C \cap X)
	= (A^{\perp} \cap B^{\perp} \cap X) \cap i (A^{\perp} \cap B^{\perp} \cap X)\subseteq A^\perp \cap (X\cap iX)=\{0\}.
\end{equation}
} for this algebra, hence we can use the cyclic-separating result \eqref{eq:Haag-cyc-sep} from above to move the prime in this second term inside the parenthesis. 
Taking the intersection with $1_A \otimes 1_B \otimes \B(\H_C)$ then simply amounts to restricting to $\ket{\psi}\in C\subseteq \H_\omega^{(1)}$ in the definition of $\A((C\cap X)')$, giving
\begin{equation} \label{eq:n-perp-reduced}
	(1_A \otimes 1_B \otimes \B(\H_C)) \cap \A(C \cap X)'
		= \A\big(\{|\psi\rangle \in C \,|\, \text{Im}\langle \psi | \eta\rangle  = 0 \text{ for all } |\eta\rangle \in C \cap X\} \big).
\end{equation}

Putting all this together, we find that $\A(X)'$ is minimally generated by (i) the set on the right-hand side of \eqref{eq:n-perp-reduced}, and (ii) the set of bounded operators acting on $\H_B.$
This latter is simply $\A((X + i X)^{\perp})$, again using the discussion near equation \eqref{eq:split-H}.
Using \eqref{eq:A-sum-is-union-A} we then find that $\A(X)'$ coincides with  $\A(Y)$, where $Y$ is the set
\begin{equation}\label{eq:one-part-Y}
	Y
	= \bar{(X + i X)^{\perp} + \{|\psi\rangle \in C \,|\, \text{Im}\langle \psi | \eta\rangle  = 0 \text{ for all } |\eta\rangle \in C \cap X\}}.
\end{equation}
It is easily verified that one has $Y \subseteq X',$\footnote{Clearly elements of $(X+iX)^{\perp}$ have real overlap with $X$, since such overlaps vanish.
Moreover, states $\ket{\psi}$ in the second term have real overlaps with both $C^\perp\cap X$ (since such overlaps vanish) and $C\cap X$ (by definition), and thus with all of $(C^\perp \cap X)\cup (C\cap X)=X$; it follows that such states are in $X'$.}
giving $\A(X)' \subseteq \A(X'),$  completing the proof.

One consequence of Haag duality that is used in appendix \ref{app:pure} is the dual of the identity $\A(\bar{X + Y}) = \A(X) \vee \A(Y)$ from equation \eqref{eq:A-sum-is-union-A}, which we proved in appendix \ref{app:cs-subspace}.
Taking commutants and using Haag duality, one discovers
\begin{equation}
	\A((X + Y)') = \A(X') \cap \A(Y'),
\end{equation}
and it is easy to compute $(X + Y)' = X' \cap Y',$ giving
\begin{equation}
	\A(X' \cap Y') = \A(X') \cap \A(Y').
\end{equation}
Relabeling $X' \to X$ and $Y' \to Y$, it is easy to see that we have $\A(X \cap Y) = \A(X) \cap \A(Y)$ for general $X, Y.$

\subsection{Non-cyclic-and-separating}
\label{app:non-cs}

We now return to our discussion of the one-particle modular operators $(s,j,\delta)$, generalizing the results of section \ref{app:1P-modular-ops} to the non-cyclic-and-separating case.
Our main reason for developing this theory is that --- as discussed in appendix \ref{app:1P-modular-ops} --- any Gaussian state for which the spectrum of $R$ contains eigenvalues $\pm i$ will not be separating for the von Neumann algebra generated by fields.
In particular, this is the case for all pure states, so it is useful to understand how to apply modular theory in this setting.

In ordinary modular theory, when one has a state $|\omega\rangle$ and a von Neumann algebra $\A$, one defines the projector $P$ onto the subspace $\bar{\A' |\omega\rangle}$ --- written $P$ because it is an element of $\A$ --- and the projector $P'$ onto the subspace $\bar{\A |\omega\rangle}$ --- written $P'$ because it is an element of $\A'.$
When one has $P \neq 1$ the state is not separating; when one has $P' \neq 1$ the state is not cyclic.
In this case one can still define modular operators --- see for example the nice review in \cite[appendix A]{Ceyhan:QNEC} --- via the formulas\footnote{The inclusion of $P$ in the definition of $S$ is necessary to ensure that $S$ is well defined in the non-separating case. 
In particular, if we have $a\ket{\omega}=0$ for some nonzero $a \in \A$, we must ensure the identity $Sa\ket{\omega}=0$ as well.
Introducing $P$ does the job, since it is easy to show that for such $a$, $a^\dag \ket{\omega}$ must be orthogonal to $\overline{\A'\ket\omega}$, thus $Pa^\dag \ket{\omega}=0$.
The formula $S(1-P')=0$ simply amounts to defining $S$ to vanish on the orthocomplement of $\A\ket{\omega}$, which is nontrivial in the non-cyclic case.
}
\begin{equation}
	S a |\omega\rangle = P a^{\dagger} |\omega\rangle , \qquad S (1-P') = 0
    \label{eq:tomita-op-non-cs}
\end{equation}
and
\begin{equation}
	S' a' |\omega\rangle = P' (a')^{\dagger} |\omega\rangle , \qquad S' (1-P) = 0.
\end{equation}
It is an easy exercise to show that $(S')^{\dagger}$ is an extension of $S$, and that $S^{\dagger}$ is an extension of $S'$, so that each operator is closable and can be replaced by its closure with no confusion.

Given a real, closed subspace $X$ of the one-particle Hilbert space $\H_{\omega}^{(1)},$ we can still define $X'$ as the real-orthocomplement of $(i X).$
We define $p$ to be the projector onto the closure of $X' + i X',$ and define $p'$ to be the projector onto the closure of $X + i X.$
By comparison to the usual modular operators, it is natural to define operators $s$ and $s'$ via
\begin{equation}
	s |x + i y\rangle = p |x - i y\rangle, \qquad s(1-p')=0
\end{equation}
and
\begin{equation}
	s' |x' + i y' \rangle = p' |x - i y \rangle, \qquad s'(1-p)=0.
\end{equation}
It is straightforward to check that $(s')^{\dagger}$ extends $s$, and that $s^{\dagger}$ extends $s'.$

In this setting, for $\psi_1, \dots, \psi_n \in X,$ we have\footnote{The exponentiated form of this equation is a straightforward application of equation \eqref{eq:tomita-op-non-cs}; taking derivatives then yields the version presented here.}
\begin{equation}
	S \chi_{\psi_1} \dots \chi_{\psi_n} |\omega\rangle
		= P \chi_{\psi_n} \dots \chi_{\psi_1} |\omega\rangle,
\end{equation}
and by using normal ordering and the isomorphism between the GNS space and a symmetric Fock space, we can conclude that $S$ acts on the real elements within each $n$-particle space as
\begin{equation}
	S | \psi_1 \otimes_S \dots \otimes_S\psi_n\rangle 
	= P | \psi_1 \otimes_S \dots \otimes_S \psi_n\rangle.
\end{equation}

Going forward, the important fact will be that $P$ acts on each $n$-particle subspace as $p^{\otimes n}.$
The reason is that $P$ is the projector onto the closure of the space $\A(X)' |\omega\rangle.$
But by the preceding subsection, we know this is the same as the closure of the space $\A(X') |\omega\rangle,$ which is simply the sum over the $n$-particle subspaces $(X' + i X')^{\otimes n}$.  
From this, we may conclude that for $\psi_1, \dots, \psi_n \in X,$ we have
\begin{equation}
	S | \psi_1 \otimes_S \dots \otimes_S\psi_n\rangle 
	=| p \psi_1 \otimes_S \dots \otimes_S p \psi_n\rangle,
\end{equation}
and extending by antilinearity, we find that so long as each $\psi_j$ is in the domain of $s,$ we have
\begin{equation}
	S | \psi_1 \otimes_S \dots \otimes_S\psi_n\rangle 
	=| s \psi_1 \otimes_S \dots \otimes_S s \psi_n\rangle.
\end{equation}
So the usual identity still holds; $S$ acts on the $n$-particle spaces as a tensor product of $s$ operators, and similarly for the operators $J, \Delta$ and $j, \delta$ that appear in the polar decompositions of $S$ and $s.$

In the special case where we take $X$ to be the subspace formed by states $\phi[f] |\omega\rangle$ with $f$ real, we computed in the preceding subsection --- within the cyclic-separating case --- the formulas $U j U^{\dagger} = \Gamma$ and $U \delta^{1/2} U^{\dagger} = \sqrt{\frac{1+iR}{1-i R}},$ where $U$ is the isometry
\begin{equation}
	U \phi[f] |\omega\rangle = \sqrt{1 - i R} |f\rangle_{\mu}.
\end{equation}
In the non-cyclic-separating case, we can still compute, for complex $f$,
\begin{align}
	\begin{split}
	U s U^{\dagger} \sqrt{1 - i R} |f\rangle_{\mu}
		& = U p \phi[f^*]\ket{\omega} \\
		& = U p U^{\dagger} \sqrt{1 - i R} \Gamma |f\rangle_{\mu} \\
		& = U p U^{\dagger} \Gamma \sqrt{1 + i R} |f\rangle_{\mu},
	\end{split}
\end{align}
and for $|\eta\rangle_{\mu}$ in the image of $\sqrt{1 - i R},$  one can conclude
\begin{align}\label{eq:UsUdag-ket}
	U s U^{\dagger} |\eta\rangle_{\mu}
	& = U p U^{\dagger} \Gamma \sqrt{\frac{1 + i R}{1 - i R}} |\eta\rangle_{\mu}.
\end{align}

We would like to say that in this setting, one can identify $U j U^{\dagger}$ with $U p U^{\dagger} \Gamma,$ and still identify $U \delta^{1/2} U^{\dagger}$ with the projection of $\sqrt{\frac{1+ i R}{1 -i R}}$ away from $R = \pm i.$

In order to do so, let us first upgrade equation \eqref{eq:UsUdag-ket} to  work for \textit{any} $\ket{\eta}_\mu$ in the domain of $UsU^\dag$, and not just for those in the image of $\sqrt{1-iR}$.
We can do this by noticing that the orthocomplement of this image is the $-i$ eigenspace of $R$, and these states are annihilated by $U^{\dagger}$ anyway, thus we can turn \eqref{eq:UsUdag-ket} into an operator equality by including, on the right-hand side, a projector away from the $-i$ eigenspace:
\begin{align}
	U s U^{\dagger}
	& = U p U^{\dagger} \Gamma \sqrt{\frac{1 + i R}{1 - i R}} (1 - \Pi_{-i}).
\end{align}
In fact, because $\sqrt{1 + i R}$ has as its kernel the $+i$ eigenspace, we can rewrite this as
\begin{align}
	U s U^{\dagger}
	& = U p U^{\dagger} \Gamma \sqrt{\frac{1 + i R}{1 - i R}} (1 - \Pi_{-i} - \Pi_{i}).
\end{align}

We now claim that the projectors $UpU^\dag$ and $1-\Pi_{-i}-\Pi_i$ actually coincide.
To understand this, recall that $p$ projects onto $\bar{X' + i X'}$, which is the orthocomplement of the space $X\cap iX$.
The key point is that in this setting, the space $X \cap i X$ is the set of all vectors that can be written either as $\phi[f] |\omega\rangle$ or as $i \phi[g] |\omega\rangle$ for distinct choices of $f, g$ real.
In other words, it is labeled by pairs $f,g$ with
\begin{equation}
	(\phi[f] - i \phi[g])|\omega\rangle = 0.
\end{equation}

These are exactly the null states of $|\omega\rangle$!
That is, whenever $|\alpha\rangle_{\mu}$  is in the $-i$ eigenspace of $R$, we have 
\begin{equation}
    0=
    \phi[\alpha] \ket{\omega}
    =
    \frac{1}{2}\phi[\alpha+\alpha^*]\ket{\omega}
    +
    \frac{i}{2} \phi[\alpha-\alpha^*]\ket{\omega}\,,
\end{equation}
so the state $\phi[\alpha+\alpha^*]\ket{\omega}$ is in $X\cap iX$, and conversely any state in $X\cap iX$ can be written in this form.
The space $\bar{X' + i X'}$ is the orthocomplement of this space, so it can be identified as the set of one-particle states $|\beta\rangle$ satisfying
\begin{equation}
	\langle \beta | \phi[\alpha + \alpha^*] \omega \rangle = 0 \text{ for all $R |\alpha\rangle_{\mu} = - i |\alpha\rangle_{\mu}$},
\end{equation}
which after acting with $U$ is the condition
\begin{equation}
	\langle U \beta | \sqrt{1 - i R} (\alpha + \alpha^*) \rangle_\mu = 0 \text{ for all $R |\alpha\rangle_{\mu} = - i |\alpha\rangle_{\mu}$},
\end{equation}
and $R$ acts on $|\alpha\rangle_\mu$ as $-i$ and on  $\ket{\alpha^*}_\mu$ as $+i,$ so this is
\begin{equation}
	\langle U \beta | \alpha^* \rangle_\mu = 0 \text{ for all $R |\alpha\rangle_{\mu} = - i |\alpha\rangle_{\mu}$}.
\end{equation}
Since complex conjugation takes $-i$ eigenstates of $R$ to $+i$ eigenstates of $R$, this is the same as saying
\begin{equation}
	\langle U \beta | \alpha \rangle_\mu = 0 \text{ for all $R |\alpha\rangle_{\mu} =  i |\alpha\rangle_{\mu}$}.
\end{equation}
So $\bar{X' + i X'}$ is the set of all states $\ket \beta$ that are mapped, under $U$, to the orthocomplement of the $+i$ eigenspace of $R$.
Since the image of $U$ is already orthogonal to the $-i$ eigenspace of $R$, we conclude
\begin{equation}
	|\beta\rangle \in \bar{X' + i X'} \quad \Leftrightarrow \quad U |\beta\rangle \text{ in $(-i, i)$ spectral subspace of $R$}.
\end{equation}
Consequently, since $p$ projects onto $\bar{X' +i X'},$ it is obvious that $U p U^{\dagger}$ projects onto the $(-i, i)$ spectral subspace of $R$.
We can write $U p U^{\dagger} = 1 - \Pi_{i} - \Pi_{-i}$, as claimed.

Equation \eqref{eq:UsUdag-ket} now becomes 
\begin{equation}
\begin{aligned}
        UsU^\dag & = U p U^{\dagger} \Gamma \sqrt{\frac{1 + i R}{1 - i R}} (U p U^{\dagger}) \\
	& = \left[ U p U^{\dagger} \Gamma U p U^{\dagger} \right] \left[ U p U^{\dagger} \sqrt{\frac{1 + i R}{1 - i R}} (U p U^{\dagger})\right],
\end{aligned}
\end{equation}
where we have also used that $UpU^{\dag}$ commutes with $\Gamma$.

\textit{Now} we may appeal to uniqueness of the polar decomposition --- and the fact that $U p U^{\dagger}$ commutes with $\Gamma$ and with $R$ --- to identify 
\begin{equation}
	U j U^{\dagger} = U p U^{\dagger} \Gamma
\end{equation}
and
\begin{equation}
	U \delta^{1/2} U^{\dagger}
		= U p U^{\dagger} \sqrt{\frac{1+i R}{1 - i R}}.
\end{equation}

\section{A modular conjugation proof}
\label{app:modular-conjugation-proof}

In section \ref{sec:canonical-checking}, we encountered two cyclic-separating states $|\omega_1\rangle_1$ and $|\omega_2\rangle_2$ in distinct Hilbert spaces $\H_{\omega_1}$ and $\H_{\omega_2}$, with modular conjugations $J_{\omega_1}$ and $J_{\omega_2}.$
These Hilbert spaces are built as GNS representations of a common $*$-algebra $\A_0,$ and the von Neumann algebras in each representation are called $\A_{\omega_1}$ and $\A_{\omega_2}.$
We also had a ``natural cone'' representative of $|\omega_2\rangle_2$ in $\H_{\omega_1},$ denoted $|\omega_2^{\natural}\rangle_1,$ and we defined an isometry $V : \H_{\omega_2} \to \H_{\omega_1}$ via the formula 
\begin{equation} \label{eq:app-natural-def-V}  
    V a |\omega_2\rangle_2
        = a |\omega_2^{\natural}\rangle_1, \qquad a \in \A_0.
\end{equation}
We claimed the identity $V^{\dagger} J_{\omega_1} V = J_{\omega_2}.$

To show this, we will use Araki's characterization of the modular conjugation from \cite[theorem 1]{Araki:natural-1}.
To show that $V^{\dagger} J_{\omega_1} V$ is the modular conjugation of $|\omega_2\rangle_2,$ we need to show the following:
\begin{enumerate}[(i)]
	\item $V^{\dagger} J_{\omega_1} V$ is antiunitary and squares to the identity.
	\item $V^{\dagger} J_{\omega_1} V$ conjugates $\A_{\omega_2}$ to $\A_{\omega_2}'$ and fixes $|\omega_2\rangle_2.$
	\item One has $\langle a^{\dagger} \omega_2 | V^{\dagger} J_{\omega_1} V a |\omega_2\rangle_2 \geq 0$ for all $a \in \A_{0}.$\footnote{The assumptions of theorem 1 in \cite{Araki:natural-1} include the assumption that this inequality is saturated only for $a=0,$ but that assumption is not actually used anywhere in the proof.} 
\end{enumerate}
To prove these, the only real ingredient we will need is that the projector $P' = V V^{\dagger}$ is in the von Neumann algebra $\A_{\omega_1}',$ which follows from (1) the fact that $P'$ is the projector onto the image of $V$, (2) the fact that the image of $V$ is an invariant space for $\A_{\omega_1},$ and (3) the general fact that projectors onto an invariant space for a von Neumann algebra $\M$ must lie in the commutant $\M'.$
We will write $P = J_{\omega_1} P' J_{\omega_1}.$
Note that $P'$ is the projector onto the closure of the space $\A_{\omega_1} |\omega_2^{\natural}\rangle_1.$
Because $J_{\omega_1}$ conjugates $\A_{\omega_1}$ to $\A_{\omega_1}'$ and fixes $|\omega_2^{\natural}\rangle_1,$ we know that $P$ is the projector onto the closure of $\A_{\omega_1}' |\omega_2^{\natural}\rangle_1.$

Using the above observations, one has
\begin{equation}
	(V^{\dagger} J_{\omega_1} V)^2
	= V^{\dagger} J_{\omega_1} P' J_{\omega_1} V
	= V^{\dagger} P V.
\end{equation}
The operator $V^{\dagger} P V$ is a projector, since one has
\begin{equation}
	(V^{\dagger} P V)^2 = V^{\dagger} P P' P V = V^{\dagger} P' P V = V^{\dagger} P V.
\end{equation}
It fixes the dense subspace $\A_{\omega_2}'|\omega_2\rangle_2,$ since conjugation by $V^{\dagger}$ maps $\A_{\omega_1}$ to $\A_{\omega_2},$ and we have
\begin{equation}
	V^{\dagger} P V a' |\omega_2\rangle_2
	= a' V^{\dagger} P V |\omega_2\rangle_2
	= a' V^{\dagger} P |\omega_2^{\natural}\rangle_1
	= a' V^{\dagger} |\omega_2^{\natural}\rangle_1
	= a' |\omega_2\rangle_2.
\end{equation}
Consequently, we may conclude that $V^{\dagger} P V$ is the identity on $\H_{\omega_2},$ which completes the proof of point (i) above.

As for point (ii), the fact that $V^{\dagger} J_{\omega_1} V$ fixes $|\omega_2\rangle_2$ is obvious from the fact that $J_{\omega_1}$ fixes $|\omega_2^{\natural}\rangle_1.$
Moreover, for any $b, c \in \A_0$ and $a' \in \A_{\omega_1}',$ we have
\begin{align}
	\begin{split} 
		V^{\dagger} a' V b c |\omega_2\rangle_2
		& = V^{\dagger} a' b c |\omega_2^{\natural}\rangle_1 \\
		& = V^{\dagger} b a' c |\omega_2^{\natural}\rangle_1 \\
		& = V^{\dagger} (V V^{\dagger}) b a' c (V V^{\dagger})|\omega_2^{\natural} \rangle_1 \\
		& = V^{\dagger} b (V V^{\dagger}) a' (V V^{\dagger}) c|\omega_2^{\natural} \rangle_1 \\
		& = b (V^{\dagger} a' V) c |\omega_2\rangle_2.
	\end{split} 
\end{align}
So $V^{\dagger}$ conjugates $\A_{\omega_1}'$ to $\A_{\omega_2}'.$
In addition, for any $a \in \A_0,$ the operator $V a V^{\dagger}$ acts as $a$ on the space $\A_{\omega_1} |\omega_2^{\natural}\rangle_1$ and annihilates the orthocomplement of this space, from which we may conclude $V a V^{\dagger} = P' a P' = a P'.$
Putting these two facts together, we have
\begin{align}
	\begin{split} 
		(V^{\dagger} J_{\omega_1} V) a (V^{\dagger} J_{\omega_1} V)
		& = V^{\dagger} J_{\omega_1} (V a V^{\dagger}) J_{\omega_1} V \\
		& = V^{\dagger} J_{\omega_1} a P' J_{\omega_1} V \\
		& = V^{\dagger} (J_{\omega_1} a J_{\omega_1}) (J_{\omega_1} P' J_{\omega_1}) V \\
		& = V^{\dagger} (J_{\omega_1} a J_{\omega_1}) P V \\
		& = V^{\dagger} (J_{\omega_1} a J_{\omega_1}) P (V V^{\dagger}) V \\
		& = (V^{\dagger} (J_{\omega_1} a J_{\omega_1}) V) V^{\dagger} P V.
	\end{split} 
\end{align}
Since we know $V^{\dagger} P V$ is the identity, and $V^{\dagger}$ conjugates $\A_{\omega_1}'$ to $\A_{\omega_2}',$ we conclude that this expression is affiliated to $\A_{\omega_2}',$ and $V^{\dagger} J_{\omega_1} V$ conjugates $\A_{\omega_2}$ to $\A_{\omega_2}',$ completing the proof of point (ii) above.

Point (iii) above follows from the identity
\begin{equation}
	\langle a^{\dagger} \omega_2 |V^{\dagger} J_{\omega_1} V a |\omega_2\rangle_2
		= \langle a^{\dagger} \omega_2^{\natural} | J_{\omega_1} a |\omega_2^{\natural}\rangle_1,
\end{equation}
together with the fact (see \cite[theorem 4]{Araki:natural-1}) that $|\omega_2^{\natural}\rangle_1$ can be written as a limit of states of the form $L J_{\omega_1} L |\omega_1\rangle_1$ for $L \in \A_{\omega_1}.$
This gives
\begin{align}
	\begin{split}
	\langle a^{\dagger} \omega_2 |V^{\dagger} J_{\omega_1} V a |\omega_2\rangle_2
		& = \lim_n \langle a^{\dagger} L_n J_{\omega_1} L_n \omega_1 | J_{\omega_1} a L_n J_{\omega_1} L_n \omega_1 \rangle_1 \\
		& = \lim_n \langle a^{\dagger} L_n (J_{\omega_1} L_n J_{\omega_1}) \omega_1 | (J_{\omega_1} a L_n J_{\omega_1}) L_n \omega_1 \rangle_1 \\	
		& = \lim_n \langle a^{\dagger} L_n \omega_1 | L_n (J_{\omega_1} L_n^{\dagger} J_{\omega_1}) (J_{\omega_1} a L_n J_{\omega_1}) \omega_1 \rangle_1 \\
		& = \lim_n \langle L_n^{\dagger} a^{\dagger} L_n \omega_1 | J_{\omega_1} (L_n^{\dagger} a L_n) \omega_1 \rangle_1 \\
		& = \lim_n \langle L_n^{\dagger} a^{\dagger} L_n \omega_1 | \Delta_{\omega_1}^{1/2} L_n^{\dagger} a^{\dagger} L_n \omega_1 \rangle_1.
	\end{split}
\end{align}
Nonnegativity holds for each $n$ by nonnegativity of the modular operator $\Delta_{\omega_1}$; this completes the proof of point (iii) above.

\section{Details of the pure-state ansatz}
\label{app:ansatz-details}

This appendix fills in some details about the pure-state ansatz from section \ref{sec:pure}.
To clean up notation, let us write $|\chi\rangle$ for $(U \otimes U) V |\omega_2^{(2)}\rangle,$ and we recall that $V$ maps each $n$-particle space $\H_{\omega}^{(n)}$ unitarily to the $n$-th symmetric tensor power of $\H_{\omega}^{(1)}$ using the inverse of equation \eqref{eq:symmetric-power-Fock}.

\subsection{Higher-particle amplitudes}
\label{app:higher-particle-amplitudes}

The ``null state constraints'' on the $n$-particle amplitudes of $|\omega_2\rangle$ are given by the equation
\begin{equation} \label{eq:app-particle-pure-constraint}
	a[L^{-1} \psi] |\omega_2^{(2 n)}\rangle + a[L^{-1} \psi^*]^{\dagger} |\omega_2^{(2n-2)}\rangle = 0
\end{equation}
for any $|\psi\rangle_{\mu_2}$ in the $-i$ eigenspace of $R_2.$
Taking an overlap of this equation with a $(2n-1)$-particle state $|:\phi[g_1] \dots \phi[g_{2n-1}]: \omega_1\rangle,$ and passing unitarily to an appropriate tensor power of $\K_{\mu_1}$, we obtain
\begin{align} \label{eq:even-particle-pure-overlap}
	\begin{split}
	& \sqrt{2n(2n-1)} \langle L^{-1} \psi^* \otimes g_1 \otimes \dots \otimes g_{2n-1} |(U \otimes \dots \otimes U) V \omega_2^{(2 n)}\rangle_{\mu_1} \\
		& \qquad \qquad = - \sum_{j} \langle \Pi_{1,i} g_j | L^{-1} \psi \rangle_{\mu_1} \langle g_1 \otimes \dots \otimes \hat{g_j} \otimes \dots \otimes g_{2n-1} |(U \otimes \dots \otimes U) V \omega_2^{(2n-2)}\rangle_{\mu_1}.
	\end{split}
\end{align}
Because states of the form $|L^{-1} \psi^*\rangle_{\mu_1}$ are dense in the image of $U$ (as discussed below equation \eqref{eq:three-particle-pure2}), this completely determines $\ket*{\omega_2^{(2n)}}$ in terms of the lower-particle amplitudes.
For $n=2$, we have
\begin{align}
	\begin{split}
		& \langle L^{-1} \psi^* \otimes g_1 \otimes g_2 \otimes g_3 |(U \otimes \dots \otimes U) V \omega_2^{(4)}\rangle_{\mu_1} \\
		& \qquad \qquad = - \frac{1}{\sqrt{12}} \sum_{j} \langle \Pi_{1,i} g_j | L^{-1} \psi \rangle_{\mu_1} \langle \dots \otimes \hat{g_j} \otimes \dots | \chi\rangle_{\mu_1}.
	\end{split}
\end{align}
Using equation \eqref{eq:two-particle-overlap-mu} to rewrite $\langle \Pi_{1,i} g_j | L^{-1} \psi \rangle_{\mu_1}$ in terms of $\ket*{\omega_2^{(2)}}$, we have
\begin{align}
	\begin{split}
		& \langle L^{-1} \psi^* \otimes g_1 \otimes g_2 \otimes g_3 |(U \otimes \dots \otimes U) V \omega_2^{(4)}\rangle_{\mu_1} \\
		& \qquad \qquad = \frac{1}{\sqrt{6}} \sum_{j} \langle L^{-1} \psi^* \otimes g_j | \chi \rangle_{\mu_1} \langle \dots \otimes \hat{g_j} \otimes \dots | \chi\rangle_{\mu_1} \\
		& \qquad \qquad =\frac{1}{\sqrt{6}} \sum_{j} \langle L^{-1} \psi^* \otimes g_j \otimes ( \dots \otimes \hat{g_j} \otimes \dots) | \chi \otimes \chi\rangle_{\mu_1} .
	\end{split}
\end{align}
We can freely symmetrize over the bra on the left-hand side, since $(U \otimes \dots \otimes U) V |\omega_2^{(4)}\rangle$ is in a symmetric tensor power space.
This gives
\begin{align}
	\begin{split}
		& \langle L^{-1} \psi^* \otimes_S g_1 \otimes_S g_2 \otimes_S g_3 |(U \otimes \dots \otimes U) V \omega_2^{(4)}\rangle_{\mu_1} \\
		& \qquad \qquad = \sqrt{\frac{3}{2}} \langle L^{-1} \psi^* \otimes_S g_1 \otimes_S g_2 \otimes_S g_3 |P_{\text{sym}} | \chi \otimes \chi \rangle_{\mu_1},
	\end{split}
\end{align}
where $P_{\text{sym}}$ is the projector onto the symmetric subspace of $\H_{\omega_1}^{\otimes 4}.$
This explicitly solves for the $4$-particle amplitude as
\begin{equation}
	|(U \otimes \dots \otimes U) V \omega_2^{(4)}\rangle_{\mu_1}
		= \sqrt{\frac{3}{2}} P_{\text{sym}} | \chi \otimes \chi \rangle_{\mu_1},
\end{equation}
and induction gives
\begin{equation}
	|(U \otimes \dots \otimes U) V \omega_2^{(2n)}\rangle_{\mu_1}
	= \frac{1}{n!} \sqrt{\frac{(2n)!}{2^n}} P_{\text{sym}}|\chi^{\otimes n}\rangle_{\mu_1}.
\end{equation}

\subsection{Symmetry}
\label{app:ansatz-symmetry}

We would like to find an explicit formula for a generic overlap $\langle f \otimes g | \chi\rangle_{\mu_1}$ that makes it manifest that $\ket{\chi}_{\mu_1}$ belongs to the \textit{symmetric} subspace of $\K_{\mu_1,{+i}}\otimes\K_{\mu_1.+i}$,  (with $\K_{\mu_1.+i}$ the $+i$ eigenspace of $R_1$).\footnote{Recall that this is required by consistency if we are working towards proving necessity as in section \ref{sec:two-part-ampl}, and it is required in order to \textit{define} $\ket*{\omega_2^{(2)}}$ when proving sufficiency.}

Recall  from equation \eqref{eq:two-particle-overlap-mu} that $\ket{\chi}_{\mu_1}$ satisfies, for any $\ket{\psi}_{\mu_1}$ in the $-i$ eigenspace of $R_2$, 
\begin{equation} \label{eq:app-chi-def}
	\langle L^{-1} \psi^* \otimes g|\chi\rangle_{\mu_1} = - \frac{1}{\sqrt{2}} \langle \Pi_{1,i} g | L^{-1} \psi \rangle_{\mu_1}.
\end{equation}
Meanwhile, by acting on equation \eqref{eq:1+Q-2} with $\Gamma$ and recalling that $\Gamma$ commutes with $L$, one finds 
\begin{equation}
	\Pi_{1, i} |L^{-1} \psi^*\rangle_{\mu_1} = \frac{1 + L^{\dagger} L}{2} |L^{-1} \psi^*\rangle_{\mu_1},
\end{equation}
and we note that the operator $(1+Q)/2$ appearing on the right hand side is invertible, where as in the main text, we write $Q = L^{\dagger} L.$
From the above, we learn that $\Pi_{1,i}$ acts in a bounded, invertible way when mapping the space $\{\ket*{L^{-1}\psi^*}_{\mu_1}\}$ to its image under $\Pi_{1,i}$.
But in section \ref{sec:pure-odd-particle}, we showed that the image of this map is dense in $\Pi_{1,i} \K_{\mu_1}.$
Since the generic $f$ appearing in the overlap $\braket{f\otimes g}{\chi}_{\mu_1}$ can be taken without loss of generality to lie in $\Pi_{1,i}\K_{\mu_1}$, we can compute
\begin{equation} \label{eq:generic-chi-overlap}
	\langle f \otimes g |\chi\rangle
	= \langle \Pi_{1,i} 2 (1 + Q)^{-1} \Pi_{1,i} f \otimes g | \chi\rangle 
	= - \sqrt{2} \langle \Pi_{1,i} g |(1 + Q)^{-1} \Gamma \Pi_{1,i} f \rangle,
\end{equation}
where in the last equality we used \eqref{eq:app-chi-def}.
Symmetry of this expression under the interchange $f \leftrightarrow g$ follows from antiunitarity of $\Gamma$ and the fact that $\Gamma$ commutes with $Q.$

\subsection{Normalizability}
\label{app:ansatz-normalizability}

The ansatz for $|\omega_2\rangle$ from section \ref{sec:pure-sufficiency} is
\begin{equation}
	|(U \otimes \dots \otimes U) V \omega_2^{(2n)}\rangle_{\mu_1}
	= \frac{1}{n!} \sqrt{\frac{(2n)!}{2^n}} P_{\text{sym}}|\chi^{\otimes n}\rangle_{\mu_1}.
\end{equation}
The norm-squared of $|\omega_2^{(2n)}\rangle$ is therefore given by
\begin{equation}
	\langle \omega_2^{(2n)} | \omega_2^{(2n)}\rangle
	= \frac{1}{(n!)^2} \frac{(2n)!}{2^n} \langle \chi^{\otimes n} | P_{\text{sym}}|\chi^{\otimes n}\rangle_{\mu_1}.
\end{equation}
Picking an arbitrary orthonormal basis $|e_j\rangle_{\mu_1}$ for the $+i$ eigenspace of $R_1$, we may think of $|\chi\rangle$ as the operator on $\K_{\mu_1}$ with matrix elements $\chi_{jk} = \langle e_j \otimes e_k | \chi\rangle_{\mu_1}.$
Then the above expression becomes
\begin{equation}
	\langle \omega_2^{(2n)} | \omega_2^{(2n)}\rangle
	= \frac{1}{(n!)^2} \frac{1}{2^n} \sum_{j_1,\dots, j_{2n}} \sum_{\pi \in S_{2n}} \chi_{j_1 j_2}^* \dots \chi_{j_{2n-1} j_{2n}}^*
	\chi_{j_{\pi(1)} j_{\pi(2)}} \dots \chi_{j_{\pi(2n-1)} j_{\pi(2n)}}.
\end{equation}

This is a fairly complicated expression.
It can be understood by treating $\chi$ as a matrix, with each term in the sum corresponding to some expression like $\tr((\chi^{\dagger} \chi)^{p_1}) \dots \tr((\chi^{\dagger} \chi)^{p_m}),$ where one has $p_1 + \dots + p_m = n,$ but the specific integers $p_1, \dots, p_m$ are determined by the cycle structure of the permutation $\pi.$
To bound any one of these terms, we can use standard Hölder-type inequalities to write
\begin{equation}
	\tr\big((\chi^{\dagger} \chi)^{p}\big)
	\leq \tr\big(\chi^{\dagger} \chi\big) \lVert \chi^{\dagger} \chi \rVert_{\infty}^{p-1},
\end{equation}
so we have
\begin{equation} \label{eq:cdagger-c-traces}
	\tr\big((\chi^{\dagger} \chi)^{p_1}\big) \cdots \tr\big((\chi^{\dagger} \chi)^{p_m}\big)
	\leq \tr\big(\chi^{\dagger} \chi\big)^{m} \lVert \chi^{\dagger} \chi \rVert_{\infty}^{n-m}.
\end{equation}

The details of the relationship between a permutation $\pi$ and the number of traces $m$ are somewhat complicated.
However, using the inequality $\tr(\chi^{\dagger} \chi) \geq \lVert \chi^{\dagger} \chi \rVert_{\infty}$, we can upper bound the right-hand side of \eqref{eq:cdagger-c-traces} by replacing $m$ with any upper bound on the number of traces. A  simple (albeit loose) upper bound is given by the number of cycles in the permutation $\pi$, since traces can only appear by the closure of a cycle.
Consequently, we have
\begin{equation}
	\langle \omega_{2}^{(2n)} | \omega_2^{(2n)} \rangle 
		\leq (\lVert \chi^{\dagger} \chi \rVert_{\infty})^n
			\sum_{\pi \in S_{2n}} \frac{1}{(n!)^2 2^n} \left( \frac{\tr(\chi^{\dagger} \chi)}{\lVert \chi^{\dagger} \chi \rVert_{\infty}} \right)^{\text{\# cycles}(\pi)}.
\end{equation}
This can be rewritten as
\begin{equation}\label{eq:norm-2n-bound}
	\langle \omega_{2}^{(2n)} | \omega_2^{(2n)} \rangle 
	\leq (\lVert \chi^{\dagger} \chi \rVert_{\infty})^n
	\sum_{m=1}^{2n} C(2n, m) \frac{1}{(n!)^2 2^n} \left( \frac{\tr(\chi^{\dagger} \chi)}{\lVert \chi^{\dagger} \chi \rVert_{\infty}} \right)^{m}.
\end{equation}
where $C(2n, m)$ is the number of permutations in $S_{2n}$ with exactly $m$ cycles.

The numbers $C(2n, m)$ are called Stirling numbers.
Their basic properties can be found in any book on combinatorics, for example \cite[chapter 1.4]{Aigner:enumeration}.
They have generating function
\begin{equation}
	\sum_{m=1}^{2n} C(2n, m) x^{m} = \frac{\Gamma(x+2n)}{\Gamma(x)},
\end{equation}
from which we can rewrite inequality \eqref{eq:norm-2n-bound} as
\begin{equation} \label{eq:penultimate-2n-upper-bound}
	\langle \omega_{2}^{(2n)} | \omega_2^{(2n)} \rangle 
	\leq (\lVert \chi^{\dagger} \chi \rVert_{\infty})^n
	\frac{1}{(n!)^2 2^n} \left. \frac{\Gamma(x+2n)}{\Gamma(x)} \right|_{x= \mathrm{tr}(\chi^{\dagger} \chi)/\lVert \chi^{\dagger} \chi\rVert_{\infty}}.
\end{equation}

Now let us think a bit about what the traces and the operator norms in this expression actually mean.
What we have been calling $\chi^{\dagger} \chi$ is an infinite-dimensional matrix with matrix elements
\begin{align} \label{eq:chidag-chi-elements}
	\begin{split}
		(\chi^{\dagger} \chi)_{jk}
		& = \sum_{\ell} \langle \chi | e_j \otimes e_{\ell} \rangle_{\mu_1} \langle e_{\ell} \otimes e_k | \chi\rangle_{\mu_1},
	\end{split}
\end{align}
where $|e_j\rangle_{\mu_1}$ is an orthonormal basis for the $+i$ eigenspace of $R_1.$
Using equation \eqref{eq:generic-chi-overlap}, we may compute
\begin{align}
	\begin{split}
		(\chi^{\dagger} \chi)_{jk}
		& = 2 \sum_{\ell} \langle (1 + Q)^{-1} \Gamma e_{j} |  e_{\ell} \rangle_{\mu_1} \langle e_{\ell} | (1 + Q)^{-1} \Gamma e_{k} \rangle_{\mu_1} \\
		& = 2 \langle (1 + Q)^{-1} \Gamma e_{j} |  \Pi_{1,i} | (1 + Q)^{-1} \Gamma e_{k} \rangle_{\mu_1}.
	\end{split}
\end{align}
Using antiunitarity of $\Gamma$, together with the fact that it commutes with $Q$ and sends $\Pi_{1,i}$ to $\Pi_{1,-i},$ we may write this as
\begin{align}
	\begin{split}
		(\chi^{\dagger} \chi)_{jk}
		& = 2 \langle e_{k} | (1 + Q)^{-1} \Pi_{1,-i}  (1 + Q)^{-1} | e_{j} \rangle_{\mu_1}.
	\end{split}
\end{align}
The same calculation as in equation \eqref{eq:1+Q-2} yields that $\Pi_{1,i}$ acts on states like $L^{-1} |R_2 = +i\rangle_{\mu_2}$ as $(1+Q)/2$, and since this produces a dense set of states in the $+i$ eigenspace of $R_1,$ one can approximate $(1 + Q)^{-1} |e_j\rangle_{\mu_1}$ arbitrarily well as a vector of the form $L^{-1} |R_2 = +i\rangle_{\mu_2}.$
Another calculation like that in equation \eqref{eq:1+Q-2} shows that $\Pi_{1,-i}$ acts on all such states as $(1-Q)/2.$
Putting these observations together, and using $\Pi_{1,-i} = \Pi_{1,-i}^2,$ we can write $(\chi^{\dagger} \chi)_{jk}$ as
\begin{align}
	\begin{split}
		(\chi^{\dagger} \chi)_{jk}
		& = \langle e_{k} | \frac{1}{2} \left( \frac{1-Q}{1+Q}\right)^2 | e_{j} \rangle_{\mu_1}.
	\end{split}
\end{align}
The function $((1-x)/(1+x))^2$ is strictly smaller than one away from $x=0.$
Since $Q$ is a lower-bounded operator --- as $L$ is invertible --- we may conclude that $((1-Q)/(1+Q))^2$ is strictly bounded away from one as an operator.
As a consequence, we obtain that there is a number $\epsilon > 0$ such that $\chi^{\dagger} \chi$ has the upper bound
\begin{equation}
	\lVert \chi^{\dagger} \chi \rVert_{\infty} \leq \frac{1 - \epsilon}{2}.
\end{equation}
If we return now to equation \eqref{eq:penultimate-2n-upper-bound}, we find the upper bound
\begin{equation} \label{eq:final-2n-upper-bound}
	\langle \omega_2^{(2n)} | \omega_2^{(2n)} \rangle
		\leq \left( \frac{1 - \epsilon}{4} \right)^n \frac{1}{(n!)^2} 
		 \left. \frac{\Gamma(x+2n)}{\Gamma(x)} \right|_{x= \tr(\chi^{\dagger} \chi)/\lVert \chi^{\dagger} \chi\rVert_{\infty}}.
\end{equation}
It is a simple task to show that the sum of this expression from $n=0$ to $n=\infty$ converges (for instance, one may use the ratio test);\footnote{Let $a_n$ be the term on the right-hand side. Using that $\Gamma(z+1)=z\Gamma(z)$ twice one finds that the ratio of $a_{n+1}/a_n$ goes to $1-\epsilon$ as $n\to\infty$, which is smaller than $1$ since $\epsilon >0$.}
consequently, the ansatz for $|\omega_2\rangle$ is normalizable.

\subsection{Domains of field operators}
\label{app:ansatz-domain}

Here we aim to show that $|\omega_2\rangle$, as defined by the ansatz of section \ref{sec:pure}, is in the domain of (the closure of) each field operator $\phi[f].$
Proceeding by induction, we will show that all states of the form $\phi[f_1] \dots \phi[f_n] |\omega_2\rangle$ make sense.

The point is that, as explained in appendix \ref{app:ladder-adjoints}, each field operator $\phi[f]$ has, as its closure, the adjoint of the field operator $\phi[f^*].$
So at first it will suffice to show that $|\omega_2\rangle$ is in the domain of the adjoint of $\phi[f^*].$
By the definition of the adjoint of an unbounded operator, this means that we must show that the map
\begin{equation}
	|\eta \rangle \mapsto \langle \omega_2 | \phi[f^*] \eta\rangle
\end{equation}
is bounded for any $|\eta\rangle$ in the original domain of definition of $\phi[f^*].$
These states $|\eta\rangle$ are the finite-particle-number states that can be obtained by acting on $|\omega_1\rangle$ with a finite number of field operators.
But since $\phi[f^*]$ is bounded on each $n$-particle space, we may as well let $|\eta\rangle$ be a generic state of finite particle number, i.e.,
\begin{equation}
	|\eta\rangle = \sum_{k=0}^{K} |\eta^{(k)}\rangle.
\end{equation}

Consider now the overlap 
\begin{equation}
	\langle \omega_2 | \phi[f^*] \eta^{(k)}\rangle.
\end{equation}
Using the decomposition $\phi[f^*] = a[f^*] + a[f]^{\dagger},$ one sees that only the $(k+1)$- and $(k-1)$-particle-number amplitudes of $|\omega_2\rangle$ contribute to this overlap.
So we have
\begin{equation}
	\langle \omega_2 | \phi[f^*] \eta^{(k)}\rangle
		= \langle \omega_2^{(k-1)} | a[f^*] \eta^{(k)}\rangle
			+ \langle \omega_2^{(k+1)} | a[f]^{\dagger} \eta^{(k)}\rangle.
\end{equation}
The operator norm of $a[f]^{\dagger}$ on the $k$-particle space is  $\sqrt{k+1} \lVert \phi[f^*] |\omega_1\rangle \rVert$ (as we computed in appendix \ref{app:ladder-adjoints}); similarly the operator norm of $a[f^*]$ on this space is $\sqrt{k} \lVert \phi[f] |\omega_1\rangle \rVert.$
Using Cauchy-Schwarz on the above expression, we find 
\begin{equation}
	|\langle \omega_2 | \phi[f^*] \eta^{(k)}\rangle|
	\leq \left( \sqrt{k+1} \lVert \phi[f^*] |\omega_1\rangle \rVert  \lVert \omega_2^{(k+1)} \rVert + \sqrt{k} \lVert \phi[f] |\omega_1\rangle \rVert \lVert \omega_2^{(k-1)} \rVert \right) \lVert \eta^{(k)}\rVert.
\end{equation}
If $k$ is even, then this vanishes, because the odd-particle amplitudes of $|\omega_2\rangle$ vanish.
On the other hand, if we put $k = 2n+1,$ then we may estimate the norms of the even-particle amplitudes of $|\omega_2\rangle$ using inequality \eqref{eq:final-2n-upper-bound}.
This gives
\begin{align}
	\begin{split}
	|\langle \omega_2 | \phi[f^*] \eta^{(2n+1)}\rangle|
		& \leq \sqrt{2n+2} \lVert \phi[f^*] |\omega_1\rangle \rVert  \left( \frac{1-\epsilon}{4} \right)^{(n+1)/2} \frac{1}{(n+1)!} \sqrt{\frac{\Gamma(x+2n+2)}{\Gamma(x)}} \lVert \eta^{(2n+1)}\rVert \\
	& \qquad + \sqrt{2n+1} \lVert \phi[f] |\omega_1\rangle \rVert \left( \frac{1-\epsilon}{4} \right)^{n/2} \frac{1}{n!} \sqrt{\frac{\Gamma(x+2n)}{\Gamma(x)}} \lVert \eta^{(2n+1)}\rVert ,
	\end{split} 
\end{align}
where $x$ is some number $\geq 1.$
Using fairly crude bounds, one can bound this as
\begin{align}
	\begin{split}
		|\langle \omega_2 | \phi[f^*] \eta^{(2n+1)}\rangle|
		& \leq \left[ \sqrt{2n+2} \left( \frac{1-\epsilon}{4} \right)^{n/2} \frac{1}{n!}    \sqrt{\frac{\Gamma(x+2n+2)}{\Gamma(x)}} \right] \left( \lVert \phi[f^*] |\omega_1\rangle \rVert + \lVert \phi[f] |\omega_1\rangle \rVert \right) \lVert \eta^{(2n+1)}\rVert.
	\end{split}
\end{align}
Call the part of the expression in square brackets $Z_{n, x, \epsilon},$ and call the part in parentheses $C_f$. One has
\begin{align}
	\begin{split}
		|\langle \omega_2 | \phi[f^*] \eta^{(2n+1)}\rangle|
		& \leq Z_{n, x, \epsilon} C_f \lVert \eta^{(2n+1)}\rVert.
	\end{split}
\end{align}
From this one may conclude
\begin{equation}
	| \langle \omega_2 | \phi[f^*] \eta \rangle | 
		\leq \sum_{n=0}^{(K-1)/2} |\langle \omega_2 | \phi[f^*] \eta^{(2n+1)} \rangle | 
		\leq C_f \sum_{n=0}^{(K-1)/2} Z_{n, x, \epsilon} \lVert \eta^{(2n+1)} \rVert.
\end{equation}
\begin{equation}
\begin{aligned}
    	| \langle \omega_2 | \phi[f^*] \eta \rangle | 
	&\leq C_f \textstyle\sqrt{\sum_{n=0}^{(K-1)/2} Z_{n, x, \epsilon}^2} \,\textstyle\sqrt{\sum_{n=0}^{(K-1)/2} \lVert \eta^{(2n+1)} \rVert^2}\\
    &\leq C_f \textstyle\sqrt{\sum_{n=0}^{(K-1)/2} Z_{n, x, \epsilon}^2} \,\lVert \eta \rVert\\
    &\leq  C_f \textstyle\sqrt{\sum_{n=0}^{\infty} Z_{n, x, \epsilon}^2} \,\lVert \eta \rVert.
\end{aligned}
\end{equation}
So $|\omega_2\rangle$ is in the domain of $\phi[f]$ provided that the sum on the right-hand side is finite.
Finiteness of this sum is established again using the ratio test; the extra factor of $\sqrt{2n+2}$ does not change anything.

Repeating the above calculation for the overlap $\langle \phi[f_1] \omega_2 | \phi[f_2^*] \eta \rangle$, the only change is an additional polynomial-in-$n$ contribution to $Z_{n,x,\epsilon}$.
This does not change summability, and one may conclude that $\phi[f_1] |\omega_2\rangle$ is in the domain of $\phi[f_2].$
Proceeding by induction, one easily makes sense of all states of the form $\phi[f_1] \dots \phi[f_n] |\omega_2\rangle.$

\section{Some Schatten-class counterexamples}

The general condition claimed in the main text for excitability $\omega_2 \prec \omega_1$ --- see theorem \ref{thm:general-theorem} --- is that one has $\mu_2 \prec \mu_1$, that the operator $\sqrt{X_2} - \sqrt{X_1}$ is Hilbert-Schmidt, and that $Q = X_2 - X_1 + 1$ has trivial kernel.
When $\omega_2$ and $\omega_1$ are both pure, instead of the condition on $\sqrt{X_2} - \sqrt{X_1}$, one may equivalently check that $X_2 - X_1$ is Hilbert-Schmidt --- see proposition \ref{prop:pure-with/out-roots}.
When $\omega_1$ is pure and $\omega_2$ is mixed, one may equivalently check that the ``diagonal blocks'' of $X_2 - X_1$ with respect to the $R_1$ decomposition are trace class --- see theorem \ref{thm:mixed-from-pure}.
In this appendix we show that these statements are tight:
\begin{itemize}
	\item First we show that for any pure state $\omega_1$ there exists a state $\omega_2$ with $X_2 - X_1$ Hilbert-Schmidt but $\sqrt{X_2} - \sqrt{X_1}$ is not Hilbert-Schmidt.
	\item
	Second we show that for any pure state $\omega_1$ there exists a state $\omega_2 \prec \omega_1$ --- so in particular the diagonal blocks of $X_2 - X_1$ are trace-class --- for which the full operator $X_2 - X_1$ is not trace-class.
\end{itemize}

\subsection{A Hilbert-Schmidt counterexample}
\label{app:counterexample-1}

Let $\omega_1$ be a pure state, and let $\K_{\mu_1}$ be the associated phase space, with operator $R_1$ implementing the symplectic form.
As in appendix \ref{app:pure}, the fact that $\omega_1$ is pure means that one has $(R_1)^2 = -1.$
Let $|e_n\rangle_{\mu_1}$ be an orthonormal basis for the $+i$ eigenspace of $R_1.$
With $\Gamma$ denoting complex conjugation, the vectors $\Gamma |e_n\rangle$ form an orthonormal basis for the $-i$ eigenspace.
We define a two-point function $\omega_2$ by the formula
\begin{equation}
	\omega_2(\phi[f]^* \phi[g])
		= \sum_{n} \frac{1}{n} \langle f | \Gamma e_n\rangle_{\mu_1} \langle \Gamma e_n | g \rangle_{\mu_1} + \sum_n \left(2 + \frac{1}{n} \right) \langle f | e_n \rangle_{\mu_1} \langle e_n | g \rangle_{\mu_1}.
\end{equation}
Leaving the basis implicit, we may write this schematically in a block decomposition as
\begin{equation}
	\omega_2(\phi[f]^* \phi[g])
		= \langle f | \begin{pmatrix} 2 + 1/n & 0 \\ 0 & 1/n \end{pmatrix} |g\rangle_{\mu_1},
\end{equation}
where the square matrix represents the operator on $\K_{\mu_1}$ that sends $\ket{e_n}_{\mu_1}$ to $(2+1/n)\ket{e_n}_{\mu_1}$ and $\ket{\Gamma e_n}_{\mu_1}$ to $1/n\ket{\Gamma e_n}_{\mu_1}$.
Clearly $\omega_2$ is positive, linear in $g,$ and antilinear in $f.$
To check that it is a valid two-point function, we need only check the canonical commutation relation.
But we have
\begin{align}
	\begin{split}
		\omega_2(\phi[f]^* \phi[g] - \phi[g] \phi[f]^*)
			& = \langle f | \begin{pmatrix} 2 + 1/n & 0 \\ 0 & 1/n \end{pmatrix} |g\rangle_{\mu_1}
				- \langle \Gamma g| \begin{pmatrix} 2 + 1/n & 0 \\ 0 & 1/n \end{pmatrix} |\Gamma f\rangle_{\mu_1} \\
			& = \langle f | \begin{pmatrix} 2 + 1/n & 0 \\ 0 & 1/n \end{pmatrix} |g\rangle_{\mu_1}
			- \langle f| \Gamma \begin{pmatrix} 2 + 1/n & 0 \\ 0 & 1/n \end{pmatrix} \Gamma |g\rangle_{\mu_1}.
	\end{split}
\end{align}
Conjugation by $\Gamma$ exchanges the eigenspaces of $R_1,$ so this may be written
\begin{align}
	\begin{split}
		\omega_2(\phi[f]^* \phi[g] - \phi[g] \phi[f]^*)
		& = \langle f | \begin{pmatrix} 2 + 1/n & 0 \\ 0 & 1/n \end{pmatrix} |g\rangle_{\mu_1}
		- \langle f| \begin{pmatrix} 1/n & 0 \\ 0 & 2 + 1/n \end{pmatrix} \Gamma |g\rangle _{\mu_1}\\
		& = 2 \langle f |\begin{pmatrix} 1 & 0 \\ 0 & -1 \end{pmatrix}| g \rangle_{\mu_1} \\
		& = - 2 i \langle f | R_1 | g \rangle_{\mu_1} \\
		& = - i \Omega[f^*, g].
	\end{split}
\end{align}

We may now compute the operators $X_2$ and $X_1.$
We have
\begin{equation}
	X_1
		= 1 - i R_1
		= \begin{pmatrix} 2 & 0 \\ 0 & 0 \end{pmatrix}
\end{equation}
and
\begin{equation}
	X_2
		= \begin{pmatrix} 2 + 1/n & 0 \\ 0 & 1/n \end{pmatrix}.
\end{equation}
We may therefore compute
\begin{equation}
	X_2 - X_1
		= \begin{pmatrix} 1/n & 0 \\ 0 & 1/n \end{pmatrix}.
\end{equation}
This is a perfectly good Hilbert-Schmidt operator, since the sum $\sum_n 1/n^2$ is finite.
But one can easily compute
\begin{equation}
	\sqrt{X_2} - \sqrt{X_1}
		= \begin{pmatrix} \textstyle\sqrt{2 + \frac{1}{n}} - \sqrt{2} & 0 \\ 0 & \frac{1}{\sqrt{n}} \end{pmatrix}.
\end{equation}
The first term is Hilbert-Schmidt, but the second is not, because $\sum_n 1/n$ does not converge.

\subsection{A trace-class counterexample}
\label{app:counterexample-2}

In the notation of the previous subsection, given a state $\omega_1,$ we can define a state $\omega_2$ by
\begin{equation}
	\omega_2(\phi[f]^* \phi[g])
		= \langle f | \begin{pmatrix} 2 + 1/n^2 & 1/n \\ 1/n & 1/n^2 \end{pmatrix} |g\rangle_{\mu_1}.
\end{equation}
The notation in the off-diagonal blocks means that this operator sends, e.g., the state $|e_n\rangle$ to $(2 + 1/n^2) |e_n\rangle + 1/n |\Gamma e_n\rangle.$
It is easy to verify that this is positive and that it respects the canonical commutation relation.
The two point operators $X_1$ and $X_2$ are given by
\begin{equation}
	X_1
	= 1 - i R_1
	= \begin{pmatrix} 2 & 0 \\ 0 & 0 \end{pmatrix}
\end{equation}
and
\begin{equation}
	X_2
	= \begin{pmatrix} 2 + 1/n^2 & 1/n \\ 1/n & 1/n^2 \end{pmatrix}.
\end{equation}
The diagonal blocks of $X_2 - X_1$ go like $1/n^2,$ so they are trace-class.
However the full operator $X_2 - X_1$ is not trace-class.
To see this, one can explicitly write
\begin{equation}
	X_2 - X_1
		= \begin{pmatrix} 1/n^2 & 1/n \\ 1/n & 1/n^2 \end{pmatrix}
\end{equation}
and compute
\begin{equation}
	|X_2 - X_1|
		= \sqrt{(X_2 - X_1)^2}
		= \begin{pmatrix} 1/n & 1/n^2 \\ 1/n^2 & 1/n \end{pmatrix},
\end{equation}
which has infinite trace.

\section{Some abelian excitability proofs}
\label{app:abelian-proofs}

\subsection{\texorpdfstring{Proof of equation \eqref{eq:pm_main}}{Proof of equation (6.45)}}
\label{app:pm-proof}

Let $\ket*{(\tau_2^\natural)_{(m)}}$ be the natural cone representative of the restriction of $\tau_2$ to the subalgebra $\A_{\tau_1, (m)}$,  within the subspace $\mathcal{H}_{\tau_1, (m)}$. 
In section \ref{sec:abelian} we noted that
\begin{equation}
\label{eq:tau-nat-app}
	|(\tau_2^{\natural})_{(m)}\rangle
		= p_m(\phi[z_1], \dots, \phi[z_m]) |\tau_1\rangle
\end{equation}
for some positive function $p_m(t_1,\dots, t_m)$ on $\mathbb R^m$, and we claimed that this function is given in terms of $Q$ by equation \eqref{eq:pm_main}.
As indicated in the main text,  $p_m$ is completely fixed by imposing that $\ket*{(\tau_2^\natural)_{(m)}}$ reproduces all correlation functions of $\tau_2$ within $\mathcal{A}_{\tau_2,(m)}$; in particular, we have equation
\eqref{eq:tau-natural-Weyl-expectation}.

Taking derivatives of equation \eqref{eq:tau-natural-Weyl-expectation}, and using that the commutator vanishes, we have
\begin{align}
	\begin{split}
	& \langle (\tau_2^{(\natural)})_{(m)} | \phi[z_1]^{k_1} \dots \phi[z_m]^{k_m} |(\tau_2^{\natural})_{(m)} \rangle \\
		& \qquad = (-i)^{k_1 + \dots + k_m} \left. \frac{ \del^{k_1 + \dots + k_m}}{\del t_1^{k_1} \dots \del t_m^{k_m}}
			\langle (\tau_2^{\natural})_{(m)} | e^{i \phi[t_1 z_1 + \dots + t_m z_m]} |(\tau_2^{\natural})_{(m)}\rangle \right|_{t_1 = \dots = t_m = 0} \\
		& \qquad = (-i)^{k_1 + \dots + k_m} \left. \frac{ \del^{k_1 + \dots + k_m}}{\del t_1^{k_1} \dots \del t_m^{k_m}}
		e^{ - \frac{1}{2} \langle t_1 z_1 + \dots + t_m z_m | Q | t_1 z_1 + \dots + t_m z_m \rangle_{\mu_1} } \right|_{t_1 = \dots = t_m = 0}\\
        &\qquad =(-i)^{k_1+\dots k_m} \left. \frac{ \partial^{k_1+\dots k_m}}{\partial t_1^{k_1} \dots \partial t_m^{k_m}}
		\exp\left[ - \frac{1}{2} \sum_{r,s} t_r Q_{r s} t_s  \right] \right|_{t_1 = \dots = t_m = 0},
	\end{split}
    \label{eq:tau-natural-correlator}
\end{align}
where we have defined $Q_{rs} = \langle z_r | Q | z_s\rangle_{\mu_1}$.
This is easily seen to describe a general polynomial expectation value with respect to an $m$-dimensional Gaussian integral with covariance matrix given by the inverse of the matrix obtained by restricting $Q$ to the span of $|z_1\rangle_{\mu_1}$ through $|z_m\rangle_{\mu_1}.$
To see this, denote by $Q_m$ the operator obtained by restricting $Q$ to this subspace.
Since $Q$ has no kernel by the results of section \ref{sec:general-necessary}, and since the subspace is finite-dimensional, the operator $Q_m$ is always invertible.
Consequently, we have, working in matrix notation,
\begin{align}
	\begin{split}
		& \langle (\tau_2^{(\natural)})_{(m)} | \phi[z_1]^{k_1} \dots \phi[z_m]^{k_m} |(\tau_2^{\natural})_{(m)} \rangle\\
		& \qquad = (-i)^{k_1 + \dots + k_m} \left. \frac{ \partial^{k_1 + \dots + k_m}}{\partial t_1^{k_1} \dots \partial t_m^{k_m}}
		\int \dd^m x\; \frac{1}{\sqrt{\det 2 \pi Q_m}} \exp\left[ - \frac{1}{2} \vec{x} \cdot Q_m^{-1} \cdot \vec{x} + i \vec{x} \cdot \vec{t} \right] \right|_{t_1 = \dots = t_m = 0} \\
		& \qquad = 	\frac{1}{\sqrt{\det 2 \pi Q_m}} \int \dd^m x\; x_1^{k_1} \dots x_m^{k_m} \exp\left[ - \frac{1}{2} \vec{x} \cdot Q_m^{-1} \cdot \vec{x} \right].
	\end{split}
    \label{eq:tau-natural-correlator2}
\end{align}

To see how this determines the function $p_m$, we rewrite the left-hand side of the above equation using equation \eqref{eq:tau-nat-app}:
\begin{align}
\label{eq:pm_expVal-app}
	\begin{split}
		& \bra*{(\tau_2^{\natural})_{(m)}}\phi[z_1]^{k_1} \dots \phi[z_m]^{k_m} \ket*{(\tau_2^{\natural})_{(m)}}\\
        & \qquad = \langle  \tau_1 | p_m(\phi[z_1], \dots, \phi[z_m])^2 \phi[z_1]^{k_1} \dots \phi[z_m]^{k_m} |\tau_1\rangle\\
		 &
         \qquad =\frac{1}{\sqrt{(2 \pi)^m}} \int \dd^m x\; p_m(x_1, \dots, x_m)^2 x_1^{k_1} \dots x_m^{k_m} \exp\left[ - \frac{1}{2} \vec{x} \cdot \vec{x} \right] .
	\end{split}
\end{align}
where to evaluate the $\tau_1$ correlator in the last line, we used similar methods as for the $\tau_2^{\natural}$ correlator; in particular, one can just replace $Q$ by the identity in equations \eqref{eq:tau-natural-correlator}-\eqref{eq:tau-natural-correlator2}.
Equality of the preceding two equations tells us
\begin{align}
	\begin{split}
	& \frac{1}{\sqrt{(2 \pi)^m}} \int \dd^m x\; p_m(x_1, \dots, x_m)^2 x_1^{k_1} \dots x_m^{k_m} \exp\left[ - \frac{1}{2} \vec{x} \cdot \vec{x} \right] \\
		& \qquad = \frac{1}{\sqrt{\det 2 \pi Q_m}} \int \dd^m x\; x_1^{k_1} \dots x_m^{k_m} \exp\left[ - \frac{1}{2} \vec{x} \cdot Q_m^{-1} \cdot \vec{x} \right],
	\end{split}
\end{align}
Since $p_m$ is positive, this uniquely determines the function $p_m$ to be given by the formula
\begin{equation} \label{eq:pm-app}
	p_m(x_1, \dots, x_m) = \frac{1}{(\det Q_m)^{1/4}} \exp\left[ \frac{1}{4} \vec{x}\cdot (1 - Q_m^{-1} ) \cdot \vec{x} \right].
\end{equation}
This completes the proof of equation \eqref{eq:pm_main}.

\subsection{\texorpdfstring{Proof of equation \eqref{eq:lower-bounding-sequence}}{Proof of equation (6.51)}}
\label{app:lower-bounding-sequence}

Here we demonstrate that
\begin{equation}
    \braket{\Pi_{(m)}P_{\tau_2}\Pi_{(m)}\tau_1}{(P_{\tau_2,(m)}-\Pi_{(m)}P_{\tau_2}\Pi_{(m)})\tau_1}
\end{equation}
is a natural cone overlap.
Since $\Pi_{(m)}P_{\tau_2}\Pi_{(m)}$ is positive and affiliated to $\mathcal{A}_{\tau_1,(m)}$, it remains to show that the same is true of $P_{\tau_2,(m)}-\Pi_{(m)}P_{\tau_2}\Pi_{(m)}.$

Let us begin with affiliation.
Both operators in the expression 
\begin{equation}\label{eq:pos-aff}
    P_{\tau_2,(m)}-\Pi_{(m)}P_{\tau_2}\Pi_{(m)}
\end{equation}
are affiliated to $\mathcal{A}_{\tau_1,(m)}$; hence, to check that the full operator is affiliated, we need only check that it is densely defined and closed.
Towards denseness, we claim that one has
\begin{equation}
\label{eq:p-2-m}
    P_{\tau_2, (m)} = |P_{\tau_2} \Pi_{(m)}|\,,
\end{equation}
and hence the domain of $P_{\tau_2,(m)}$ coincides with that of $P_{\tau_2}\Pi_{(m)}$, which in turn coincides with that of $\Pi_{(m)}P_{\tau_2}\Pi_{(m)}$.
This common domain is dense, so $P_{\tau_2,(m)}-\Pi_{(m)}P_{\tau_2}\Pi_{(m)}$ is densely defined.
We now see that $P_{\tau_2,(m)}-\Pi_{(m)}P_{\tau_2}\Pi_{(m)}$ is the difference of two closed operators with the same domain, and hence it is itself closed on this domain.
We conclude $P_{\tau_2,(m)}-\Pi_{(m)}P_{\tau_2}\Pi_{(m)}$ is affiliated to $\mathcal{A}_{\tau_1,(m)}$, as desired.

Before showing positivity of this operator, let us justify equation \eqref{eq:p-2-m}.
This is the statement that acting on $|\tau_1\rangle$ with $|P_{\tau_2} \Pi_{(m)}|$ produces the natural cone representative $|(\tau_2^{\natural})_{(m)}\rangle$.
By similar reasoning to that used below equation \eqref{eq:first-natural-cone-overlap} in the main text, we have that $|P_{\tau_2} \Pi_{(m)}|$ is affiliated to $\A_{\tau_1, (m)}$. 
Since this operator is manifestly positive, it only remains to show that $|P_{\tau_2} \Pi_{(m)}| |\tau_1\rangle$ reproduces all the correlation functions of $|\tau_2\rangle$ on this algebra.
Writing the polar decomposition 
\begin{equation}
	P_{\tau_2} \Pi_{(m)}
		= V_m |P_{\tau_2} \Pi_{(m)}|,
\end{equation}
and recalling that $P_{\tau_2} \Pi_{(m)}$ commutes with $\A_{\tau_1, (m)},$ one learns that the isometry $V_m$ commutes with $\A_{\tau_1, (m)}$ as well.
So for any $a$ in this algebra, we have
\begin{align}
	\begin{split}
	\langle (|P_{\tau_2} \Pi_{(m)}|) \tau_1 | a | (|P_{\tau_2} \Pi_{(m)}|) \tau_1 \rangle
		& = \langle V_m^{\dagger} P_{\tau_2} \Pi_{(m)} \tau_1 | a | V_m^{\dagger} P_{\tau_2} \Pi_{(m)} \tau_1 \rangle \\
		& = \langle V_m^{\dagger} P_{\tau_2} \tau_1 | a | V_m^{\dagger} P_{\tau_2} \tau_1 \rangle \\
		& = \langle \tau_2^{\natural} | V_m a V_m^{\dagger} | \tau_2^{\natural} \rangle \\
		& = \langle \tau_2^{\natural} | a V_m V_m^{\dagger} | \tau_2^{\natural} \rangle.
	\end{split}
\end{align}
The operator $V_m V_m^{\dagger}$ is the projector onto the image of $P_{\tau_2} \Pi_{(m)}.$
Obviously $|\tau_2^{\natural} \rangle$ is in this image --- since it is just produced by acting with this operator on the state $|\tau_1\rangle$ --- so one finds
\begin{align}
	\begin{split}
		\langle (|P_{\tau_2} \Pi_{(m)}|) \tau_1 | a | (|P_{\tau_2} \Pi_{(m)}|) \tau_1 \rangle
		& = \langle \tau_2^{\natural} | a | \tau_2^{\natural} \rangle.
	\end{split}
\end{align}
This lets us conclude $P_{\tau_2, (m)}
		= |P_{\tau_2} \Pi_{(m)}|$ as desired.

Finally, we check positivity of $P_{\tau_2,(m)}-\Pi_{(m)}P_{\tau_2}\Pi_{(m)}.$
We compute, for $\ket{\psi}$ in the shared domain,
\begin{equation}
\begin{aligned}
        \bra{\psi}\Pi_{(m)}P_{\tau_2}\Pi_{(m)}\ket{\psi} &= \bra{\psi}P_{\tau_2}\Pi_{(m)}\ket{\psi}\\
        &= \bra{\psi}V_m P_{\tau_2, (m)}\ket{\psi}\\
        &= \bra*{P_{\tau_2, (m)} ^{1/2}\psi}V_m \ket*{P_{\tau_2, (m)} ^{1/2}\psi},
\end{aligned}
\end{equation}
where in the first equality we used that $\Pi_{(m)}$ formally commutes with $P_{\tau_2}$, in the second we used equation \eqref{eq:p-2-m}, and in the third we used that $V_m$ formally commutes with $P_{\tau_2,(m)}$.
Since $V_m$ is a partial isometry, we have 
\begin{equation}
    \bra*{P_{\tau_2, (m)} ^{1/2}\psi}V_m \ket*{P_{\tau_2, (m)} ^{1/2}\psi} \leq \bra*{P_{\tau_2, (m)} ^{1/2}\psi}\ket*{P_{\tau_2, (m)} ^{1/2}\psi} = \bra*{\psi}P_{\tau_2, (m)} \ket*{\psi},
\end{equation}
which concludes the proof.

\subsection{\texorpdfstring{Proof of equation \eqref{eq:tauOverlaps}}{Proof of equation (6.56)}}
\label{app:overlap-taus}

Here we compute the inner product $\braket*{(\tau_2^\natural)_{(m)}}{(\tau_2^\natural)_{(n)}}$ with
\begin{equation}
	|(\tau_2^{\natural})_{(m)}\rangle
		= p_m(\phi[z_1], \dots, \phi[z_m]) |\tau_1\rangle,
\end{equation}
and $p_m$ given in equation \eqref{eq:pm-app}.
As in the main text, we take $m\geq n$. 
We have
\begin{equation}
\begin{aligned}
        \braket*{(\tau_2^\natural)_{(m)}}{(\tau_2^\natural)_{(n)}} &=\bra{\tau_1}p_m(\phi[z_1], \dots, \phi[z_m])p_n(\phi[z_1], \dots, \phi[z_n]) \ket{\tau_1}\,.
\end{aligned}
\end{equation}
Using the explicit expression for $\tau_1$  correlators discussed below equation \eqref{eq:pm_expVal-app}, we have
\begin{equation}
    \braket*{(\tau_2^\natural)_{(m)}}{(\tau_2^\natural)_{(n)}} = \frac{1}{\sqrt{(2\pi)^m}}\int\dd^m x \,e^{-\frac{1}{2}\vec{x}\cdot \vec{x}}p_m(x_1,\dots, x_m)p_n(x_1,\dots, x_n),
\end{equation}
which, upon substituting equation \eqref{eq:pm-app} for $p_m$, becomes
\begin{equation}
    \braket*{(\tau_2^\natural)_{(m)}}{(\tau_2^\natural)_{(n)}} = 
    \frac{1}{\sqrt{(2 \pi)^m}(\det Q_m)^{1/4}(\det Q_n)^{1/4}} \int \dd^m x\; e^{- \frac{1}{2} \vec{x} \cdot \vec{x}}e^{\frac{1}{4} \vec{x} \cdot (1 - Q_m^{-1}) \cdot \vec{x}} e^{\frac{1}{4} \vec{x} \cdot (1_n - Q_n^{-1}) \cdot \vec{x}},
\end{equation}
where $1_n$ denotes the projector onto the first $n$ entries $x_1, \dots , x_n$.
Combining the terms in the exponent we get 
\begin{equation}
    \braket*{(\tau_2^\natural)_{(m)}}{(\tau_2^\natural)_{(n)}} =  \frac{1}{\sqrt{(2 \pi)^m}(\det Q_m)^{1/4}(\det Q_n)^{1/4}} \int \dd^m x\; e^{- \frac{1}{4} \vec{x} \cdot (Q_m^{-1} +Q_n^{-1} + 1_{m-n}) \cdot \vec{x}}
\end{equation}
with $1_{m-n}$ the projector onto $x_{n+1},\dots , x_{m}$.
Finally, evaluating the last Gaussian integral, and substituting $\det Q_n = \det(Q_n + 1_{m-n}),$ this is
\begin{align}
	\begin{split} 
		\braket*{(\tau_2^\natural)_{(m)}}{(\tau_2^\natural)_{(n)}} &=   (\det Q_m)^{-1/4} (\det (Q_n + 1_{m-n}))^{-1/4}\left({\det(\tfrac{Q_m^{-1} + (Q_n + 1_{m-n})^{-1}}{2})}\right)^{-1/2}\\
        &= \det(\frac{Q_{n, \perp}^{1/4} Q_m^{-1/2} Q_{n, \perp}^{1/4} + Q_{n, \perp}^{-1/4} Q_m^{1/2} Q_{n, \perp}^{-1/4}}{2})^{-1/2},
	\end{split} 
\end{align}
where in the second line we used the multiplicative property of determinants and we have introduced the notation $Q_{n, \perp}$ for $Q_n + 1_{m-n}.$

Writing $\lambda^{(m, n)}_{j}$ for the eigenvalues of $Q_{n, \perp}^{-1/4} Q_{m}^{1/2} Q_{n, \perp}^{-1/4}$, we finally get
\begin{align} \label{eq:tau-diff-det-app}
	\begin{split} 
		\braket*{(\tau_2^\natural)_{(m)}}{(\tau_2^\natural)_{(n)}}
			= \left(\prod_{j=1}^{m} \frac{\lambda^{(m,n)}_j + 1/\lambda^{(m,n)}_j}{2}\right)^{-1/2},
	\end{split} 
\end{align}
which reproduces equation \eqref{eq:tauOverlaps}, as desired.

\subsection{\texorpdfstring{Invertibility of $Q$}{Invertibility of Q}}
\label{app:why-q-is-invertible}

Above equation \eqref{eq:HS-mu} in the main text, we claimed that $Q$ is invertible whenever we have the Gaussian-state excitability relation $\tau_2 \prec \tau_1$ on an abelian scalar field algebra.
Notation used in this appendix is the same as the notation used in section \ref{sec:abelian}.

Invertibility of $Q$ is the statement that $\mu_1$ and $\mu_2$ induce equivalent inner products on the space of smearing functions.
From $\tau_2\prec\tau_1$, we already have $\mu_2\prec \mu_1$.
Hence, it remains to show that $\mu_1\prec \mu_2$.
To do so, we will show that any sequence of test functions $f_n$ that converges to zero with respect to $\mu_2$ must  also converge to zero with respect to $\mu_1$.
As in section \ref{sec:mu-boundedness}, it suffices to consider real sequences.

We begin by evaluating 
\begin{equation}
\begin{aligned}
        \bra*{\tau_2^\natural}e^{i\phi[g]}e^{i\phi[f_n]}e^{i\phi[h]}\ket*{\tau_2^\natural}  &=  \bra*{\tau_2^\natural}e^{i\phi[g+f_n +h]}\ket*{\tau_2^\natural}\\
        &= \tau_2(e^{i\phi[g+f_n +h]})\\
        &= e^{-\norm{g+f_n+h}^2_{\mu_2}/2}\\
        &\to  e^{-\norm{g+h}^2_{\mu_2}/2}\\
        &= \bra*{\tau_2^\natural}e^{i\phi[g]}e^{i\phi[h]}\ket*{\tau_2^\natural}.
\end{aligned}
\end{equation}
This shows that if $f_n \to 0$ in $\K_{\mu_2}$, then  $e^{i\phi[f_n]} \to 1$ with respect to matrix elements for a dense set of states within the subspace of $\H_{\tau_1}$ generated by the action of Weyl operators on $\ket*{\tau_2^\natural}$.
Since the sequence $e^{i \phi[f_n]}$ is uniformly bounded, that means that the sequence actually converges with respect to the weak topology for this subspace.

We now claim that this subspace coincides with the entire GNS Hilbert space $\H_{\tau_1}$; that is, we claim $\ket*{\tau_2^\natural}$ is cyclic for $\mathcal{A}_{\tau_1}$.
This follows from the fact that $\ket*{\tau_2^\natural}$ lies in the natural cone for  $\mathcal{A}_{\tau_1}$ (by construction) and is separating for   $\mathcal{A}_{\tau_1}$ (since it is separating on each subalgebra $\mathcal{A}_{\tau_1,(m)}$),\footnote{Proof: 
We first argue that $\ket*{\tau_2^\natural}$ is separating for each subalgebra $\mathcal{A}_{\tau_1,(m)}$.
Let $a\in \mathcal{A}_{\tau_1,(m)}$ be such that $\expval*{a^*a}{\tau_2^\natural}=0$. 
Since $\mathcal{A}_{\tau_1,(m)}$ is generated by a finite set of commuting, independent Hermitian operators, $a$ can be written in terms of some function $a(x_1,\dots, x_m)$ on $\mathbb R^m$. Similarly to equation \eqref{eq:pm_expVal-app}, the corresponding expectation value is given by integrating the non-negative function $|a(\vec x)|^2$ against the positive measure $p_m(\vec x)^2\exp(-\vec x \cdot \vec x /2)/\sqrt{(2\pi)^m}$, with $p_m(\vec x)$ given by equation \eqref{eq:pm_main}. This integral vanishes only if $|a(\vec x)|^2=0$, which implies that $a=0$. So $\ket*{\tau_2^\natural}$ is separating for each subalgebra.
This means $\ket*{\tau_2^\natural}$ is also separating for $\bigcup_m \mathcal{A}_{\tau_1,(m)}$.
In turn, this means $\ket*{\tau_2^\natural}$ is cyclic for $(\bigcup_m \mathcal{A}_{\tau_1,(m)})'$ which finally means $\ket*{\tau_2^\natural}$ is separating for $(\bigcup_m \mathcal{A}_{\tau_1,(m)})''=\A_{\tau_1}$ (this last equality is proved in section \ref{sec:suff-abelian}.).
} and any separating state in the natural cone is automatically cyclic \cite[theorem 5.4]{Araki:natural-1}.

So, $e^{i\phi[f_n]}$ converges weakly to 1 for the full space ${\cal H}_{\tau_1}$.
In particular, evaluating on the GNS vector $\ket{\tau_1}$ gives
\begin{equation}
    \bra*{\tau_1}e^{i\phi[f_n]}\ket*{\tau_1} = e^{-\norm{f_n}^2_{\mu_1}/2}\to 1,
\end{equation}
which implies that $f_n$ must converge to zero also with respect to $\mu_1$, which concludes the proof.

\subsection{\texorpdfstring{Proof of equation \eqref{eq:HS-mu}}{Proof of equation (6.58)}}
\label{app:bound-lambdas}

Letting $\lambda^{(m, n)}_{j}$ be the eigenvalues of $Q_{n, \perp}^{-1/4} Q_{m}^{1/2} Q_{n, \perp}^{-1/4}$, we claim that there exists a positive constant $C$, independent of $m$ and $n$, such that
\begin{equation} \label{eq:HS-mu-app}
	\prod_{j=1}^{m} \frac{\lambda^{(m,n)}_j + 1/\lambda^{(m,n)}_j}{2}
		\geq 1 + C \sum_{j=1}^{m} (\lambda^{(m,n)}_j - 1)^2.
\end{equation}
This is equation \eqref{eq:HS-mu} in the main text.
To obtain this, we note that, in light of appendix \ref{app:why-q-is-invertible}, $Q$ is bounded and invertible, hence we have constants $0 < u < v$ with $u \leq Q \leq v.$
This same property is inherited by $Q_{m}$, i.e., we have $u \leq Q_{m} \leq v,$ and as for $Q_{n, \perp},$ we have the relation
\begin{equation}
	\text{min}(1, u) \leq Q_{n, \perp} \leq \text{max}(1, v),
\end{equation}
which is made obvious by writing $Q_{n,\perp}$ in blocks corresponding to the decomposition $E_n {\cal K}_{\mu_1} \oplus (E_m-E_n) {\cal K}_{\mu_1}$.
We might as well take $u < 1$ and $v > 1$ to avoid having to use the minima and maxima in the above equation.
Then we obtain the inequality 
\begin{equation}
	\left( \frac{u}{v} \right)^{1/2} \leq Q_{n,\perp}^{-1/4} Q_m^{1/2} Q_{n, \perp}^{-1/4} \leq \left(\frac{v}{u}\right)^{1/2},
\end{equation}
hence
\begin{equation}
	\left( \frac{u}{v} \right)^{1/2} \leq \lambda^{(m,n)}_j \leq \left(\frac{v}{u}\right)^{1/2}.
\end{equation}
The function
\begin{equation}
	\frac{e^{\sum_j \log((x_j + x_j^{-1})/2)} - 1}{\sum_j (x_j-1)^2}
\end{equation}
is continuous and positive on any subset of $\mathbb R^m_{\geq 0}$ that is bounded away from the axes, so in particular, when all $x_j$ are bounded between $(u/v)^{1/2}$ and $(v/u)^{1/2},$ this function has a minimum; call it $C$.
Plugging in $x_j = \lambda^{(m,n)}_j$ gives the desired inequality \eqref{eq:HS-mu-app}.

\bibliographystyle{JHEP}
\bibliography{bibliography}

\providecommand{\href}[2]{#2}\begingroup\raggedright\begin{thebibliography}{10}

\bibitem{Powers:quasi}
R.~T. Powers and E.~St{\o}rmer, {\it {Free states of the canonical
  anticommutation relations}},  {\em Commun. Math. Phys.} {\bf 16} (1970)
  1--33.

\bibitem{Araki:quasi}
H.~Araki, {\it {On quasifree states of the canonical commutation relations:
  II.}}, .

\bibitem{VanDaele:quasi}
A.~Van~Daele, {\it {Quasi-equivalence of quasi-free states on the weyl
  algebra}},  {\em Commun. Math. Phys.} {\bf 21} (1971) 171--191.

\bibitem{Araki-Yamagami}
H.~Araki and S.~Yamagami, {\it On quasi-equivalence of quasifree states of the
  canonical commutation relations},  {\em {Publications of the Research
  Institute for Mathematical Sciences}} {\bf 18} (1982), no.~2 703--758.

\bibitem{Verch:hadamard}
R.~Verch, {\it {Local definiteness, primarity and quasiequivalence of quasifree
  Hadamard quantum states in curved space-time}},  {\em Commun. Math. Phys.}
  {\bf 160} (1994) 507--536.

\bibitem{Jensen:JSS}
K.~Jensen, J.~Sorce, and A.~J. Speranza, {\it {Generalized entropy for general
  subregions in quantum gravity}},  {\em JHEP} {\bf 12} (2023) 020,
  [\href{http://arxiv.org/abs/2306.01837}{{\tt arXiv:2306.01837}}].

\bibitem{Sorce:analyticity}
J.~Sorce, {\it {Analyticity and the Unruh effect: a study of local modular
  flow}},  {\em JHEP} {\bf 24} (2024) 040,
  [\href{http://arxiv.org/abs/2403.18937}{{\tt arXiv:2403.18937}}].

\bibitem{Caminiti:2025hjq}
J.~Caminiti, F.~Capeccia, L.~Ciambelli, and R.~C. Myers, {\it {Geometric
  modular flows in 2d CFT and beyond}},  {\em JHEP} {\bf 08} (2025) 166,
  [\href{http://arxiv.org/abs/2502.02633}{{\tt arXiv:2502.02633}}].

\bibitem{Leutheusser:HSMT}
S.~A.~W. Leutheusser and H.~Liu, {\it {Emergent Times in Holographic Duality}},
   {\em Phys. Rev. D} {\bf 108} (2023), no.~8 086020,
  [\href{http://arxiv.org/abs/2112.12156}{{\tt arXiv:2112.12156}}].

\bibitem{Witten:largeN}
E.~Witten, {\it Why does quantum field theory in curved spacetime make sense?
  and what happens to the algebra of observables in the thermodynamic limit?},
  \href{http://arxiv.org/abs/2112.11614}{{\tt arXiv:2112.11614}}.

\bibitem{Leutheusser:subalgebra}
S.~Leutheusser and H.~Liu, {\it {Subregion-subalgebra duality: Emergence of
  space and time in holography}},  {\em Phys. Rev. D} {\bf 111} (2025), no.~6
  066021, [\href{http://arxiv.org/abs/2212.13266}{{\tt arXiv:2212.13266}}].

\bibitem{Shale:unitary}
D.~Shale, {\it Linear symmetries of free boson fields},  {\em {Transactions of
  the American Mathematical Society}} {\bf 103} (1962), no.~1 149--167.

\bibitem{Wald:particle-creation}
R.~M. Wald, {\it {On Particle Creation by Black Holes}},  {\em Commun. Math.
  Phys.} {\bf 45} (1975) 9--34.

\bibitem{Wald:s-matrix}
R.~M. Wald, {\it {Existence of the S matrix in quantum field theory in curved
  space-time}},  {\em Annals Phys.} {\bf 118} (1979) 490--510.

\bibitem{Woronowicz:purification}
S.~L. Woronowicz, {\it {On the purification of factor states}},  {\em Commun.
  Math. Phys.} {\bf 28} (1972) 221--235.

\bibitem{Longo:quasi}
R.~Longo, {\it {Modular Structure of the Weyl Algebra}},  {\em Commun. Math.
  Phys.} {\bf 392} (2022), no.~1 145--183,
  [\href{http://arxiv.org/abs/2111.11266}{{\tt arXiv:2111.11266}}].

\bibitem{Conti:quasi}
R.~Conti and G.~Morsella, {\it {Quasi-free Isomorphisms of Second Quantization
  Algebras and Modular Theory}},  {\em Math. Phys. Anal. Geom.} {\bf 27}
  (2024), no.~2 8, [\href{http://arxiv.org/abs/2305.07606}{{\tt
  arXiv:2305.07606}}].

\bibitem{Verch:algebra}
R.~Verch, {\it {Continuity of symplectically adjoint maps and the algebraic
  structure of Hadamard vacuum representations for quantum fields on curved
  space-time}},  {\em Rev. Math. Phys.} {\bf 9} (1997) 635--674,
  [\href{http://arxiv.org/abs/funct-an/9609004}{{\tt funct-an/9609004}}].

\bibitem{Segal:distributions}
I.~E. Segal, {\it {Distributions in Hilbert space and canonical systems of
  operators}},  {\em {Transactions of the American Mathematical Society}} {\bf
  88} (1958), no.~1 12--41.

\bibitem{Araki-Shiraishi}
H.~Araki and M.~Shiraishi, {\it {On quasifree states of the canonical
  commutation relations (I)}},  {\em {Publications of the Research Institute
  for Mathematical Sciences}} {\bf 7} (1971), no.~1 105--120.

\bibitem{Sorce:paper1}
J.~Sorce, {\it {Continuum canonical purifications}},
  \href{http://arxiv.org/abs/2512.17014}{{\tt arXiv:2512.17014}}.

\bibitem{Haag-Kastler}
R.~Haag and D.~Kastler, {\it {An Algebraic approach to quantum field theory}},
  {\em J. Math. Phys.} {\bf 5} (1964) 848--861.

\bibitem{Sorce:types}
J.~Sorce, {\it {Notes on the type classification of von Neumann algebras}},
  {\em Rev. Math. Phys.} {\bf 36} (2024), no.~02 2430002,
  [\href{http://arxiv.org/abs/2302.01958}{{\tt arXiv:2302.01958}}].

\bibitem{Stratila:book}
{Str\v{a}til\v{a}, Serban and Zsid\'{o}, L\'{a}szl\'{o}}, {\em {Lectures on von
  Neumann algebras}}.
\newblock {Cambridge University Press}, 2019.

\bibitem{Sorce:modular}
J.~Sorce, {\it {An intuitive construction of modular flow}},  {\em JHEP} {\bf
  12} (2023) 079, [\href{http://arxiv.org/abs/2309.16766}{{\tt
  arXiv:2309.16766}}].

\bibitem{Conway:book}
J.~B. Conway, {\em A course in operator theory}, vol.~21.
\newblock {American Mathematical Society}, 2025.

\bibitem{Hollands:axioms}
S.~Hollands and R.~M. Wald, {\it {Axiomatic quantum field theory in curved
  spacetime}},  {\em Commun. Math. Phys.} {\bf 293} (2010) 85--125,
  [\href{http://arxiv.org/abs/0803.2003}{{\tt arXiv:0803.2003}}].

\bibitem{Hollands:review}
S.~Hollands and R.~M. Wald, {\it {Quantum fields in curved spacetime}},  {\em
  Phys. Rept.} {\bf 574} (2015) 1--35,
  [\href{http://arxiv.org/abs/1401.2026}{{\tt arXiv:1401.2026}}].

\bibitem{Slawny:1972iq}
J.~Slawny, {\it {On factor representations and the c*-algebra of canonical
  commutation relations}},  {\em Commun. Math. Phys.} {\bf 24} (1972) 151--170.

\bibitem{Bratteli:1996xq}
O.~Bratteli and D.~W. Robinson, {\em {Operator algebras and quantum statistical
  mechanics. Vol. 2: Equilibrium states. Models in quantum statistical
  mechanics}}.
\newblock N.d., 1996.

\bibitem{Araki:natural-1}
H.~Araki, {\it {Some properties of modular conjugation operator of von Neumann
  algebras and a non-commutative Radon-Nikodym theorem with a chain rule}},
  {\em {Pacific Journal of Mathematics}} {\bf 50} (1974), no.~2 309--354.

\bibitem{Skripka:book}
A.~Skripka and A.~Tomskova, {\em Multilinear operator integrals}.
\newblock Springer, 2019.

\bibitem{Bach:revisited}
V.~Bach, A.~F.~M. ter Elst, and J.~Rehberg, {\it {The Birman--Solomyak theorem
  revisited: a novel elementary proof, generalisation, and applications}},
  \href{http://arxiv.org/abs/2511.11058}{{\tt arXiv:2511.11058}}.

\bibitem{sorce-blog-Pick}
J.~Sorce, {\it Pick functions and operator monotones},  Sep, 2024.
\newblock
  \url{https://sorcenotes.blogspot.com/2024/09/pick-functions-and-operator-monotones.html}.

\bibitem{Takeaski:I}
M.~Takesaki, {\em {Theory of operator algebras I}}.
\newblock Springer, 1979.

\bibitem{Satishchandran:uniqueness}
G.~Satishchandran and J.~Sorce, {\it {Uniqueness of null-local modular flow}},
  {\em to appear}.

\bibitem{Witten:background}
E.~Witten, {\it {A background-independent algebra in quantum gravity}},  {\em
  JHEP} {\bf 03} (2024) 077, [\href{http://arxiv.org/abs/2308.03663}{{\tt
  arXiv:2308.03663}}].

\bibitem{Gaiotto:higher-form}
D.~Gaiotto, A.~Kapustin, N.~Seiberg, and B.~Willett, {\it {Generalized Global
  Symmetries}},  {\em JHEP} {\bf 02} (2015) 172,
  [\href{http://arxiv.org/abs/1412.5148}{{\tt arXiv:1412.5148}}].

\bibitem{Wald:QFTCS}
R.~M. Wald, {\em {Quantum Field Theory in Curved Space-Time and Black Hole
  Thermodynamics}}.
\newblock Chicago Lectures in Physics. University of Chicago Press, Chicago,
  IL, 1995.

\bibitem{Friedlander:book}
F.~G. Friedlander, {\em The wave equation on a curved space-time}, vol.~2.
\newblock Cambridge university press, 1975.

\bibitem{Driessler:unbounded-to-vN}
W.~Driessler, S.~J. Summers, and E.~H. Wichmann, {\it {On the Connection
  Between Quantum Fields and Von Neumann Algebras of Local Operators}},  {\em
  Commun. Math. Phys.} {\bf 105} (1986) 49--84.

\bibitem{Buchholz:unbounded-to-vN}
D.~Buchholz, {\it {On quantum fields which generate local algebras}},  {\em J.
  Math. Phys.} {\bf 31} (1990) 1839--1846.

\bibitem{Araki:lattice}
H.~Araki, {\it {A lattice of von Neumann algebras associated with the quantum
  theory of a free Bose field}},  {\em {Journal of Mathematical Physics}} {\bf
  4} (1963), no.~11 1343--1362.

\bibitem{Longo:lectures}
R.~Longo, {\it {Lectures on conformal nets}}, .

\bibitem{Figliolini:tomita}
F.~Figliolini and D.~Guido, {\it {THE TOMITA OPERATOR FOR THE FREE SCALAR
  FIELD}},  {\em Ann. Inst. H. Poincare Phys. Theor.} {\bf 51} (1989) 419--435.

\bibitem{Witten:notes}
E.~Witten, {\it {APS Medal for Exceptional Achievement in Research: Invited
  article on entanglement properties of quantum field theory}},  {\em Rev. Mod.
  Phys.} {\bf 90} (2018), no.~4 045003,
  [\href{http://arxiv.org/abs/1803.04993}{{\tt arXiv:1803.04993}}].

\bibitem{Ceyhan:QNEC}
F.~Ceyhan and T.~Faulkner, {\it {Recovering the QNEC from the ANEC}},  {\em
  Commun. Math. Phys.} {\bf 377} (2020), no.~2 999--1045,
  [\href{http://arxiv.org/abs/1812.04683}{{\tt arXiv:1812.04683}}].

\bibitem{Aigner:enumeration}
M.~Aigner, {\em A course in enumeration}.
\newblock Springer, 2007.

\end{thebibliography}\endgroup

\end{document}